\documentclass[11pt,a4paper]{article}
\usepackage[left=2.5cm,right=2.5cm,bottom=3cm,top=3cm,bindingoffset=0.5cm]{geometry}

\usepackage[normalem]{ulem}
\usepackage[utf8]{inputenc}
\usepackage[english]{babel}
\usepackage{amsmath,mathtools}

\usepackage{amssymb}
\usepackage{amsthm}
\usepackage{fancyhdr}
\usepackage{graphicx}
\usepackage{subfigure}
\usepackage{boxedminipage}
\usepackage{cite}
\usepackage{lastpage}
\usepackage{textcase}
\usepackage{xcolor}
\usepackage{verbatim}
\usepackage{bbm}
\usepackage[final,pdfusetitle]{hyperref}
\hypersetup{%
	pdfauthor={S{\o}ren Fournais and Jan Philip Solovej}
}
\usepackage{authblk}

\theoremstyle{definition} \newtheorem{definition}{Definition}[section]
\theoremstyle{plain} \newtheorem{theorem}[definition]{Theorem}
\theoremstyle{plain} \newtheorem{assumption}[definition]{Assumption}
\theoremstyle{plain} \newtheorem{proposition}[definition]{Proposition}
\theoremstyle{plain} \newtheorem{lemma}[definition]{Lemma}
\theoremstyle{plain} \newtheorem{corollary}[definition]{Corollary}
\theoremstyle{plain} \newtheorem{remark}[definition]{Remark}
\theoremstyle{definition} 
\theoremstyle{plain} 
\numberwithin{equation}{section}
\fancypagestyle{plain}{%
  \fancyhf{}%
}

\newcommand{\N}{\mathbb{N}}

\newcommand{\R}{\mathbb{R}}
\newcommand{\C}{\mathbb{C}}

\newcommand{\one}{\mathbbm{1}}

\newcommand{\err}{\mathcal{R}_1}

\newcommand{\M}{\mathcal{M}}

\newcommand{\cA}{\mathcal{A}}
\newcommand{\cB}{\mathcal{B}}

\newcommand{\rmu}{\rho_{\mu}}

\newcommand{\ML}{{\mathcal M}^{L}}
\newcommand{\ER}{\varepsilon_R}

\newcommand{\tr}{\text{tr}}

\newcommand{\abs}[1]{\lvert {#1} \rvert }

\DeclareMathOperator{\Ran}{\mathrm{Ran}}

\DeclareMathOperator{\supp}{\mathrm{supp}}
\DeclareMathOperator{\Spec}{\mathrm{Spec}}

\newcommand{\cM}{{M}} 

\renewcommand{\epsilon}{\varepsilon}
\renewcommand{\phi}{\varphi}

\newcommand{\ta}{d}
\newcommand{\cst}{\text{(cst.)}}

\begin{document}

\title{The energy of dilute Bose gases II: \\The general case}
\author[1]{S\o ren Fournais}
\author[2]{Jan Philip Solovej}
\affil[1]{\small{Department of Mathematics, Aarhus University\\ Ny Munkegade 118\\ DK-8000 Aarhus C\\ Denmark}}
\affil[2]{\small{Department of Mathematics\\ University of Copenhagen\\ Universitetsparken 5\\ DK-2100 Copenhagen \O\\Denmark }}
\maketitle

\begin{abstract}
For a dilute system of non-relativistic bosons interacting through a
positive potential $v$ with scattering length $a$ we prove
that the ground state energy density satisfies the bound $e(\rho) \geq
4\pi a \rho^2 (1+ \frac{128}{15\sqrt{\pi}} \sqrt{\rho a^3} +o(\sqrt{\rho a^3}\,))$, thereby proving a lower bound consistent with the Lee-Huang-Yang formula for
the energy density. 
The proof allows for potentials with large $L^1$-norm, in particular, the case of hard core interactions is included.
Thereby, we solve a problem in mathematical physics that had been a major challenge since the 1960's.
\end{abstract}

\tableofcontents

\section{Introduction}
\subsection{Introduction and result for hard core potential}
Since the foundational paper of Lee, Huang and Yang \cite{LHY}, it has been a  fundamental open problem in mathematical physics to establish the so-called Lee-Huang-Yang term of the energy of the dilute, hard core Bose gas. More precisely, this means to prove the 
two-term asymptotics of the ground state energy per unit volume in the thermodynamic limit for the $3$-dimensional, dilute Bose gas in the case of hard core interactions. In this paper we will give the lower bound for this ground state energy density for a very large class of potentials including the hard core case.
Our method allows for the potentials to have both a hard core part and a fairly large support/slow decay, thereby being remarkably close to the expected optimal conditions.

It is interesting to remark that in the same number of Physical Review as \cite{LHY}, Dyson \cite{dyson} gave rigorous upper and lower bounds to the ground state energy. Thus it appears that already at that time the importance  of establishing the correct mathematical way of analysing this fundamental many-body problem was recognised. Notice in passing that the upper bound of Dyson's 1957 paper remains unsurpassed in the hard core case, despite significant recent progress on the upper bound in general \cite{LY,YY,AA,BCS}.

In \cite{FS} we gave the first proof of a two-term lower bound agreeing with the LHY correction in the thermodynamic limit. 
The analysis in \cite{FS} requires, most importantly, an $L^1$-condition on the  potential, thereby excluding the hard core case.
Combining the result of \cite{FS} with the known upper bound of \cite{YY} (see \cite{BCS} for a recent improvement of the upper bound to $L^3$-potentials) establishes the correctness of the LHY-term in the thermodynamic limit. Notice that surprisingly the conditions under which we have the correct upper bounds are much more restrictive than those necessary for the lower bounds. In the present paper we extend and improve the methods of \cite{FS} combined with new insights in order to remove the $L^1$-assumption on the potential therefore allowing us to treat the hard core case. The improvements also allow for potentials of much larger support, which leads to the inclusion of potentials with slow decay through an approximation argument.

For more description of the history of the subject and related results
we refer to the accompanying paper \cite{FS} or to the many recent
works on the subject of which we mention in particular
\cite{ABS,BCS,BBCS0,BBCS,BBCS2,BCaS,BSS,BSS2,BFS,BS,ESY,F,GS,NT,NNRT,BogFuncI,BogFuncII,St}. Notice
also the inspiring review \cite{LSSY} of the status of the subject in
2005.  For a more recent review of various mathematical aspects of the
Bose gas see \cite{rougerie2021scaling}, for a point of view closer to
Physics, see \cite{ChevySalomon,Navon2010}.  The present paper is
essentially self-contained, in particular, it is an independent
improvement of \cite{FS}.

We proceed to more precisely define the model under study.
We consider $N$
bosons in $3$ dimensions described by the Hamiltonian
\begin{align}
\label{eq:Hamiltonian}
{\mathcal H}_N= {\mathcal H}_N(v)=\sum_{i=1}^N-\Delta_i+\sum_{1\leq i<j\leq N}v(x_i-x_j).
\end{align}
One of our main concerns in this paper will be to prove the LHY-correction for the hard core potential, i.e. for $v = v_{\text{hc}}$ with
\begin{align}\label{eq:v_hc}
v_{\text{hc}}(x) = \begin{cases}
0, & |x| > a, \\
+\infty, & |x| \leq a.
\end{cases}
\end{align}

We are interested in the thermodynamic limit of the ground state
energy density as a function of the particle density $\rho$.
\begin{align}\label{eq:e_rho}
e(\rho, v)=\lim_{L\to\infty\atop N/L^3\to \rho}L^{-3}\inf_{\Psi\in
  C^\infty_0([0,L]^N)\setminus\{0\}}\frac{\langle\Psi,{\mathcal
    H}_N(v) \Psi\rangle}{\|\Psi\|^2}.
\end{align}
We will omit the dependence on $v$ from the notation and just write
$e(\rho)$, when the potential is clear from the context.  Here the
inner product $\langle\cdot,\cdot\rangle$ and the corresponding norm
$\|\cdot\|$ are in the Hilbert space $L^2(\Omega^N)$, where we have
denoted $\Omega=[0,L]^3$. When considering bosons the infimum above
should be over all symmetric function in $C^\infty_0(\Omega^N)$. It is
however a well-known fact that the infimum over all functions is
actually the same as if constrained to symmetric functions. When we
restrict to functions with compact support in $\Omega$ we are
effectively using Dirichlet boundary conditions, but it is not
difficult to see that the thermodynamic energy is independent of the
boundary condition used.

Our result for the hard core potential is the following:

\begin{theorem}[The Lee-Huang-Yang Formula]\label{thm:LHY-hc}
Suppose $v = v_{\text{hc}}$ is the hard core potential defined in \eqref{eq:v_hc}. Then, 
\begin{equation}\label{eq:LHY-hc}
  e(\rho)\geq 4\pi\rho^2 a\left(1 +\frac{128}{15\sqrt{\pi}} \sqrt{\rho a^3}-{\mathcal C}(\rho a^3)^{\frac{1}{2}+\eta}\right),
\end{equation}
for some $\eta>0$ and ${\mathcal C}>0$.
\end{theorem}
\begin{remark}
Our proof gives that the choice $\eta=\frac{1}{4640}$ works in Theorem~\ref{thm:LHY-hc}. However, this exponent is likely an artefact of the proof and not a true reflection of the next correction term.
\end{remark}

\subsection{Result for general class of potentials}
An interesting feature of the LHY-formula is the universality property that \eqref{eq:LHY-hc} holds for a large class of potentials $v$, where the parameter $a$ has to be understood as the {\it scattering length} of the potential.
The definition of the scattering length is recalled in Section~\ref{scattering} below.
We will denote the scattering length of the potential $v$ by $a(v)$ and write $a$ instead of $a(v)$, when the potential is clear from the context.
Our more general result is indeed to show this universality, allowing for both potentials with large $L^1$-norm, for instance a hard core part, and a slowly decaying tail.
\begin{theorem}\label{thm:LHY-hc-gen}
Let ${\mathcal C}, S_0, \varepsilon_0>0$ be given. Then there exist $\eta \in (0,\frac{1}{4})$ and $C>0$ (only depending on ${\mathcal C}, S_0, \varepsilon_0$) such that if $v: {\mathbb R}^3 \rightarrow [0,+\infty]$ is a measurable, spherically symmetric potential with scattering length $a\in (0,\infty)$ and satisfying the following decay condition
\begin{align}\label{eq:decay}
\frac{1}{8\pi a} \int_{\{ |x| \geq R \}} v(x)\,dx \leq {\mathcal C} (R/a)^{-1-\varepsilon_0}, \qquad \text{ for all } R\geq S_0 a.
\end{align}
Then
\begin{align}\label{eq:LHY-hc-gen}
e(\rho) \geq 4 \pi a \rho^2 \left( 1 + \frac{128}{15\sqrt{\pi}}\sqrt{\rho a^3} - C (\rho a^3)^{\frac{1}{2} + \eta} 
\right).
\end{align}

\end{theorem}

\begin{remark}
Our proof gives that the choice $\eta= \frac{1}{5} \min\{ \frac{\varepsilon_0}{4(1+\varepsilon_0)},\frac{1}{928}\}$ works in Theorem~\ref{thm:LHY-hc-gen}.
Again, we do not expect this exponent in the error bound to be optimal.
\end{remark}

\begin{remark}
Theorem~\ref{thm:LHY-hc-gen} is quite satisfactory in its generality. One could imagine the LHY-correction to be true for all positive potentials with finite scattering length. Notice that the finiteness of the scattering length implies integrability at infinity by \cite{MR2923973}. The decay condition \eqref{eq:decay} is reasonably close to this.
One could of course imagine to allow for potentials with positive scattering length but which are not everywhere positive. 
In this generality the Lee-Huang-Yang formula cannot hold. This is discussed in \cite{LSSY} where a conjecture is formulated (see  \cite{Yin} for partial results).
\end{remark}

Clearly, Theorem~\ref{thm:LHY-hc} follows from Theorem~\ref{thm:LHY-hc-gen}, since the scattering length $a$ of the hard core potential is the same as 
the radius of its support and the decay condition \eqref{eq:decay} is trivially satisfied.
Notice that in Theorem~\ref{thm:LHY-hc-gen}, the potential $v$ is not assumed to be $L^1$, indeed the hard core potential $v_{\text{hc}}$ has $\int v_{\text{hc}} = \infty$.
To prove Theorem~\ref{thm:LHY-hc-gen}, we will approximate $v$ from below by a $\rho$-dependent $L^1$-potential in such a way that
the difference between the scattering lengths is sufficiently small while the integral of the approximating potential is not too large.
The existence of such an approximating potential is guaranteed by the following uniform approximation result, 
which could be of independent interest and will be proved in Section~\ref{scattering}.

\begin{theorem}[Approximation of scattering length]\label{thm:ApprScatGen}
Suppose that $v: {\mathbb R}^3 \rightarrow [0,+\infty]$ is radial and has finite scattering length $a(v)$.
Then for all $T>1$ there exists $v_T \in L^1({\mathbb R}^3)$, with compact support and such that 
\begin{align}\label{eq:lowerpt}
0 \leq v_T(x) \leq v(x),
\end{align}
for all $x \in {\mathbb R}^3$, and furthermore
\begin{align}\label{eq:UnifBoundaThm}
(8 \pi)^{-1} \int v_T &\leq T a(v), \nonumber \\
a(v_T) &\geq a(v) \Big(1-   \Big(1+ 2\frac{\sqrt{5}}{\sqrt{T}}\Big)   T^{-1} \Big).
\end{align}
\end{theorem}

From an explicit calculation in the case of the hard core potential (see Remark~\ref{rem:HardCore} below) it follows that the error term in \eqref{eq:UnifBoundaThm} has an optimal dependence on $T$ to leading order.
A result in the same spirit as Theorem~\ref{thm:ApprScatGen} was given in \cite[Lemma 1]{St} but with a non-optimal error bound, which would not be sufficient for our purpose.

Using the approximation result above, we need to prove a Lee-Huang-Yang formula for $L^1$-potentials 
but with an explicit dependence on the parameters of the problem. That is the result of the next theorem, which also allows for potentials whose support has large but finite diameter. This is needed in order to approximate potentials as in Theorem~\ref{thm:LHY-hc-gen} that are allowed to have a long tail.

\begin{theorem}\label{thm:LHY}
For all $\kappa \in (0, \frac{1}{4}), {\mathcal C} >0$,
there exist constants $C>0$, $\eta_1 \in (0, \frac{1}{4})$ (depending only on $\kappa, {\mathcal C}$) such that if $0 \leq v \in L^1({\mathbb R}^3)$ is 
spherically symmetric with support in $\overline{B(0,R)}$, and such that $a:=a(v), R, \int v$ and the density $\rho$ satisfy that 
\begin{align}\label{eq:Assumptionsv}
\rho a^3 \leq  C^{-1}, \qquad
\frac{R}{a} \leq  {\mathcal C} (\rho a^3)^{\kappa - \frac{1}{2}},  \qquad
{\mathcal R} \leq   {\mathcal C} (\rho a^3)^{-\eta_1 - \frac{1}{2}},
\end{align}
where
\begin{equation}\label{eq:DefCalR}
{\mathcal R}  :=\frac{1}{8\pi
a} \int v.
\end{equation}
Then
\begin{equation}\label{eq:LHY}
e(\rho) \geq 4 \pi a \rho^2 \left( 1 + \frac{128}{15\sqrt{\pi}}\sqrt{\rho a^3} - 
C (\rho a^3)^{\frac{1}{2} + \eta_1}  \right).
\end{equation}
\end{theorem}

\begin{remark}
Our proof gives the choice 
\begin{align}\label{eq:Choiceeta2}
\eta_1:=\frac{1}{5} \min\{ \frac{1}{928}, \frac{\kappa}{11}\}.
\end{align}
\end{remark}

We end this introduction by showing how Theorem~\ref{thm:LHY-hc-gen} follows from Theorem~\ref{thm:LHY} and by a remark detailing the formal definition of the Hamiltonian ${\mathcal H}_N(v)$ for potentials as singular as $v_{\text{hc}}$.
In the following Section~\ref{sec:Desc}, we will give an overview the paper, in particular of the proof of Theorem~\ref{thm:LHY}.

\begin{proof}[Proof of Theorem~\ref{thm:LHY-hc-gen}]
Choose $\kappa = \frac{\varepsilon_0}{4(1+\varepsilon_0)}$ and ${\mathcal C} = 2$.
Let $\eta_1$ be given by Theorem~\ref{thm:LHY} for the chosen $\kappa, {\mathcal C}$ and define
$T = (\rho a^3)^{-\frac{1}{2} - \eta_1}$.
Let $v$ satisfying the assumptions of Theorem~\ref{thm:LHY-hc-gen} be given and let $a=a(v)$ be its scattering length.
By the trivial bound $e(\rho) \geq 0$, and by choosing $C$ larger than some universal constant it suffices to consider small values of $\rho a^3$.
Choose, using Theorem~\ref{thm:ApprScatGen}, an $L^1$-potential with compact support, $\widetilde{v}$ (depending on $T$) such that  \eqref{eq:UnifBoundaThm} is satisfied.

Let $R = a (\rho a^3)^{-\frac{1}{2} + \kappa}$, and define
$v' = \widetilde{v} \one_{\{|x| \leq R \}}$.
Notice that $R \geq S_0 a$ for sufficiently small values of $\rho a^3$.
By \eqref{eq:ScatInt}, \eqref{eq:decay} and the choice of $\kappa$, we have
\begin{align}
a(\widetilde{v} - v') \leq \frac{1}{8\pi} \int  \widetilde{v} \one_{\{|x| \geq R \}} \leq  C (R/a)^{-1-\varepsilon_0} a= C (\rho a^3)^{\frac{1}{2} + \frac{\varepsilon_0}{4}}a.
\end{align}
Therefore,
\begin{align}\label{eq:CollectScattering}
a= a(v) &\leq a(\widetilde{v}) + C a (\rho a^3)^{\frac{1}{2} + \eta_1} \nonumber \\
&\leq a(v') + a(\widetilde{v} - v') + C  (\rho a^3)^{\frac{1}{2} +\eta_1}a\nonumber \\
&\leq a(v') + C\left( (\rho a^3)^{\frac{1}{2} + \frac{\varepsilon_0}{4}} +  (\rho a^3)^{\frac{1}{2} + \eta_1}\right)a.
\end{align}
Also, by construction, $v'$ has support in $\{|x| \leq R\}$ and $\frac{1}{ 8 \pi}\int v' \leq T a$. 
So $v'$ satisfies \eqref{eq:Assumptionsv} with ${\mathcal C} = 2$ (for $\rho a^3$ sufficiently small).
Using monotonicity of the energy as a function of the potential and Theorem~\ref{thm:LHY}, 
we find
\begin{align}
e(\rho,v) \geq e(\rho,v') 
&\geq 4\pi\rho^2 a(v')\left(1 +\frac{128}{15\sqrt{\pi}} \sqrt{\rho a(v')^3}
- C (\rho a(v')^3)^{\frac{1}{2} + \eta_1} \right) .
\end{align}
At this point we insert \eqref{eq:CollectScattering} to get \eqref{eq:LHY-hc-gen} (with $\eta = \min(\frac{\varepsilon_0}{4}, \eta_1)$).
\end{proof}

\begin{remark}[Definition of the Hamiltonian]
Since we allow for very singular but positive potentials, we recall the precise definition of the realization of ${\mathcal H}_N(v)$ as a self-adjoint operator.

Consider the quadratic form 
\begin{align}
Q_N(\Psi) = \int_{\Lambda^N} \sum_{j=1}^N |\nabla_j \Psi|^2 + \sum_{i<j} v(x_i-x_j) |\Psi|^2\,dx,
\end{align}
defined on the domain
\begin{align}
D(Q_N):= \Big\{ \Psi \in H_0^1(\Omega^N)\,\Big{|}\, \Big(\sum_{i<j} v(x_i-x_j)\Big)^{1/2} \Psi \in L^2(\Lambda^N) \Big\}.
\end{align}
Consider furthermore,
\begin{align}
{\mathcal H}_Q = \overline{ \{  \Psi \in L^2(\Omega^N) \,|\, Q_N(\Psi) + \| \Psi \|_2^2 < \infty \}},
\end{align}
where the closure is taken in $L^2(\Omega^N)$.
Since ${\mathcal H}_Q$ is a closed, linear subspace of the Hilbert space $L^2(\Omega^N)$, ${\mathcal H}_Q$ becomes a Hilbert space in its own right having $D(Q_N)$ as a dense subspace.
Therefore, clearly $Q_N$ defines a densely defined, Hermitian, quadratic form in the Hilbert space ${\mathcal H}_Q$. It is straight forward to check that $Q_N$ is a closed form and it therefore follows that $Q_N$ defines a unique self-adjoint operator ${\mathcal H}_N$.
\end{remark}

\paragraph{Acknowledgements.}
SF was partially supported by a Sapere Aude grant from the
Independent Research Fund Denmark, Grant number DFF--4181-00221, by the Charles Simonyi Endowment, and by an EliteResearch Prize from the Danish Ministry
of Higher Education and Science.
JPS was partially supported by the Villum Centre of Excellence for the
Mathematics of Quantum Theory (QMATH). 

\section{Description of the proof}\label{sec:Desc}

In this section we give an overview of the paper, most notably the proof of Theorem~\ref{thm:LHY}.

From a fundamental point of view, the work provides a rigorous understanding of the correctness and the limits of the Bogoliubov pairing theory.

Compared to our previous article \cite{FS}, we have to deal with the complication that the potential can have large $L^1$-norm and support. In particular, this causes our control of high moments of the excitations (the localization of large matrices argument) to break down. It turns out that a substitute for this is to bound only high moments of the low-frequency excitations.
To better describe this complication, we first give the common outline of the present article and \cite{FS} (see \cite[Section 2]{FS} for a more detailed description) 
and below comment on their differences.

Section~\ref{scattering} gives background on the scattering length $a$ and related quantities. 
In this paper we consider very general potentials, for example of hard core type. We need to approximate these potentials from below with a control of the relative difference in scattering length. Therefore, 
a key new result is Theorem~\ref{thm:ApprScatGen}, which we prove in Section~\ref{scattering}.

A main challenge for the Bogoliubov approach to work is to establish Bose-Einstein condensation for the interacting Bose gas. This is a main challenge in mathematical physics, and at present out of reach in the setting of the thermodynamic limit. However, we can prove condensation---with optimal estimates on the condensate depletion---on length scales much {\it shorter} than the healing length $1/\sqrt{\rho a}$. In order to get the LHY-term correct, it is necessary to consider length scales much {\it longer} than the healing length. The way out of this apparent problem is to make a double localization, first to length scales much longer than the healing length (the length scales $\ell$ to be introduced in \eqref{eq:def_ell} below). On these 'large' boxes, one then carries out a second localization in order to obtain almost optimal a priori information on condensation and the condensate fraction.
This double localization has to 1) keep the condensate untouched; and 2) allow for the algebraic structure of the Bogoliubov approach to remain intact. These two constraints unfortunately result in the localized kinetic energy being somewhat complicated (see \eqref{eq:DefT_new} below). This causes several technical complications in the rest of the proof.
Since we here allow the potentials to have a long range, the analysis in the small boxes has to be improved compared to \cite{FS}. We can therefore not simply refer to \cite{FS} but have to carry out the necessary analysis in detail in Appendix~\ref{SmallBoxes}.
For further results on condensation 

As a technical tool in order to avoid keeping track of how the particles distribute themselves between the localized boxes, we reformulate the problem in a grand canonical setting with a chemical potential term. This and the localization procedure are carried out in Section~\ref{sec:Box} below.

A two-body potential can be thought of as producing two outgoing momenta starting from two incoming momenta, a process involving a total of $4$ momenta. The condensate being described by the particles having momentum $0$, we can accordingly split the potential into $0Q$, $1Q$,\ldots, $4Q$-terms, labelled according to the number of non-zero momenta involved. 
Section~\ref{sec:Box} contains the crucial Lemma~\ref{lm:potsplit} which uses the positivity of the potential to estimate it from below by an effective interaction where only $0Q$ to $3Q$ terms appear.
At the same time, this estimate renormalizes the interaction so that from that point on all potentials will have $L^1$-norm controlled by the scattering length.

Section~\ref{sec:Apriori} starts by establishing a priori estimates on, among other things, the condensation in the large boxes. 
Part of this analysis is the second localization to small boxes, which is carried out in Appendix~\ref{SmallBoxes}.
Furthermore, a number of estimates needed to control the excitation of particles outside the condensate are established in this section. Here a fundamental difference between the current paper and \cite{FS} appears, which will be discussed further below. The actual control of the particles outside the condensate is the content of Section~\ref{sec:LocMatrices}.

It is an important insight from \cite{YY, FS} and the current paper that a correct treatment of the $3Q$-term is essential to the precision of the LHY-term: The so-called 'soft-pairs' of two high momenta (of magnitude $a^{-1}$) producing one zero momentum and one low momentum (of magnitude $\sqrt{\rho a}$) or vice versa, are a part of the $3Q$-term and have to be calculated precisely. However, to essentially reduce the $3Q$-term to the soft-pair contribution requires a few estimates, some of which are carried out in the short Section~\ref{sec:3Q}---since they are easiest in position-space---and some are carried out in Section~\ref{sec:Second}, being easier in $2$nd quantized formalism.

In Section~\ref{sec:Second} we reformulate the problem on the large boxes in terms of $2$nd quantization. This is the setting in which the Bogoliubov diagonalization of quadratic operators is well understood. Furthermore, in Section~\ref{sec:Second} we apply the technique of $c$-number substitution, which allows to reduce our quantum mechanical problem to a family of problems in the non-condensed particles alone, parametrized by the number of particles in the condensate, denoted by $\rho_z \ell^3$.
After this, we need to establish rough bounds on $\rho_z \ell^3$, which is the content of Section~\ref{sec:rough}.
This sets the stage for the precise calculation of the ground state energy in the large boxes in Section~\ref{sec:Precise}. Here the quadratic Bogoliubov-type Hamiltonian is diagonalized, yielding the right energy to LHY-precision. Several terms are left out in this calculation, some are positive and some are error terms, of which the most severe corresponds to the soft-pairs from the $3Q$-term. It is another main insight of this paper and \cite{FS} that these soft pairs can be controlled by the positive terms, in particular, by part of the Bogoliubov-diagonalized Hamiltonian.

In the short Section~\ref{sec:proofCombined} the estimates are combined to give the proof of the LHY-asymptotical formula Theorem~\ref{thm:LHY}. Several appendices contain technical details.
Of particular importance is Appendix~\ref{sec:params}: The proof involves a number of parameters that have to be chosen large or small and have to satisfy several relations between them. Appendix~\ref{sec:params} lists the relations between these parameters and also the concrete choice made at the end in order to finish the proof of Theorem~\ref{thm:LHY}.

In the previous work \cite{FS} we proved the lower bound for the Lee-Huang-Yang correction for potentials with bounded $L^1$-norm. As we pointed out in that paper, there was one essential inequality, where $\int v$ appeared (instead of the smaller quantity $a$), namely \cite[(7.4) in Theorem 7.1]{FS} which was an essential ingredient in the localization of large matrices argument. This is needed in order to control moments of $n_{+}$---the operator counting the number of excitations out of the condensate---which is crucial for the finer analysis in later sections.
One of the main observations in the present work, is that it suffices to make a localization of large matrices argument on a momentum localized version of $n_{+}$---counting the excitations in the low momenta, denoted by $n_{+}^{\rm L} $ and defined in \eqref{eq:def_nplusL}.
In particular, it is important that in the control of the 'soft-pairs' in Section~\ref{sec:Precise} it suffices to have the localization of large matrices result for $n_{+}^{\rm L}$.

To get sufficiently sharp bounds to carry through the localization of large matrices for $n_{+}^{\rm L} $ requires quite a bit of technical work (see in particular Appendix~\ref{sec:proofs}).
It is worth noticing that this part of the analysis would have been substantially easier---and the resulting estimates sharper---if the localized problem had been periodic with the standard periodic Laplacian as kinetic energy instead of the localized kinetic energy in \eqref{eq:DefT_new}. However, we do not know how to compare the problem on a periodic box to the thermodynamic limit.
Furthermore, in order for these estimates to be valid (see in particular Lemma~\ref{lem:LowMomentaOperator}), we have had to modify the kinetic energy compared to \cite{FS} essentially by extracting the extra term $\varepsilon_N (-\Delta^{\mathcal N})$ in \eqref{eq:DefT_new} to assure a degree of uniform ellipticity on the box. 
So the localized operator (on the large box) in the present paper is different from the one of \cite{FS}.

Also, let us give some idea as to why it appears necessary to work with a different quantity than $n_{+}$:
The localization of large matrices allows us to work with states $\Psi$ satisfying for suitable ${\mathcal M}>0$ that $\one_{(0,{\mathcal M}]}(n_{+}) \Psi = \Psi$, which clearly gives us control of expectations of powers of $n_{+}$ in the state $\Psi$.
Since $\Psi$ is an $n$-particle state, one can always take the trivial choice ${\mathcal M} = n$, but in order to be useful for our purposes we need to improve that bound by a non-trivial power of our small parameter $\rmu a^3$.

We found in  \cite{FS}  that $\M$ had to be sufficiently {\it large} in order to localize $n_+$ with a localization error of smaller order than LHY and had to be sufficiently {\it small} in order to control expectations of powers of $n_+$ in the later estimates of the proof. The two conditions we found in  \cite{FS}  (see \cite[(5.9),(5.10), and (5.12)]{FS}) are \begin{equation}\label{eq:oldMest}
(\int v/a)^{1/2}(\rmu a^3)^{-1/4}\ll \M\ll (\rmu a^3)^{-1/2}.
\end{equation}
Here and in the rest of the paper $f\ll g$ is used in the precise meaning that $(f/g)\leq (\rmu
a^3)^{\varepsilon}$ for some positive $\varepsilon$ and likewise for $f \gg g$.
The first condition in \eqref{eq:oldMest} ensures that the error in \cite[Lemma~8.3]{FS} is smaller than the LHY order. The second condition is not only needed in  \cite{FS}  but also in this paper (see estimate \eqref{cond:KH3-n} on $\M^{\rm L}$ used in 
Lemma~\ref{lem:Q3-splitting2} below).

To prove the LHY asymptotics for the large class of potentials we consider here, in particular, the hard core potential, we need to approximate by $L^1$ potentials in such a way that the scattering length is approximated to order $o(\sqrt{\rmu a^3})$. From 
Theorem~\ref{thm:ApprScatGen} we see that this would require $\int v/a \gg (\rmu a^3)^{-1/2}$, which is in contradiction with
\eqref{eq:oldMest}.

In this paper we use two observations to circumvent this problem. Firstly, as mentioned above we only need to control fluctuations of the number of low momentum excitations $n_+^{\rm L}$. Secondly, the first condition in \eqref{eq:oldMest}, required to localize fluctuations, can be improved if we only have to localize the low momentum fluctuations $n_+^{\rm L}$. In fact, it follows from 
Proposition~\ref{prop:LocMatrices} below that the conditions in this paper replacing \eqref{eq:oldMest} are
$$
(\int v/a)^{1/2}(\rmu a^3)^{-3/16}\ll \ML \ll (\rmu a^3)^{-1/2}.
$$
Here $\ML$ controls $n_+^{\rm L}$ in essentially the same way that
$\M$ controlled $n_+$ in  \cite{FS}. These conditions are now adequate for the $L^1$ approximation.

\section{Facts about the scattering solution in ${\mathbb R}^3$}\label{scattering}
\subsection{Basic theory}\label{basic}
In this section we establish notation and results concerning the scattering length and associated quantities. 
For simplicity we restrict ourselves to the $3$-dimensional situation. 
We refer to \cite[Appendix C]{LSSY} for more details. We will always assume that the potential $v:  {\mathbb R}^3 \rightarrow [0,+\infty]$ is radial and positive.
Often, but not always, we will also assume that $v$ has compact support, i.e. $v(x) =0$ unless  $|x| \leq R$, for some $R>0$.

\begin{definition}\label{def:scatteringlength}
Given a potential $v$ with compact support,
the {\it scattering length} $a=a(v)$ is defined by
\begin{align}\label{eq:ScatLengthMin}
\frac{4\pi a}{1- a/\widetilde{R}} = \inf_{
\{ \phi \in H^1(B(0,\widetilde{R}): \phi_{|x|=\widetilde{R}} =1
\} }
\Big\{ \int_{\{|x|\leq \widetilde{R}\} }|\nabla \phi(x)|^2 + \frac{1}{2}v(x) |\phi(x)|^2\,dx
\Big\}.
\end{align}
Here $\widetilde{R} > R$ is arbitrary.
\end{definition}

It follows from an analysis of the minimisation problem that $a$ is independent of the choice of $\widetilde{R} > R,$ and satisfies 
\begin{equation}
\label{eq: upper bound on scattering length finite range}
a\leq R.\end{equation}
By choosing $\varphi=1$ as a variational state, one gets
\begin{align}\label{eq:ScatInt}
a \leq \frac{1}{8\pi} \int v.
\end{align}

Furthermore, there is a unique minimizer $\phi_{v,\widetilde{R}}$ in \eqref{eq:ScatLengthMin}, which is radial and satisfies
$\phi_{v,\widetilde{R}}= \left(1-a/\widetilde{R}\right)^{-1} \phi_v(x)$.
Here the function $\phi_v(x)$ is independent of $\widetilde{R}$, radial, non-negative, monotone non-decreasing as a function of $|x|$ and satisfies (in the sense of distributions on the set where $v$ is $L^1_{\mathrm{loc}}$)
\begin{equation}\label{eq:scattering_equation}
-\Delta \phi_v + \frac{1}{2} v \phi_v = 0.
\end{equation}
Furthermore,
\begin{equation}\label{eq:ScatOutside}
\phi_v(x) = 1 - \frac{a}{|x|}, \qquad \text{for }|x| \geq R.
\end{equation}

We introduce the notation
\begin{equation}
\varphi_v = 1 - \omega_v.
\end{equation}
When there is no possible confusion, we will drop the subscript and write $\varphi = \varphi_v$ and $\omega=\omega_v$.
By the properties of $\varphi$, we find that
$\omega(x) = a/|x|$ for $x$ outside
$\supp\, v$.
Furthermore, $\omega$ is radially symmetric and
non-increasing with (see \cite[Appendix C]{LSSY})
\begin{align}
	0\leq \omega(x)\leq \min\left(\frac{a}{|x|}, 1\right) .\label{omegabounds}
\end{align}
We introduce the function
\begin{equation}\label{eq:gdef}
g := v(1-\omega).
\end{equation}
The scattering equation \eqref{eq:scattering_equation} can be reformulated as
\begin{align}
\label{eq:Scattering3}
-\Delta \omega = \frac{1}{2} g.
\end{align}
From this we deduce that, if $v \in L^1_{\rm loc}$ so that \eqref{eq:scattering_equation} is valid on all of ${\mathbb R}^3$, then
\begin{align}\label{eq:ScatLengthBasic}
a = (8\pi)^{-1} \int g,
\end{align}
and the Fourier transform satisfies
\begin{equation}\label{es:scatteringFourier}
\widehat{\omega}(k) = \frac{\hat{g}(k)}{2 k^2}.
\end{equation}
\subsection{Potentials without compact support}
For potentials that do not have compact support, the scattering length is defined as the limit of the scattering lengths for a sequence of localized versions, i.e.
\begin{align}\label{eq:InfRange}
a(v) = \lim_{n \rightarrow \infty} a( \one_{\{|x| \leq n\}} v).
\end{align}
Since $v_1 \leq v_2$ implies that $a(v_1) \leq a(v_2)$ (as is immediate from Definition~\ref{def:scatteringlength}), the limit in \eqref{eq:InfRange} exists in ${\mathbb R}$ if and only if the sequence $\{ a( \one_{\{|x| \leq n\}} v) \}_n$ is bounded.
By \cite[Lemma 1]{MR2923973} this is true if and only if $v$ is $L^1$ near infinity, i.e. if and only if there exists $b>0$, such that
\begin{align*}
\int_{\{|x|\geq b \}} v(x)\,dx < \infty.
\end{align*}

\subsection{Potentials with large integral}
We end this section by proving the approximation result for potential with large $L^1$-norm, Theorem~\ref{thm:ApprScatGen}.

For a radial potential $v:{\mathbb R}^3 \rightarrow [0,+\infty]$ with compact support, it is well-known \cite[Lemma 3.2]{BFS} that a monotone convergence result holds for the scattering length, i.e. that
\begin{align}\label{eq:Scatcutoff}
a(\max\{v,n\}) \nearrow a(v).
\end{align}
By combining \eqref{eq:InfRange} and \eqref{eq:Scatcutoff}, i.e. by simultaneously cutting off the top and the tail of the potential, we reach the following conclusion: Given a radial potential $v:{\mathbb R}^3 \rightarrow [0,+\infty]$, we can define $v_n:= \one_{\{|x|\leq n\}} \max\{v,n\}$ and get
\begin{align}
a(v_n) \nearrow a(v).
\end{align}
Notice that with this definition, the potential $v_n$ has compact support and belongs to $L^1({\mathbb R}^3)$ (actually, we even have $v_n \in L^{\infty}({\mathbb R}^3)$).
However, this approximation can be too slow for our purposes. For example for the hard core potential $v_{\text{hc}}$ defined in \eqref{eq:v_hc}, we have $v_n= n \one_{\{|x| \leq R\}}$ and therefore $\int v_n =  \text{Const}\cdot n$. A calculation shows that $a(v_{\text{hc}}) - a(v_n) \approx \text{Const} \cdot n^{-1/2}$.
This means that in order to approximate the scattering length to precision $\sqrt{\rho a^3}$ (with $a = a(v_{\text{hc}}) = R$) one would have an integral $\int v_n$ of magnitude $a (\rho a^3)^{-1}$.
This would not be allowed by our method---as can for instance be seen from the error term in Theorem~\ref{thm:LHY}.

Our result in Theorem~\ref{thm:ApprScatGen} shows that one can make a better approximation than that; not only for hard core potentials, but uniformly for general potentials.
Our approximation result is optimal to leading order as can be seen from the explicit calculation in the case of the hard core potential given below in Remarik~\ref{rem:HardCore}.

\begin{proof}[Proof of Theorem~\ref{thm:ApprScatGen}]
The main approximation result is given below as Lemma~\ref{lem:ApprScat} and is valid for $L^1$-potentials with compact support. Here we only make the simple reduction to that case.

Let $v$ be as in the theorem.
By the discussion above we can approximate $v$ arbitrarily well from below by potentials with finite $L^1$-norm and compact support. So for all $\delta>0$, there exists a radial potential $v' \in L^1({\mathbb R}^3)$ with compact support and with $0 \leq v' \leq v$ and such that 
\begin{align}\label{eq:primitive}
a(v') \leq a(v) \leq a(v') + \delta a(v).
\end{align}
Applying Lemma~\ref{lem:ApprScat} to $v'$ we find $v_T$ satisfying \eqref{eq:lowerpt} and \eqref{eq:UnifBoundaThm} as well as 
\begin{align}
a(v') \leq a(v_T) +  \Big(1+\frac{\sqrt{5}}{\sqrt{T}}\Big)  T^{-1} a(v') \leq a(v_T) +  \Big(1+\frac{\sqrt{5}}{\sqrt{T}}\Big) T^{-1} a(v).
\end{align}
Combining this with \eqref{eq:primitive} and using that $\delta$ was arbitrary finishes the proof.
\end{proof}

\begin{lemma}\label{lem:ApprScat}
Suppose that $v: {\mathbb R}^3 \rightarrow [0,+\infty]$ is radial and of class $L^1$ and that there exists $R>0$ such that $v(x) = 0$, for all $|x| > R$.
Then for all $T>1$ there exists $v_T \in L^1({\mathbb R}^3)$, with 
\begin{align}
0 \leq v_T(x) \leq v(x),
\end{align}
for all $x \in {\mathbb R}^3$, and such that
\begin{align}\label{eq:UnifBounda}
(8\pi)^{-1} \int v_T &\leq T a(v), \\
a(v_T) &\geq a(v) \Big(1- \Big(1+\frac{\sqrt{5}}{\sqrt{T}}\Big)  T^{-1}\Big).\label{eq:UnifBounda_itself}
\end{align}
\end{lemma}

\begin{proof}
We will write $a$ instead of $a(v)$ for simplicity.
We may assume that $(8\pi a)^{-1} \int v > T$, because if not there is nothing to prove.
Define $R_T = \inf\{ R'>0\,:\, \int_{\{|x| \geq R'\}} v \,dx < 8 \pi T a\}$, and
\begin{align}
v_< := v \one_{\{|x| \leq R_T\}}, \qquad v_T := v \one_{\{|x| > R_T\}}.
\end{align}
Clearly,
\begin{align}
\int v_T = 8\pi a T,
\end{align}
and $R_T > 0$ since $\int v > 8\pi a T$.

Let $\varphi$ be the scattering solution for the potential $v$. Similarly, we let $a_T$ be the scattering length of the potential $v_T$ and $\varphi_T$ the associated scattering solution. The same convention is used to introduce $a_{<}$ and $\varphi_{<}$.

We have from \eqref{eq:ScatLengthBasic}, using that $\varphi_T$ is a non-decreasing function,
\begin{align}
8\pi a_T = \int v_T \varphi_T \geq \varphi_T(R_T) \int v_T = 8\pi \varphi_T(R_T) T a,
\end{align}
so we find
\begin{align}\label{eq:phiRT}
\varphi_T(R_T) \leq \frac{a_T}{a} \frac{1}{T} \leq \frac{1}{T}.
\end{align}
The same argument applied to $\varphi$ gives
\begin{align}\label{eq:phiRT-2}
\varphi(R_T) \leq T^{-1}.
\end{align}

Now we choose $c=R_T \varphi_T(R_T)$ and use $u:= \one_{\{|x| \geq R_T\}} ( \varphi_T - \frac{c}{|x|})$ as a trial state in the functional for $a$.
Notice that $\varphi_T$ is constant on $\{|x| \leq R_T\}$, since it is harmonic there. Also, note that $v=v_T$ on $\{|x| \geq R_T\}$.
Therefore,
\begin{align}
4\pi a &\leq \int_{\{|x| \geq R_T\}} |\nabla u|^2 + \frac{1}{2} v |u|^2 
= E_1 + E_2 + E_3,
\end{align}
with
\begin{align}
E_1 &= \int_{\{|x| \geq R_T\}} |\nabla \varphi_T |^2 + \frac{1}{2} v_T |\varphi_T|^2 = 4\pi a_T,\nonumber \\
E_2 &= -2 c \int_{\{|x| \geq R_T\}} \nabla \varphi_T \cdot \nabla |x|^{-1} + \frac{1}{2} v_T \varphi_T |x|^{-1},\nonumber \\
E_3 &= c^2 \int_{\{|x| \geq R_T\}} |\nabla |x|^{-1}|^2 +\frac{1}{2} v_T |x|^{-2}.
\end{align}
Notice that $E_2=0$ by integration by parts, since $\Delta \varphi_T = \frac{1}{2} v_T \varphi_T$, and where the boundary term from the integration by parts disappears since $\varphi_T$ is constant on $\{|x| < R_T\}$.

We therefore find, using the monotonicity of the scattering length for the first inequality,
\begin{align}
0 \leq 4\pi (a-a_T) &\leq E_3 \leq c^2 (4 \pi R_T^{-1} + R_T^{-2} 4\pi  a T) 
\leq 4\pi \varphi_T(R_T)^2 ( R_T +   a T),
\end{align}
using the choice of $c$.

We conclude using the bound \eqref{eq:phiRT} that
\begin{align}\label{eq:smallRT}
0 \leq \frac{a-a_T}{a} \leq \frac{1}{T} \left( \frac{R_T}{a}  \frac{1}{T} + 1\right).
\end{align}

For large values of $\frac{R_T}{a}$ we will use a different inequality.
Assume that $\frac{a}{R_T} <   \lambda $, for some $\lambda  \leq 1$. 
We will choose the optimal $\lambda$ as a function of $T$ at the end of the proof.
Notice first that on $\{|x| < R_T\}$ we have $\varphi = c_0 \varphi_{<}$ with $c_0 = \varphi(R_T)/\varphi_{<}(R_T)$, since the two functions satisfy the same equation on the ball.
Also notice that since $v \geq v_T$ we have $\varphi \leq \varphi_T$ by \cite[Lemma C.2]{LSSY}.
Therefore, we can estimate
\begin{align}
8\pi a &= \int v_{<} \varphi + \int v_T \varphi \leq \frac{\varphi(R_T)}{\varphi_{<}(R_T)} 8\pi a_{<} + 8\pi a_T
= \frac{\varphi(R_T)}{1-\frac{a_{<}}{R_T}} 8\pi a_{<} + 8\pi a_T .
\end{align}
Using that $a_< \leq a$, $\frac{a}{R_T}< \lambda \leq 1$, and \eqref{eq:phiRT-2} we find
\begin{align}\label{eq:largeRT}
\frac{a-a_T}{a} \leq \frac{1}{1-\frac{a}{R_T}} T^{-1}.
\end{align}

We use \eqref{eq:smallRT} for $\frac{a}{R_T} \geq   \lambda$, \eqref{eq:largeRT} for $\frac{a}{R_T} <   \lambda $, and choose the optimal value
$\lambda=\frac{\sqrt{4T+1}-1}{2T}$.
Notice that this choice satisfies that $\lambda < 1$, since $T\geq 1$.
Inserting this optimal choice we find, for all $T \geq 1$,
\begin{align}
\frac{a-a_T}{a} \leq\Big( 1 + \frac{1+\sqrt{4T+1}}{2T}\Big) T^{-1},
\end{align}
from which the simplified version \eqref{eq:UnifBounda_itself} follows by an elementary estimate.
\end{proof}

\begin{remark}\label{rem:HardCore}
Suppose that  $v$ is the hard core potential with unit scattering length, i.e.,
\begin{align}
v(x) = \begin{cases} \infty, & |x| \leq 1, \\ 0, & |x| >1.
\end{cases}
\end{align}
Let $T\geq 1$ be given and assume that $\widetilde{v}$ is a radial $L^1$-potential with $0\leq \widetilde{v} \leq v$ (i.e. $\supp \widetilde{v} \subseteq \{ |x| \leq 1\}$) with $\int \widetilde{v} \leq 8\pi T$.
Let $\varphi_{\widetilde{v}}$ be the corresponding scattering solution. Since $\varphi_{\widetilde{v}}$ is radial and non-decreasing we find that
\begin{align}\label{eq:radialIncreasing}
4\pi a(\widetilde{v})
=
\int |\nabla \varphi_{\widetilde{v}}|^2 + \frac{1}{2} \widetilde{v} |\varphi_{\widetilde{v}}|^2 
\leq \int |\nabla \varphi_{\widetilde{v}}|^2 + \frac{1}{2} \overline{v} |\varphi_{\widetilde{v}}|^2,
\end{align}
with $\overline{v} = 8\pi T \delta_{\{|x| = 1\}}$ and $\delta_{\{|x| = 1\}}$ being the normalized surface measure on the unit sphere.
Although $\overline{v}$ is a measure (not a function), it is not difficult to generalize the discussion of Section~\ref{basic} to this case.
It follows that we can find the scattering length $a(\overline{v})$---satisfying $a(\overline{v}) \geq a(\widetilde{v})$
by \eqref{eq:radialIncreasing} and \eqref{eq:ScatLengthMin}---by solving the equation \eqref{eq:scattering_equation} subject to \eqref{eq:ScatOutside}.

In terms of the continuous function $u(r) = r \varphi(x)$, with $r = |x|$, we rewrite the equation as
\begin{align}
u''(r) = 0, \quad \text{ on } \quad \{r<1\} \cup  \{r>1\} ,
\end{align}
with the boundary conditions that $u(0)= 0$, $u(r) = r- a(\overline{v})$,  for $r>1$ and the jump condition $T u(1)= \lim_{r \searrow 1} u'(r) - \lim_{r \nearrow 1} u'(r)$ (corresponding to the $\delta$-function in the potential.
These equations are easily solved and yield $a(\overline{v})= \frac{T}{1+T}$.

In particular, we get that 
\begin{align}
a(v) - a(\widetilde{v}) = 1 -  a(\widetilde{v}) \geq \frac{1}{1+T},
\end{align}
for all potentials $\widetilde{v}\leq v$ with $\int \widetilde{v} \leq 8\pi T$.
Therefore, we can conclude that in general the error term $\Big(1+ 2\frac{\sqrt{5}}{\sqrt{T}}\Big)   T^{-1}$ in \eqref{eq:UnifBoundaThm} has the correct behavior $T^{-1}$  for large values of $T$ and with the correct coefficient ($=1$). However, the second term---in $T^{-3/2}$ in \eqref{eq:UnifBoundaThm}---can possibly be improved both in terms of the power and the coefficient.
\end{remark}

\section{Reduction to box Hamiltonian}\label{sec:Box}
To prove Theorem~\ref{thm:LHY} we will study localized problems. In order to control how particles distribute themselves among the boxes, it is convenient to introduce a chemical potential. This is done in subsection~\ref{subsec:chem}.

\subsection{Chemical potential}\label{subsec:chem}
We start by reformulating the problem grand canonically on Fock space. 
Consider, for given $\rmu >0$, the following operator ${\mathcal
  H}_{\rmu}$ on the symmetric Fock space $ {\mathcal F}_{\rm
  s}(L^2(\Omega))$. The operator ${\mathcal H}_{\rmu}$ commutes with
particle number and satisfies, with ${\mathcal H}_{\rmu,N}$ denoting
the restriction of ${\mathcal H}_{\rmu}$ to the $N$-particle subspace
of ${\mathcal F}_{\rm s}(L^2(\Omega))$,
\begin{align}\label{eq:BackgroundH}
{\mathcal H}_{\rmu,N} &= {\mathcal H}_N- 8\pi a \rmu N=\sum_{i=1}^N -\Delta_i 
+ \sum_{i<j} v(x_i-x_j)
- 8\pi a \rmu N\\
&=\sum_{i=1}^N \left( -\Delta_i - \rmu \int_{{\mathbb R}^3} g(x_i-y)\,dy \right)
+ \sum_{i<j} v(x_i-x_j).
 \nonumber 
\end{align}
Notice that the new term in ${\mathcal H}_{\rmu,N}$ plays the role of a chemical potential justifying the notation.

Define the corresponding ground state energy density,
\begin{align}
e_0(\rmu):= \lim_{|\Omega| \rightarrow \infty} |\Omega|^{-1}
\inf_{\Psi \in {\mathcal F_{\rm s}}(L^2(\Omega))\setminus \{0\}} \frac{\langle
  \Psi, {\mathcal H}_{\rmu} \Psi \rangle}{\| \Psi \|^2}.
\end{align}

\noindent We formulate the following result, which will be a
consequence of Theorems~\ref{thm:CompareBoxEnergy} and
\ref{thm:LHY-Box} below.

\begin{theorem}\label{thm:LHY-Background}
For all $\kappa \in (0,\frac{1}{4}), {\mathcal C}>0$, there exist $\varepsilon_2 >0, C>0$ and $\eta_2 \in (0,\frac{1}{4})$ (only depending on $\kappa, {\mathcal C}$) such that if $v\geq 0$ 
is spherically symmetric of class $L^1({\mathbb R}^3)$ with
compact support in $\overline{B(0,R)}$ and (with ${\mathcal R}$ as defined in \eqref{eq:DefCalR})
\begin{align}\label{eq:assumps_thm_background}
\rmu a^3 \leq C^{-1}, \qquad
\frac{R}{a} \leq {\mathcal C}(\rmu a^3)^{\kappa -\frac{1}{2}}, \qquad
{\mathcal R} \leq {\mathcal C}(\rmu a^3)^{-\eta_2 -\frac{1}{2}}.
\end{align}
Then the
thermodynamic ground state energy density of ${\mathcal H}_{\rmu}$
satisfies that
\begin{align}
  e_0(\rmu) &\geq -4\pi \rmu^2 a \left(1 - \frac{128}{15\sqrt{\pi}} (\rmu a^3)^{\frac{1}{2}}\right) 
  - C\rmu^2 a  (\rmu a^3)^{\frac{1}{2}+\varepsilon_2} .
\end{align}
\end{theorem}

\begin{remark}
Our proof gives the choice $\eta_2=5\eta_1$, with $\eta_1$ as stated in \eqref{eq:Choiceeta2} and $\varepsilon_2:= \eta_1$.
\end{remark}

\begin{proof}[Proof of Theorem~\ref{thm:LHY}]
Let $\kappa, {\mathcal C}$ be given and
let $\eta_2, C, \varepsilon_2$ be the resulting constants from Theorem~\ref{thm:LHY-Background}. 
If $v$ satisfies \eqref{eq:Assumptionsv} with $\eta_1=\eta_2$, then $v$ also satisfies \eqref{eq:assumps_thm_background} with $\rmu = \rho$.
By our choice of $\eta_1$, this means that \eqref{eq:assumps_thm_background} is satisfied for $\rmu =\rho$.
By inserting the ground state of
${\mathcal H}_N$ as a trial state in ${\mathcal H}_{\rmu}$ we get in
the thermodynamic limit and with $\rmu=\rho >0$
\begin{align}\label{eq:CompareGC}
e(\rho\,) &\geq e_0(\rmu) + 8\pi a \rho \rmu \nonumber \\
&\geq 8\pi a \rho \rmu
-4\pi \rmu^2 a \left(1 - \frac{128}{15\sqrt{\pi}} (\rmu a^3)^{\frac{1}{2}}\right)
- C \rmu^2 a  (\rmu a^3)^{\frac{1}{2}+\varepsilon_2},
\end{align}
where we have used the lower bound from Theorem~\ref{thm:LHY-Background}.
Recalling that $\rmu = \rho$,
this gives \eqref{eq:LHY} with $\eta_1:= \min\{\eta_2, \varepsilon_2\}$.
\end{proof}

\subsection{Localization: Setup and notation}
The main part of the analysis will be carried out on a box $\Lambda=[-\ell/2,\ell/2]^3$ of size 
$\ell$ given in \eqref{eq:def_ell}.
In this section we will carry out
the localization to the box $\Lambda$. The main result is given at the
end of the section as Theorem~\ref{thm:CompareBoxEnergy} which states
that for a lower bound it suffices to consider a `box energy',
i.e. the ground state energy of a Hamiltonian localized to a box of
size $\ell$.  For convenience, in Theorem~\ref{thm:LHY-Box} we state
the bound on the box energy that will suffice in order to prove
Theorem~\ref{thm:LHY-Background}.

Notice that the assumption $\frac{R}{a} \leq {\mathcal C}(\rmu a^3)^{\kappa- \frac{1}{2}}$ from Theorem~\ref{thm:LHY-Background} implies that
\begin{align}
\frac{R}{\ell} = \frac{1}{K_{\ell}}\frac{R}{a} \sqrt{\rmu a^3} \leq {\mathcal C} \frac{(\rmu a^3)^{\kappa}}{K_{\ell} }.
\end{align}
So $R$ is much smaller than $\ell$ for all sufficiently small values of $\rmu a^3$.

It will be important to make an explicit choice of a localization
function $\chi\in C_0^{M}(\R^3)$, for $M\in\N$ and with support in
$[-1/2,1/2]^3$. 
It is given in Appendix~\ref{sec:chiproperties}. The
function will not be smooth but it will be important in the analysis
that we choose $M\in 4 \N$ finite but sufficiently large. The condition for the choice of $M$ is given in \eqref{eq:choiceM} below.
The explicit choice of $\chi$ plays a role in the double localization argument which was carried out in \cite{FS}. However, in the present paper, we only use some of the consequences of this argument, most notably Theorem~\ref{thm:aprioribounds} below. Therefore, the reader of the present paper may safely disregard the specific choice and only observe that $\chi \in C_0^{M}$ is such that
$\chi$ is even, 
\begin{align}\label{eq:chinormalization}
0 \leq \chi, \qquad \int \chi^2 = 1.
\end{align}
For some proofs it will be useful to use the following structure of $\chi$,
\begin{align}\label{eq:Splitchi}
\chi = f^2,
\end{align}
where $f\in C_0^{\frac{M}{2}}(\R^3)$.

We will also use the notation 
\begin{align}
\chi_{\Lambda}(x) := \chi(x/\ell). 
\end{align}

For given $u \in {\mathbb R}^3$, we define 
\begin{align}
\chi_u(x) = \chi(\frac{x}{\ell}-u)=\chi_\Lambda(x-u\ell).
\end{align}
Notice that $\chi_u$ localizes to the box $\Lambda(u) := \ell u + [-\ell/2,\ell/2]^3$.

We will also need the sharp localization function $\theta_u$ to the box $\Lambda(u)$, i.e.
\begin{align}\label{eq:Theta}
\theta_u := \one_{\Lambda(u)}.
\end{align}

Define $P_u, Q_u$ to be the orthogonal projections in $L^2({\mathbb R}^3)$ defined by
\begin{align}\label{def: projections}
P_u \varphi := \ell^{-3} \langle \theta_u, \varphi\rangle \theta_u, \qquad  Q_u \varphi:= \theta_u \varphi - \ell^{-3} \langle \theta_u, \varphi \rangle \theta_u.
\end{align}
In the case $u=0$, we will use the notations
\begin{equation}
\theta_{u=0} = \theta, \qquad
  P_{u=0}=P_\Lambda=P,\qquad Q_{u=0}=Q_\Lambda=Q.
\end{equation}
Define furthermore 
\begin{align}\label{eq:3.5}
W(x) := \frac{v(x)}{\chi*\chi(x/\ell)}.
\end{align}
Since $R \ll \ell$ it is clear that $W$ is well-defined for sufficiently small values of $\rmu a^3$. 
Manifestly $W$ depends on $\ell$ and thus $\rho_\mu$, but we will not reflect this in our notation.

Define the localized potentials
\begin{align}\label{eq:w_u}
w_u(x,y) := \chi_u(x) W(x-y) \chi_u(y), \qquad w(x,y) := w_{u=0}(x,y).
\end{align}
Notice the translation invariance,
\begin{align}\label{eq:transInv}
w_{u+\tau}(x,y) = w_u(x-\ell \tau,y-\ell \tau).
\end{align}
For some estimates it is convenient to invoke the scattering solution
and thus we introduce the notation, which again is well-defined for
$\rho_\mu a^3$ sufficiently small,
\begin{align}\label{eq:defW_12}
W_1(x) &:= W(x) (1-\omega(x)) = \frac{g(x)}{\chi*\chi(x/\ell)},\qquad w_1(x,y):=w(x,y)(1-\omega(x-y)), \nonumber \\
W_2(x) &:= W(x) (1-\omega^2(x)) = \frac{g(x)+g\omega(x)}{\chi*\chi(x/\ell)},\qquad w_2(x,y):=w(x,y)(1-\omega^2(x-y)).
\end{align}
If we add a subscript $u$ we mean as above the translated versions $w_{1,u}(x,y)=w_1(x-\ell u,y-\ell u)$. 
For $\rmu a^3$ sufficiently small a simple change of variables yields, for all $u\in{\mathbb R}^3$, the identities
\begin{align}\label{eq:DefU}
\ell^{-3} \iint_{{\mathbb R}^3\times{\mathbb R}^3} \chi(\frac{x}{\ell})\chi(\frac{y}{\ell}) W_1(x-y)\, dx\,dy
&=\ell^{-3} \iint_{{\mathbb R}^3\times{\mathbb R}^3} w_{1}(x,y)\,dx\,dy \nonumber \\
&= \int g(x)\,dx = 8 \pi a,
\end{align}
and likewise
\begin{align}\label{eq:w2int}
\ell^{-3} \iint_{{\mathbb R}^3\times{\mathbb R}^3} w_2(x,y)\,dx\,dy 
&= \int g(1+\omega) \,dx= 8\pi a + \int g \omega \,dx .
\end{align}
The following basic lemma will often be useful.

\begin{lemma}\label{lem:Simple}
Assuming that $R/\ell$ is smaller than some universal constant, we have
\begin{itemize}
\item 
\begin{equation}\label{eq:W1-g-new}
0 \leq {W}_1(x) -  g(x)  \leq C g(x) \frac{\min\{|x|^2,R^2\}}{\ell^2} .
\end{equation}
\item Suppose that $f \in L^1({\mathbb R}^3)$ satisfies $\supp f \subset B(0,R)$ and $f(-x)=f(x)$.
Then
\begin{align}\label{eq:I2-integral-2}
\left| f*\chi_{\Lambda}(x) - \chi_{\Lambda}(x) \int f \right| \leq \max_{i,j}\| \partial_i\partial_j \chi \|_{\infty} 
\frac{1}{\ell^2}  \int |x|^2  |f|\,dx \leq
\max_{i,j}\| \partial_i\partial_j \chi \|_{\infty} 
\frac{R^2}{\ell^2}  \int  |f|\,dx 
\end{align}
\item
For some universal constant $C>0$ we have
\begin{align}\label{eq:I2-integral-new}
\left|(2\pi)^{-3}  \int \frac{\widehat{W}_1(k)^2}{2k^2} \,dk -
\widehat{g\omega}(0) \right| \leq  C \frac{R a^2}{\ell^2}.
\end{align}
Furthermore,
\begin{equation}\label{eq:I2-integral-2-new}
\int \frac{(\widehat{W}_1(k) - \widehat{g}(k))^2}{2k^2} \,dk \leq C \frac{R^3 a^2}{\ell^4}.
\end{equation}
\end{itemize}
\end{lemma}

\begin{proof}
To prove \eqref{eq:W1-g-new} it suffices to estimate $\chi*\chi(x/\ell)$ for $|x|< R$, since $\supp g \subset \overline{B(0,R)}$.
Clearly $\chi*\chi(y) \leq 1$ for all $y \in {\mathbb R}^3$. On the other hand
\begin{align}
\left| 1 - \chi*\chi(y) \right|= \left| \int \chi(u) [ \chi(-u) - \chi(y-u)] \,du \right| \leq C y^2 \max_{i,j} \| \partial_i \partial_j \chi \|_{\infty} ,
\end{align}
by Taylor's theorem, using that $\int y \chi(y) \,dy = 0$. This gives \eqref{eq:W1-g-new} and \eqref{eq:I2-integral-2} follows in the same way.

Recall that $\widehat{\omega}(k)=\frac{\widehat{g}(k)}{2k^2}$ by \eqref{es:scatteringFourier}. 
Using the Fourier transformation and \eqref{eq:W1-g-new} we get
\begin{align}
\left|(2\pi)^{-3} \int \frac{\widehat{W}_1^2(k)- \widehat{g}^2(k)}{2k^2} \, dk\right|
&= C   \iint \frac{(W_1-g)(x) (W_1+g)(y)}{|x-y|}\,dx\,dy  \nonumber \\
&\leq \widetilde{C} C \frac{1}{\ell^2}  \iint \frac{|x|^2 g(x) g(y)}{|x-y|}\,dx\,dy  \nonumber \\
&=  C' \frac{1}{\ell^2}\int |x|^2 g(x) \omega(x) \,dx .\label{eq: HLS}
\end{align}
At this point we insert \eqref{omegabounds} to get
\begin{align}
\left|(2\pi)^{-3} \int \frac{\widehat{W}_1^2(k)- \widehat{g}^2(k)}{2k^2} \, dk\right|
\leq C'  \frac{a}{\ell^2}\int |x| g(x)\,dx \leq 8\pi C'  \frac{a^2 R}{\ell^2},
\end{align}
since $|x|\leq R$ on $\supp g$.
This finishes the proof of \eqref{eq:I2-integral-new}.
The proof of \eqref{eq:I2-integral-2-new} follows from a similar calculation and is omitted.
\end{proof}

\subsection{The localized hamiltonian}
Define the kinetic energy operator on $L^2(\Lambda)$ as
\begin{align}
\label{eq:DefT_new}
      {\mathcal T} &:=
       \varepsilon_N (-\Delta^{\mathcal N})+ (1-\varepsilon_N)\widetilde{\mathcal T}, 
\end{align}
with $\Delta^{\mathcal N}$ being the Neumann Laplacian on $\Lambda$ and 
\begin{align}\label{eq:DefT_tilde}
      \widetilde{\mathcal T}:=  {\mathcal T}' +
      \frac12 \varepsilon_T (d \ell)^{-2} \frac{-\Delta^{\mathcal N}}{-\Delta^{\mathcal N}+(d\ell)^{-2}} 
      + b \ell^{-2} Q 
      + b\varepsilon_T (d\ell)^{-2} Q\one_{(d^{-2}\ell^{-1},\infty)}(\sqrt{-\Delta})Q,
\end{align}
where
\begin{equation}\label{eq:T'u_new}
  {\mathcal T}' :=
  Q \chi_{\Lambda} \tau(-i\partial_x) \chi_{\Lambda} Q,
\end{equation}
with $\tau$ being the Fourier multiplier,
\begin{align}\label{eq:Def_tau}
\tau(p) = \Big\{
  (1-\varepsilon_T) \Big[ |p| - \frac{1}{2}(s\ell)^{-1} \Big]_{+}^2
  + \varepsilon_T \Big[ |p| - \frac{1}{2}(d s\ell)^{-1} \Big]_{+}^2
  \Big\}.
\end{align}
Here the parameters $0<s, d,  \varepsilon_N, \varepsilon_T$ are chosen in Appendix~\ref{sec:params} in such a way that Assumption~\ref{assump:params} is satisfied, and $b$ is a universal constant chosen in Theorem~\ref{thm:CompareBoxEnergy} below.

\begin{remark} The kinetic energy operator in \eqref{eq:DefT_new} looks complicated. This is partly because 
we need to localize it even further into smaller boxes in order to get
a priori estimates (see \cite[Appendix B]{FS} and Theorem~\ref{thm:aprioribounds} below) but also because we need different terms to remedy certain problems of the main kinetic energy operator ${\mathcal T}'$.  

To describe the different terms, let us start with ${\mathcal T}'$, which is the main kinetic energy term in the (large)
boxes. This will  be used in the Bogoliubov-type calculation---together with the localized potential energy---to give the correct energy up to and including the LHY-term. However, ${\mathcal T}'$ does not have a gap, and also---due to the factors of $\chi$---it does not behave correctly close to boundary of the box, which will cause problems for establishing certain technical lemmas along the way. The remaining terms in ${\mathcal T}$ are included to solve these (and similar) problems.

The Neumann Laplacian in \eqref{eq:DefT_new} compensates the loss of ellipticity of ${\mathcal T}'$ near the boundary and is essential in
Lemma~\ref{lem:LowMomentaOperator}.
The second term in
\eqref{eq:DefT_tilde} will give us a Neumann gap in the small boxes, which is important in the analysis in \cite[Appendix B]{FS}.
The
third term in \eqref{eq:DefT_tilde} is a Neumann gap in the large
boxes. The fourth term in \eqref{eq:DefT_tilde}will control errors coming
from excited particles with very large momenta (see Lemma~\ref{lem:Q3-splitting1} and the estimate \eqref{eq:EstimateQ3Tilde1}  in Lemma~\ref{lem:EstimatesOnQ3-js}). 
\end{remark}

The localized Hamiltonian ${\mathcal H}_{\Lambda}$ will be an
operator on the symmetric Fock space over $L^2({\mathbb R}^3)$
preserving particle number. Its action on the $N$-particle sector is
as
\begin{align}\label{eq:Def_HB_new}
({\mathcal H}_{\Lambda}(\rmu))_{N} :=
 \sum_{i=1}^N \mathcal{T}^{(i)} -
\rmu \sum_{i=1}^N \int w_{1}(x_i,y)\,dy + \sum_{1\leq i<j\leq N} w(x_i,x_j),
\end{align}
with $w, w_1$ from \eqref{eq:w_u} and \eqref{eq:defW_12}.

Using the splitting ${\mathcal F}_{\rm s}(L^2({\mathbb R}^3)) = {\mathcal F}_{\rm s}( L^2(\Lambda)) \otimes  {\mathcal F}_{\rm s}((L^2({\mathbb R}^3\setminus \Lambda)))$ we may also consider ${\mathcal H}_{\Lambda}(\rmu))$ as an operator on ${\mathcal F}_{\rm s}:={\mathcal F}_{\rm s}( L^2(\Lambda))$. This operator will have the same spectrum as the original operator on ${\mathcal F}_{\rm s}(L^2({\mathbb R}^3))$.

\begin{theorem}\label{thm:CompareBoxEnergy}
Assume that the conditions of Assumption~\ref{assump:params} 
are satisfied and that $\rmu a^3$ is small enough. 
Define the ground state energy and energy density in the box, by
\begin{align}
E_{\Lambda}(\rmu) &:= \inf \Spec {\mathcal H}_{\Lambda}(\rmu), \\
e_{\Lambda}(\rmu) &:= \ell^{-3} \inf \Spec {\mathcal H}_{\Lambda}(\rmu) = \ell^{-3} E_{\Lambda}(\rmu),
\end{align}
with ${\mathcal H}_{\Lambda}(\rmu)$ as defined in \eqref{eq:Def_HB_new} above.
Then, if the parameter $b$ from \eqref{eq:DefT_tilde} is smaller than a universal constant, we have
\begin{align}
e_0(\rmu) \geq e_{\Lambda}(\rmu).
\end{align}
\end{theorem}

It is clear, using Theorem~\ref{thm:CompareBoxEnergy}, that Theorem~\ref{thm:LHY-Background} is a consequence of the following theorem on the box Hamiltonian.
Therefore, the remainder of the paper will be dedicated to the proof of Theorem~\ref{thm:LHY-Box} below.

\begin{theorem}\label{thm:LHY-Box}
Let $\kappa \in (0,\frac{1}{4}), {\mathcal C}>0$. Then there exists $C>0$ such that if
$v\geq 0$ 
is spherically symmetric of class $L^1({\mathbb R}^3)$ 
and satisfies \eqref{eq:assumps_thm_background} with $\eta_2= 5\eta_1$ with $\eta_1$ given by \eqref{eq:Choiceeta2},
and if the parameters are as given in \eqref{eq:Xparameter}-\eqref{eq:choiceM}
we have, for $\rmu a^3$ sufficiently small,
\begin{align}\label{eq:EnergyBoxRes1}
  e_\Lambda(\rmu) \geq -4\pi \rmu^2 a +  4\pi \rmu^2 a  \frac{128}{15\sqrt{\pi}}(\rmu a^3)^{\frac{1}{2}}
  -C\rmu^2 a (\rmu a^3)^{\frac{1}{2}}
X^{\frac{1}{5}}.
\end{align}
\end{theorem}

\begin{remark}
Our proof gives that $\eta_2$ can be chosen as $\eta_2=5 \eta_1$, with $\eta_1$ as in \eqref{eq:Choiceeta2},
and the bound on the error term $X^{\frac{1}{5}}$ in \eqref{eq:EnergyBoxRes1} is
\begin{align}
X^{\frac{1}{5}} = (\rmu a^3)^{\frac{\eta_2}{5}}. 
\end{align}
We believe that these powers are technical artefacts of our proof and leave room for improvement.
\end{remark}

\begin{proof}[Proof of Theorem~\ref{thm:CompareBoxEnergy}]
The proof of Theorem~\ref{thm:CompareBoxEnergy} follows by a sliding technique as in \cite{BS,BFS,FS}.
For completeness we sketch the argument here.

By a direct calculation of the integral, we get
\begin{align}\label{eq:Loc_pot}
  -\rmu \sum_{i=1}^N \int & g(x_i-y)\,dy + \sum_{i<j} v(x_i-x_j)  
 \nonumber \\
 &=
  \int_{\R^3} \Big[-\rmu \sum_{i=1}^N \int w_{1,u}(x_i,y)\,dy + \sum_{i<j} w_u(x_i,x_j)\Big] du.
\end{align}

We next consider the kinetic energy.
It follows from \cite[Lemma 6.4]{FS} that if the regularity parameter $M$ of $\chi$ satisfies $M \geq 5$, then for all $\ell>0$ and all sufficiently small values of $\varepsilon_T,d,s,b$, we have 
\begin{equation}\label{eq:LowerKin}
  \int_{\R^3} \widetilde{\mathcal T}_u  du\leq -\Delta,
\end{equation}
where $\widetilde{\mathcal T}_u$ is a translated version of the kinetic operator defined in \eqref{eq:DefT_tilde}.

Furthermore, it is a standard result that (in the form sense)
\begin{align}
\sum_{u \in {\mathbb Z}^3} - \Delta_u^{{\mathcal N}}  \leq - \Delta,
\end{align}
where $- \Delta_u^{{\mathcal N}}$ is the Neumann Laplace operator on $L^2(\Lambda_u)$.
Therefore, by averaging, 
\begin{align}
\int - \Delta_u^{{\mathcal N}} \,du = \int_{[0,1]^3}  \sum_{u \in {\mathbb Z}^3} - \Delta_{u+v}^{{\mathcal N}}\,dv \leq - \Delta.
\end{align}
Combining this with \eqref{eq:LowerKin} and \eqref{eq:Loc_pot}, we find
\begin{align}
{\mathcal H}_{\rmu,N}(\rmu) \geq
\int_{\ell^{-1}(\Omega + B(0,\ell/2))} ({\mathcal H}_{\Lambda,u}(\rmu))_{N} \,du  \geq \ell^{-3} | \Omega + B(0,\ell/2)| E_{\Lambda}(\rmu),
\end{align}
where ${\mathcal H}_{\Lambda,u}(\rmu)$ is the natural translated version of ${\mathcal H}_{\Lambda}(\rmu)$, and where the second inequality uses that $({\mathcal H}_{\Lambda,u}(\rmu))_{N}$ and $({\mathcal H}_{\Lambda,u'}(\rmu))_{N}$ are unitarily equivalent by \eqref{eq:transInv},
and the definition of $E_{\Lambda}(\rmu)$.
Now Theorem~\ref{eq:Def_HB_new} follows upon using that $|\Omega+B(0,\ell/2)|/|\Omega| \rightarrow 1$ in the thermodynamic limit.
\end{proof}

\subsection{Potential energy splitting}\label{sec:potsplit}
Using that $P+Q = \one_{\Lambda}$, we 
will in Lemma~\ref{lm:potsplit} below arrive at a very useful decomposition of the potential.

Define the (commuting) operators
\begin{align}
n_0 = \sum_{i=1}^N P_i,\qquad n_{+} = \sum_{i=1}^N Q_i,\qquad n=\sum_{i=1}^N\one_{\Lambda,i}=n_0+n_+ .
\end{align}
We furthermore define
\begin{align}\label{eq:Densities}
\rho_{+} := n_{+} \ell^{-3}, \qquad \rho_0:= n_{0} \ell^{-3}.
\end{align}
A crucial idea in this paper is to write the potential energy in the
form given in the next lemma, where the important observation is to
identify the positive term ${\mathcal Q}_4^{\rm ren}$ which can be dropped for a 
lower bound.
\begin{lemma}[Potential energy decomposition]\label{lm:potsplit}
We have
\begin{equation} \label{eq:potsplit}
-\rmu \sum_{i=1}^N \int w_1(x_i,y)\,dy+
\frac{1}{2} \sum_{i\neq j}  w(x_i, x_j)
= {\mathcal Q}_0^{\rm ren}+{\mathcal Q}_1^{\rm ren}
+{\mathcal Q}_2^{\rm ren}+{\mathcal Q}_3^{\rm ren} + {\mathcal Q}_4^{\rm ren},
\end{equation}
where
\begin{align}
{\mathcal Q}_4^{\rm ren}:=&\,
\frac{1}{2} \sum_{i\neq j} \Big[ Q_i Q_j + (P_i P_j + P_i Q_j + Q_i P_j)\omega(x_i-x_j) \Big] w(x_i,x_j) \nonumber \\
&\,\qquad \qquad \times
\Big[ Q_j Q_i + \omega(x_i-x_j) (P_j P_i + P_j Q_i + Q_j P_i)\Big],\label{eq:DefQ4}\\
{\mathcal Q}_3^{\rm ren}:=&\,
\sum_{i\neq j} P_i Q_j w_1(x_i,x_j) Q_j Q_i + h.c. \label{eq:DefQ3} \\
{\mathcal Q}_2^{\rm ren}:=&\, \sum_{i\neq j} P_i Q_j w_2(x_i,x_j) P_j Q_i
+ \sum_{i\neq j} P_i Q_j w_2(x_i,x_j) Q_j P_i \nonumber \\&\,- \rmu \sum_{i=1}^N Q_i \int w_1(x_i,y)\,dy Q_i 
+ \frac{1}{2}\sum_{i\neq j} (P_i P_j w_1(x_i,x_j) Q_j Q_i + h.c.),\label{eq:DefQ2}\\
{\mathcal Q}_1^{\rm ren}:=&\, \sum_{i,j}P_jQ_iw_2(x_i,x_j)P_iP_j-\rmu
  \sum_{i} Q_i \int w_1(x_i,y)\,dy P_i +h.c.
\label{eq:DefQ1} \\
{\mathcal Q}_0^{\rm ren}:=&\,
\frac{1}{2} \sum_{i\neq j} P_i P_j w_2(x_i,x_j) P_j P_i - \rmu \sum_i P_i \int w_1(x_i,y)\,dy P_i\label{eq:DefQ0}
\end{align}
\end{lemma}
\begin{proof}
The identity \eqref{eq:potsplit} follows using simple algebra and the
identitites \eqref{eq:DefU} and \eqref{eq:w2int}.  We simply write
$P_i+Q_i = 1_{\Lambda,i}$ for all $i$.  Inserting this identity in
both $i$ and $j$ on both sides of $w(x_i,x_j)$ and expanding yields
$16$ terms, which we have organized in a positive ${\mathcal Q}_4$
term and terms depending on the number of $Q$'s occuring.
\end{proof}
It will be useful to rewrite and estimate these terms as in the following lemma. 
\begin{lemma}\label{lem:QDecompSimpler}
If $v$ and hence $W_1$ are non-negative we have
\begin{align}
  {\mathcal Q}_0^{\rm ren}=&\,
  \frac{n_0(n_0-1)}{2|\Lambda|^2} \iint w_2(x,y)\,dx dy - \rmu \frac{n_0}{|\Lambda|}  \iint w_1(x,y)\,dx dy \nonumber\\
  =&\,\frac{n_0(n_0-1)}{2|\Lambda|} \Big(\widehat{g}(0) + \widehat{g\omega}(0)\Big)
  - \rmu n_0  \widehat{g}(0),\label{eq:Q0n0}
\end{align}
\begin{align}  
  {\mathcal Q}_1^{\rm ren}= &\, (n_0|\Lambda|^{-1} -\rmu) \sum_{i} Q_i \chi_{\Lambda}(x_i) W_1*\chi_{\Lambda}(x_i)  P_i + h.c. 
  \nonumber\\ 
  &\,+ n_0|\Lambda|^{-1}  \sum_{i} Q_i \chi_{\Lambda}(x_i)  (W_1\omega)*\chi_{\Lambda}(x_i) P_i + h.c.
  \label{eq:Q1n0}
\end{align}
and
\begin{align}    
      {\mathcal Q}_2^{\rm ren}\geq &\, 
      \sum_{i\neq j} P_i Q_j w_2(x_i,x_j) P_j Q_i
      + \frac{1}{2}\sum_{i\neq j} (P_i P_j w_1(x_i,x_j) Q_j Q_i + h.c.) \nonumber \\
      &\, + \Big((\rho_0 -\rmu) \widehat{W_1}(0) + \rho_0 \widehat{W_1 \omega}(0)\Big)\sum_{i} Q_i \chi_{\Lambda}(x_i)^2 Q_i 
      -C(\rmu+\rho_0)a (R/\ell)^2 n_+\, .\label{eq:Q2n0}
\end{align}
\end{lemma}
\begin{proof} The rewriting of ${\mathcal Q}_0$ is straightforward.
The rewriting of ${\mathcal Q}_1^{\rm ren}$ follows from
\begin{align}
  {\mathcal Q}_1^{\rm ren} = &\,
  \Big((n_0|\Lambda|^{-1} -\rmu) \sum_{i} Q_i \int w_1(x_i,y)\,dy P_i + h.c. \Big) \nonumber\\&\,+
  \Big(n_0|\Lambda|^{-1}  \sum_{i} Q_i \int w_1(x_i,y) \omega(x_i-y) \,dy P_i + h.c. \Big). \nonumber
\end{align}
We carry out the similar calculation on the part of the
$2Q$-term where $P$ acts in the same variable on both sides of the
potential,
\begin{align}
{\mathcal Q}_2^{\rm ren,1} &\,=\sum_{i\neq j} P_i Q_j w_2(x_i,x_j) P_j Q_i
+ \frac{1}{2}\sum_{i\neq j} (P_i P_j w_1(x_i,x_j) Q_j Q_i + h.c.) \nonumber \\
&\,\quad + (\rho_0 -\rmu) \sum_{i} Q_i \chi_{\Lambda}(x_i) W_1*\chi_{\Lambda}(x_i)  Q_i + 
\rho_0  \sum_{i} Q_i \chi_{\Lambda}(x_i)  (W_1\omega)*\chi_{\Lambda}(x_i) Q_i . \nonumber
\end{align}
At this point we invoke Lemma~\ref{lem:Simple} to get, for example,
\begin{align}
\sum_{i} Q_i \chi_{\Lambda}(x_i) W_1*\chi_{\Lambda}(x_i)  Q_i
\geq &\,
\Bigl(\int W_1\Bigr) \sum_{i} Q_i \chi_{\Lambda}(x_i)^2 Q_i\nonumber\\
&\, - \max_{i,j}\| \partial_i\partial_j \chi \|_{\infty} (R/\ell)^2 \Bigl(\int W_1\Bigr) \| \chi \|_{\infty} n_{+}.
\end{align}
\end{proof}

\section{A priori bounds on particle number and excited particles}\label{sec:Apriori}

In this section we will give some important a priori bounds on the
particle number $n$, the number of excited particles $n_+$, as well as
on the potential energy term ${\mathcal Q}_4^{\rm ren}$. The bounds on $n$ and $n_+$
essentially say that for states with sufficiently low energy $n$ is
close to what one would expect, i.e., $\rmu\ell^3$ and the expectation
of $n_+$ is smaller with a factor which is not much worse than the
relative LHY error. 
The bounds on $n$ and $n_{+}$ were given in \cite[Theorem 7.1]{FS} for a similar operator.
To prove these bounds requires the localization to smaller boxes---the double localization mentioned above---and to facilitate this our localization function $\chi$ has only finite smoothness $M$.

One of the main new contributions in this article compared to \cite{FS} is a bound on the momentum-localized excitations $n_+^L$ to be defined in \eqref{eq:def_nplusL} below. In order to achieve this bound, we need the improved bound on ${\mathcal Q}_4^{\rm ren}$ given in \eqref{eq:Q4apriori_2} 
below.

\begin{theorem}[{A priori bounds}]\label{thm:aprioribounds}
Assume that the conditions of Assumption~\ref{assump:params} 
are satisfied and that $\rmu a^3$ is small enough. 
Then there is a universal constant $C>0$
such that if $\Psi \in {\mathcal F}_s(L^2(\Lambda))$ is an $n$-particle normalized state in the
bosonic Fock space over $L^2(\Lambda)$ satisfying
\begin{align}\label{eq:aprioriPsi}
  \langle \Psi, {\mathcal H}_{\Lambda}(\rmu) \Psi \rangle \leq
  - 4\pi \rmu^2 a \ell^3(1-K_B^3 (\rmu a^3)^{\frac{1}{2}})
\end{align}
then 
\begin{align}\label{eq:apriori_n}
|n\ell^{-3}-\rmu| &\leq  C \rmu K_B^{3/2}K_\ell(\rmu a^3)^{1/4},\\
\label{eq:apriorinn+NEW}
\langle\Psi,n_+\Psi\rangle &\leq C \rmu\ell^3K_B^3K_\ell^2(\rmu a^3)^{1/2}, \\
\intertext{and}
\label{eq:Q4apriori_2}
  0\leq \langle \Psi, {\mathcal Q}_4^{\rm ren}\Psi\rangle&\leq C\rmu^2 a \ell^3 K_B^3 K_{\ell}^2 (\rmu a^3)^{1/2}.
\end{align}
\end{theorem}

The proof of Theorem~\ref{thm:aprioribounds} is given in Appendix~\ref{SmallBoxes}.

There will be important regions in momentum space corresponding to low and high momenta.

We first define the `low' and `high' momentum regions as follows\footnote{The ${\mathcal P}_{\rm high}$ region was also used in \cite{FS}, but denoted $P_H$ and the parametrization was different. Therefore we here use a `doubleprime' on $K_H''$ to distinguish it from the parameter $K_H$ used in that paper.
The notation $K_H'$ will be used for another related constant, see \eqref{eq:Def_QL'}.
}
\begin{align}\label{eq:Momomta}
{\mathcal P}_{\rm low} := \{ |p| \leq d^{-2} \ell^{-1} \} \subset {\mathbb R}^3,\qquad {\mathcal P}_{\rm high}  := \{|p| \geq K_H'' \ell^{-1}\} \subset {\mathbb R}^3.
\end{align}
where $d\ll1$ and $K_H''\gg 1$  were defined in Section~\ref{sec:params}.
We will always assume that \eqref{cond:disjoint} is satisfied.
This assures that ${\mathcal P}_{\rm low}$ and ${\mathcal P}_{\rm high}$ are disjoint.
The choice of ${\mathcal P}_{\rm low}$ here is in agreement with a term in the kinetic energy, which will assure that we have a good bound on the number of particles with momentum outside ${\mathcal P}_{\rm low}$.

We will define the low momentum localization operator $Q_L$ as follows.
Let $f \in C^{\infty}({\mathbb R})$ be a monotone non-increasing function satisfying that $f(t) = 1 $ for $t\leq 1$ and $f(t) = 0$ for $t\geq 2$.
We further define
\begin{align}\label{eq:def_fL}
f_L(t) := f(d^2 \ell t).
\end{align}
I.e.  $f_L$ is a smooth localization to the low momenta ${\mathcal P}_{\rm low}$. With this notation, we define
\begin{align}\label{eq:Ql-bar}
Q_L := Q f_L(\sqrt{-\Delta}), \qquad \overline{Q_L} := Q (1- f_L(\sqrt{-\Delta})).
\end{align}
Notice that $Q_L$ is not self-adjoint.

We define
\begin{align}\label{eq:nplusH}
n_{+}^{H} := \sum Q \one_{(d^{-2} \ell^{-1}, \infty)}(\sqrt{-\Delta}) Q.
\end{align}
With this definition, we have
\begin{align}\label{eq:QLs_giveNplusH}
\sum \overline{Q}_L (\overline{Q}_L)^{*} \leq n_{+}^{H}.
\end{align}
Notice that $n_{+}^{H}$ can be thought of as counting the number of particles outside the condensate and with momentum outside ${\mathcal P}_{\rm low}$---{\it not} the number of excited particles with momentum in ${\mathcal P}_{\rm high}$ as one might be led to think.

We will also need the splitting corresponding intuitively to particles with nonzero momentum in or outside ${\mathcal P}_{\rm high}$. However, for technical reasons (we would like the corresponding projections to commute with the kinetic energy), we define these in terms of the kinetic energy.
Therefore, we define for  $K_H' >1$, 
\begin{align}\label{eq:Def_QL'}
Q_L'  = \one_{(0, (K_H' \ell^{-1})^2)}( {\mathcal T}),\qquad Q_H'  = \one_{[(K_H' \ell^{-1})^2, \infty)}( {\mathcal T}), 
\end{align}
with the kinetic energy operator ${\mathcal T}$ as defined in \eqref{eq:DefT_new}.
In particular,
\begin{align}
Q = Q_L' + Q_H',  \qquad
1_{\Lambda} = P + Q_L' + Q_H'.
\end{align}
Since the kinetic energy is more complicated than the standard Laplacian, we cannot choose $K_H'=K_H''$, but
need them to satisfy \eqref{cond:disjoint}.

We define also
\begin{align}\label{eq:def_nplusL}
\widetilde{n}_{+}^{\rm H} := \sum_{j} (Q_H')_j, \qquad n_{+}^{\rm L} := \sum_{j} (Q_L')_j.
\end{align}
We warn the reader that with these definitions, we have $n_{+}^{\rm L} + n_{+}^{\rm H} \neq n_{+}$.

It is one of the main achievements of the present paper to establish good bounds on moments on $n_{+}^{\rm L}$---even for potentials with large $L^1$-norm.

In the remainder of this section we will establish two crucial technical results. One of these, Lemma~\ref{lem:pseudolocal} below, compares $Q_L' $ (which is defined in terms of a spectral projection of the kinetic energy) with a localization to comparable momenta.
For the localization of large matrices result, Proposition~\ref{prop:LocMatrices}, it is convenient to consider a projection commuting with the kinetic energy. However, for later calculations, a momentum localization is better. Thus the need for Lemma~\ref{lem:pseudolocal} that compares the two in \eqref{eq:MomentComp}.

We also establish Lemma~\ref{lem:Estimate_dLs} below. This lemma gives a bound on the parts of the Hamiltonian that change $n_{+}^{\rm L}$ and is the key input to the localization of large matrices; Proposition~\ref{prop:LocMatrices}.
To prove Lemma~\ref{lem:Estimate_dLs}, we need Lemma~\ref{lem:LowMomentaOperator}.
The main technical input to the proof of Lemma~\ref{lem:LowMomentaOperator} has been placed in Corollary~\ref{cor:Derivatives} and Lemma~\ref{lem:SobolevHom} in Appendix~\ref{sec:proofs}.

\begin{lemma}\label{lem:pseudolocal}
Suppose that the parameters satisfy Assumption~\ref{assump:params}.
Then, we have
\begin{align}\label{eq:MomentComp}
Q \chi_{\Lambda} \one_{\{|p| \leq K_H'' \ell^{-1}\}} \chi_{\Lambda}  Q \leq C Q_L' + C \left( \left( \frac{K_H''}{K_H'} \right)^{M} +  \varepsilon_N^{\frac{3}{2}}\right) Q.
\end{align}

Furthermore, we have a similar estimate without $\chi_{\Lambda}$'s:
\begin{align}
Q  \one_{\{|p| \leq K_H'' \ell^{-1}\}} Q \leq C Q_L' + C \left( \left( \frac{K_H''}{K_H'} \right)^{M} +  \varepsilon_N^{\frac{3}{2}}\right) Q,
\end{align}
\end{lemma}

The proof of Lemma~\ref{lem:pseudolocal} is a bit long and technical and is given in Appendix~\ref{sec:proofs}.
We next state and prove Lemma~\ref{lem:LowMomentaOperator}.

\begin{lemma}\label{lem:LowMomentaOperator}
Assume that the conditions of Assumption~\ref{assump:params} 
are satisfied and that $\rmu a^3$ is small enough. 
Then,
\begin{align}\label{eq:OperatorNorm}
\| (Q_L' \otimes I) w(x_j,x_{j'}) (Q_L' \otimes I) \| \leq C \err (K_H')^3 \ell^{-3} (\int v),
\end{align}
where
\begin{equation}\label{eq:Def_R1}
\err := 
\left(1 + \varepsilon_N^{-\frac{1}{2}} \ell^{-1} (a+ R{\mathcal R}^{-1}) \right) 
\varepsilon_N^{-\frac{1}{2}}.
\end{equation}
\end{lemma}

\begin{proof}
Notice first that ${\mathcal T}$ commutes with $Q$ and that $\Ran Q_L' \subseteq \Ran Q$.

Let $\varphi \in \Ran Q_L'$ and be normalized in $L^2$.
We will use \eqref{eq:Sob-box} on $\varphi$ and \eqref{eq:Sob-fullspace} on $\chi_{\Lambda} \varphi$, using Corollary~\ref{cor:Derivatives} to control the derivatives.

We estimate
\begin{align}
\langle \varphi, (Q_L')_j w(x_j,x_{j'}) (Q_L')_j \varphi \rangle_{L^2(\Lambda_j)} \leq
I_1 + I_2,
\end{align}
where
\begin{align}
I_1&:= \int \chi_{\Lambda}(x_j)^2 |\varphi(x_j)|^2 v(x_j - x_{j'}) \,dx_j, \nonumber \\
I_2&:= \int \chi_{\Lambda}(x_j) | \chi_{\Lambda}(x_j) - \chi_{\Lambda}(x_{j'})| \, |\varphi(x_j)|^2 v(x_j - x_{j'}) \,dx_j.
\end{align}
To estimate $I_1$, we use \eqref{eq:Sob-fullspace}  and \eqref{eq:DerivsChiTerms} 
to get
\begin{align}
I_1 \leq (\int v)  \| \chi_{\Lambda} \varphi \|_{\infty}^2 \leq C (\int v )
\| \nabla (\chi_{\Lambda} \varphi) \|
\| (-\Delta^{\mathcal N}) (\chi_{\Lambda} \varphi) \|
\leq C \ell^{-3} (\int v) \varepsilon_N^{-1/2} (K_H')^3 .
\end{align}
This is in agreement with \eqref{eq:OperatorNorm}.

The estimate on $I_2$ uses the same main ideas. 
Start by writing 
\begin{align}
v = g + \omega v.
\end{align}
Then notice that
\begin{align}
| \chi_{\Lambda}(x_j) - \chi_{\Lambda}(x_{j'})|  v(x_j - x_{j'}) \leq C \left( \frac{R}{\ell} g(x_j - x_{j'}) + \frac{a}{\ell} v(x_j - x_{j'}) \right),
\end{align}
by Taylor's formula and the support properties of $v$ as well as \eqref{omegabounds}.
Therefore, we can estimate
\begin{align}
I_2 \leq C \left( \frac{R}{\ell} (\int g )+ \frac{a}{\ell} (\int  v ) \right) \| \chi_{\Lambda} \varphi \|_{\infty}  \| \varphi \|_{\infty} 
\leq C \frac{a}{\ell} \left( R + (\int  v ) \right) \| \chi_{\Lambda} \varphi \|_{\infty}  \| \varphi \|_{\infty} .
\end{align}
Therefore, this term also can be estimated in agreement with \eqref{eq:OperatorNorm}
upon noticing that the estimates in \eqref{eq:DerivsPhiTerms} contain an extra factor of $\varepsilon_N^{-1/2}$ compared to those in \eqref{eq:DerivsChiTerms}.
\end{proof}

We aim to get a good bound on moments of $n_{+}^{\rm L}$ by the method of localization of large matrices.
To this end define $d_1^L$ and $d_2^L$ as the parts of the Hamiltonian do not commute with $n_{+}^{\rm L}$, more precisely, they respectively change $n_{+}^{\rm L}$ by $\pm1$ and $\pm2$.
They are defined as
\begin{align}\label{eq:def_d2L}
d_2^L &:= \sum_{i \neq j} (P_i + Q_{H,i}') (P_j + (Q_H')_j) w(x_i,x_j) (Q_L')_j (Q_L')_i + h.c. \nonumber \\
&= d_{2,1}^L + d_{2,2}^L + d_{2,3}^L,
\end{align}
with
\begin{align}
d_{2,1}^L &:= \sum_{i \neq j} P_i P_j  w(x_i,x_j) (Q_L')_j (Q_L')_i + h.c. , \nonumber \\
d_{2,2}^L &:= \sum_{i \neq j} (P_i (Q_H')_j +  (Q_H')_i P_j) w(x_i,x_j) (Q_L')_j (Q_L')_i + h.c. , \nonumber \\
d_{2,3}^L &:= \sum_{i \neq j} (Q_H')_i (Q_H')_j w(x_i,x_j) (Q_L')_j (Q_L')_i + h.c.,   \nonumber
\end{align}
and
\begin{align}\label{eq:def_d1L}
d_1^L &:= - \rmu \sum_i (P_i + Q_{H,i}') \int w_1(x_i,y) \,dy Q_{L,i}' + h.c. \nonumber \\
&\quad+ \sum_{i\neq j} \big( (P_i + Q_{H,i}') Q_{L,j}' w(x_i,x_j) Q_{L,j}'Q_{L,i}' + h.c. \big) \nonumber \\
& \quad+  \sum_{i\neq j} \big( Q_{L,i}' (P_j + Q_{H,j}') w(x_i,x_j)  (P_j + Q_{H,j}')  (P_i + Q_{H,i}')+ h.c. \big) \nonumber \\
&= \sum_{j=1}^{10} d_{1,j}^L,
\end{align}
with
\begin{align}
d_{1,1}^L &:= - \rmu \sum_i P_i  \int w_1(x_i,y) \,dy Q_{L,i}' + h.c. \nonumber \\
d_{1,2}^L &:= - \rmu \sum_i Q_{H,i}' \int w_1(x_i,y) \,dy Q_{L,i}' + h.c. \nonumber \\
d_{1,3}^L &:= \sum_{i\neq j}  P_i Q_{L,j}' w(x_i,x_j) Q_{L,j}'Q_{L,i}' + h.c.  \nonumber \\
d_{1,4}^L &:=\sum_{i\neq j}  Q_{H,i}' Q_{L,j}' w(x_i,x_j) Q_{L,j}'Q_{L,i}' + h.c. \nonumber\\
d_{1,5}^L &:=\sum_{i\neq j}  Q_{L,i}' P_j w(x_i,x_j) P_j P_i+ h.c. \nonumber  \\
d_{1,6}^L &:=\sum_{i\neq j}   Q_{L,i}' Q_{H,j}'  w(x_i,x_j) P_j P_i+ h.c.  \nonumber \\
d_{1,7}^L &:=\sum_{i\neq j}  Q_{L,i}' P_j w(x_i,x_j) \left(Q_{H,j}' P_i+ P_j Q_{H,i}'\right)+ h.c.   \nonumber\\
d_{1,8}^L &:= \sum_{i \neq j}  Q_{L,i}' Q_{H,j}' w(x_i,x_j) \left( P_j Q_{H,i}' + Q_{H,j}' P_i\right) + h.c   \nonumber\\
d_{1,9}^L &:=\sum_{i\neq j}  Q_{L,i}' P_j w(x_i,x_j) Q_{H,j}' Q_{H,i}'+ h.c.  \nonumber \\
d_{1,10}^L &:= \sum_{i \neq j}  Q_{L,i}' Q_{H,j}' w(x_i,x_j) Q_{H,j}' Q_{H,i}'  + h.c.  \nonumber
\end{align}

The following Lemma~\ref{lem:Estimate_dLs} controls the magnitude of $d_1^L$ and $d_2^L$ in a state with low energy, i.e. the magnitude of the non-diagonal part of the Hamiltonian with respect to $n_{+}^{\rm L}$.
This is a crucial input for the localization of large matrices.

\begin{lemma}\label{lem:Estimate_dLs}
Let $\Psi \in {\mathcal F}_s(L^2(\Lambda))$ be an $n$-particle normalized state in the
bosonic Fock space over $L^2(\Lambda)$ satisfying \eqref{eq:aprioriPsi}.
Assume that the conditions of Assumption~\ref{assump:params} 
are satisfied and that $\rmu a^3$ is small enough. 
Suppose that, for some $\widetilde{{\mathcal M}} \geq 0$, we have $\Psi = 1_{[0, \widetilde{{\mathcal M}} ]}(n_{+}^{\rm L} ) \Psi$.
Then, with $\err$ as defined in \eqref{eq:Def_R1},
\begin{align}\label{eq:GoodEstimatedSimple}
&|\langle \Psi, d_1^L \Psi \rangle | + |\langle \Psi, d_2^L \Psi \rangle | \leq \nonumber \\
&
C a \rmu^2 \ell^3  (\int v/a) \Big\{ \frac{ \langle \Psi, n_{+} \Psi \rangle^{1/2} }{n^{1/2}}+ \Big(\frac{ \err (K_H')^3  \widetilde{{\mathcal M}}}{n} \Big)^{1/2} \frac{ \langle \Psi, n_{+} \Psi \rangle^{1/2} }{n^{1/2}}
+ \frac{\err (K_H')^3 \widetilde{{\mathcal M}}}{n} \frac{ \langle \Psi, n_{+} \Psi \rangle }{n}
\Big\} \nonumber \\
&\quad +C  \langle \Psi, {\mathcal Q}_4^{\rm ren}\Psi\rangle.
\end{align}
\end{lemma}

\begin{proof}
We will prove that
\begin{align}\label{eq:GoodEstimated1d2-v2}
|\langle \Psi, d_1^L \Psi \rangle &| + |\langle \Psi, d_2^L \Psi \rangle | 
\nonumber \\
& \leq C a \rmu^2 \ell^3  (\int v/a) \Big\{ n^{-1/2} \langle \Psi, n_{+} \Psi \rangle^{1/2} +  [\err (K_H')^3 ]^{1/2} n^{-1} \langle \Psi, (n_{+}^L)^2 \Psi \rangle^{1/2} \nonumber \\
&\quad\quad\quad\quad\quad\quad\quad\quad  + [\err (K_H')^3 ] n^{-2} \langle \Psi,  n_{+}^{\rm L}  \widetilde{n}_{+}^{\rm H} \Psi \rangle^{1/2} \langle \Psi,  (n_{+}^{\rm L} )^2 \Psi \rangle^{1/2} \nonumber \\
&\quad\quad\quad\quad\quad\quad\quad\quad+ [\err (K_H')^3 ]^{1/2} n^{-1} \langle \Psi,  n_{+}^{\rm L}  \widetilde{n}_{+}^{\rm H} \Psi \rangle^{1/2} 
+ [\err (K_H')^3 ] n^{-2}  \langle \Psi,  (n_{+}^{\rm L} )^2 \Psi\rangle\nonumber \\
&\quad\quad\quad\quad\quad\quad\quad\quad+ [\err (K_H')^3 ] n^{-2} \langle \Psi,  n_{+}^{\rm L}  \widetilde{n}_{+}^{\rm H} \Psi \rangle  \Big\}
+C  \langle \Psi, {\mathcal Q}_4^{\rm ren}\Psi\rangle .
\end{align}
The estimate \eqref{eq:GoodEstimatedSimple} is a simple consequence of \eqref{eq:GoodEstimated1d2-v2} upon estimating $n_{+}^L \leq \widetilde{{\mathcal M}}$ in occurences with higher than first moments.

We estimate each term in $d_1^L$ and $d_2^L$ individually. 
They will all be estimated by fairly simple Cauchy-Schwarz inequalities, however the proof is long because there are many terms.
Let us indicate where some of the bounds in \eqref{eq:GoodEstimated1d2-v2} come from.
The error bound in terms of ${\mathcal Q}_4^{\rm ren}$ comes from the analysis of the terms $d_{1,9}^L$, $d_{1,10}^L$ and $d_{2,3}^L$.
Next we list the $6$ terms in $\{ \cdot \}$ in \eqref{eq:GoodEstimated1d2-v2}, we only list the terms where we believe that our inequalities cannot be significantly improved.
The first term in \eqref{eq:GoodEstimated1d2-v2} comes from $d_{1,5}^L$, the second from $d_{2,1}^L$, the third term comes from $d_{1,4}^L$, the fourth term comes from $d_{1,6}^L$, the fifth from $d_{2,3}^L$ and the sixth from $d_{1,10}^L$.

We start with $d_{1,5}^L$.
Using Cauchy-Schwarz, we find for arbitrary $\varepsilon >0$,
\begin{align}\label{eq:A_d_term}
|\langle \Psi, d_{1,5}^L \Psi \rangle | \leq  n \ell^{-3} (\int v )
\left( \varepsilon^{-1} \langle \Psi,  n_{+}^{\rm L}  \Psi \rangle + \varepsilon n \right).
\end{align}
With $\varepsilon = \sqrt{ n^{-1}\langle \Psi,  n_{+}^{\rm L}  \Psi \rangle}$ we get an estimate in agreement with \eqref{eq:GoodEstimated1d2-v2}.

To estimate  $d_{1,1}^L$ we use Cauchy-Schwarz in a similar manner and 
the fact that $0\leq \int W_1 \leq Ca$, to get
\begin{align}
|\langle \Psi, d_{1,1}^L \Psi \rangle | \leq C n \ell^{-3} (\int g )
\left( \varepsilon^{-1}  \langle \Psi,  n_{+}^{\rm L}  \Psi \rangle + \varepsilon n \right).
\end{align}
Choosing $\varepsilon = \sqrt{ n^{-1} \langle \Psi,  n_{+}^{\rm L}  \Psi \rangle}$ , and using that $\int v/a \geq 1$, this term is consistent with the first term in  \eqref{eq:GoodEstimated1d2-v2}.

Applying Cauchy-Schwarz to the $d_{1,2}^L$-term, we find
\begin{align}
|\langle \Psi, d_{1,2}^L \Psi \rangle | \leq C \rmu (\int g )
\left( \varepsilon^{-1} \langle \Psi,  n_{+}^{\rm L}  \Psi \rangle + \varepsilon  \langle \Psi,  \widetilde{n}_{+}^{\rm H}  \Psi \rangle \right).
\end{align}
Since $ \widetilde{n}_{+}^{\rm H} \leq n$ and $(\int g) \leq (\int v)$, this term is also consistent with \eqref{eq:GoodEstimated1d2-v2}.

After application of Cauchy-Schwarz to the $d_{1,3}^L$-term, we find after using Lemma~\ref{lem:LowMomentaOperator} on one of the resulting terms,
\begin{align}
|\langle \Psi, d_{1,3}^L \Psi \rangle | \leq n \ell^{-3} (\int v )
\left( \varepsilon^{-1} \langle \Psi,  n_{+}^{\rm L}  \Psi \rangle + \varepsilon  \err (K_H')^3 n^{-1}\langle \Psi,  (n_{+}^{\rm L} )^2 \Psi \rangle \right).
\end{align}
After using the (suboptimal) estimate $n_{+}^{\rm L} \leq n$ on the first term, we can choose $\varepsilon^{-1} = (\err (K_H')^3)^{1/2} n^{-1}\langle \Psi,  (n_{+}^{\rm L} )^2 \Psi \rangle^{1/2}$ to get an estimate consistent with the second term in \eqref{eq:GoodEstimated1d2-v2}.

We also apply Cauchy-Schwarz to the $d_{1,4}^L$-term, and use Lemma~\ref{lem:LowMomentaOperator} on both resulting terms.
This yields,
\begin{align}
|\langle \Psi, d_{1,4}^L \Psi \rangle | \leq C \err (K_H')^3 \ell^{-3} (\int v )
\left( \varepsilon^{-1} \langle \Psi,  n_{+}^{\rm L}  \widetilde{n}_{+}^{\rm H} \Psi \rangle + \varepsilon  \langle \Psi,  (n_{+}^{\rm L} )^2 \Psi \rangle \right).
\end{align}
Upon optimizing in $\varepsilon$, we find the third term in \eqref{eq:GoodEstimated1d2-v2}.

After applying Cauchy-Schwarz and Lemma~\ref{lem:LowMomentaOperator} to the $d_{1,6}^L$-term, we find 
\begin{align}
|\langle \Psi, d_{1,6}^L \Psi \rangle | \leq C \ell^{-3} (\int v )
\left( \varepsilon^{-1}  \err (K_H')^3  \langle \Psi,  n_{+}^{\rm L}  \widetilde{n}_{+}^{\rm H} \Psi \rangle + \varepsilon  n^2  \right).
\end{align}
Optimizing in $\varepsilon$ leads to the fourth term in \eqref{eq:GoodEstimated1d2-v2}.

The term with $d_{1,7}^L$ is easily estimated by Cauchy-Schwarz as
\begin{align}
|\langle \Psi, d_{1,7}^L \Psi \rangle | \leq C \ell^{-3} (\int v ) n \langle \Psi, n_{+} \Psi \rangle.
\end{align}
Since $n_{+} \leq n$ this term in agreement with the first term in \eqref{eq:GoodEstimated1d2-v2}.

To estimate $d_{1,8}^L$ we use Cauchy-Schwarz and Lemma~\ref{lem:LowMomentaOperator} and get (where the second inequality is clearly suboptimal),
\begin{align}
|\langle \Psi, d_{1,8}^L \Psi \rangle | &\leq C \ell^{-3} (\int v )
\left( \varepsilon^{-1}  \err (K_H')^3  \langle \Psi,  n_{+}^{\rm L}  \widetilde{n}_{+}^{\rm H} \Psi \rangle + \varepsilon  n   \langle \Psi,  \widetilde{n}_{+}^{\rm H} \Psi \rangle \right) \nonumber \\
&\leq C \ell^{-3} (\int v )
\left( \varepsilon^{-1}  \err (K_H')^3  \langle \Psi,  n_{+}^{\rm L}  \widetilde{n}_{+}^{\rm H} \Psi \rangle + \varepsilon  n^2    \right),
\end{align}
and we recognize that this term can be estimated in the same fashion as $ d_{1,6}^L$ therefore contributing to the fourth term in \eqref{eq:GoodEstimated1d2-v2}.

The estimates of $d_{1,9}^L$ and $d_{1,10}^L$ are a bit different since there are two factors of $Q_H'$ on the same side. We will complete the $Q_H'$'s to $Q = Q_H' + Q_L'$ and then complete to $QQ+ \omega(PP + PQ + QP)$ before using Cauchy-Schwarz. Thereby, we will be able to use the bound \eqref{eq:Q4apriori_2} on ${\mathcal Q}_4^{\rm ren}$. The new terms appearing through this completion procedure can be estimated in the same manner as the previous $d_{1,j}^L$'s.

Let ${\mathcal E}$ denote the error bound in \eqref{eq:GoodEstimated1d2-v2}. We get
\begin{align}
|\langle \Psi, d_{1,9}^L \Psi \rangle | &\leq | \langle \Psi, \sum_{i\neq j}  (Q_{L,i}' P_j w(x_i,x_j) Q_j Q_i+ h.c. ) \Psi \rangle| \\
&\quad+| \langle \Psi, \sum_{i\neq j}  (Q_{L,i}' P_j w(x_i,x_j) (Q_{H,j}' Q_{L,i}'+ Q_{L,j}'Q_{H,i}') + h.c. ) \Psi \rangle| \nonumber
+|\langle \Psi, d_{1,3}^L \Psi \rangle |
\end{align}
Here the second term on the right can be estimated, after an application of Cauchy-Schwarz, in the same manner as $d_{1,6}^L$.
Continuing the procedure, we therefore find,
\begin{align}
|\langle \Psi, d_{1,9}^L \Psi \rangle | &\leq  | \langle \Psi, \sum_{i\neq j}  (Q_{L,i}' P_j w(x_i,x_j) (Q_j Q_i + \omega(P_j Q_i + Q_j P_i + P_j P_i))  + h.c. ) \Psi \rangle| \nonumber \\
&\quad + | \langle \Psi, \sum_{i\neq j}  (Q_{L,i}' P_j w(x_i,x_j) \omega (P_j Q_i + Q_j P_i ) + h.c. ) \Psi \rangle| \nonumber \\
&\quad + | \langle \Psi, \sum_{i\neq j}  (Q_{L,i}' P_j w(x_i,x_j) \omega P_j P_i+ h.c. ) \Psi \rangle| + {\mathcal E}.
\end{align}
Here the third term is estimated as $d_{1,5}^L$ using that $0\leq \omega \leq 1$. Similarly, the second term is estimated as $d_{1,7}^L$. 
Finally, the first term. After using Cauchy-Schwarz, this is controlled by
\begin{align}
 \langle \Psi, {\mathcal Q}_4^{\rm ren}\Psi\rangle + C n \ell^{-3} (\int v) \langle \Psi, n_{+} \Psi \rangle ,
\end{align}
which is in agreement with  \eqref{eq:GoodEstimated1d2-v2}.

The analysis of $d_{1,10}^L$ is similar. We write
\begin{align}
|\langle \Psi, d_{1,10}^L \Psi \rangle | &\leq | \langle \Psi, \sum_{i\neq j}  (Q_{L,i}' Q_{H,j}'  w(x_i,x_j) Q_j Q_i+ h.c. ) \Psi \rangle|  \\
&\quad+| \langle \Psi, \sum_{i\neq j}  (Q_{L,i}' Q_{H,j}'  w(x_i,x_j) (Q_{H,j}' Q_{L,i}'+ Q_{L,j}'Q_{H,i}') + h.c. ) \Psi \rangle|
+|\langle \Psi, d_{1,4}^L \Psi \rangle | \nonumber
\end{align}
Here the second term is easily estimated by the sixth term in \eqref{eq:GoodEstimated1d2-v2}. Therefore, we continue the estimate as
\begin{align}
|\langle \Psi, d_{1,10}^L \Psi \rangle | &\leq | \langle \Psi, \sum_{i\neq j}  \left(Q_{L,i}' Q_{H,j}'  w(x_i,x_j) (Q_j Q_i + \omega(P_j Q_i + Q_j P_i + P_j P_i)  + h.c. \right) \Psi \rangle| \nonumber \\
& \quad+ | \langle \Psi, \sum_{i\neq j}  (Q_{L,i}' Q_{H,j}'  w(x_i,x_j) \omega (P_j Q_i + Q_j P_i ) + h.c. ) \Psi \rangle| \nonumber \\
&\quad  + | \langle \Psi, \sum_{i\neq j}  (Q_{L,i}' Q_{H,j}'  w(x_i,x_j) \omega P_j P_i+ h.c. ) \Psi \rangle| + {\mathcal E}.
\end{align}
The third term can be estimated as $d_{1,6}^L$ using that $0\leq \omega \leq 1$. The second term can similarly be estimated as $d_{1,8}^L$.
Applying Cauchy-Schwarz to the first term, we therefore find
\begin{align}
|\langle \Psi, d_{1,10}^L \Psi \rangle | &\leq  \langle \Psi, {\mathcal Q}_4^{\rm ren}\Psi\rangle + C \err (K_H')^3 \ell^{-3} (\int v)  \langle \Psi, n_{+}^L  \widetilde{n}_{+}^H \Psi\rangle + {\mathcal E}.
\end{align}
So we conclude that also  $d_{1,10}^L$ can be estimated in agreement with \eqref{eq:GoodEstimated1d2-v2}.

Using Cauchy-Schwarz, we find for arbitrary $\varepsilon >0$, and using Lemma~\ref{lem:LowMomentaOperator} in the last inequality,
\begin{align}
| \langle \Psi, d_{2,1}^L \Psi \rangle| &\leq \sum_{i \neq j} \left\langle \Psi,  \left(\varepsilon  P_i P_j  w(x_i,x_j) P_j P_i + \varepsilon^{-1}(Q_L')_i  (Q_L')_j w(x_i,x_j) (Q_L')_j (Q_L')_i\right) \Psi \right\rangle  \nonumber \\
&\leq C\varepsilon n^2 \ell^{-3} \int v + C \varepsilon^{-1} \langle \Psi, (n_{+}^L)^2 \Psi \rangle \err (K_H')^3\ell^{-3}  \int v,
\end{align}
Making the choice  $\varepsilon = n^{-1} [ \err (K_H')^3]^{1/2} \langle \Psi, (n_{+}^L)^2 \Psi \rangle^{1/2} $ this estimate  corresponds to the second term in \eqref{eq:GoodEstimated1d2-v2}.

We estimate $d_{2,2}^L$ similarly by the Cauchy-Schwarz inequality. 
\begin{align}
\pm d_{2,2}^L  &\leq \varepsilon \sum_{i \neq j} (Q_H')_i  P_j  w(x_i,x_j) P_j (Q_H')_i  + \varepsilon^{-1} \sum_{i \neq j} (Q_L')_i  (Q_L')_j w(x_i,x_j) (Q_L')_j (Q_L')_i  .
\end{align}
The first sum on the right is controlled by $n \ell^{-3} n_{+}^H \int v \leq n^2 \ell^{-3}  \int v$ (using the suboptimal estimate $n_{+}^H \leq n$). Therefore, the $d_{2,2}^L$-term is controlled by the same expression as for $d_{2,1}^L$.

We next estimate $d_{2,3}^L$. We write
\begin{align}
d_{2,3}^L &= \sum_{i \neq j} Q_i Q_j w(x_i,x_j) (Q_L')_j (Q_L')_i + h.c. \nonumber \\
&\quad- \sum_{i \neq j} (Q_L')_i (Q_L')_j w(x_i,x_j) (Q_L')_j (Q_L')_i + h.c.  \nonumber \\
&\quad- 2 d_{1,4}^L 
\end{align}
The second term here can be estimated as $C \ell^{-3} \err (K_H')^3 (\int v) (n_+^L)^2$. This is the fifth term in \eqref{eq:GoodEstimated1d2-v2}.
Therefore, we can continue the estimate as follows, where ${\mathcal E}$ denote the error bound in \eqref{eq:GoodEstimated1d2-v2},
\begin{align}
|\langle \Psi, d_{2,3}^L \Psi \rangle | &\leq | \langle \Psi, \sum_{i\neq j}  \left(Q_{L,i}' Q_{L,j}' w(x_i,x_j) [Q_j Q_i + \omega(P_j Q_i + Q_j P_i + P_j P_i)]  + h.c. \right) \Psi \rangle| \nonumber \\
& + | \langle \Psi, \sum_{i\neq j}  (Q_{L,i}' Q_{L,j}'w(x_i,x_j) \omega (P_j Q_i + Q_j P_i ) + h.c. ) \Psi \rangle| \nonumber \\
& + | \langle \Psi, \sum_{i\neq j}  (Q_{L,i}' Q_{L,j}'w(x_i,x_j) \omega P_j P_i+ h.c. ) \Psi \rangle| + {\mathcal E}.
\end{align}
Here the third term can be estimated in the same way as $d_{2,1}^L$ and the second as $d_{1,3}^L$, using in both cases that $0\leq \omega \leq 1$. So we find after applying Cauchy-Schwarz to the first term
\begin{align}
|\langle \Psi, d_{2,3}^L \Psi \rangle | &\leq  \langle \Psi, {\mathcal Q}_4^{\rm ren}\Psi\rangle + C \ell^{-3} \err (K_H')^3 (\int v) (n_+^L)^2 + {\mathcal E},
\end{align}
which is in agreement with \eqref{eq:GoodEstimated1d2-v2}.
\end{proof}

\section{Localization of the momentum localized number of excited particles $n_+^L$}\label{sec:LocMatrices}

In this section we prove Proposition~\ref{prop:LocMatrices} below. That result allows us to work with a state $\widetilde{\Psi}$ that is localized in $n_{+}^L$ and by \eqref{eq:EnergyBound} has an energy which is correct to LHY-order.
We introduce a parameter $\ML$ satisfying \eqref{eq:DefML}.

\begin{proposition}[Restriction on $n_+^L$]\label{prop:LocMatrices}
Assume that the conditions of Assumption~\ref{assump:params} 
are satisfied and that $\rmu a^3$ is small enough. 
There is then a universal constant $C>0$, such that
if there is a normalized $n$-particle $\Psi\in{\mathcal F}_s(L^2(\Lambda))$ 
satisfying (notice the factor of $\frac{1}{2}$ compared to \eqref{eq:aprioriPsi}),
\begin{align}\label{eq:aprioriPsiHalf}
  \langle \Psi, {\mathcal H}_{\Lambda}(\rmu) \Psi \rangle \leq
  - 4\pi \rmu^2 a \ell^3(1-\frac{1}{2}K_B^3 (\rmu a^3)^{\frac{1}{2}})
\end{align}
under the assumptions in
Theorem~\ref{thm:aprioribounds}, then
there is also a normalized $n$-particle wave function $\widetilde{\Psi}\in
{\mathcal F}_{\rm s}(L^2(\Lambda))$ with the property that
\begin{align}\label{eq:LocalizedNPlus}
  \widetilde{\Psi} = 1_{[0,\ML]}(n_{+}^L) \widetilde{\Psi},
\end{align}
i.e., only values
of $n_+^L$ smaller than $\ML$ appear in $\widetilde{\Psi}$, and
such that
\begin{align}\label{eq:EnergyBound}
\langle\widetilde\Psi,{\mathcal H}_\Lambda(\rmu)\widetilde\Psi\rangle&
\leq \langle\Psi,{\mathcal H}_\Lambda(\rmu)\Psi\rangle
+C \rmu^2 a \ell^3 (\rmu a^3)^{\frac{5}{8}}  \widetilde{K}_{{\mathcal M}}^2 K_{\mathcal R}  K_N^{\frac{1}{4}} (K_H')^{\frac{3}{2}} K_B^{\frac{3}{2}} K_{\ell}^{-5} .
\end{align}
\end{proposition}

\begin{proof}
We start by proving that with notation from \eqref{eq:def_d1L} and \eqref{eq:def_d2L},
we have
\begin{align}\label{eq:ds}
|\langle \Psi, d_1^L \Psi \rangle | + |\langle \Psi, d_2^L \Psi \rangle | \leq 
C \rmu^2 a \ell^3 K_{\mathcal R}  K_N^{\frac{1}{4}} (K_H')^{\frac{3}{2}} K_B^{\frac{3}{2}} K_{\ell} (\rmu a^3)^{-\frac{3}{8}}.
\end{align}
To get this, we apply \eqref{eq:GoodEstimatedSimple} with $\widetilde{M} = n$.
By \eqref{eq:Q4apriori_2} and \eqref{con:KBKell} the estimate is immediate for the ${\mathcal Q}_4^{\rm ren}$-part in \eqref{eq:GoodEstimatedSimple}.
Also, since ${\mathcal R_1} (K_H')^3 \geq 1$, we only have to prove that
\begin{align}
\rmu^2 a \ell^3 \left\{ {\mathcal R} \sqrt{ {\mathcal R_1} (K_H')^3  \frac{ \langle \Psi, n_{+} \Psi \rangle }{n}} + 
{\mathcal R}  {\mathcal R_1} (K_H')^3 
 \frac{ \langle \Psi, n_{+} \Psi \rangle }{n}\right\},
\end{align}
is bounded by the right hand side in \eqref{eq:ds}. This follows by inserting \eqref{eq:apriori_n}, \eqref{eq:apriorinn+NEW}, the definition \eqref{eq:Def_R1} of $\err$ and using \eqref{con:Constants} to realize that the largest term is the one with the square root. 
Notice in particular, that ${\mathcal R} \err \leq 3 K_{\mathcal R} (\rmu a^3)^{-\frac{1}{2}} \varepsilon_N^{-\frac{1}{2}}$, by \eqref{con:RNew}, \eqref{con:intv} and \eqref{con:Constants}.
This establishes \eqref{eq:ds}.

Using \eqref{con:newML} and \eqref{con:KBellKM} we can therefore apply Lemma~\ref{lem:LocMatrices_toiterate} to get \eqref{eq:EnergyBound}.
\end{proof}

\begin{remark}
In the proof above, we used \eqref{eq:GoodEstimatedSimple} with $\widetilde{{\mathcal M}} = n$. It is clear that one could now use  \eqref{eq:GoodEstimatedSimple} with $\widetilde{{\mathcal M}} = \ML$ and thus iterate to get improved bounds.
However, the gain of doing this is limited, and we have refrained from doing so, since it is not necessary for our purposes.
\end{remark}

\begin{lemma}[Restriction on $n_+^L$---a priori version]\label{lem:LocMatrices_toiterate}
Assume that the conditions of Assumption~\ref{assump:params} 
are satisfied and that $\rmu a^3$ is small enough. 
Then there exists a universal $C>0$ such that
if there is a normalized $n$-particle $\Psi\in{\mathcal F}_s(L^2(\Lambda))$ 
satisfying \eqref{eq:aprioriPsiHalf} under the assumptions in
Theorem~\ref{thm:aprioribounds}, and
if for some $\widetilde{\ML}$ we have
\begin{align}\label{eq:ML_large}
\widetilde{\ML} \gg K_B^3 K_{\ell}^5,
\end{align}
and
\begin{equation}\label{eq:IMS_small}
\frac{|\langle\Psi,d_1^L \Psi\rangle| + |\langle\Psi,d_2^L \Psi\rangle|}{\widetilde{\ML}^2} \ll a \rmu^2 \ell^3 (\rmu a^3)^{1/2},
\end{equation}
with $d_1^L$ from \eqref{eq:def_d1L} and $d_2^L$ from \eqref{eq:def_d2L},
then
there also exists a normalized $n$-particle wave function $\widetilde{\Psi}\in
{\mathcal F}_{\rm s}(L^2(\Lambda))$ with the property that
\begin{align}\label{eq:LocalizedNPlus_toiterate}
  \widetilde{\Psi} = 1_{[0,\widetilde{\ML}]}(n_{+}^L) \widetilde{\Psi},
\end{align}
and
such that
\begin{equation}\label{eq:EnergyBound_toiterate}
\langle\widetilde\Psi,{\mathcal H}_\Lambda(\rmu)\widetilde\Psi\rangle\leq
\langle\Psi,{\mathcal H}_\Lambda(\rmu)\Psi\rangle
+ \frac{|\langle\Psi,d_1^L \Psi\rangle| + |\langle\Psi,d_2^L \Psi\rangle|}{(\widetilde{\ML})^2}.
\end{equation}
\end{lemma}

We shall use the following theorem from \cite{LS}.
\begin{theorem}[Localization of large matrices]~\label{thm:Localizing large matrices} Suppose that 
  ${\mathcal A}$ is an $(N+1)\times (N+1)$ 
  Hermitean matrix and let ${\mathcal A}^{(k)}$, with $k=0,1,\dots ,N$, denote the
  matrix consisting of the $k^{\mathrm{th}}$ supra- and infra-diagonal
  of ${\mathcal A}$. Let $\psi\in \C^{N+1}$ be a normalized vector and set
  $d_k=\langle\psi,{\mathcal A}^{(k)}\psi\rangle$ and $\lambda=\langle\psi,{\mathcal A}
    \psi\rangle=\sum_{ k=0}^{N}d_k$ ($\psi$ need not be an eigenvector of
  ${\mathcal A}$). Choose some positive integer ${\mathcal M}'\leq N+1$. Then, with
  ${\mathcal M}'$ fixed, there is some $n'\in [0,N+1-{\mathcal M}']$ and some normalized
  vector $\phi\in \C^{N+1}$ with the property that $\phi_j=0$ unless
  $n'+1\leq j\leq n'+{\mathcal M}'$ (i.e., $\phi$ has localization length ${\mathcal M}'$) and such that
  \begin{align}
    \langle\phi,{\mathcal A}\phi\rangle\leq \lambda
    +\frac{C}{{\mathcal M}'^2}\sum_{k=1}^{{\mathcal M}'-1}k^2\abs{d_k}+C\sum_{k={\mathcal M}'}^{N}\abs{d_k},\label{eq:Localizing
      of large matrices}
  \end{align}
  where $C>0$ is a universal constant. (Note that the first sum starts at $k=1$.)
\end{theorem}

\begin{proof}[Proof of Lemma~\ref{lem:LocMatrices_toiterate}]
We may assume by \eqref{eq:ML_large} that $\widetilde{\ML} \geq 5$ and that 
$\widetilde{\ML}\leq n$, since otherwise there is nothing to prove.
 
We apply Theorem~\ref{thm:Localizing large matrices} to the 
$(n+1)\times(n+1)$-matrix with elements
$$
{\mathcal A}_{i,j}=\|\one_{\{n_+^L=i\}}\Psi\|^{-1}\|\one_{\{n_+^L=j\}}\Psi\|^{-1}
\langle \one_{\{n_+^L=i\}}\Psi, H_\Lambda(\rmu)\one_{\{n_+^L=j\}}\Psi\rangle.
$$
(If any of the norms are zero, we set the corresponding element to zero.) 
Then we get a normalized vector $\psi=(\|\one_{\{n_+^L=0\}}\Psi\|,\ldots,\|\one_{\{n_+^L=n\}}\Psi\|)$ in $\C^{n+1}$ satisfying
$$
\langle\psi,{\mathcal A} \psi\rangle= \langle \Psi, H_\Lambda(\rmu)\Psi\rangle. 
$$
Moreover, using the notation of Theorem~\ref{thm:Localizing large matrices}, only the ${\mathcal A}^{(k)}$ with $k=0,1,2$ are non-vanishing, and
\begin{equation}
 d_1=\langle \psi,{\mathcal A}^{(1)} \psi \rangle=\langle\Psi,d_1^L \Psi\rangle,
 \qquad
d_2=\langle \psi,{\mathcal A}^{(2)} \psi \rangle=
\langle\Psi,d_2^L\Psi \rangle.
\end{equation}
Notice that here it is important that $Q_L'$ commutes with the kinetic energy.

By Theorem~\ref{thm:Localizing large matrices} we can find a normalized $\phi\in \C^{n+1}$ with localization length ${\mathcal M}'$ equal to the integer part of $\widetilde{\ML}/2$,
such that 
\begin{align}
  \langle \phi, {\mathcal A}\phi\rangle\leq&\, \langle\psi,{\mathcal A} \psi\rangle+C\frac{|d_1| + |d_2|}{(\widetilde{\ML})^2}
  = \langle \Psi, H_\Lambda(\rmu)\Psi\rangle+C\frac{|d_1| + |d_2|}{(\widetilde{\ML})^2}.
\end{align}
Let $\widetilde\phi\in\C^{n+1}$ be given by
$\widetilde\phi_i=\phi_i$ if $\|\one_{\{n_+^L=i\}}\Psi\|\ne 0$ and
$\widetilde\phi_i=0$ if $\|\one_{\{n_+^L=i\}}\Psi\|=0$.  Then
$\|\widetilde\phi\|\leq 1$. 
We then have 
\begin{align}
  \langle \widetilde \phi, {\mathcal A}\widetilde\phi\rangle=\langle \phi, {\mathcal A}\phi\rangle\leq
  \langle \Psi, H_\Lambda(\rmu)\Psi\rangle+C\frac{|d_1| + |d_2|}{(\widetilde{\ML})^2}< 0.
\end{align}
where the negativity follows from 
\eqref{eq:aprioriPsiHalf},
\eqref{con:KBKell}, and \eqref{eq:IMS_small}.
In particular,
$\widetilde\phi\ne0$. Define
$$
\widetilde \Psi=\|\widetilde\phi\|^{-1}\sum_{i=0}^n\widetilde\phi_i\|\one_{\{n_+^L=i\}}\Psi\|^{-1}\one_{\{n_+^L=i\}}\Psi.
$$
Then $\widetilde \Psi$ is normalized and satisfies 
$$
\langle\widetilde\Psi,H_\Lambda(\rmu)\widetilde\Psi\rangle=
\|\widetilde\phi\|^{-2}\langle \widetilde \phi, {\mathcal A}\widetilde\phi\rangle\leq 
\langle \widetilde \phi, {\mathcal A}\widetilde\phi\rangle,
$$ since the term on the right is negative and
$\|\widetilde\phi\|^{-2}\geq1$. This proves that $\widetilde\Psi$
satisfies \eqref{eq:EnergyBound_toiterate}. It remains to prove that
$\widetilde\Psi$ satisfies \eqref{eq:LocalizedNPlus_toiterate}.  We know from
the construction that the possible values of $n_+^L$ that occur in
$\widetilde \Psi$ lie in an interval of length ${\mathcal M}'$. 
Furthermore, since we have 
\eqref{eq:aprioriPsiHalf}
as well as \eqref{eq:EnergyBound_toiterate} and \eqref{eq:IMS_small}, we may 
use the a priori
bound \eqref{eq:apriorinn+NEW} on the expectation value of $n_+^L$ in
$\widetilde \Psi$. 
This implies that the interval of $n_+^L$ values in
$\widetilde \Psi$ must be contained in
$[0,\widetilde{\ML}]$,
since otherwise we get a contradiction between \eqref{eq:apriorinn+NEW}  and \eqref{eq:ML_large}.
\end{proof}

\section{Localization of the $3Q$-term}
\label{sec:3Q}
In this section we will absorb an unimportant part of the $3Q$ term in the positive $4Q$ term. 
This result, Lemma~\ref{lem:Q3-splitting1} below, as well as its proof is almost identical to \cite[Lemma 9.1]{FS}, but has an improved dependence on $R$ in the error bound. For completeness, we have included the short proof below.

\begin{lemma}\label{lem:Q3-splitting1}
Define
\begin{equation}
\widetilde{Q}_3^{(1)} := 
\sum_{i\neq j} (P_i Q_{L,j} w_1(x_i,x_j) Q_j Q_i + h.c.) ,
\end{equation}
Assume that the conditions of Assumption~\ref{assump:params} 
are satisfied and that $\rmu a^3$ is small enough. 
With the notation from \eqref{eq:DefQ4}, \eqref{eq:DefQ3}, we get,
\begin{align}\label{eq:ReduceQ3}
{\mathcal Q}_3^{\rm ren} + \frac{1}{4} {\mathcal Q}_4^{\rm ren} +\frac{b}{100}  \left( \ell^{-2} n_{+} + \varepsilon_T (d\ell)^{-2} n_{+}^H \right) 
\geq \widetilde{Q}^{(1)}_3 
 - C \rmu^2 a \ell^3   \left( d^{2M} + \frac{R a}{\ell^2} \right) .
\end{align}
\end{lemma}

\begin{proof}
Using \cite[Corollary 6.12]{FS}, with $Q' = \overline{Q}_{L,}$ and $\varepsilon= c K_{\ell}^{-2}$ for some sufficiently small constant $c$, as well as \eqref{eq:QLs_giveNplusH}
we find 
\begin{align}
\frac{1}{4} {\mathcal Q}_4^{\rm ren} + \frac{b}{100 \ell^2} n_{+}+ 
{\mathcal Q}_3^{\rm ren}
- \widetilde{Q}^{(1)}_3 
&\geq
  \sum_{i,j}\Bigl(P_j \overline{Q}_{L,i} w_1(x_i,x_j)\omega(x_i-x_j)P_iP_j+ h.c. \Bigr) \nonumber \\
  &\quad   -C \ell^{-2} K_{\ell}^4 n_{+}^H. \label{eq:pq'pp-sf} 
\end{align}
Using \eqref{con:eTdK} it is clear that the $n_{+}^H$ term is dominated by half of the positive $n_{+}^H$ term from \eqref{eq:ReduceQ3}.

To estimate the remaining terms in \eqref{eq:pq'pp-sf} we start by using the estimate \eqref{eq:I2-integral-2} on the convolution from Lemma~\ref{lem:Simple} to get 
\begin{align}
&- \sum_{i\neq j} \left( P_i \overline{Q}_{L,j}  w_1(x_i,x_j) \omega(x_i-x_j) P_j P_i + h.c. \right) \nonumber \\
&\geq - I \ell^{-3} \Big( n_0 \sum_j \overline{Q}_{L,j} \chi_{\Lambda}^2(x_j) P_j + h.c.\Big)
- C a n^2 \ell^{-3} \frac{R a}{\ell^2},
\end{align}
where $I := \int W_1(y) \omega(y) \leq C a$, and where we used \eqref{omegabounds} to estimate the error term in \eqref{eq:I2-integral-2}.

To complete the proof we write, with $N = M/2 \in {\mathbb N}$ 
\begin{align}
\overline{Q}_{L} \chi_{\Lambda}^2 P + h.c.
=
\overline{Q}_{L} (\ell^{-2} - \Delta)^{-N} \left[(\ell^{-2} - \Delta)^{N}\chi_{\Lambda}^2\right]  P + h.c.
\end{align}
and notice that
\begin{equation}
|(\ell^{-2} - \Delta)^{N}\chi_{\Lambda}^2| \leq C \ell^{-2N}.
\end{equation}
Therefore,
\begin{equation}
\overline{Q}_{L} \chi_{\Lambda}^2 P + h.c. \leq
\varepsilon_2 P + \varepsilon_2^{-1} \ell^{2N} \overline{Q}_{L} (\ell^{-2} - \Delta)^{-2N}
(\overline{Q}_{L})^*
\leq \varepsilon_2 P + \varepsilon_2^{-1} d^{4N} \overline{Q}_{L}(\overline{Q}_{L})^*.
\end{equation}
Choosing $\varepsilon_2=d^{4N} = d^{2M}$ and using again \eqref{con:eTdK} we get \eqref{eq:ReduceQ3}
upon summing this estimate in the particle indices
and absorbing the $n_{+}^H$ term as before.
\end{proof}

\section{Second quantized operators}
\label{sec:Second}
\subsection{Creation/annihilation operators}
We will use $a, a^{\dagger}$ to denote the standard bosonic annihilation/creation operators on the bosonic Fock space ${\mathcal F}_{\rm s}(L^2(\Lambda))$.

We define $a_0$ as the annihilation operator associated to the condensate function for the box $\Lambda$, i.e. $a_0 = \ell^{-3/2} a(\theta)$, where we recall that $\theta$ defined in \eqref{eq:Theta} is the characteristic function of the box. 

We will also introduce creation/annihilation operators for other momenta. These will come in two forms, depending on whether they contain the localization function $\chi_{\Lambda}$ or not.

For $k \in {\mathbb R}^3\setminus\{0\}$ we let 
\begin{align}
\widetilde{a}_k := \ell^{-3/2} a(Q(e^{ikx} \theta)), \qquad \widetilde{a}_k^{\dagger} := \ell^{-3/2} a^{\dagger}(Q(e^{ikx} \theta)).
\end{align}
Clearly, for $k,k' \in {\mathbb R}^3\setminus\{0\}$, we have the commutation relations 
\begin{align}\label{eq:ak-akdagger}
[\widetilde{a}_k, \widetilde{a}_{k'}] = 0, \qquad [\widetilde{a}_k, \widetilde{a}_{k'}^{\dagger}] = \ell^{-3} \langle e^{ikx} \theta, Q e^{ik'x} \theta \rangle.
\end{align}
We also define, for $k \in {\mathbb R}^3\setminus \{0\}$,
\begin{align}
a_k := \ell^{-3/2} a( Q(e^{ikx} \chi_\Lambda))\qquad \textrm{and}\qquad a_k^{\dagger} := \ell^{-3/2} a^{\dagger}( Q(e^{ikx} \chi_\Lambda)).
\end{align}
Then, for all $k,k' \in {\mathbb R}^3\setminus \{ 0 \}$,
\begin{align}
[ a_k, a_{k'} ] = 0, 
\end{align}
and, using $Q = \one - P$ on $L^2(\Lambda)$,
\begin{align}\label{eq:Commutator_general}
[a_k, a_{k'}^{\dagger}] &= \ell^{-3}
\langle Q(e^{ikx} \chi_\Lambda), Q (e^{ik'x} \chi_\Lambda) \rangle 
= \widehat{\chi^2}((k-k')\ell) -  \widehat{\chi}(k \ell) \overline{\widehat{\chi}(k'\ell)}.
\end{align}
In particular,
\begin{align}\label{eq:Commutator}
[a_k, a_{k}^{\dagger}] \leq  1.
\end{align}

We will repeatedly need the following consequence of Lemma~\ref{lem:pseudolocal}
\begin{lemma}\label{lem:UsingML}
Assume that the conditions of Assumption~\ref{assump:params} 
are satisfied and that $\rmu a^3$ is small enough. 
Suppose that $\widetilde{\Psi} \in {\mathcal F}_{\rm s}(L^2(\Lambda))$ is normalized and satisfies
\begin{align}
\one_{[0,2\ML]}(n_+^L)\widetilde{\Psi} = \widetilde{\Psi} , \qquad
\one_{[0, 2 \rmu \ell^3]} (n_+) \widetilde{\Psi} = \widetilde{\Psi} .
\end{align}

Then,
\begin{align}\label{eq:DomByML}
\langle \widetilde{\Psi} , 
\ell^3 \int_{\{|k|\leq 2 K_H'' \ell^{-1}\}}  
(a_k^{\dagger} a_k + \widetilde{a}_k^{\dagger} \widetilde{a}_k)\,dk\,
\widetilde{\Psi} \rangle
\leq C \ML.
\end{align}

Furthermore,
\begin{align}\label{eq:DomByML_ext}
\langle \widetilde{\Psi} , 
\ell^3 \int 
(a_k^{\dagger} a_k + \widetilde{a}_k^{\dagger} \widetilde{a}_k) \,dk\,
\widetilde{\Psi} \rangle
\leq C \ML + C \langle  \widetilde{\Psi}, n_+^{H}  \widetilde{\Psi} \rangle.
\end{align}

\end{lemma}

\begin{proof}
We carry out the proof for the $a_k$'s and only comment on how the same proof works for the $\widetilde{a}_k$'s.

An explicit calculation shows that, on the $N$-particle subspace,
\begin{align}\label{eq:BackToFirst}
(2\pi)^{-3} \ell^3 \int_{\{|k|\leq 2 K_H'' \ell^{-1}\}}  a_k^{\dagger} a_k \,dk
&=  \sum_{j=1}^N Q_j \chi_{\Lambda}(x_j) \one_{(0,2 K_H'' \ell^{-1}]}(|p_j|) \chi_{\Lambda}(x_j) Q_j\nonumber \\
&\leq 
C \sum_{j=1}^N \Big\{ (Q_L')_j +  \left( \left( \frac{K_H''}{K_H'} \right)^{M} +  \varepsilon_N^{\frac{3}{2}}\right) Q_j \Big\},
\end{align}
where the inequality follows from Lemma~\ref{lem:pseudolocal}.

The estimate \eqref{eq:DomByML}  for the $a_k$'s now follows easily using \eqref{cond:KHs-ML}, since we have that $\varepsilon_N^{\frac{3}{2}} \rmu \ell^3 \ll 1$ using the definition of $\varepsilon_N$ in \eqref{eq:Cond_epsilonN}, as well as \eqref{con:KBKell} and \eqref{con:Constants}.

The estimate \eqref{eq:DomByML}  for the $\widetilde{a}_k$'s follows in the same way with the only change that the operator on the $N$-particle subspace in \eqref{eq:BackToFirst} becomes $\sum_j Q_j \one_{(0,2 K_H'' \ell^{-1}]}(|p_j|) Q_j$, which can also be compared to $\sum_j (Q_L')_j$ by Lemma~\ref{lem:pseudolocal}.

To prove \eqref{eq:DomByML_ext} for the $a_k$'s we therefore have to estimate the complementary integral
\begin{align}
\ell^3 \int_{\{|k|\geq 2 K_H'' \ell^{-1}\}}  
a_k^{\dagger} a_k \,dk.
\end{align}
This is the second quantization of the one-particle operator $Q\chi_{\Lambda} \one_{[2 K_H'' \ell^{-1}, \infty)}(\sqrt{-\Delta}) \chi_{\Lambda}Q$.
Using Cauchy-Schwarz we can estimate this as
\begin{align}
&Q\chi_{\Lambda}\one_{[2 K_H'' \ell^{-1}, \infty)} \chi_{\Lambda}Q \nonumber\\
&\leq 2 \| \chi_{\Lambda}\|_{\infty} Q \one_{[d^{-2}\ell^{-1}, \infty)}(\sqrt{-\Delta}) Q
\nonumber \\
&\quad 
+ 2 Q \one_{[0, d^{-2}\ell^{-1})}(\sqrt{-\Delta}) \chi_{\Lambda}\one_{[2 K_H'' \ell^{-1}, \infty)}(\sqrt{-\Delta})  \chi_{\Lambda} \one_{[0, d^{-2}\ell^{-1})}(\sqrt{-\Delta}) Q \nonumber \\
&\leq 2 \| \chi_{\Lambda} \|_{\infty} Q \one_{[d^{-2}\ell^{-1}, \infty)}(\sqrt{-\Delta}) Q + 2 \| \one_{[0, d^{-2}\ell^{-1})}(\sqrt{-\Delta}) \chi_{\Lambda}\one_{[2 K_H'' \ell^{-1}, \infty)}(\sqrt{-\Delta}) \|^2 Q.
\end{align}
Here the operator norm satisfies
\begin{align}
&\| \one_{[0, d^{-2}\ell^{-1})}(\sqrt{-\Delta}) \chi_{\Lambda} \one_{[2 K_H'' \ell^{-1}, \infty)}(\sqrt{-\Delta}) \| \nonumber \\
&= \left\| \left(\one_{[0, d^{-2}\ell^{-1})}(\sqrt{-\Delta}) \chi_{\Lambda} (-\Delta)^{M/2}\right) \left( (-\Delta)^{-M/2}\one_{[2 K_H'' \ell^{-1}, \infty)}(\sqrt{-\Delta}) \right)\right\| \nonumber \\
&\leq C (d^2 K_H'')^{-M},
\end{align}
where we used that $\chi$ has $M\in 2{\mathbb N}$ bounded derivatives. So we see that
\begin{align}
\ell^3 \int_{\{|k|\geq 2 K_H'' \ell^{-1}\}}  
a_k^{\dagger} a_k \,dk \leq C n_{+}^{H} + C (d^2 K_H'')^{-2M} n_{+}.
\end{align}
Now \eqref{eq:DomByML_ext} for the $a_k$'s follows using the assumption on $ \widetilde{\Psi}$ and \eqref{cond:KHs-ML}.

The proof of  \eqref{eq:DomByML_ext} for the $\widetilde{a}_k$'s is easier, we just observe that $\one_{[d^{-2}\ell^{-1},\infty)}(\sqrt{-\Delta})  \geq \one_{[2 K_H'' \ell^{-1}, \infty)}(\sqrt{-\Delta})$.
\end{proof}

We will now formulate a lower bound on the Hamiltonian $ {\mathcal H}_{\Lambda}(\rmu) $ in second quantization.
Recall the Fourier-multiplier $\tau(p)$ defined in \eqref{eq:Def_tau}.

\begin{proposition}\label{prop:Hamilton2ndQuant}
Assume that $\widetilde{\Psi}$ is an $n$-particle wavefunction which satisfies \eqref{eq:aprioriPsi} and \eqref{eq:LocalizedNPlus},
and that the parameters satisfy 
Assumption~\ref{assump:params} and that $\rmu a^3$ is small enough.
Then, in 2nd quantization the operator ${\mathcal H}_{\Lambda}(\rmu)$ defined in \eqref{eq:Def_HB_new} satisfies
\begin{align}
\langle \widetilde{\Psi}, {\mathcal H}_{\Lambda}(\rmu) \widetilde{\Psi} \rangle 
\geq \langle \widetilde{\Psi}, {\mathcal H}_{\Lambda}^{\rm 2nd}(\rmu) \widetilde{\Psi} \rangle
- C \rmu^2 a \ell^3   \left( d^{2M} + \frac{R a}{\ell^2} \right) , 
\end{align}
where
\begin{align}\label{eq:HLambdaSecond}
{\mathcal H}_{\Lambda}^{\rm 2nd} &=
(2\pi)^{-3} \ell^3 \int (1-\varepsilon_N) \tau(k) a_k^{\dagger} a_k\,dk + \frac{b}{2 \ell^2} n_{+} + \varepsilon_T\frac{b}{8 d^2 \ell^2} n_{+}^{H} 
+\varepsilon_T\frac{b}{16 d^2 \rmu \ell^5} n_0 n_{+}^{H} 
\nonumber \\
&\quad+ \frac{1}{2} \ell^{-3} a_0^{\dagger} a_0^{\dagger} a_0 a_0 
\Big(\widehat{g}(0) + \widehat{g\omega}(0)\Big)
  - \rmu   \widehat{g}(0) a_0^{\dagger} a_0 \nonumber \\
  &\quad +\Big( (\ell^{-3} a_0^{\dagger} a_0 - \rmu) \widehat{W}_1(0) (2\pi)^{-3} \int \widehat{\chi}_{\Lambda}(k) a_k^{\dagger} a_0 \,dk + h.c. \Big) \nonumber \\
  &\quad +\Big( \ell^{-3} a_0^{\dagger} a_0 \widehat{W \omega}_1(0)(2\pi)^{-3} \int \widehat{\chi}_{\Lambda}(k) a_k^{\dagger} a_0 \,dk + h.c. \Big) \nonumber \\
  &\quad +
  (2\pi)^{-3} \int \left(\widehat{W}_1(k) + \widehat{W_1 \omega}(k)\right) a_0^{\dagger} a_k^{\dagger} a_k a_0 
  + \frac{1}{2} \widehat{W}_1(k) \left( a_0^{\dagger} a_0^{\dagger} a_k a_{-k} + a_k^{\dagger} a_{-k}^{\dagger} a_0 a_0 \right) \,dk \nonumber \\
  & \quad + \Big((\ell^{-3} a_0^{\dagger} a_0 -\rmu) \widehat{W_1}(0) + \ell^{-3} a_0^{\dagger} a_0\widehat{W_1 \omega}(0)\Big) (2\pi)^{-3} \ell^{-3} \int a_k^{\dagger} a_k\,dk  \nonumber\\
  &\quad + \widetilde{Q}_3,
\end{align}
where
\begin{align}\label{eq:3Qtilde}
\widetilde{Q}_3 := \ell^3 (2\pi)^{-6} \iint_{\{ k \in P_H\}}
f_L(s)
\widehat{W}_1(k) (a_0^{\dagger} \widetilde{a}_s^{\dagger}  a_{s-k} a_k
+ a_k^{\dagger}  a^{\dagger}_{s-k} \widetilde{a}_s a_0).
\end{align}
\end{proposition}

\begin{proof}
Notice that \eqref{eq:apriori_n} and \eqref{eq:apriorinn+NEW} hold, using \eqref{eq:aprioriPsi} and Theorem~\ref{thm:aprioribounds}.
Notice also that $n_0 \leq n \leq 2 \rmu \ell^3$.

We start by explicitly calculating the second quantization of the kinetic energy. Here we drop the two positive terms in \eqref{eq:DefT_new} depending on the Neumann Laplace operator $-\Delta^{\mathcal N}$.
Also, using the definition \eqref{eq:nplusH} of $n_{+}^{H}$ we easily obtain the last two terms in the first line of \eqref{eq:HLambdaSecond}---for later use this $n_{+}^{H}$-term has been split in two by artificially inserting a factor of $n_0/(\rmu \ell^3)\leq 2$ in one of the resulting terms.
Also, part of the $n_+$-term---the `gap' in the kinetic energy---has been dropped in \eqref{eq:HLambdaSecond}: It is used to control error terms appearing in the potential terms treated below.

To calculate the potential energy, 
we apply Lemma~\ref{lm:potsplit}. 
For the operators ${\mathcal Q}_0^{\rm ren}$ and ${\mathcal Q}_1^{\rm ren}$ we use the simplifications of Lemma~\ref{lem:QDecompSimpler} before making the explicit calculation of their 2nd quantizations.
For ${\mathcal Q}_2^{\rm ren}$ we also use the simplifications of Lemma~\ref{lem:QDecompSimpler}. The error term in \eqref{eq:Q2n0} is absorbed in the gap in the kinetic energy. This uses that $R \ll (\rmu a)^{-1/2}$, which is \eqref{con:ER}, and the relation $n \approx \rmu \ell^3$ 
from \eqref{eq:apriori_n}.

Finally we consider ${\mathcal Q}_3^{\rm ren}$ and ${\mathcal Q}_4^{\rm ren}$.
By Lemma~\ref{lem:Q3-splitting1} and the positivity of $v$ we have the lower bound
\eqref{eq:ReduceQ3}. What remains of ${\mathcal Q}_4^{\rm ren}$ will be discarded for a lower bound. The application of \eqref{eq:ReduceQ3} also costs a bit of the gap in the kinetic energy.
We have left to compare $\widetilde{Q}^{(1)}_3$ with $\widetilde{Q}_3$.
But that is the content of Lemma~\ref{lem:Q3-splitting2} below.
Notice that using \eqref{cond:KH3-n} the error term from \eqref{eq:Q3-1-reduce} can also be absorbed in the gap in the kinetic energy.
This finishes the proof of Proposition~\ref{prop:Hamilton2ndQuant}.
\end{proof}

In the above proof we used the following localization of the $3Q$-term.

\begin{lemma}\label{lem:Q3-splitting2}
Assume that $\widetilde{\Psi}$ is an $n$-particle wavefunction which satisfies \eqref{eq:aprioriPsi} and \eqref{eq:LocalizedNPlus},
and that the parameters satisfy 
Assumption~\ref{assump:params} and that $\rmu a^3$ is small enough.
Let $\widetilde{Q}^{(1)}_3$ be as defined in Lemma~\ref{lem:Q3-splitting1} and
$\widetilde{Q}_3$ from \eqref{eq:3Qtilde}.

Then, 
\begin{equation}\label{eq:Q3-1-reduce}
\langle \widetilde{\Psi}, \widetilde{Q}^{(1)}_3 \widetilde{\Psi} \rangle
\geq \langle \widetilde{\Psi}, \widetilde{Q}_3 \widetilde{\Psi} \rangle - C \left( \frac{K_{\ell}^4 (K_H'')^3 \ML}{\rmu \ell^3} \right)^{1/2}  \ell^{-2} \langle \widetilde{\Psi}, n_{+} \widetilde{\Psi} \rangle.
\end{equation}
\end{lemma}

\begin{proof}
Notice that \eqref{eq:apriori_n} holds, using  \eqref{eq:aprioriPsi}
and Theorem~\ref{thm:aprioribounds}.

In second quantization we have
\begin{align}
\widetilde{Q}^{(1)}_3
= 
\ell^3 (2\pi)^{-6} \iint
f_L(s)
\widehat{W}_1(k) (a_0^{\dagger} \widetilde{a}_s^{\dagger}  a_{s-k} a_k
+ a_k^{\dagger}  a^{\dagger}_{s-k} \widetilde{a}_s a_0) \,dk\,ds,
\end{align}
so we have to estimate the part of the integral where $k \notin P_H$.
Notice that
\begin{align}
\ell^6 \langle  \widetilde{\Psi}, \iint_{\{ |k| \leq K_H'' \ell^{-1} \}}a_k^{\dagger} a_{s-k}^{\dagger} a_{s-k} a_k \,dsdk  \,\widetilde{\Psi} \rangle
\leq C \ML  \langle  \widetilde{\Psi}, n_{+} \widetilde{\Psi} \rangle,
\end{align}
using Lemma~\ref{lem:UsingML}.
Therefore, we get for $\varepsilon >0$, \begin{align}\label{eq:EstimQ3-reduce}
&\langle  \widetilde{\Psi}, \ell^3 (2\pi)^{-6} \iint_{\{ |k| \leq K_H'' \ell^{-1} \}}
f_L(s)
\widehat{W}_1(k) (a_0^{\dagger} \widetilde{a}_s^{\dagger}  a_{s-k} a_k
+ a_k^{\dagger}  a^{\dagger}_{s-k} \widetilde{a}_s a_0) \widetilde{\Psi} \rangle dk ds \nonumber \\
&\geq
- C \| \widehat{W}_1(k) \|_{\infty}  \ell^3 (2\pi)^{-6} \iint_{\{ |k| \leq K_H'' \ell^{-1} \}}
f_L(s)  \langle  \widetilde{\Psi},  \left(
\varepsilon  \widetilde{a}_s^{\dagger} a_0^{\dagger} a_0 \widetilde{a}_s  +
\varepsilon^{-1} a_k^{\dagger} a_{s-k}^{\dagger} a_{s-k} a_k
\right) \widetilde{\Psi} \rangle dk ds \nonumber \\
&\geq
- C a n \frac{\langle  \widetilde{\Psi}, n_{+} \widetilde{\Psi} \rangle }{\ell^3} \left( \varepsilon  (K_H'')^{3} + \varepsilon^{-1} \frac{\ML}{n}\right).
\end{align}
From here  \eqref{eq:Q3-1-reduce} follows upon choosing $\varepsilon = \left( \frac{\ML}{n (K_H'')^3} \right)^{1/2}$ and 
the relation $n \approx \rmu \ell^3$ 
from \eqref{eq:apriori_n}.
\end{proof}

It will also be useful to notice the following representation in terms of the operators $\widetilde{a}_k$.
\begin{lemma}
We have the identities
\begin{align}
\left((2\pi)^{-6} \ell^6 \iint \widetilde{a}_k^{\dagger} \widehat{\chi^2}\left((k-k')\ell\right) \widetilde{a}_{k'}
\right)_{N} = \sum_{j=1}^N Q_j \chi_{\Lambda}^2(x_j) Q_j,
\end{align}
and
\begin{align}\label{eq:RepInTermsOfatilde}
\left((2\pi)^{-6} \ell^6 \iint  f_L(k)f_L(k')\widetilde{a}_k^{\dagger} \widehat{\chi^2}\left((k-k')\ell\right) \widetilde{a}_{k'}
\right)_{N} = \sum_{j=1}^N Q_{L,j}^{*} \chi_{\Lambda}^2(x_j) Q_{L,j}.
\end{align}

\end{lemma}

\subsection{$c$-number substitution}
It is convenient to apply the technique of $c$-number substitution as described in \cite{LSYc}.

Let $\Psi \in {\mathcal F}_{\rm s}(L^2(\Lambda))$.
We can think of $L^2(\Lambda) = \Ran(P) \oplus \Ran(Q)$, with $\Ran(P)$ being, of course, spanned by the constant vector $\theta$ (defined in \eqref{eq:Theta}).
This leads to the splitting ${\mathcal F}_{\rm s}(L^2(\Lambda)) = {\mathcal F}_{\rm s}(\Ran(P)) \otimes {\mathcal F}_{\rm s}(\Ran(Q))$. We let $\Omega$ denote the vacuum vector in ${\mathcal F}_{\rm s}(L^2(\Lambda))$.

For $z \in {\mathbb C}$ we define 
\begin{align}
|z \rangle:= \exp\left(-\frac{|z|^2}{2} - z a_0^{\dagger}\right) \Omega.
\end{align}
Given $z$ and $\Psi$ we can define 
\begin{align}
\Phi(z) := \langle z | \Psi \rangle \in {\mathcal F}_{\rm s}(\Ran(Q)),
\end{align}
where the inner product is 
considered as a partial inner product induced by the representation ${\mathcal F}_{\rm s}(L^2(\Lambda)) = {\mathcal F}_{\rm s}(\Ran(P)) \otimes {\mathcal F}_{\rm s}(\Ran(Q))$.

It is a simple calculation that
\begin{align}\label{eq:PartUnity}
1 = \int_{{\mathbb C}} |z \rangle \langle z | \,d^2z, \qquad 
\text{ and } \qquad 
a_0  |z \rangle = z  |z \rangle.
\end{align}

\begin{theorem}\label{thm:Kafz}
Define
\begin{align}
\rho_z := |z|^2 \ell^{-3},
\end{align}
and
\begin{align}\label{eq:Kz}
{\mathcal K}(z) 
&=
 \frac{b}{2 \ell^2} n_{+} + \varepsilon_T\frac{b}{8 d^2 \ell^2} n_{+}^{H} +  \varepsilon_T\frac{b |z|^2}{16 d^2 \rmu \ell^5} n_{+}^{H}
 +\ER (\rmu - \rho_z)^2 a \ell^3 \nonumber \\
&\quad +  
 \frac{1}{2} \rho_z^2 \ell^{3} 
\Big(\widehat{g}(0) + \widehat{g\omega}(0)\Big)
  - \rmu   \widehat{g}(0) \rho_z \ell^3 
  \nonumber \\
  &\quad +
  (2\pi)^{-3} \ell^3 \int \left( (1-\varepsilon_N)\tau(k) + \rho_z \widehat{W}_1(k) \right) a_k^{\dagger} a_k
  + \frac{1}{2} \rho_z \widehat{W}_1(k) \left(  a_k a_{-k} + a_k^{\dagger} a_{-k}^{\dagger}  \right) \,dk \nonumber \\
  & \quad + (\rho_z -\rmu) \widehat{W_1}(0)  (2\pi)^{-3} \ell^{3} \int a_k^{\dagger} a_k\,dk  \nonumber\\
  &\quad + {\mathcal Q}_1(z)+  {\mathcal Q}_1^{\rm ex}(z)+
  {\mathcal Q}_2^{\rm ex}(z) + {\mathcal Q}_3(z),
\end{align}
with
\begin{align}\label{eq:DefQ1z}
 {\mathcal Q}_1(z)&:= 
  \Big( (\rho_z - \rmu) \widehat{W}_1(0) (2\pi)^{-3} \int \widehat{\chi}_{\Lambda}(k) a_k^{\dagger} z \,dk + h.c. \Big),  \\
  \label{eq:DefQ1zex}
 {\mathcal Q}_1^{\rm ex}(z) &:= \Big( \rho_z \widehat{W_1 \omega}(0)(2\pi)^{-3} \int \widehat{\chi}_{\Lambda}(k) a_k^{\dagger} z \,dk + h.c. \Big),\\
  \label{eq:kroelleQ3z}
 {\mathcal Q}_3(z)&:= \ell^3 (2\pi)^{-6} z \iint_{\{ k \in P_H\}}
f_L(s)
\widehat{W}_1(k) ( \widetilde{a}_s^{\dagger}  a_{s-k} a_k
+ a_k^{\dagger}  a^{\dagger}_{s-k} \widetilde{a}_s),
\end{align}
and
\begin{align}\label{eq:DefQ2ex}
{\mathcal Q}_2^{\rm ex} = {\mathcal Q}_2^{\rm ex}(z) :=
(2\pi)^{-3} \rho_z \ell^3  \int \big(\widehat{W_1 \omega}(k) + \widehat{W_1 \omega}(0)\big)  a_k^{\dagger} a_k .
\end{align}

Assume that $\widetilde{\Psi}$ is an $n$-particle wavefunction which satisfies \eqref{eq:aprioriPsi} and \eqref{eq:LocalizedNPlus},
and that the parameters satisfy 
Assumption~\ref{assump:params} and that $\rmu a^3$ is small enough.

Then,
\begin{align}\label{eq:IntroK(z)}
\langle \widetilde{\Psi}, {\mathcal H}_{\Lambda}^{\rm 2nd}(\rmu) \widetilde{\Psi} \rangle
\geq
\inf_{z \in {\mathbb R}_{+} } \inf_{\Phi} \langle \Phi, {\mathcal K(z)} \Phi \rangle
- C \left( K_{\ell}^{-3} + \ER K_{\ell}^2 K_B^3\right) K_{\ell}^3 \rmu a,
\end{align}
where the second infimum is over all normalized $\Phi  \in {\mathcal F}_{\rm s}(\Ran(Q))$ with 
\begin{align}\label{eq:LocalizedAftercNumber}
\Phi =  1_{[0,\ML]}(n_{+}^L) \Phi,
\end{align}
and
\begin{equation}\label{eq:Localized_n_AftercNumber}
\Phi =  1_{[0,n]}(n_{+}) \Phi.
\end{equation}
\end{theorem}

\begin{remark}
The last two terms in the first line of \eqref{eq:Kz} are new compared to the corresponding definition in \cite{FS}.
The first of these terms is needed in Section~\ref{sec:rough} as part of the proof that the ground state energy is too large unless $|z|^2 \approx \rmu \ell^3$.
This modification is related to the fact that in this paper, since we consider singular potentials, we can only control the number of excited particles with low momenta through the localization of large matrices. Therefore, an extra effort is needed to control the number of excited particles with high momenta.
The second term is needed because we in this paper improve the dependence on the support of the potential. Therefore, we have to include this extra term to control a negative term in \eqref{eq:BogHamDiag} in Theorem~\ref{thm:BogHamDiag} that has a bad dependence on the radius $R$ of the potential.
\end{remark}

\begin{proof}
Notice that \eqref{eq:apriori_n} and \eqref{eq:apriorinn+NEW} hold for $\widetilde{\Psi}$, using \eqref{eq:aprioriPsi} and Theorem~\ref{thm:aprioribounds}.

We start by adding and subtracting the term $\ER (\rmu - \frac{n_0}{\ell^3})^2 a \ell^3$ in ${\mathcal H}_{\Lambda}^{\rm 2nd}$ and estimating the negative term. We get by \eqref{eq:apriori_n} and \eqref{eq:apriorinn+NEW}, upon using the simple estimate $n_{+}^2 \leq n n_{+}$,
\begin{align}
- \langle \widetilde{\Psi}, \ER (\rmu - \frac{n_0}{\ell^3})^2 a \ell^3 \widetilde{\Psi} \rangle
&\geq  - 2 \ER a \ell^{-3}\left\{ (n- \rmu \ell^3)^2
+  \langle \widetilde{\Psi}, n_{+}^2 \widetilde{\Psi} \rangle
\right\} \nonumber \\
&\geq
- C  \ER a \ell^{3} K_B^3 K_{\ell}^2 (\rmu a^3)^{\frac{1}{2}} \rmu^2 \nonumber \\
&=
-C \ER K_{\ell}^5 K_B^3 \rmu a.
\end{align}
So this term can be absorbed in the error term in \eqref{eq:IntroK(z)}.

The corresponding positive term, $\rmu a R^2 (\rmu - \frac{n_0}{\ell^3})^2 a \ell^3,$ is written in second quantization as 
\begin{align}\label{eq:ExtraQuadratic}
\rmu a^2 R^2 \ell^{-3} (\rmu \ell^3  - a_0^{\dagger} a_0)^2 
= \rmu a^2 R^2 \ell^{-3} \left( a_0^{\dagger} a_0^{\dagger} a_0 a_0 - (2 \rmu \ell^3-1) a_0^{\dagger} a_0 +  (\rmu \ell^3)^2\right),
\end{align}
and added to ${\mathcal H}_{\Lambda}^{\rm 2nd}$ defined in \eqref{eq:HLambdaSecond} above.
We define $\widetilde{\mathcal K}(z)$ to be the operator ${\mathcal H}_{\Lambda}^{\rm 2nd}+\ER (\rmu - \frac{n_0}{\ell^3})^2 a \ell^3$, but where the following substitutions have been performed:
\begin{align}\label{eq:subst-z}
&a_0^{\dagger} a_0^{\dagger} a_0 a_0 \mapsto |z|^4 - 4 |z|^2 + 2, \nonumber \\
&a_0^{\dagger} a_0 a_0   \mapsto |z|^2 z - 2 z, \qquad
a_0 a_0^{\dagger} a_0^{\dagger}    \mapsto |z|^2 \overline{z} , \nonumber \\
&a_0 a_0 \mapsto z^2,\qquad a_0^{\dagger} a_0^{\dagger} \mapsto \overline{z}^2,\qquad
a_0^{\dagger} a_0 \mapsto |z|^2 - 1, \nonumber \\
&a_0 \mapsto z,\qquad a_0^{\dagger} \mapsto \overline{z}.
\end{align}
Then, we will prove that
\begin{align}\label{eq:Csubst}
\langle \widetilde{\Psi}, \left( {\mathcal H}_{\Lambda}^{\rm 2nd} + \rmu a R^2 (\rmu - \frac{n_0}{\ell^3})^2 a \ell^3\right) \widetilde{\Psi} \rangle 
&= \Re \int \langle \Phi(z), \widetilde{\mathcal K}(z)   \Phi(z) \rangle \,d^2z \nonumber \\
&=  \Re \int \langle \widetilde{\Phi}(z), \widetilde{\mathcal K}(z)   \widetilde{\Phi}(z) \rangle n^2(z) \,d^2z,
\end{align}
where $n(z) = \| \Phi(z) \|_{{\mathcal F}_{\rm s}(\Ran(Q))}$ and $ \widetilde{\Phi}(z)  =  \Phi(z)/n(z)$.

To obtain \eqref{eq:Csubst}  we write all polynomials in $a_0, a_0^{\dagger}$ in anti-Wick ordering, for example $a_0^{\dagger} a_0 = a_0 a_0^{\dagger} - 1$. Therefore, using \eqref{eq:PartUnity},
\begin{align}
\langle \Psi, a_0^{\dagger} a_0 \Psi \rangle = \int \langle a_0^{\dagger} \Psi |z \rangle \langle z | a_{0}^{\dagger} \Psi \rangle - \langle \Psi| z \rangle \langle z | \Psi \rangle \,d^2 z = 
\int (|z|^2-1) \langle \Phi(z) | \Phi(z) \rangle \,d^2 z .
\end{align}
Performing this type of calculation for each term in $ {\mathcal H}_{\Lambda}^{\rm 2nd}$ yields \eqref{eq:Csubst}.

If $\widetilde{\Psi} \in {\mathcal F}_{\rm s}(L^2(\Lambda))$ satisfies \eqref{eq:LocalizedNPlus} then, for all $z \in {\mathbb C}$,
\begin{align}\label{eq:ControlNPlusRemains}
\widetilde{\Phi}(z)  =  1_{[0,\ML]}(n_{+}^L) \widetilde{\Phi}(z).
\end{align}
with $\widetilde{\Phi}(z) = \langle z | \widetilde{\Psi} \rangle$ as above, since $n_{+}^L$ acts in the orthogonal complement to $P$.

By a similar argument, one obtains \eqref{eq:Localized_n_AftercNumber}.

The next step of the proof is to remove the lower order terms coming from the substitutions in \eqref{eq:subst-z} above.

Let us first consider the negative term $-4|z|^2$ in the substitution of $a_0^{\dagger} a_0^{\dagger} a_0 a_0$ in $ {\mathcal H}_{\Lambda}^{\rm 2nd}$.
By undoing the integrations leading to $\widetilde{\mathcal K}(z)$ for this term, we see that it contributes with
\begin{align}\label{eq:ErrorInZ4} 
\int \langle \Phi(z), 
-4  \frac{1}{2} |z|^2 \ell^{-3} 
\Big(\widehat{g}(0) + \widehat{g\omega}(0)\Big)\Phi(z) \rangle\,d^2z 
&\geq
-C a \ell^{-3} \langle \widetilde{\Psi}, a_0 a_0^{\dagger} \widetilde{\Psi} \rangle  \nonumber \\
&\geq -C a \ell^{-3} (n+1),
\end{align}
in agreement with the error term in \eqref{eq:IntroK(z)} (using that by \eqref{eq:apriori_n} $n \approx \rmu \ell^3 \gg 1$).

An other term of a similar kind comes from the $a_0^{\dagger} a_0^{\dagger} a_0 a_0$ term in $\ER (\rmu - \frac{n_0}{\ell^3})^2 a \ell^3$. This is easily estimated by $C \ER \rmu a$, and can therefore also be absorbed in the error term in \eqref{eq:IntroK(z)}.
We notice in passing that the negative $(-1)$-part of the cross term in \eqref{eq:ExtraQuadratic} also has this magnitude and can therefore also be absorbed in the error term in \eqref{eq:IntroK(z)}.

We also estimate the term linear in $z$ coming from the substitution of $a_0^{\dagger} a_0 a_0$ in \eqref{eq:subst-z}. This substitution occurs twice, but we will only explicitly treat one of them, namely the term
\begin{align}
& \Re \int \Big\langle \Phi(z), -2 \ell^{-3} \widehat{W}_1(0) (2\pi)^{-3} \int \widehat{\chi}_{\Lambda}(k) a_k^{\dagger} z \,dk \Phi(z) \Big\rangle d^2 z \nonumber \\
&=
-2 \ell^{-3} \widehat{W}_1(0) (2\pi)^{-3}
\int \Big\langle \Phi(z),  \int \widehat{\chi}_{\Lambda}(k) (a_k^{\dagger} z + a_k \overline{z}) \,dk \Phi(z) \Big\rangle d^2 z \nonumber \\
&\geq -C a \ell^{-3} 
\int \Big\langle \Phi(z),  \int |\widehat{\chi}_{\Lambda}(k)| (  a_k^{\dagger} a_k +  |z|^2) \,dk \Phi(z) \Big\rangle d^2 z.
\end{align}
Notice that $|\widehat{\chi}_{\Lambda}(k)| = \ell^{3} |\widehat{\chi}(k\ell)|$ and that $\widehat{\chi} \in L^1({\mathbb R}^3)$ for $M \geq 4$.
Notice also that \eqref{eq:Localized_n_AftercNumber} holds for each $\Phi(z)$.
Redoing the calculation in \eqref{eq:ErrorInZ4} we therefore find 
\begin{align}
\Re \int \Big\langle \Phi(z), -2 \ell^{-3} \widehat{W}_1(0) (2\pi)^{-3} \int \widehat{\chi}_{\Lambda}(k) a_k^{\dagger} z \,dk \Phi(z) \Big\rangle d^2 z
\geq -C a \ell^{-3} (n+1),
\end{align}
in agreement with the error term in \eqref{eq:IntroK(z)}.

The other error terms from the substitutions are \eqref{eq:subst-z} estimated in a similar manner and we will leave out the details

Finally, we need to restrict to non-negative $z$.
Suppose $z = |z| e^{i\phi}$.
In the operator we can replace $a_{\pm k}$ by $e^{i\phi} a_{\pm k}$. In this way all occurences of $z$ will be replaced by $|z|$. 
Notice that this substitution will not affect the commutation relations. 
This finishes the proof.
\end{proof}

\section{First energy bounds}
\label{sec:rough}

In this section we will make a rough estimate on the energy. 
This rough estimate will be used to eliminate the values of $\rho_z$ that are far away from $\rmu$.

\begin{lemma}\label{lem:Roughbounds}
Assume that the parameters satisfy 
Assumption~\ref{assump:params} and that $\rmu a^3$ is small enough.
For any normalized state $\Phi \in {\mathcal F}_{\rm s}(\Ran(Q))$ satisfying  \eqref{eq:LocalizedAftercNumber} and \eqref{eq:Localized_n_AftercNumber} we have
\begin{align}
\label{eq:RoughBelow}
\langle \Phi, {\mathcal K}(z) \Phi \rangle &\geq - \frac{\widehat{g}(0)}{2} \rmu^2 \ell^3 + \frac{\widehat{g}(0)}{2} (\rmu - \rho_z)^2 \ell^3 - a (\rho_z + \rmu)^{3/2} \rmu^{1/2} \ell^3 \delta_1- \rho_z^2 a \ell^3 \delta_2
- \delta_0 \rho_z a \rmu \ell^3 
\nonumber  \\
&\quad 
 - C \rmu^2 a \ell^3\sqrt{\rmu a^3}.
\end{align}
with 
\begin{align}
\delta_0 &:= C \frac{\varepsilon_T^2}{d^2 K_{\ell}^4 \widetilde{K}_{{\mathcal M}}}   (1 + (K_H'')^{-2} K_{\ell}^2) \nonumber \\
 \delta_1&:=  \widetilde{K}_{{\mathcal M}}^{-\frac{1}{2}}, \nonumber \\
\delta_2 &:= C  \left(\frac{d^4 K_{\ell}^4}{\varepsilon_T^2}+ \frac{R a}{\ell^2} 
 \right).
\end{align}
\end{lemma}

Before we give the proof of Lemma~\ref{lem:Roughbounds} we will state its main consequence, Proposition~\ref{prop:FarAway} below.

Notice that by Section~\ref{sec:params} our choice of parameters ensures that $\delta_0 + \delta_1 + \delta_2 \ll 1$.

\begin{proposition}\label{prop:FarAway}
Suppose, for some sufficiently large universal constant $C>0$, we have
\begin{align}\label{eq:Close-I}
|\rho_z - \rmu | \geq C \rmu \max\left( \Big(\delta_0+ \delta_1 + \delta_2\big)^{\frac{1}{2}}, (\rmu a^3)^{\frac{1}{4}} \right).
\end{align}
and that $\delta_0+ \delta_1+\delta_2\leq 1/C$.
Then, for any state $\Phi$ satisfying \eqref{eq:LocalizedAftercNumber}, we have
\begin{align}
\langle \Phi, {\mathcal K}(z) \Phi \rangle
\geq - \frac{\widehat{g}(0)}{2} \rmu^2 \ell^3+ 2 \rmu^2 a \ell^3 \frac{128}{15 \sqrt{\pi}} \sqrt{\rmu a^3}.
\end{align}
\end{proposition}

\begin{proof}
By convexity, \eqref{eq:RoughBelow} implies the bound
\begin{align}
\langle \Phi, {\mathcal K}(z) \Phi \rangle &\geq - \frac{\widehat{g}(0)}{2} \rmu^2 \ell^3 + \frac{\widehat{g}(0)}{2} (1 - C(\delta_0+ \delta_1 +\delta_2))  (\rmu - \rho_z)^2 \ell^3
 \nonumber \\
&\quad  
- C \rmu^2 a \ell^3 \Big(\delta_0+ \delta_1 + \delta_2 +\sqrt{\rmu a^3}\Big) \nonumber \\
 &\geq
 - \frac{\widehat{g}(0)}{2} \rmu^2 \ell^3 + \frac{\widehat{g}(0)}{4}   (\rmu - \rho_z)^2 \ell^3 
 - C \rmu^2 a \ell^3 \Big(\delta_0+ \delta_1 + \delta_2 + \sqrt{\rmu a^3}\Big).
\end{align}
If \eqref{eq:Close-I} is satisfied, then the $(\rmu - \rho_z)^2$ term dominates both the error term above and the LHY correction. This finishes the proof of Proposition~\ref{prop:FarAway}.
\end{proof}

\begin{proof}[Proof of Lemma~\ref{lem:Roughbounds}]
For the proof of Lemma~\ref{lem:Roughbounds}] we will drop the positive term $\ER (\rmu - \rho_z)^2 a \ell^3$ from \eqref{eq:Kz}.

Since, for any $\delta'>0$,
$z a_k^{\dagger} + \overline{z} a_k \leq \delta' |z|^2 + (\delta')^{-1} a_k^{\dagger} a_k$,
we find, using \eqref{eq:LocalizedAftercNumber}, \eqref{eq:Localized_n_AftercNumber} and Lemma~\ref{lem:UsingML},
\begin{align}
& \langle \Phi, \int \widehat{\chi_{\Lambda}}(k) (z a_k^{\dagger} + \overline{z} a_k)\,dk\, \Phi  \rangle\nonumber \\
 &\leq \delta'  |z|^2 \int | \widehat{\chi_{\Lambda}}(k) |\,dk + |\widehat{\chi_{\Lambda}}(0) |
(\delta')^{-1} \langle \Phi,  \int_{\{|k| \leq K_H'' \ell^{-1}\}}  a_k^{\dagger} a_k\,dk\,\Phi \rangle \nonumber \\
&\quad + 
(\delta')^{-1} \langle \Phi,  \int_{\{|k| \geq K_H'' \ell^{-1}\}}  | \widehat{\chi_{\Lambda}}(k) | a_k^{\dagger} a_k\,dk\,\Phi \rangle \nonumber \\
 &\leq C ( \delta' |z|^2 + (\delta')^{-1} (\ML + \rmu \ell^3 \sup_{\{|k| \geq K_H'' \ell^{-1}\}} ( \ell^{-3}  | \widehat{\chi_{\Lambda}}(k) |))).
\end{align}
Using \eqref{cond:KHs-ML} and \eqref{eq:DecayPH-new}, both $ (\delta')^{-1} $ terms can be estimated by $\ML$.
Therefore, we easily get, setting $\delta' = \sqrt{ \ML/|z|^2)}$ and the definitions in \eqref{eq:DefQ1z} and \eqref{eq:DefQ1zex},
\begin{align}
\langle \Phi, \big({\mathcal Q}_1(z) + {\mathcal Q}_1^{\rm ex}(z) \big) \Phi \rangle \geq
-C  a \sqrt{\ML|z|^2} ( |\rho_z-\rmu| + \rho_z) ,
\end{align}
in agreement with the $\delta_1$ term in \eqref{eq:RoughBelow}.

Notice that quadratic terms of the form
$\ell^{3} \int \widehat{W}_1(k) a_k^{\dagger} a_k \,dk$ are easily estimated, using \eqref{eq:DomByML_ext},  as
\begin{align}
\pm \langle \Phi, \ell^{3} \int \widehat{W}_1(k) a_k^{\dagger} a_k \,dk \Phi \rangle
\leq C a (\ML + \langle \Phi, n_{+}^{H} \Phi \rangle).
\end{align}
This allows us to estimate all the quadratic terms in ${\mathcal K}(z)$ except the kinetic energy and the off-diagonal quadratic terms and to absorb the corresponding terms in the error in \eqref{eq:RoughBelow} (using in particular that $\ML \ll \rmu \ell^3$) and in the gap term(s) $n_{+}^{H}$ in ${\mathcal K}(z)$. To absorb the $n_+^H$-terms in the corresponding gap terms, recall that \eqref{con:eTdK} implies that $\frac{\varepsilon_T}{ (dK_\ell)^2}\gg 1$.

Therefore, to establish \eqref{eq:RoughBelow} we only have left to estimate the sum of the kinetic energy, ${\mathcal Q}_3(z)$ and the `off-diagonal' quadratic terms.
This we will do by first adding and subtracting an $n_{+}$ term, which is easily estimated as above.
We will prove the following 3 inequalities, where $\Phi$ is a state satisfying \eqref{eq:LocalizedAftercNumber} and \eqref{eq:Localized_n_AftercNumber}, and where $\varepsilon$ is chosen to be
\begin{align}\label{eq:ChoiceEpsilon}
\varepsilon=C_0 \varepsilon_{T}^{-2} d^4 K_{\ell}^4 ,
\end{align}
with $C_0$ a sufficiently large constant. The condition \eqref{con:eTdK} assures that $\varepsilon \ll 1$.

\begin{align}\label{eq:Rough0}
-  \left\langle \Phi, (2\pi)^{-3} \ell^3  \rho_z \varepsilon^{-1/2} a \int  a_k^{\dagger} a_k \,dk \,\Phi \right\rangle
 \geq - C \varepsilon^{-1/2}  \rho_z a (\ML + \langle \Phi, n_{+}^{H}\Phi \rangle )
\end{align}
\begin{align}\label{eq:Rough1}
\left\langle \Phi,  \left((2\pi)^{-3} \ell^3 \int \varepsilon \tau (k) a_k^{\dagger} a_k \,dk + {\mathcal Q}_3(z)\right) \, \Phi \right\rangle 
 \geq -\varepsilon^{-1} C \rho_z a \rmu \ell^3  \frac{\ML}{\rmu \ell^3} d^{-6} \Big(1 +  (K_H'')^{-2}  K_{\ell}^2
\Big),
\end{align}
and
\begin{align}\label{eq:Rough2}\frac{1}{2} \rho_z^2 \ell^3 \widehat{g\omega}(0)  +
(2\pi)^{-3} \ell^3 &\int \Big(\widetilde{\mathcal A}_1(k) a_k^{\dagger} a_k
 +\frac{1}{2}{\mathcal B}_1(k) \left( a_k^{\dagger} a_{-k}^{\dagger} + a_k a_{-k}\right) \Big) dk  \nonumber \\
 &\geq
 - 
 C \ell^3 \rho_z^2 a \left(\varepsilon + \frac{R a}{\ell^2}
 \right)
 - C \rmu^2 a \ell^3 \sqrt{\rmu a^3},
\end{align}
where we have introduced
\begin{align}
\widetilde{\mathcal A}_1(k) := (1-\varepsilon-\varepsilon_N)\tau(k) + \rho_z \varepsilon^{-1/2} a,\qquad
{\mathcal B}_1(k) :=  \rho_z \widehat{W}_1(k).
\end{align}

Notice that \eqref{eq:Rough0} is easy given the discussion above. The choice of $\varepsilon$ in \eqref{eq:ChoiceEpsilon} is dictated by the necessity of absorbing the $n_{+}^{H}$-term in \eqref{eq:Rough0} in the corresponding gap in ${\mathcal K}(z)$.

We proceed to prove \eqref{eq:Rough2}.  Using \eqref{eq:Cond_epsilonN}, \eqref{con:eTdK}, \eqref{con:sdKellKB} and \eqref{con:KBKell}, we find
\begin{align}
\label{cond:epsilonNdKlepsilonT}
\epsilon_N \ll d^4 K_{\ell}^4 \varepsilon_T^{-2}.
\end{align}
Notice that \eqref{cond:epsilonNdKlepsilonT} implies that for a lower bound we may replace $\widetilde{\mathcal A}_1(k)$ by the smaller
\begin{align}
{\mathcal A}_1(k) := (1-2\varepsilon)\left[ |k| - \frac{1}{2}(ds\ell)^{-1}\right]_{+}^2 + \rho_z \varepsilon^{-1/2} a.
\end{align}
We rewrite in a symmetrized form
\begin{align}\label{eq:symmetrized}
(2\pi)^{-3} \ell^3 &\int \Big({\mathcal A}_1(k) a_k^{\dagger} a_k
 +\frac{1}{2}{\mathcal B}_1(k) a_k^{\dagger} a_{-k}^{\dagger} + \frac{1}{2}{\mathcal B}_1(k)a_k a_{-k}\Big) dk  \\
 &=
 \frac{1}{2} (2\pi)^{-3} \ell^3 \int \Big({\mathcal A}_1(k) a_k^{\dagger} a_k
 +{\mathcal A}_1(k) a_{-k}^{\dagger} a_{-k}
 +{\mathcal B}_1(k) a_k^{\dagger} a_{-k}^{\dagger} + {\mathcal B}_1(k)a_k a_{-k}\Big) dk. \nonumber\\
 & \geq
 -  \frac{1}{2} (2\pi)^{-3} \ell^3 \int \left( {\mathcal A}_1(k)  - \sqrt{ {\mathcal A}_1(k)^2 - |{\mathcal B}_1(k) |^2}\right)\,dk,
\end{align}
where we used Lemma~\ref{eq:BogIneq} and \eqref{eq:Commutator}.

At this point we use Lemma~\ref{lm:aprioriintegral_new} with $\kappa_A=(1-2\varepsilon)$, $K_1= \frac{1}{2}\rho_z \varepsilon^{-1/2}$, $K_2=2\rho_z$ (using  \eqref{eq:W1-g-new}  and $R \ll \ell$), and $P_1=(ds\ell)^{-1}$.
Therefore, for $\varepsilon$ sufficiently small, the assumptions of Lemma~\ref{lm:aprioriintegral_new} are satisfied and we get
\begin{align}
& (2\pi)^{-3} \int \left( {\mathcal A}_1(k)  - \sqrt{ {\mathcal A}_1(k)^2 - |{\mathcal B}_1(k) |^2}\right)\,dk \nonumber \\
&\leq
\frac{1}{1-2\varepsilon} \rho_z^2 (2\pi)^{-3} \int \frac{\widehat{W}_1(k)^2}{2k^2} 
 + C (a (ds\ell)^{-1} + \varepsilon) \rho_z^2 \int \frac{\widehat{W}_1(k)^2}{2k^2}
+C \varepsilon^{1/2} \rho_z a (ds\ell)^{-3} \nonumber \\
&\quad +C (\rho_z a)^2 (ds\ell)^{-1} \log( d s \ell/a).
\end{align}
Using now \eqref{eq:I2-integral-new} we therefore find
\begin{align}\label{eq:NewIntegral}
& (2\pi)^{-3} \int \left( {\mathcal A}_1(k)  - \sqrt{ {\mathcal A}_1(k)^2 - |{\mathcal B}_1(k) |^2}\right)\,dk \nonumber \\
&\leq \rho_z^2 \widehat{g\omega}(0) + C \rho_z^2 a \left( \varepsilon + \frac{Ra}{\ell^2}
+a (ds\ell)^{-1}\left(1 + \log( d s \ell/a) \right)\right) + a (ds\ell)^{-6}.
\end{align}
Combing \eqref{con:eTdK}, \eqref{con:sdKellKB} and \eqref{con:KBKell} with the choice \eqref{eq:ChoiceEpsilon} and the fact that $s \ll 1$, we get that
$\varepsilon \gg  \frac{a}{d s \ell} \big(1 + \log\big( \frac{d s \ell}{a}\big)$.
Also, combining \eqref{con:sdKellKB} and \eqref{con:KBKell} with the fact that $K_{\ell} \gg 1$, we get that 
$a (d s \ell )^{-6} \ll \rmu^2 a  \sqrt{\rmu a^3}$.
Therefore, \eqref{eq:Rough2} follows upon combing \eqref{eq:NewIntegral} with \eqref{eq:symmetrized}.

To prove \eqref{eq:Rough1} we use a similar approach. Notice that by definition \eqref{eq:kroelleQ3z}, ${\mathcal Q}_3(z)$ only lives in the high momentum region $P_H$.
For these momenta we have $\tau(k) \geq \frac{1}{2} k^2$. Therefore, dropping a part of the kinetic energy, it suffices to bound
\begin{align}
\ell^3 (2\pi)^{-3} \int_{\{ k \in P_H\}} \left(\frac{\varepsilon}{2} k^2 a_k^{\dagger} a_k +
(2\pi)^{-3} \int f_L(s)
\widehat{W}_1(k) (\overline{z} \widetilde{a}_s^{\dagger}  a_{s-k} a_k
+ a_k^{\dagger}  a^{\dagger}_{s-k} \widetilde{a}_s z) \,ds\right)\,dk.
\end{align}
We estimate, with $\widetilde{b}_k := a_k + 2(2\pi)^{-3} \int f_L(s)
\frac{ \widehat{W}_1(k)}{\varepsilon k^2} z a_{s-k}^{\dagger} \widetilde{a}_s   \,ds$,
\begin{align}\label{eq:RoughQ3_1}
&\ell^3 (2\pi)^{-3} \int_{\{ k \in P_H\}} \left(\frac{\varepsilon}{2} k^2 a_k^{\dagger} a_k +
(2\pi)^{-3} \int f_L(s)
\widehat{W}_1(k) (\overline{z} \widetilde{a}_s^{\dagger}  a_{s-k} a_k
+ a_k^{\dagger}  a^{\dagger}_{s-k} \widetilde{a}_s z) \,ds\right)\,dk\nonumber \\
&=\ell^3 (2\pi)^{-3} \int_{\{ k \in P_H\}}\left(
\frac{\varepsilon}{2} k^2 \widetilde{b}_k^{\dagger} \widetilde{b}_k
- 4 (2\pi)^{-6} 
\iint f_L(s) f_L(s') \frac{ \widehat{W}_1(k)^2}{\varepsilon k^2} |z|^2 \widetilde{a}_{s'}^{\dagger}  a_{s'-k} a_{s-k}^{\dagger} \widetilde{a}_s\right)dk \nonumber \\
&\geq- 4\varepsilon^{-1}
\ell^3 (2\pi)^{-9} \int_{\{ k \in P_H\}}  \frac{ \widehat{W}_1(k)^2}{ k^2} |z|^2
\iint f_L(s) f_L(s') \widetilde{a}_{s'}^{\dagger}  (a_{s-k}^{\dagger} a_{s'-k} +[a_{s'-k}, a_{s-k}^{\dagger}]) \widetilde{a}_s .
\end{align}
On the term without a commutator, we estimate $\widetilde{a}_{s'}^{\dagger}  a_{s-k}^{\dagger} a_{s'-k}\widetilde{a}_s$ by Cauchy-Schwarz and (since $k \in P_H= \{ |k| \geq K_H'' \ell^{-1} \} $),
$\frac{ \widehat{W}_1(k)^2}{k^2} \leq C (K_H'')^{-2} a^2 \ell^2$.
For a $\Phi$ satisfying \eqref{eq:LocalizedAftercNumber} and \eqref{eq:Localized_n_AftercNumber} we find, using \eqref{eq:DomByML} and the support of $f_L$ as defined in \eqref{eq:def_fL},
\begin{align}
\langle \Phi, \ell^6 \iint f_L(s) \widetilde{a}_{s}^{\dagger}  a_{s'-k}^{\dagger} a_{s'-k}  \widetilde{a}_s \,dsdk\Phi \rangle 
&=
\langle \Phi, \ell^6 \iint f_L(s)  a_{k}^{\dagger}\widetilde{a}_{s}^{\dagger} \widetilde{a}_s   a_{k}  \,dsdk\Phi \rangle \nonumber \\
&\leq C \ML \rmu \ell^3.
\end{align}
Therefore, it follows that
\begin{align}\label{eq:RoughQ3_2}
\langle \Phi, &
\ell^3 \int_{\{ k \in P_H\}}  \frac{ \widehat{W}_1(k)^2}{ k^2} |z|^2 \iint f_L(s) f_L(s') \widetilde{a}_{s'}^{\dagger}  a_{s-k}^{\dagger} a_{s'-k}  \widetilde{a}_s \Phi \rangle \nonumber \\
&\leq
C \rho_z  (\int_{\{|s| \leq 2 d^{-2} \ell^{-1} \}}\,ds )(K_H'')^{-2} a^2 \ell^2  \ML \rmu \ell^3 \nonumber \\
&\leq
C a \rho_z  \rmu \ell^3 d^{-6}   (K_H'')^{-2} (a / \ell)  \ML .
\end{align}

For the commutator term, we estimate (using \eqref{eq:Commutator_general} and the Cauchy-Schwarz inequality)
$$
\widetilde{a}_{s'}^{\dagger} [a_{s'-k}, a_{s-k}^{\dagger}]) \widetilde{a}_s + h.c.\leq 2 \widetilde{a}_{s'}^{\dagger}\widetilde{a}_{s'} + 2 \widetilde{a}_s^{\dagger} \widetilde{a}_s
$$
and $\int \frac{ \widehat{W}_1(k)^2}{k^2} \leq Ca$.
This leads to (for a $\Phi$ satisfying \eqref{eq:LocalizedAftercNumber}),
\begin{align}\label{eq:RoughQ3_3}
\langle \Phi, \ell^3 &\int_{\{ k \in P_H\}}  \frac{ \widehat{W}_1(k)^2}{ k^2} |z|^2
\iint f_L(s) f_L(s') \widetilde{a}_{s'}^{\dagger} [a_{s'-k}, a_{s-k}^{\dagger}] \widetilde{a}_s \Phi \rangle \nonumber \\
&
\leq 
C \ML a |z|^2 \int_{\{|s| \leq 2 d^{-2} \ell^{-1} \}}\,ds \nonumber \\
&\leq C a \rho_z \ML d^{-6}.
\end{align}
Combining the estimates \eqref{eq:RoughQ3_1}, \eqref{eq:RoughQ3_2} and \eqref{eq:RoughQ3_3} proves \eqref{eq:Rough1}.

We will add the estimates of \eqref{eq:Rough0}, \eqref{eq:Rough1} and \eqref{eq:Rough2} with $\varepsilon$ from \eqref{eq:ChoiceEpsilon}. 
Notice that, by choosing $C_0$ sufficiently large, the $n_{+}^{H}$ term in \eqref{eq:Rough0} can be absorbed in the $n_{+}^{H}$-gap in ${\mathcal K}(z)$.

Notice that since $\varepsilon^{-1/2} \leq \varepsilon^{-1}$ the 
$\ML$ contribution from \eqref{eq:Rough0} will be smaller than the corresponding term in \eqref{eq:Rough1}.
Clearly, with this choice of $\varepsilon$, the term in \eqref{eq:Rough1} is controlled by the $\delta_0$-term in \eqref{eq:RoughBelow}.
Similarly, inserting the choice of $\varepsilon$ in \eqref{eq:Rough2} we get the $\delta_2$ term as well as the final term in \eqref{eq:RoughBelow}.

This finishes the proof of \eqref{eq:RoughBelow}.
\end{proof}

\section{More precise energy estimates}
\label{sec:Precise}
From Proposition~\ref{prop:FarAway} above, we see that the energy is too high unless $\rho_z \approx \rmu$. In this section we will give precise energy bounds in the complementary regime. We will always assume that
\begin{align}\label{eq:Close}
|\rho_z - \rmu | \leq \rmu C \max\left( \big(\delta_0+ \delta_1 + \delta_2\big)^{\frac{1}{2}}, (\rmu a^3)^{\frac{1}{4}} \right),
\end{align}
with the notation from Proposition~\ref{prop:FarAway}.

We will need the condition that
\begin{align}\label{eq:NewConditionNew}
K_{\ell}^2 \max\left( \big(\delta_0+\delta_1 + \delta_2\big)^{\frac{1}{2}}, (\rmu a^3)^{\frac{1}{4}} \right) \leq C^{-1},
\end{align}
for some sufficiently large universal constant.
This condition is satisfied by \eqref{con:ER}, \eqref{con:eTdK}, \eqref{eq:maerkedobbelt} and \eqref{cond:KH3-n}.

Notice, using \eqref{eq:Close} and \eqref{eq:NewConditionNew}, that
\begin{align}
\frac{|\rho_z-\rmu|}{\rmu} \leq C^{-1} K_{\ell}^{-2}.
\end{align}

We define the quadratic Bogoliubov Hamiltonian as follows,
\begin{align}\label{eq:BogHamPositive-1}
{\mathcal K}^{\rm Bog} =
\frac{1}{2} (2\pi)^{-3} \ell^3 \int \Big(& {\mathcal A}(k) (a_k^{\dagger} a_k + a^{\dagger}_{-k} a_{-k}) + {\mathcal B}(k) (a_k^{\dagger} a_{-k}^{\dagger} +  a_{k} a_{-k} )\nonumber \\
&+
{\mathcal C}(k) (a^{\dagger}_k + a^{\dagger}_{-k}+a_k + a_{-k}) \Big)
 \,dk ,
\end{align}
with
\begin{align}\label{eq:defABC}
{\mathcal A}(k) :=
(1-\varepsilon_N) \tau(k) + \rho_z \widehat{W_1}(k), \quad
{\mathcal B}(k) :=\rho_z \widehat{W_1}(k),\quad
{\mathcal C}(k):=  \ell^{-3} (\rho_z - \rmu) \widehat{W_1}(0) \widehat{\chi_{\Lambda}}(k) z\,.\quad
\end{align}
With this notation, we can rewrite/estimate ${\mathcal K}(z)$ from \eqref{eq:Kz} as follows,
\begin{align}\label{eq:SimplifyKz}
{\mathcal K}(z) 
& = {\mathcal K}^{\rm Bog} 
+ \frac{1}{2} \rho_z^2 \ell^{3} 
\Big(\widehat{g}(0) + \widehat{g\omega}(0)\Big)
  - \rmu   \widehat{g}(0) \rho_z \ell^3
  + \ER (\rmu - \rho_z)^2 a \ell^3
  \nonumber \\
  &\quad  +  \frac{b}{2 \ell^2} n_{+} + \varepsilon_T\frac{b}{8 d^2 \ell^2} n_{+}^{H} 
  +\varepsilon_T\frac{b |z|^2}{16 d^2 \rmu \ell^5} n_{+}^{H}
  + (\rho_z - \rmu) \widehat{W}_1(0) (2\pi)^{-3} \ell^{3} \int a_k^{\dagger} a_k\,dk \nonumber \\
  &\quad
  + {\mathcal Q}_1^{\rm ex}(z)+
  {\mathcal Q}_2^{\rm ex}(z) + {\mathcal Q}_3(z) \nonumber \\
 & \geq
  - \frac{1}{2} \rmu^2 \ell^{3} \widehat{g}(0)
+ \frac{1}{2} \rho_z^2 \ell^{3} \widehat{g\omega}(0)
+ \frac{1}{2} (\rho_z-\rmu)^2 \ell^{3} \widehat{g}(0) + \ER (\rmu - \rho_z)^2 a \ell^3+ {\mathcal K}^{\rm Bog} \nonumber \\
  &\quad  +  \frac{b}{4 \ell^2} n_{+} + \varepsilon_T\frac{b}{8 d^2 \ell^2} n_{+}^{H} 
  + {\mathcal Q}_1^{\rm ex}(z)+
  {\mathcal Q}_2^{\rm ex}(z) + {\mathcal Q}_3(z).
\end{align}
with the notations ${\mathcal Q}_1^{\rm ex}(z)$, ${\mathcal Q}_2^{\rm ex}(z)$,  and ${\mathcal Q}_3(z)$ from \eqref{eq:DefQ1zex}, \eqref{eq:kroelleQ3z}, and \eqref{eq:DefQ2ex}, and
where we used \eqref{eq:NewConditionNew} to absorb the quadratic operator $(\rho_z - \rmu) \widehat{W}_1(0) (2\pi)^{-3} \ell^{3} \int a_k^{\dagger} a_k\,dk$ in the gap.

\subsection{The Bogoliubov Hamiltonian}
\begin{theorem}[Analysis of Bogoliubov Hamiltonian]\label{thm:BogHamDiag}
Assume that $\Phi$ satisfies \eqref{eq:LocalizedAftercNumber} and \eqref{eq:Localized_n_AftercNumber}, that $\frac{9}{10} \rmu \leq \rho_z \leq \frac{11}{10} \rmu$ as well as Assumption~\ref{assump:params} and that $\rmu a^3$ is sufficiently small.
Then,
\begin{align}\label{eq:BogHamDiag}
&\frac{1}{2} \rho_z^2 \ell^3\widehat{g\omega}(0)+
 \frac{1}{2} (\rho_z-\rmu)^2 \ell^{3} \widehat{g}(0)
+ \langle \Phi, {\mathcal K}^{\rm Bog} \Phi \rangle \nonumber \\
&\geq
(2\pi)^{-3} \ell^3 \langle \Phi, \int {\mathcal D}_k b_k^{\dagger} b_{k}\,dk \, \Phi \rangle  
+ 4\pi  \frac{128 }{15\sqrt{\pi}} \rho_z^2 a \sqrt{\rho_z a^3} \ell^3 - C (\rmu - \rho_z)^2 a \ell^3 \rmu a R^2
 \nonumber\\
&\quad
- C \rmu^2 a \ell^3  \left( \frac{R a}{\ell^2} 
+ \varepsilon_N 
+ \epsilon(\rmu)\sqrt{\rmu a^3} 
+ \sqrt{\rmu \ell^3}( K_H'')^{4-M}\right).
\end{align}
Here
\begin{align}\label{eq:LebesgueError-Total}
\epsilon(\rmu) &= (\rmu a)^{\frac{1}{4}} \sqrt{R} + \varepsilon_T +
 (K_{\ell} s)^{-1} \left( 1 + \log(d^{-1}) + \log( \frac{K_{\ell} d s}{(\rmu a^3)^{1/2}}) \right) \nonumber \\
 &\quad 
+ \varepsilon_T (K_{\ell} d s)^{-1} \left( 1 +\log( \frac{K_{\ell} d s}{(\rmu a^3)^{1/2}}) \right).
\end{align}
Also, with ${\mathcal A}, {\mathcal B}$ as defined in \eqref{eq:defABC},
\begin{align}\label{eq:defDk}
{\mathcal D}_k := \frac{1}{2}\left( {\mathcal A}(k) + \sqrt{{\mathcal A}(k)^2 - {\mathcal B}(k)^2}\right),
\end{align}
and
\begin{align}\label{eq:Defb_k}
b_k :=  a_{k} + \alpha_k a_{-k}^{\dagger} + c_k,
\end{align}
with
\begin{align}
\alpha_k:= {\mathcal B}(k)^{-1} \left( {\mathcal A}(k) - \sqrt{{\mathcal A}(k)^2 - {\mathcal B}(k)^2}\right),
\end{align}
and
\begin{align}
c_k:=\begin{cases}\frac{2{\mathcal C}(k)}{{\mathcal A}(k) + {\mathcal B}(k) + \sqrt{{\mathcal A}(k)^2 - {\mathcal B}(k)^2}}, & |k| \leq \frac{1}{2}
K_H'' \ell^{-1},\\
0, & |k| > \frac{1}{2} K_H'' \ell^{-1}.
\end{cases}
\end{align}
\end{theorem}

\begin{remark}\label{rem:commutatorsMixed}
We notice that following commutation relations (using the ones for the $a_k$'s \eqref{eq:Commutator_general} and that $\widehat{\chi}$ is even and real).
\begin{align}\label{eq:bComms1}
[ b_k, b_{k'}]
= (\alpha_k - \alpha_{k'}) \Big(  \widehat{\chi^2}((k+k')\ell) - \widehat{\chi}(k \ell) \widehat{\chi}(k'\ell) \Big).
\end{align}
Also,
\begin{align}\label{eq:bComms2}
[ b_k, b_{k'}^{\dagger}] =
(1-\alpha_k \alpha_{k'}) \Big(\widehat{\chi^2}((k-k')\ell) -  \widehat{\chi}(k \ell) \widehat{\chi}(k'\ell) \Big).
\end{align}

We also have
\begin{align}
[\widetilde{a}_s^{\dagger} , b_{-k}^{\dagger}] = \alpha_{-k} [ \widetilde{a}_s^{\dagger}, a_k ] =  \alpha_{-k} \ell^{-3} \langle e^{isx}, Q \chi_{\Lambda} e^{ikx}\rangle,
\end{align}
and
\begin{align}
[\widetilde{a}_s, b_{-k}^{\dagger}] = [\widetilde{a}_s, a_{-k}^{\dagger}] = \ell^{-3} \langle e^{isx}, Q \chi_{\Lambda} e^{-ikx}\rangle
= \widehat{\chi}((s+k)\ell) - \widehat{\theta}(s\ell) \widehat{\chi}(k\ell).
\end{align}

\end{remark}

\begin{proof}[Proof of Theorem~\ref{thm:BogHamDiag}]
To simplify later calculations we start by removing the linear terms, i.e. ${\mathcal C}(k)$, for $|k| > \frac{1}{2} K_H'' \ell^{-1}$ from ${\mathcal K}^{\rm Bog}$, so we aim to prove
\begin{align}\label{eq:RemoveC's}
\langle \Phi, 
\frac{1}{2} (2\pi)^{-3} \ell^3 &\int_{\{|k| > \frac{1}{2} K_H'' \ell^{-1}\}} {\mathcal C}(k) (a^{\dagger}_k + a^{\dagger}_{-k}+a_k + a_{-k}) \Big)
 \,dk\,  \Phi \rangle \nonumber \\
 &\geq - C \rmu^2 a \ell^3 \sqrt{\rmu \ell^3}( K_H'')^{4-M}.
\end{align}
Since $a_k+ a_k^{\dagger} \leq a_k^{\dagger} a_k + 1$,
we find
\begin{align}
\frac{1}{2}& (2\pi)^{-3} \ell^3 \int_{\{|k| > \frac{1}{2} K_H''\ell^{-1}\}} {\mathcal C}(k) (a^{\dagger}_k + a^{\dagger}_{-k}+a_k + a_{-k}) 
 \,dk \nonumber \\
 &\geq
 - (2\pi)^{-3} |\rho_z -\rmu| \widehat{W_1}(0) |z|
 \int_{\{|k| > \frac{1}{2} K_H'' \ell^{-1}\}} |\widehat{\chi_{\Lambda}}(k)| (a_k^{\dagger} a_k +1)\,dk \nonumber \\
 &\geq
 -C |\rho_z -\rmu| \widehat{W_1}(0) |z| ( n_{+} + 1) \epsilon(\chi),
\end{align}
where
\begin{align}
\epsilon(\chi) := \ell^{-3}  \sup_{\{|k| > \frac{1}{2} K_H'' \ell^{-1}\}} (1+(k\ell)^2)^2  |\widehat{\chi_{\Lambda}}(k)|
\leq
C ( K_H'')^{4-M},
\end{align}
where we used Lemma~\ref{lem:DecaychiHat} to get the last estimate.
Estimating $n_{+}$ using \eqref{eq:Localized_n_AftercNumber} and using the assumed control of $|z|$, 
we get \eqref{eq:RemoveC's}.

By the estimate above, it suffices to consider
\begin{align}
\widetilde{\mathcal K}^{\rm Bog} :=
\frac{1}{2} (2\pi)^{-3} \ell^3 \int \Big( &{\mathcal A}(k) (a_k^{\dagger} a_k + a^{\dagger}_{-k} a_{-k}) + {\mathcal B}(k) (a_k^{\dagger} a_{-k}^{\dagger} +  a_{k} a_{-k} )\nonumber \\
&+
\widetilde{\mathcal C}(k) (a^{\dagger}_k + a^{\dagger}_{-k}+a_k + a_{-k}) \Big)
 \,dk ,
\end{align}
with ${\mathcal A}, {\mathcal B}$ from \eqref{eq:BogHamPositive-1} and
\begin{align}
\widetilde{\mathcal C}(k):=  \begin{cases} 0, & |k| \geq \frac{1}{2} K_{H}'' \ell^{-1}, \\
\ell^{-3} (\rho_z - \rmu) \widehat{W_1}(0) \widehat{\chi_{\Lambda}}(k) z, & \text{ else}.
\end{cases}
\end{align}

With the notation from Theorem~\ref{thm:BogHamDiag} and using Theorem~\ref{thm:bogolubov-complete} combined with \eqref{eq:Commutator} we find
\begin{align}\label{eq:LowerBoundHBog}
\widetilde{\mathcal K}^{\rm Bog} 
&\geq
(2\pi)^{-3} \ell^3 \int {\mathcal D}_k b_k^{\dagger} b_{k}\,dk  \nonumber\\
&\quad-\frac{1}{2} (2\pi)^{-3} \ell^3 \int  \left( {\mathcal A}(k) - \sqrt{{\mathcal A}(k)^2 - {\mathcal B}(k)^2}\right) \,dk \nonumber \\
&\quad - (\rho_z - \rmu)^2 \widehat{W_1}(0)^2 z^2  (2\pi)^{-3} \ell^{-3} \int_{\{|k| \leq \frac{1}{2} K_{H}'' \ell^{-1}\}} \frac{|{\widehat \chi}_{\Lambda}(k)|^2}{{\mathcal A}(k)+ {\mathcal B}(k)}.
\end{align}
The calculation of the integral $\int {\mathcal A}(k) - \sqrt{{\mathcal A}(k)^2 - {\mathcal B}(k)^2} \,dk$ is carried out in Lemma~\ref{lem:BogIntegral}.
The result is that
\begin{align}\label{eq:BogIntegral_new}
&-\frac{1}{2} (2\pi)^{-3} \ell^3 \int  \left( {\mathcal A}(k) - \sqrt{{\mathcal A}(k)^2 - {\mathcal B}(k)^2}\right) \,dk + \frac{\widehat{g\omega}(0)}{2} \rho_z^2 \ell^3
\nonumber \\
&
\geq 
4\pi  \frac{128 }{15\sqrt{\pi}} \rho_z^2 a \sqrt{\rho_z a^3} \ell^3 
- C \epsilon(\rmu) \rmu^2 a\sqrt{\rmu a^3} \ell^3
- C \rmu^2 \ell^3 \frac{R a}{\ell^2} a
-C \varepsilon_N a \rmu^2 \ell^3,
\end{align}
with $\epsilon(\rmu)$ defined in \eqref{eq:LebesgueError-Total}.

It is elementary, using that $W_1$ is even, that 
\begin{align}
\left| \widehat{W}_1(k) - \widehat{W}_1(0) \right|
\leq C a (kR)^2.
\end{align}
So for $|k| \leq C \sqrt{\rmu a}$, we find that ${\mathcal A}(k)+ {\mathcal B}(k)  \geq 2 \rho_z \widehat{W_1(0)} (1 - C (\rmu a^3) \frac{R^2}{a^2})$ (for sufficiently small values of $\rmu a^3$) using \eqref{con:ER} and \eqref{eq:W1-g-new}.
This lower bound extends to all $k$, since the kinetic energy is dominating unless $|k| \leq C \sqrt{\rmu a}$. Therefore, we find
that the last term in \eqref{eq:LowerBoundHBog} becomes controlled as
\begin{align}
&(\rho_z - \rmu)^2 \widehat{W_1}(0)^2 z^2  (2\pi)^{-3} \ell^{-3} \int_{\{|k| \leq \frac{1}{2} K_{H}'' \ell^{-1}\}} \frac{|{\widehat \chi}_{\Lambda}(k)|^2}{{\mathcal A}(k)+ {\mathcal B}(k)} \nonumber \\
&\leq
(\rho_z - \rmu)^2 \frac{\widehat{W_1}(0)}{2} \ell^3 (1 + C (\rmu a^3) \frac{R^2}{a^2})\nonumber \\
&\leq
(\rho_z - \rmu)^2 \frac{\widehat{g}(0)}{2} \ell^3 (1 + C (\rmu a^3)\frac{R^2}{a^2}),
\end{align}
where we used \eqref{eq:W1-g-new} and that $\ell^{-2} \ll \rmu a$ to get the last estimate.

Combining this estimate with \eqref{eq:RemoveC's} and \eqref{eq:BogIntegral_new} finishes the proof of Theorem~\ref{thm:BogHamDiag}.
\end{proof}

\subsection{The control of ${\mathcal Q}_3(z)$}
\label{subsec:Q3}
Recall the notations ${\mathcal Q}_1^{\rm ex}(z)$, ${\mathcal Q}_2^{\rm ex}(z)$,  and ${\mathcal Q}_3(z)$ from \eqref{eq:DefQ1zex}, \eqref{eq:kroelleQ3z}, and \eqref{eq:DefQ2ex}.
The quadratic Hamiltonian $(2\pi)^{-3} \ell^3 \int {\mathcal D}_k b_k^{\dagger} b_{k}\,dk$ from \eqref{eq:BogHamDiag} turns out to control the $3Q$-term ${\mathcal Q}_3(z)$ from \eqref{eq:kroelleQ3z}.
This we summarize as follows

\begin{theorem}\label{thm:Control3Q}
Assume that $\Phi$ satisfies \eqref{eq:LocalizedAftercNumber} and \eqref{eq:Localized_n_AftercNumber}.
Assume furthermore that \eqref{eq:Close} and 
Assumption~\ref{assump:params} are satisfied and that $\rmu a^3$ is sufficiently small.
Then,
\begin{align}\label{eq:Q3Cancels}
\Big\langle &
\Phi, \Big((2\pi)^{-3} \ell^3 \int_{\{|k| \geq \frac{1}{2} K_H'' \ell^{-1}\}} {\mathcal D}_k b_k^{\dagger} b_{k}\,dk
+
{\mathcal Q}_3(z) + {\mathcal Q}_2^{\rm ex}
+ {\mathcal Q}_1^{\rm ex}(z)
+ \frac{b}{50} (\frac{1}{\ell^2} n_{+} + \frac{\varepsilon_{T}}{(d \ell)^2} n_{+}^H)
\Big) \Phi \Big\rangle \nonumber \\
&\geq
-C \rmu^2 a \ell^3 \Big[
\sqrt{\frac{(K_H'')^2}{\widetilde{K}_{\mathcal M}}} K_{\ell}^{-1} (\rmu a^3)^{\frac{1}{2}} + \widetilde{K}_{\mathcal M}^{-\frac{1}{2}} d^{2M}
 + ( K_H'' )^{-4M-4} K_{\ell} d^{-12} (\rmu a^3)^{\frac{1}{2}}\nonumber \\
&\quad\quad\quad\quad\quad\quad +
(\rmu a^3) (K_H'')^{-6}  K_{\ell}^8 d^{-6}
\Big].
\end{align}
\end{theorem}

\begin{proof}[Proof of Theorem~\ref{thm:Control3Q}]

Notice that (with the notation from \eqref{eq:defABC})
\begin{align}\label{eq:BoverA}
|{\mathcal B}(k)/{\mathcal A}(k)|  
\leq C \rho_z a |k|^{-2} \leq 
C K_{\ell}^2 (K_H'')^{-2} , \qquad \forall |k| \geq \frac{1}{2} K_H'' \ell^{-1}.
\end{align}
In particular, $|{\mathcal B}(k)/{\mathcal A}(k)| \leq \frac{1}{2}$, for $\rmu$ sufficiently small and $|k| \geq \frac{1}{2} K_H'' \ell^{-1}$.

This implies, by expansion of the square root that for all $|k| \geq \frac{1}{2} K_H'' \ell^{-1}$,
\begin{align}\label{eq:alphaInPH}
| \alpha_k| = | {\mathcal B}(k)|^{-1} \left( {\mathcal A}(k) - \sqrt{{\mathcal A}(k)^2 - {\mathcal B}(k)^2}\right)| 
\leq C \rho_z a |k|^{-2} 
\leq C K_{\ell}^2 (K_H'')^{-2}.
\end{align}
In particular, \eqref{eq:BoverA} and \eqref{eq:alphaInPH} are valid for $k=k'-s$, when $s \in {\mathcal P}_{\rm low}$ and $k' \in {\mathcal P}_{\rm high}$.

Similarly, we get for ${\mathcal D}_k$ defined in \eqref{eq:defDk} and for all $|k| \geq \frac{1}{2} K_H'' \ell^{-1}$,
\begin{align}\label{eq:Dk_in_PH}
{\mathcal D}_k \geq \frac{1}{2} k^2 \geq \frac{1}{8} (K_H'')^2 \ell^{-2}
\end{align}

For later convenience, we reformulate the first-order operator in \eqref{eq:Q3Cancels} in terms of the $\widetilde{a}_s$. We get
\begin{align}
{\mathcal Q}_1^{\rm ex}(z)&=\rho_z z \widehat{W_1 \omega}(0)
(2\pi)^{-3} \int \widehat{\chi_{\Lambda}}(s)  ( a_s^{\dagger} +  a_s)\,ds \nonumber \\
&=
\rho_z z \widehat{W_1 \omega}(0)
(2\pi)^{-3} \int \widehat{\chi_{\Lambda}^2}(s)  ( \widetilde{a}_s^{\dagger} + \widetilde{a}_s)\,ds \nonumber \\
&= \rho_z z \widehat{W_1 \omega}(0)
(2\pi)^{-3}\ell^3  \int \widehat{\chi^2}(s\ell)  ( \widetilde{a}_s^{\dagger} + \widetilde{a}_s)\,ds.
\end{align}

We start by rewriting ${\mathcal Q}_3(z)$ in terms of the $b_k$'s defined in \eqref{eq:Defb_k}.
Notice that $c_k, c_{s-k}=0$ if $k \in {\mathcal P}_{\rm high}$ and $s\in {\mathcal P}_{\rm low}$.
We find the basic relation (we will freely use that all involved functions are symmetric, e.g. $\alpha_k = \alpha_{-k}$)
\begin{align}
a_{s-k} &= \frac{1}{1-\alpha_{s-k}^2} \Big(b_{s-k} - \alpha_{s-k} b_{k-s}^{\dagger}
\Big), &
a_k &= \frac{1}{1-\alpha_k^2} \Big(b_k - \alpha_k b_{-k}^{\dagger}
\Big).
\end{align}
Therefore,
\begin{align}\label{eq:asInTermsOfbs}
&a_{s-k} a_k  \\
&= \frac{1}{1-\alpha_k^2}  \frac{1}{1-\alpha_{s-k}^2}\Big(
b_{s-k} b_k  - \alpha_k b_{-k}^{\dagger} b_{s-k} - \alpha_{s-k} b_{k-s}^{\dagger} b_k
+ \alpha_k \alpha_{s-k} b_{k-s}^{\dagger}  b_{-k}^{\dagger} 
-\alpha_{k} [b_{s-k}, b_{-k}^{\dagger}]
\Big) .\nonumber
\end{align}

We will decompose ${\mathcal Q}_3(z)$ according to the different terms in \eqref{eq:asInTermsOfbs}, i.e.
\begin{align}
{\mathcal Q}_3(z) = {\mathcal Q}_3^{(1)}(z) + {\mathcal Q}_3^{(2)}(z) + {\mathcal Q}_3^{(3)}(z) + {\mathcal Q}_3^{(4)}(z), 
\end{align}
where
\begin{align}
{\mathcal Q}_3^{(1)}(z) &:= z \ell^3 (2\pi)^{-6} \iint_{\{ k \in {\mathcal P}_{\rm high}\}}
\frac{f_L(s)
\widehat{W}_1(k)}{(1-\alpha_k^2)(1-\alpha_{s-k}^2)} \left(\widetilde{a}_s^{\dagger}  
b_{s-k} b_k + \alpha_k \alpha_{s-k} \widetilde{a}_s^{\dagger}   b_{k-s}^{\dagger} b_{-k}^{\dagger} + h.c.
 \right), \nonumber \\
{\mathcal Q}_3^{(2)}(z) &:= - z \ell^3 (2\pi)^{-6} \iint_{\{ k \in {\mathcal P}_{\rm high}\}}
\frac{f_L(s)
\widehat{W}_1(k)\alpha_k}{(1-\alpha_k^2)(1-\alpha_{s-k}^2)} \left(\widetilde{a}_s^{\dagger}  b_{-k}^{\dagger} b_{s-k}
+ b_{s-k}^{\dagger}  b_{-k} \widetilde{a}_s \right), \nonumber \\
{\mathcal Q}_3^{(3)}(z) &:= - z \ell^3 (2\pi)^{-6} \iint_{\{ k \in {\mathcal P}_{\rm high}\}}
\frac{f_L(s)
\widehat{W}_1(k)\alpha_{s-k}}{(1-\alpha_k^2)(1-\alpha_{s-k}^2)} \left(\widetilde{a}_s^{\dagger}  b_{k-s}^{\dagger} b_{k}
+ b_{k}^{\dagger}  b_{k-s} \widetilde{a}_s \right), \nonumber \\
\intertext{and}
{\mathcal Q}_3^{(4)}(z) &:= (2\pi)^{-6} z \ell^3 \iint_{\{k \in {\mathcal P}_{\rm high}\}} f_L(s) \widehat{W}_1(k) \frac{-\alpha_k}{(1-\alpha_k^2)(1-\alpha_{s-k}^2)}[b_{s-k},b_{-k}^{\dagger}] (\widetilde{a}_s^{\dagger} + \widetilde{a}_s).
\end{align}
The different ${\mathcal Q}_3^{(j)}(z)$'s will be estimated individually. The result of this is summarized in Lemma~\ref{lem:EstimatesOnQ3-js}.
Theorem~\ref{thm:Control3Q} follows by adding the estimates of Lemma~\ref{lem:EstimatesOnQ3-js}.
We have used \eqref{con:Constants} and the definition \eqref{eq:DefML} to simplify the total remainder.
This finishes the proof.
\end{proof}

\begin{lemma}\label{lem:EstimatesOnQ3-js}
Assume that $\Phi$ satisfies \eqref{eq:LocalizedAftercNumber} and \eqref{eq:Localized_n_AftercNumber}
Assume furthermore that \eqref{eq:Close} and Assumption~\ref{assump:params} are satisfied and $\rmu a^3$ is sufficiently small.
Then,
\begin{align}
\label{eq:EstimateQ3Tilde1} 
&\Big\langle \Phi, \Big( {\mathcal Q}_3^{(1)}(z)
+ \left(1-\frac{K_{\ell}^4}{ (K_H'')^{4}} \right) (2\pi)^{-3} \ell^3 \int_{\{|k| \geq \frac{1}{2} K_H'' \ell^{-1}\}} {\mathcal D}_k b_{k}^{\dagger} b_k \,dk\nonumber \\
&\qquad\qquad\qquad\qquad\qquad\qquad\qquad\qquad + {\mathcal Q}_2^{\rm ex}
+ \frac{b}{100} (\frac{1}{\ell^2} n_{+} + \frac{\varepsilon_{T}}{(d \ell)^2} n_{+}^H)
\Big) \Phi \Big\rangle   \\
&\geq  -C \rmu^2 a \ell^3 (\rmu a^3)^{\frac{1}{2}}  \left\{ (\rmu a^3)^{\frac{1}{2}} \left( (K_H'')^{-6} K_{\ell}^{8} d^{-6} + (K_H'')^{-2M-3} d^{-12}\right)
+ d^{-12} K_{\ell}  (K_H'')^{-2M-4}\right\},  \nonumber \\
\label{eq:EstimateQ3Tilde2+3}
&\Big\langle \Phi, \Big( {\mathcal Q}_3^{(2)}(z) + {\mathcal Q}_3^{(3)}(z)
+ \Big(\frac{K_{\ell}}{ K_H''}\Big)^4 (2\pi)^{-3} \ell^3 \int_{\{|k| \geq \frac{1}{2} K_H'' \ell^{-1}\}} {\mathcal D}_k b_{k}^{\dagger} b_k\,dk
\Big) \Phi \Big\rangle 
\nonumber \\ 
& \geq 
-C \rmu^2 a \ell^3 
( K_H'' )^{-2M-4} K_{\ell} d^{-12} (\rmu a^3)^{\frac{1}{2}},\\
\intertext{and}
\label{eq:EstimateQ3Tilde4}
&\Big\langle \Phi, \Big( {\mathcal Q}_3^{(4)}(z) + \rho_z z \widehat{W_1 \omega}(0)
(2\pi)^{-3} \int \widehat{\chi_{\Lambda}^2}(s)  ( \widetilde{a}_s^{\dagger} + \widetilde{a}_s)\,ds 
+  \frac{1}{100} \frac{b}{\ell^2} n_{+}
\Big) \Phi \Big\rangle \nonumber \\
&\geq
-C \rmu^2 a \ell^3 \sqrt{\frac{\ML}{|z|^2}} \Big(
K_H'' \frac{a}{\ell} +d^{2M}
\Big) .
\end{align}

\end{lemma}

\begin{proof}[Proof of Lemma~\ref{lem:EstimatesOnQ3-js}]
The proofs of \eqref{eq:EstimateQ3Tilde1}, \eqref{eq:EstimateQ3Tilde2+3} and \eqref{eq:EstimateQ3Tilde4} are each rather lengthy and will be carried out individually.

\begin{proof}[Proof of \eqref{eq:EstimateQ3Tilde4}]
Notice, using Lemma~\ref{lem:DecaychiHat} applied to $\chi^2$ that
\begin{align}
\sup_{\{|k| \geq d^{-2} \ell^{-1}\}} \widehat{\chi_{\Lambda}^2}(k)  \leq C_0 \ell^3 d^{4M}.
\end{align}

Therefore, by Cauchy-Schwarz weighted with $\sqrt{\ML}$, we get for any state $\Phi$ satisfying \eqref{eq:LocalizedAftercNumber} and \eqref{eq:Localized_n_AftercNumber} and using \eqref{cond:KHs-ML} as well as Lemma~\ref{lem:UsingML},
\begin{align}\label{eq:CSon1a-terms}
&\Big| \langle \Phi, \int \widehat{\chi_{\Lambda}^2}(s)  ( \widetilde{a}_s^{\dagger} + \widetilde{a}_s)\,ds \, \Phi \rangle \Big|\nonumber \\
& \leq C \sqrt{\ML} + \frac{1}{\sqrt{\ML}} \langle \Phi,  \Big( \int_{k \in {\mathcal P}_{\rm high}} |\widehat{\chi_{\Lambda}^2}(s)| \,\widetilde{a}_s^{\dagger} \widetilde{a}_s\,ds+ \int_{k \notin {\mathcal P}_{\rm high}}| \widehat{\chi_{\Lambda}^2}(s)|\, \widetilde{a}_s^{\dagger} \widetilde{a}_s\,ds \Big) \Phi \rangle
\nonumber \\
&\leq C \sqrt{\ML}, 
\end{align}
and similarly,
\begin{align}
\Big| \langle \Phi, \int \widehat{\chi_{\Lambda}^2}(s) \Big(1-f_L(s)\Big) ( \widetilde{a}_s^{\dagger} + \widetilde{a}_s)\,ds \, \Phi \rangle \Big|
&\leq C d^{2M} \sqrt{\rmu \ell^3}.
\end{align}
Therefore, using Lemma~\ref{lem:EstimateIntegral} below
 to estimate the $k$-integral, we find
\begin{align}
\Big| \langle \Phi,  \Big( z \rho_z  \widehat{W_1 \omega}(0) (2\pi)^{-3} & \int \widehat{\chi_{\Lambda}^2}(s)  ( \widetilde{a}_s^{\dagger} + \widetilde{a}_s)\,ds\, 
\nonumber \\
&\quad 
- z (2\pi)^{-6} \iint_{k \in {\mathcal P}_{\rm high}}  \widehat{W}_1(k) \alpha_{k} \widehat{\chi_{\Lambda}^2}(s)  f_L(s) ( \widetilde{a}_s^{\dagger} + \widetilde{a}_s)\,ds \Big) \Phi \rangle \Big| \nonumber \\
&
\leq
C \rho_z^2 a \ell^3 \Big(K_H'' \frac{a}{\ell} \sqrt{\frac{\ML}{|z|^2}} + d^{2M} \Big) 
\end{align}
The estimate is in agreement with the error term in \eqref{eq:EstimateQ3Tilde4}.

What remains in order to prove \eqref{eq:EstimateQ3Tilde4} is to estimate a difference of two integrals over the same domain. Writing out the commutator using \eqref{eq:bComms2} we have to estimate
\begin{align}\label{eq:1Q-integral}
z (2\pi)^{-6} \ell^3 \iint_{k \in {\mathcal P}_{\rm high}} \widehat{W}_1(k) \alpha_k 
\widehat{\chi^2}(s\ell) f_L(s) \left( 1- \frac{1-\alpha_{s-k} \alpha_{-k}}{(1-\alpha_k^2)(1-\alpha_{s-k}^2)}\right) ( \widetilde{a}_s^{\dagger} + \widetilde{a}_s),
\end{align}
and
\begin{align}\label{eq:1Q-integral2}
z (2\pi)^{-6} \ell^3 \iint_{k \in {\mathcal P}_{\rm high}} \widehat{W}_1(k) \alpha_k f_L(s) 
\frac{1-\alpha_{s-k} \alpha_{-k}}{(1-\alpha_k^2)(1-\alpha_{s-k}^2)}
\widehat{\chi}(k\ell) \widehat{\chi}\big((k-s)\ell\big)
 ( \widetilde{a}_s^{\dagger} + \widetilde{a}_s).
\end{align}

\begin{remark}
In the estimate of the term corresponding to \eqref{eq:1Q-integral} in \cite{FS}, i.e. \cite[(12.56)]{FS} there is an unimportant error. The factor $z$ in  \eqref{eq:1Q-integral} (or equivalently \cite[(12.56)]{FS}) is forgotten in the estimate \cite[(12.58)]{FS}.
which should have been,
\begin{align}
(12.56) \leq \ldots \leq C |z| \rho_z a \delta^2 (\varepsilon^{-1} + \varepsilon n_{+}).
\end{align}
In order to absorb the $n_{+}$-term in the `gap', the choice of $\varepsilon$ therefore needs an extra factor of $1/|z|$, and consequently the error term in \cite[(12.59)]{FS}---which also should replace the last error term in \cite[(12.51)]{FS} and the corresponding term in \cite[(12.41)]{FS}---becomes
\begin{align}
C z^2 \rho_z^2 a \ell^3 \delta^4 \frac{a}{\ell}.
\end{align}
Since $|z|^2 \leq 2 \rmu \ell^3$ and $\delta = (\rmu a^3)^{\frac{1}{6}} \widetilde{K}_H^2$ (see \cite[(12.4)]{FS}) this is estimated as
\begin{align}
Cz^2 \rho_z^2 a \ell^3 \delta^4 \frac{a}{\ell} \leq C \rmu^2 a \ell^3 \delta^4 K_{\ell}^2 = C \rmu^2 a \ell^3 \sqrt{\rmu a^3} 
(\rmu a^3)^{\frac{1}{6}} K_{\ell}^2
 \widetilde{K}_H^8.
\end{align}
Notice that with the choices of parameters in \cite{FS} $K_{\ell} = X^{-3/2}$, $\widetilde{K}_H = X^{-8/3}$ with $X=(\rmu a^3)^{\frac{1}{323}}$, we get
$$
(\rmu a^3)^{\frac{1}{6}} K_{\ell}^2
 \widetilde{K}_H^8 = (\rmu a^3)^{\frac{1}{6}(1-\frac{146}{323})} \ll 1. 
$$
So this is still a lower order error term and the proof of Theorem~1.2 in \cite[Section 13]{FS} is not affected.
\end{remark}

To estimate the expectation value of the expression in \eqref{eq:1Q-integral} in the state $\Phi$, we use \eqref{eq:alphaInPH} and a Cauchy-Schwarz inequality (weighted with $\sqrt{\ML}$) to get 
\begin{align}
|\langle \Phi \eqref{eq:1Q-integral} \Phi \rangle| 
&\leq
z (2\pi)^{-6} \ell^3\iint_{k \in {\mathcal P}_{\rm high}}  |\widehat{W}_1(k)| |\alpha_k| | \widehat{\chi^2}(s\ell) | f_L(s) (\alpha_k^2 + \alpha_{s-k}^2) \nonumber \\
&\qquad\qquad\qquad \times \left(\sqrt{\ML}+ \frac{1}{\sqrt{\ML}} \langle \Phi, \widetilde{a}_s^{\dagger} \widetilde{a}_s \Phi \rangle \right)\,dsdk\nonumber \\
&\leq C z (K_H'')^{-3} \ell^3 \rho_z^3 a^4 \sqrt{\ML}
\nonumber \\
&\leq C \rmu^2 a \ell^3 \sqrt{\frac{\ML}{|z|^2}} (K_H'')^{-3}  K_{\ell}^3 \sqrt{\rmu a^3},
\end{align}
where we used Lemma~\ref{lem:UsingML} to estimate the $s$-integral.

Using that by \eqref{eq:maerkedobbelt} we have $K_{\ell} \ll K_H''$, this term can be absorbed in the error term in \eqref{eq:EstimateQ3Tilde4}.

In the second integral \eqref{eq:1Q-integral2} the terms $\widehat{\chi}(k \ell)$ are very small due to regularity of $\chi$ and the fact that $k \in {\mathcal P}_{\rm high}$. Therefore this integral is much smaller.
We easily get, upon estimating the $s$-integral as above and the $k$-integral by Lemma~\ref{lem:EstimateIntegral}
 below,
\begin{align}
&\langle \Phi, z (2\pi)^{-6} \ell^3 \iint_{k \in {\mathcal P}_{\rm high}} \widehat{W}_1(k) \alpha_k f_L(s) 
\frac{1-\alpha_{s-k} \alpha_{-k}}{(1-\alpha_k^2)(1-\alpha_{s-k}^2)}
\widehat{\chi}(k\ell) \widehat{\chi}\big((k-s)\ell\big)
 ( \widetilde{a}_s^{\dagger} + \widetilde{a}_s) \Phi \rangle \nonumber \\
& \geq
 - C z \rmu a  \sup_{k \in {\mathcal P}_{\rm high}} |\widehat{\chi}(k\ell)| \sqrt{\ML} \nonumber \\
 &\geq - C \rmu^2 a \ell^3 \sqrt{\frac{\ML}{|z|^2}} (K_H'')^{-M}
\end{align}
where we used Lemma~\ref{lem:DecaychiHat} to get the last estimate. This error term is clearly in agreement with \eqref{eq:EstimateQ3Tilde4}, since 
by \eqref{eq:maerkedobbelt} we have $d^{-2} \ll K_H''$.

This finishes the proof of \eqref{eq:EstimateQ3Tilde4}.
\end{proof}

In the proof of \eqref{eq:EstimateQ3Tilde4} we used the following result.
\begin{lemma}\label{lem:EstimateIntegral}
Suppose \eqref{eq:Close} as well as 
Assumption~\ref{assump:params}.
Then for sufficiently small values of $\rmu a^3$ we have,
\begin{align}\label{eq:CompIntegrals-Final}
\left|  \rho_z \widehat{W_1 \omega}(0) - (2\pi)^{-3} \int_{k \in {\mathcal P}_{\rm high}}  \widehat{W}_1(k) \alpha_{k} \,dk \right|
\leq C \rho_z a   K_H'' (a/\ell).
\end{align}
Furthermore,
\begin{align}\label{eq:CompIntegrals-Final-2}
\left|    \widehat{W_1 \omega}(0) - (2\pi)^{-3} \int_{k \in {\mathcal P}_{\rm high}}  \frac{\widehat{W}_1(k)^2}{ 2{\mathcal D}_k} \,dk \right|
\leq C  a   K_H'' (a/\ell).
\end{align}
\end{lemma}

\begin{proof}
Collecting the estimates below, we really get
\begin{align}\label{eq:CompIntegrals}
&\left|  \rho_z \widehat{W_1 \omega}(0) - (2\pi)^{-3} \int_{k \in {\mathcal P}_{\rm high}}  \widehat{W}_1(k) \alpha_{k} \,dk \right| \nonumber \\
&\qquad \leq C \rho_z a \Big\{ K_H'' a \ell^{-1}+ K_{\ell}^3 (K_H'')^{-3} \sqrt{\rmu a^3} + \varepsilon_N
+\frac{R a }{\ell^2} \nonumber \\
&\qquad\qquad+ K_{\ell} (K_H'')^{-1} \sqrt{\rmu a^3} + \frac{a}{ds\ell} \left(1 + \log\left(\frac{K_{\ell}}{K_H'' \sqrt{\rmu a^3}}\right)
\right)\Big\}.
\end{align}
From this \eqref{eq:CompIntegrals-Final} follows upon using \eqref{eq:Cond_epsilonN}, \eqref{eq:maerkedobbelt} and \eqref{con:ER} to compare the different terms.

We will calculate on $(1-\varepsilon_N)^{-1}\rho_z \widehat{W_1 \omega}(0)$ instead of $\rho_z \widehat{W_1 \omega}(0)$. The difference between the two contributes with the $\varepsilon_N$-error term in \eqref{eq:CompIntegrals}.
We get,
\begin{align}
&(1-\varepsilon_N)^{-1}\rho_z \widehat{W_1 \omega}(0) - (2\pi)^{-3} \int_{k \in {\mathcal P}_{\rm high}}  \widehat{W_1}(k) \alpha_{k} \,dk\\
&=
(2\pi)^{-3} \int_{k \in {\mathcal P}_{\rm high}}  \widehat{W_1}(k) \Big( \rho_z \frac{\widehat{g}(k)}{2(1-\varepsilon_N)k^2} - \alpha_{k}\Big) \,dk 
+
(2\pi)^{-3} \int_{k \notin {\mathcal P}_{\rm high}}  \rho_z \widehat{W_1}(k) \frac{\widehat{g}(k)}{2(1-\varepsilon_N)k^2}\,dk. \nonumber 
\end{align}
We first estimate the last integral, using that $\varepsilon_N \ll 1$,
\begin{align}
\Big| \int_{k \notin {\mathcal P}_{\rm high}}  \widehat{W_1}(k) \frac{\widehat{g}(k)}{2(1-\varepsilon_N)k^2}\,dk \Big|
\leq 
C a^2 \int_{\{ |k| \leq  K_H'' \ell^{-1}\}} k^{-2} \,dk 
=
C a^2 K_H'' \ell^{-1}.
\end{align}
This is consistent with the error term in \eqref{eq:CompIntegrals}.

To continue, we write
\begin{align}
\widehat{W_1}(k) \alpha_k = \rho_{z}^{-1} {\mathcal A}(k) \Big( 1 - \sqrt{1- {\mathcal B}(k)^2/{\mathcal A}(k)^2}\Big).
\end{align}

Notice that $|{\mathcal B}(k)/{\mathcal A}(k)| \leq \frac{1}{2}$, for $\rmu$ sufficiently small using \eqref{eq:BoverA} and \eqref{eq:maerkedobbelt}.
Therefore,
\begin{align}
\Big| \widehat{W_1}(k) \alpha_k - \frac{ \rho_z \widehat{W_1}(k)^2}{2 {\mathcal A}(k)}\Big|
\leq C \rho_z^3 \frac{ \widehat{W_1}(k)^4}{{\mathcal A}(k)^3} \leq C \rho_z^3 a^4 k^{-6},
\end{align}
where we used that ${\mathcal A}(k) \geq \frac{1}{2} k^2$ in ${\mathcal P}_{\rm high}$.
Upon integrating over ${\mathcal P}_{\rm high}$ we find a term of magnitude
\begin{align}\label{eq:Comparealpha-A}
\int_{{\mathcal P}_{\rm high}} \Big| \widehat{W_1}(k) \alpha_k - \frac{ \rho_z \widehat{W_1}(k)^2}{2 {\mathcal A}(k)}\Big| \leq C \rho_z^3 a^4 (K_H'' \ell^{-1})^{-3},
\end{align}
in agreement with \eqref{eq:CompIntegrals}.

Finally, we consider,
\begin{align}\label{eq:SlightlyComp}
\rho_z &\Big| \int_{k \in {\mathcal P}_{\rm high} } \widehat{W_1}(k) \Big( \frac{\widehat{g}(k)}{2(1-\varepsilon_N)k^2} - \frac{\widehat{W}_1(k)}{2{\mathcal A}(k)} \Big) \Big|\\
&\leq  \rho_z \Big| \int_{k \in {\mathcal P}_{\rm high} } \widehat{W_1}(k)  \frac{\widehat{W}_1(k)-\widehat{g}(k)}{2(1-\varepsilon_N)k^2}\Big|+ \rho_z \Big|\int_{k \in {\mathcal P}_{\rm high} } \frac{\widehat{W_1}(k)^2}{2(1-\varepsilon_N)k^2}  \Big( 1 - \frac{(1-\varepsilon_N)k^2}{{\mathcal A}(k)} \Big)\Big|.\nonumber 
\end{align}
We estimate the first term in \eqref{eq:SlightlyComp} as 
\begin{align}\label{eq:goodinR}
\rho_z \Big| \int_{k \in {\mathcal P}_{\rm high} } \widehat{W_1}(k)  \frac{\widehat{W}_1(k)- \widehat{g}(k)}{2(1-\varepsilon_N)k^2}\Big| &
\leq 2 \rho_z \Big| \int_{{\mathbb R}^3} \widehat{W_1}(k)  \frac{\widehat{W}_1(k)-\widehat{g}(k)}{2 k^2}\Big| 
+  C \rho_z \frac{a^2 R^2}{\ell^2} K_H'' \ell^{-1}  \nonumber \\
&\leq
C \rho_z \ell^{-2} \iint \frac{g(x) |y|^2 g(y)}{|x-y|} \,dx dy +  C \rho_z \frac{a^2 R^2}{\ell^2} K_H'' \ell^{-1}  \nonumber \\
&\leq
C \rho_z \ell^{-2} \int \omega(y) |y|^2 g(y) \,dy +  C \rho_z \frac{a^2 R^2}{\ell^2} K_H'' \ell^{-1} \nonumber \\
&\leq C \rho_z R a^2 \ell^{-2}  +  C \rho_z \frac{a^2 R^2}{\ell^2} K_H'' \ell^{-1},
\end{align}
where we used \eqref{eq:Scattering3}, \eqref{omegabounds} and \eqref{eq:W1-g-new}. 
Note that in the second inequality in \eqref{eq:goodinR} we used the first inequality in \eqref{eq:W1-g-new} to conclude that the integral is positive.
Clearly, using \eqref{con:ER} this is in agreement with \eqref{eq:CompIntegrals}.

We estimate the second term in \eqref{eq:SlightlyComp}, using $0 \leq k^2 - \tau(k) \leq 2|k| (d s \ell)^{-1}$ in ${\mathcal P}_{\rm high}$,
\begin{align}
&\rho_z \Big|\int_{k \in {\mathcal P}_{\rm high} } \frac{\widehat{W_1}(k)^2}{2(1-\varepsilon_N)k^2}  \Big( 1 - \frac{(1-\varepsilon_N)k^2}{{\mathcal A}(k)} \Big)\Big| \nonumber \\
&\leq  C \rho_z^2 a^3 \int_{k \in {\mathcal P}_{\rm high} } k^{-4}  
 +
C\rho_z (d s \ell)^{-1} \Big(\int_{\{K_H'' \ell^{-1} \leq  |k| \leq a^{-1} \}} a^2 |k|^{-3}+ a \int \frac{\widehat{W_1}(k)^2}{2k^2} \Big) \nonumber \\
&\leq C\rho_z a K_{\ell} (K_H'')^{-1} \sqrt{\rmu a^3} +
C\rho_z a^2(d s \ell)^{-1} \left(1+ \log\big(\frac{K_{\ell}}{K_H'' \sqrt{\rmu a^3}}\big)\right)
\end{align}
Since this also agrees with \eqref{eq:CompIntegrals}, this finishes the proof of \eqref{eq:CompIntegrals},
and therefore of \eqref{eq:CompIntegrals-Final}.

The proof of \eqref{eq:CompIntegrals-Final-2} is similar. One can for instance use \eqref{eq:CompIntegrals-Final} and \eqref{eq:Comparealpha-A} and the fact that
$|1 - \frac{{\mathcal A}(k)}{{\mathcal D}_k}| \leq C \frac{\mathcal B(k)^2}{{\mathcal A}(k)^2} \leq C \rmu^2 a^2 k^{-4}$ in ${\mathcal P}_{\rm high}$. Then \eqref{eq:CompIntegrals-Final-2} follows.
\end{proof}

\begin{proof}[Proof of \eqref{eq:EstimateQ3Tilde2+3}]
The two operators $ {\mathcal Q}_3^{(2)}(z)$ and ${\mathcal Q}_3^{(3)}(z)$ are very similar and can be estimated in identical fashion, so we will only explicity consider the first.
We decompose 
\begin{align}
{\mathcal Q}_3^{(2)}(z) = I + II,
\end{align}
where
\begin{align}\label{eq:OneAndTwo}
I&:= - z \ell^3 (2\pi)^{-6} \iint_{\{ k \in {\mathcal P}_{\rm high}\}}
\frac{f_L(s)
\widehat{W}_1(k)\alpha_k}{(1-\alpha_k^2)(1-\alpha_{s-k}^2)} \left( b_{-k}^{\dagger} \widetilde{a}_s^{\dagger}  b_{s-k}
+ b_{s-k}^{\dagger} \widetilde{a}_s b_{-k} \right) , \nonumber \\
II&:= - z \ell^3 (2\pi)^{-6} \iint_{\{ k \in {\mathcal P}_{\rm high}\}}
\frac{f_L(s)
\widehat{W}_1(k)\alpha_k}{(1-\alpha_k^2)(1-\alpha_{s-k}^2)} \left([\widetilde{a}_s^{\dagger} , b_{-k}^{\dagger}] b_{s-k}
+ b_{s-k}^{\dagger}  [b_{-k}, \widetilde{a}_s] \right).
\end{align}
The second term $II$ will be very small, due to the smallness of the commutator (notice that $s$ and $k$ are `far apart' since $s \in {\mathcal P}_{\rm low}$ and $k \in {\mathcal P}_{\rm high}$).
So the main term is $I$, which we estimate using Cauchy-Schwarz and \eqref{eq:alphaInPH} as
\begin{align}
I \geq
-C \ell^3 z a K_{\ell}^2 (K_H'')^{-2}  \iint_{\{k \in {\mathcal P}_{\rm high}\}} f_L(s) \left( \varepsilon b_{-k}^{\dagger} \widetilde{a}_s^{\dagger} \widetilde{a}_s b_{-k} + \varepsilon^{-1} b_{s-k}^{\dagger} b_{s-k} \right).
\end{align}
Using Lemma~\ref{lem:UsingML}, we estimate
\begin{align}
\langle \Phi, \int f_L(s) b_{-k}^{\dagger} \widetilde{a}_s^{\dagger} \widetilde{a}_s b_{-k} \,ds \Phi \rangle \leq C \ell^{-3} \ML \langle \Phi,b_{-k}^{\dagger} b_{-k} \Phi \rangle.
\end{align}
Therefore, after doing the $s$-integral in the other term, we choose $\varepsilon = \sqrt{d^{-6}/\ML}$ 
and use \eqref{eq:Dk_in_PH} to insert ${\mathcal D}_k$.
This leads to the estimate
\begin{align}\label{eq:LoseABitOfb's}
\langle \Phi, I  \Phi \rangle & \geq - Cz a K_{\ell}^2 (K_H'')^{-2} 
d^{-3} (\ML)^{1/2}
\langle \Phi, \int_{\{|k| \geq \frac{1}{2} K_H'' \ell^{-1}\}} b_{k}^{\dagger} b_k \Phi \rangle \nonumber \\
&\geq - C K_{\ell}^4 (K_H'')^{-4} \ell^3 \sqrt{\frac{d^{-6} \ML}{\rmu \ell^3}} \langle \Phi, \int_{\{|k| \geq \frac{1}{2} K_H'' \ell^{-1}\}} {\mathcal D}_k b_{k}^{\dagger} b_k \Phi \rangle .
\end{align}
Notice that 
\begin{align}
d^{-6} \ML \ll \rmu \ell^3,
\end{align}
using \eqref{cond:KH3-n} and \eqref{cond:disjoint}.
Therefore, $I$ can be absorbed in the positive ${\mathcal D}_k b_{k}^{\dagger} b_k$ term in \eqref{eq:EstimateQ3Tilde2+3}.

We now return to the term $II$ from \eqref{eq:OneAndTwo}, where we use Remark~\ref{rem:commutatorsMixed} to evaluate the commutators.
Therefore, the term $II$ from \eqref{eq:OneAndTwo} can be estimated, using \eqref{eq:alphaInPH} and Lemma~\ref{lem:EstimateIntegral}, for arbitrary $\varepsilon >0$, as
\begin{align}
II &\geq - 2 z \ell^3 (2\pi)^{-6} 
\iint_{\{ k \in {\mathcal P}_{\rm high}\}}  | [\widetilde{a}_s^{\dagger} , b_{-k}^{\dagger}] |
f_L(s)
|\widehat{W}_1(k)\alpha_k| \left( \varepsilon b_{s-k}^{\dagger} b_{s-k}
+ \varepsilon^{-1} \right)\\
&\geq
-C z (\frac{K_{\ell}}{ K_H''})^2  \Big(\sup_{\{|p|\geq \frac{1}{2} K_H'' \ell^{-1}\}} |\widehat{\chi}(p\ell)| \Big) 
d^{-6}
\left(\varepsilon^{-1} \rmu a
+ \varepsilon a \frac{K_{\ell}^2}{ (K_H'')^4} \ell^2 \int_{\{|k| \geq \frac{1}{2} K_H'' \ell^{-1}\}} {\mathcal D}_k b_k^{\dagger} b_k \right).\nonumber 
\end{align}

Upon choosing $\varepsilon^{-1}$ proportional to $z (\sup_{\{|p|\geq \frac{1}{2} K_H'' \ell^{-1}\}} |\widehat{\chi}(p\ell)|) d^{-6} (K_H'')^{-2} a/\ell$, 
the $b_k^{\dagger} b_k$-term is easily absorbed in the $K_{\ell}^4 (K_H'')^{-4} \ell^3 \int_{\{|k| \geq \frac{1}{2} K_H'' \ell^{-1}\}} {\mathcal D}_k b_{k}^{\dagger} b_k$ term in \eqref{eq:EstimateQ3Tilde2+3}.
Therefore, using \eqref{eq:Close} and Lemma~\ref{lem:DecaychiHat}, $II$ contributes with an error term of order
\begin{align}
\rmu^2 a \ell^3 
( K_H'' )^{-2M-4} K_{\ell} d^{-12} (\rmu a^3)^{\frac{1}{2}}
\end{align}
to \eqref{eq:EstimateQ3Tilde2+3}.

This finishes the proof of \eqref{eq:EstimateQ3Tilde2+3}.
\end{proof}

\begin{proof}[Proof of \eqref{eq:EstimateQ3Tilde1}]
Finally, we estimate ${\mathcal Q}_3^{(1)}(z)$. We rewrite
\begin{align}
{\mathcal Q}_3^{(1)}(z) =
 z \ell^3 (2\pi)^{-6} \iint_{\{ k \in {\mathcal P}_{\rm high}\}}
\frac{f_L(s)
\widehat{W}_1(k)}{(1-\alpha_k^2)(1-\alpha_{s-k}^2)} \left(\widetilde{a}_s^{\dagger}  
b_{s-k} b_k + \alpha_k \alpha_{s-k} \widetilde{a}_{-s}^{\dagger}   b_{s-k}^{\dagger} b_{k}^{\dagger} + h.c.
 \right),
\end{align}
where we performed a change of variables in the second term to get the equality.

We combine this term with the diagonalized Bogoliubov Hamiltonian.
We leave a (small) $K_{\ell}^4 (K_H'')^{-4}$-part of this operator in order to control error terms appearing below.

Therefore, we consider
\begin{align}\label{eq:CentralEstimateOfQ31}
&{\mathcal Q}_3^{(1)}(z) + (2\pi)^{-3} \ell^3 \int_{\{ k \in {\mathcal P}_{\rm high}\}} (1- 2 K_{\ell}^4 (K_H'')^{-4}) {\mathcal D}_k b_k^{\dagger} b_{k}\,dk \nonumber \\
 &=
 (2\pi)^{-3} \ell^3 \int_{\{ k \in {\mathcal P}_{\rm high}\}} (1-2K_{\ell}^4 (K_H'')^{-4}){\mathcal D}_k c_k^{\dagger} c_k +
 T_1(k) +T_2(k) \nonumber \\
 &\geq (2\pi)^{-3} \ell^3 \int_{\{ k \in {\mathcal P}_{\rm high}\}} 
 T_1(k) +T_2(k).
 \end{align}
Here we have introduced the operators,
\begin{align}
c_k &:= b_k +\\
&\quad z (2\pi)^{-3} \int \frac{f_L(s)
\widehat{W}_1(k)}{(1-2K_{\ell}^4 (K_H'')^{-4}){\mathcal D}_k(1-\alpha_k^2)(1-\alpha_{s-k}^2)} 
\left( b_{s-k}^{\dagger} \widetilde{a}_s + \alpha_k \alpha_{s-k} \widetilde{a}_{-s} b_{s-k}^{\dagger}\right)ds,  \nonumber \\
\label{eq:T1_def}
T_1(k)  &:= - z (2\pi)^{-3} \int \frac{f_L(s)
\widehat{W}_1(k) \alpha_k \alpha_{s-k} }{(1-\alpha_k^2)(1-\alpha_{s-k}^2)} \left( [b_k^{\dagger}, \widetilde{a}_{-s} b_{s-k}^{\dagger}] + h.c. \right)ds,
\intertext{and}
\label{eq:T2_def}
T_2(k) &:= -\frac{|z|^2 \widehat{W}_1(k)^2}{(1-2 K_{\ell}^4 (K_H'')^{-4}){\mathcal D}_k(1-\alpha_k^2)^2} (2\pi)^{-6} \iint \frac{f_L(s)f_L(s')
}{(1-\alpha_{s-k}^2)(1-\alpha_{s'-k}^2)}  \nonumber \\
&\quad \times 
\left(\widetilde{a}_{s'}^{\dagger} b_{s'-k}      + \alpha_k \alpha_{s'-k} b_{s'-k} \widetilde{a}_{-s'}^{\dagger} \right)
\left( b_{s-k}^{\dagger} \widetilde{a}_s + \alpha_k \alpha_{s-k} \widetilde{a}_{-s} b_{s-k}^{\dagger}\right)\,ds\,ds'\nonumber \\
&\geq
-\left(1+ C K_{\ell}^4 (K_H'')^{-4} \right) |z|^2\frac{ \widehat{W}_1(k)^2}{{\mathcal D}_k} (2\pi)^{-6} \iint \frac{f_L(s)f_L(s')}{(1-\alpha_{s-k}^2)(1-\alpha_{s'-k}^2)} 
 \nonumber \\
&\quad\times 
\left(\widetilde{a}_{s'}^{\dagger} b_{s'-k}      + \alpha_k \alpha_{s'-k} b_{s'-k} \widetilde{a}_{-s'}^{\dagger} \right)
\left( b_{s-k}^{\dagger} \widetilde{a}_s + \alpha_k \alpha_{s-k} \widetilde{a}_{-s} b_{s-k}^{\dagger}\right)\,ds\,ds',
\end{align}
where we used \eqref{eq:alphaInPH} to get the estimate on $T_2$.

We will prove that
\begin{align}\label{eq:FinalT1}
&(2\pi)^{-3} \ell^3 \int_{\{ k \in {\mathcal P}_{\rm high}\}} T_1(k) \,dk + \frac{b}{200} \frac{1}{\ell^2} n_{+}
+ \frac{K_{\ell}^4}{ 2 (K_H'')^4} \ell^3 \int_{\{ |q| \geq \frac{1}{2} K_H'' \ell^{-1}\}} D_q b_q^{\dagger} b_q \,dq\nonumber \\
&\geq  - C \rmu^2 a \ell^3 (\rmu a^3) \left( (K_H'')^{-6} d^{-6} K_{\ell}^8 + (K_H'')^{-2M -3} d^{-12}\right).
\end{align}
and
\begin{align}\label{eq:FinalT2}
&(2\pi)^{-3} \ell^3 \int_{\{ k \in {\mathcal P}_{\rm high}\}} T_2(k) \,dk + {\mathcal Q}_2^{\rm ex} 
+ \frac{K_{\ell}^4}{ 2 (K_H'')^{4}} 
(2\pi)^{-3} \ell^3 \int_{\{|k| \geq \frac{1}{2} K_H'' \ell^{-1}\}} {\mathcal D}_k b_{k}^{\dagger} b_k \,dk \nonumber \\
&\quad  + \frac{b}{200} (\frac{1}{\ell^2} n_{+} + \frac{\varepsilon_{T}}{(d \ell)^2} n_{+}^H) \nonumber \\
&\geq  -C \rmu^2 a \ell^3 (\rmu a^3)^{\frac{1}{2}}  d^{-12} K_{\ell}  (K_H'')^{-2M-4}.
\end{align}
Clearly, \eqref{eq:EstimateQ3Tilde1} follows from \eqref{eq:FinalT1}, \eqref{eq:FinalT2} and \eqref{eq:CentralEstimateOfQ31}.

\noindent {\bf Proof of \eqref{eq:FinalT2}.}
We start by proving \eqref{eq:FinalT2}.
Notice that
\begin{align}
\widetilde{a}_{s'}^{\dagger} b_{s'-k}+ \alpha_k \alpha_{s'-k} b_{s'-k} \widetilde{a}_{-s'}^{\dagger}
=
\left( \widetilde{a}_{s'}^{\dagger}+ \alpha_k \alpha_{s'-k} \widetilde{a}_{-s'}^{\dagger}\right) b_{s'-k} + \alpha_k \alpha_{s'-k} [b_{s'-k}, \widetilde{a}_{-s'}^{\dagger}].
\end{align}
The contribution from the commutator term is very small, both due to the factors of $\alpha$ and to the commutator,
since $k \in {\mathcal P}_{\rm high}$, $s'\in {\mathcal P}_{\rm low}$. 
Therefore, we estimate for arbitrary $\varepsilon>0$,
\begin{align}
T_2(k) \geq (1+\varepsilon) T_2'(k) + (1+ \varepsilon^{-1}) T_2''(k),
\end{align}
where
\begin{align}
T_2'(k) &:= - \left(1+ C K_{\ell}^4 (K_H'')^{-4} \right) |z|^2 \frac{ \widehat{W}_1(k)^2}{{\mathcal D}_k} (2\pi)^{-6} \iint \frac{f_L(s)f_L(s')}{(1-\alpha_{s-k}^2)(1-\alpha_{s'-k}^2)} 
 \nonumber \\
&\qquad \times 
\left( \widetilde{a}_{s'}^{\dagger}      + \alpha_k \alpha_{s'-k} \widetilde{a}_{-s'}^{\dagger}\right) b_{s'-k} b_{s-k}^{\dagger} 
\left( \widetilde{a}_s + \alpha_k \alpha_{s-k} \widetilde{a}_{-s} \right)\,ds\,ds', \nonumber \\
T_2''(k) &:=- \left(1+ C K_{\ell}^4 (K_H'')^{-4} \right) |z|^2 \frac{ \widehat{W}_1(k)^2}{{\mathcal D}_k} (2\pi)^{-6} \nonumber \\
&\qquad \times \iint \frac{f_L(s)f_L(s')}{(1-\alpha_{s-k}^2)(1-\alpha_{s'-k}^2)}  |\alpha_k|^2 \alpha_{s'-k} \alpha_{s-k}
[b_{s'-k}, \widetilde{a}_{-s'}^{\dagger}] [ \widetilde{a}_{-s}, b_{s-k}^{\dagger}] .
\end{align}
For simplicity, we choose $\varepsilon = K_{\ell}^4 (K_H'')^{-4}  \ll 1$ and can therefore absorb the factor of $(1+\varepsilon)$ in $T_2'(k)$ by simply changing the value of $C$.
With this choice, we estimate using \eqref{eq:CompIntegrals-Final-2}, \eqref{eq:alphaInPH} and Remark~\ref{rem:commutatorsMixed},
\begin{align}
(2\pi)^{-3} \ell^3 \int_{k \in {\mathcal P}_{\rm high}} (1+ \varepsilon^{-1}) T_2''(k) \,dk
&\geq 
-C \rho_z a d^{-12} K_{\ell}^4 (K_H'')^{-4}  \sup_{k\in {\mathcal P}_{\rm high}, s \in {\mathcal P}_{\rm low}}
| [ \widetilde{a}_{-s}, b_{s-k}^{\dagger}] |^2 \nonumber \\
&\geq
- C \rmu^2 a \ell^3 (\rmu a^3)^{\frac{1}{2}} d^{-12} K_{\ell} (K_H'')^{-4} \sup_{\{|k| \geq \frac{1}{2} K_H'' \ell^{-1}\}}
| \widehat{\chi}(k\ell)|^2 \nonumber \\
&\geq - C
\rmu^2 a \ell^3 (\rmu a^3)^{\frac{1}{2}} d^{-12} K_{\ell} (K_H'')^{-4} (K_H'')^{-2M},
\end{align}
by Lemma~\ref{lem:DecaychiHat}. This is in agreement with the error bound in \eqref{eq:FinalT2}.

We continue to estimate the other part of $T_2(k)$.
\begin{align}\label{eq:T2prime}
T_2'(k) &:= - \left(1+ C K_{\ell}^4 (K_H'')^{-4}  \right) |z|^2 \frac{ \widehat{W}_1(k)^2}{{\mathcal D}_k} (2\pi)^{-6} \iint \frac{f_L(s)f_L(s')}{(1-\alpha_{s-k}^2)(1-\alpha_{s'-k}^2)} 
 \nonumber \\
&\qquad \times 
\left( \widetilde{a}_{s'}^{\dagger}      + \alpha_k \alpha_{s'-k} \widetilde{a}_{-s'}^{\dagger}\right) b_{s'-k} b_{s-k}^{\dagger} 
\left( \widetilde{a}_s + \alpha_k \alpha_{s-k} \widetilde{a}_{-s} \right)\,ds\,ds' 
\nonumber \\
&=T_{2,{\rm comm}}'(k) + T_{2,{\rm op}}'(k),
\end{align}
with
\begin{align}
T_{2,{\rm comm}}'(k)&:=- \left(1+ C K_{\ell}^4 (K_H'')^{-4}  \right) |z|^2\frac{ \widehat{W}_1(k)^2}{{\mathcal D}_k} (2\pi)^{-6} \iint \frac{f_L(s)f_L(s')}{(1-\alpha_{s-k}^2)(1-\alpha_{s'-k}^2)} 
 \nonumber \\
&\qquad \times 
\left( \widetilde{a}_{s'}^{\dagger}      + \alpha_k \alpha_{s'-k} \widetilde{a}_{-s'}^{\dagger}\right) [b_{s'-k} ,b_{s-k}^{\dagger} ]
\left( \widetilde{a}_s + \alpha_k \alpha_{s-k} \widetilde{a}_{-s} \right)\,ds\,ds' ,\nonumber \\
T_{2,{\rm op}}'(k)&:=-\left(1+ C K_{\ell}^4 (K_H'')^{-4}  \right) |z|^2
\frac{ \widehat{W}_1(k)^2}{{\mathcal D}_k} (2\pi)^{-6} \iint \frac{f_L(s)f_L(s')}{(1-\alpha_{s-k}^2)(1-\alpha_{s'-k}^2)} 
 \nonumber \\
&\qquad \times 
\left( \widetilde{a}_{s'}^{\dagger}+ \alpha_k \alpha_{s'-k} \widetilde{a}_{-s'}^{\dagger}\right) b_{s-k}^{\dagger} b_{s'-k}
\left( \widetilde{a}_s + \alpha_k \alpha_{s-k} \widetilde{a}_{-s} \right)\,ds\,ds'.
\end{align}
We start by estimating the last term in \eqref{eq:T2prime}.
We introduce the notation
\begin{align}
{\mathcal C} := \sup_{s,s' \in {\mathcal P}_{\rm low}, k \in {\mathcal P}_{\rm high} } \left| [ \widetilde{a}_{s'}^{\dagger}      + \alpha_k \alpha_{s'-k} \widetilde{a}_{-s'}^{\dagger}, b_{s-k}^{\dagger}]\,\right| \leq 1,
\end{align}
In fact,
it follows from Remark~\ref{rem:commutatorsMixed} and Lemma~\ref{lem:DecaychiHat} that
\begin{align}\label{eq:StoerrelseC}
{\mathcal C} \leq C K_{\ell}^2 (K_H'')^{-2}  (K_H'')^{-M} \ll 1.
\end{align}

To estimate the last term in \eqref{eq:T2prime} we first apply Cauchy-Schwarz, then commute the $\widetilde{a}$'s through the $b$'s and apply Cauchy-Schwarz to the commutator terms.
This yields,
\begin{align}
&\left\langle \Phi, \iint \frac{f_L(s)f_L(s')}{(1-\alpha_{s-k}^2)(1-\alpha_{s'-k}^2)}   \left( \widetilde{a}_{s'}^{\dagger}      + \alpha_k \alpha_{s'-k} \widetilde{a}_{-s'}^{\dagger}\right) b_{s-k}^{\dagger} b_{s'-k}
\left( \widetilde{a}_s + \alpha_k \alpha_{s-k} \widetilde{a}_{-s} \right) \Phi \right\rangle\nonumber \\
&\leq 
3 \iint f_L(s)f_L(s') \left\langle \Phi, b_{s-k}^{\dagger}
\left( \widetilde{a}_{s'}^{\dagger}      + \alpha_k \alpha_{s'-k} \widetilde{a}_{-s'}^{\dagger}\right) 
\left( \widetilde{a}_{s' }+ \alpha_k \alpha_{s'-k} \widetilde{a}_{-s'} \right)  b_{s-k} \Phi \right\rangle\nonumber \\
&\quad 
+C {\mathcal C}^2  \iint f_L(s)f_L(s')\,ds\,ds'
\nonumber \\
&\leq
C \ell^{-3} \ML  \int f_L(s)  \langle \Phi, b_{s-k}^{\dagger}b_{s-k} \Phi \rangle 
+ C d^{-12} \ell^{-6} {\mathcal C}^2,
\end{align}
where we applied \eqref{eq:DomByML} of Lemma~\ref{lem:UsingML} to $b_{s-k} \Phi$ for fixed $s,k$ to get the last estimate.

Therefore, using \eqref{eq:Dk_in_PH}
and \eqref{eq:CompIntegrals-Final-2},
\begin{align}\label{eq:T2opprime}
&\left\langle \Phi, (2\pi)^{-3} \ell^3 \int_{k \in {\mathcal P}_{\rm high}} T_{2,{\rm op}}'(k) \Phi \right\rangle\nonumber \\
&\geq
- C  \left( K_{\ell} d^{-6} (K_H'')^{-4} \ML \sqrt{\rmu a^3} \right)\ell^3 \langle \Phi \int_{\{|q| \geq \frac{1}{2} K_H'' \ell^{-1} \}} {\mathcal D}_q b_q^{\dagger} b_q \Phi \rangle - C \rmu  a d^{-12} {\mathcal C}^2.
\end{align}
Notice that, using \eqref{cond:KH3-n} and \eqref{eq:maerkedobbelt} that
$$
K_{\ell} d^{-6} (K_H'')^{-4} \ML \sqrt{\rmu a^3} \ll  K_{\ell}^4  (K_H'')^{-4}.
$$
Therefore, the negative ${\mathcal D}_q b_q^{\dagger} b_q$-term in \eqref{eq:T2opprime} can be absorbed in a fraction of the similar (positive) term in 
\eqref{eq:FinalT2}.
Using \eqref{eq:StoerrelseC} and the identity $\rmu a = \rmu^2 a \ell^3 (\rmu a^3)^{\frac{1}{2}} K_{\ell}^{-3}$, the error term in \eqref{eq:T2opprime} is easily seen to agree with \eqref{eq:FinalT2}.

We next consider the commutator term $T_{2,{\rm comm}}'(k)$ from \eqref{eq:T2prime}.
From \eqref{eq:bComms2} and using Lemma~\ref{lem:DecaychiHat} and that $M\geq 2$, we see that
\begin{align}\label{eq:bcommutator}
\left|  [b_{s'-k} ,b_{s-k}^{\dagger} ]  - \widehat{\chi^2}((s-s')\ell) \right| \leq 
C K_{\ell}^4 (K_H'')^{-4}.
\end{align}
Therefore, using also $\ell^3 \int \widetilde{a}_s^{\dagger} \widetilde{a}_s \leq C n_{+}$,
\begin{align}\label{eq:EstimT2comm}
T_{2,{\rm comm}}'(k) &\geq
-\big(1+C K_{\ell}^4 (K_H'')^{-4} \big) |z|^2 \frac{ \widehat{W}_1(k)^2}{{\mathcal D}_k} (2\pi)^{-6} \iint f_L(s)f_L(s')
\widetilde{a}_{s'}^{\dagger} \widehat{\chi^2}((s-s')\ell) \widetilde{a}_s \nonumber \\
&\quad -
C  |z|^2
\frac{ \widehat{W}_1(k)^2}{{\mathcal D}_k} K_{\ell}^4 (K_H'')^{-4} d^{-6} \ell^{-6} n_{+}.
\end{align}
Using \eqref{eq:RepInTermsOfatilde} and \eqref{eq:CompIntegrals-Final-2} we see that 
\begin{align}
-(2\pi)^{-3} &\ell^3 \int_{k \in {\mathcal P}_{\rm high}} \big(1+C K_{\ell}^4 (K_H'')^{-4} \big) |z|^2 \frac{ \widehat{W}_1(k)^2}{{\mathcal D}_k} (2\pi)^{-6} \iint f_L(s)f_L(s')
\widetilde{a}_{s'}^{\dagger} \widehat{\chi^2}((s-s')\ell) \widetilde{a}_s \nonumber \\
&=
-\rho_z \big(1+C K_{\ell}^4 (K_H'')^{-4} \big) \big((2\pi)^{-3} \int_{k \in {\mathcal P}_{\rm high}} \frac{ \widehat{W}_1(k)^2} {{\mathcal D}_k} \big)\sum_j Q_{L,j}^{*} \chi_{\Lambda}^2(x_j)  Q_{L,j} \nonumber \\
&\geq
- 2 \rho_z \big(1+C[K_H'' \frac{a}{\ell} + K_{\ell}^4 (K_H'')^{-4}] \big) \widehat{W_1 \omega}(0) \sum_j Q_{L,j}^{*} \chi_{\Lambda}^2(x_j) Q_{L,j}.
\end{align}

We now notice that, for all $0 < \varepsilon < 1$,
\begin{align}
\sum_j Q_{L,j}^{*}\chi_{\Lambda}^2(x_j)  Q_{L,j} \leq (1+\varepsilon) \sum_j Q_{j} \chi_{\Lambda}^2(x_j)  Q_{j} +
C \varepsilon^{-1} n_{+}^H.
\end{align}
We notice that $\rmu a = (d K_{\ell})^2 \frac{1} {d^2 \ell^2}$. Therefore, choosing $\varepsilon$ proportional to $\varepsilon_T^{-1} (d K_{\ell})^2$, (notice that $\varepsilon_T^{-1} (d K_{\ell})^2\ll 1$ by \eqref{con:eTdK}) we find, 
\begin{align}
-(2\pi)^{-3} &\ell^3 \int_{k \in {\mathcal P}_{\rm high}} \big(1+C  K_{\ell}^4 (K_H'')^{-4} \big) |z|^2 \frac{ \widehat{W}_1(k)^2}{{\mathcal D}_k} (2\pi)^{-6} \iint f_L(s)f_L(s')
\widetilde{a}_{s'}^{\dagger} \widehat{\chi^2}((s-s')\ell) \widetilde{a}_s \nonumber \\
&\geq
- 2 \rho_z \widehat{W_1 \omega}(0) \sum_j Q_j \chi_{\Lambda}^2(x_j)  Q_j 
\nonumber \\
&\quad 
- \frac{b}{500} \frac{\varepsilon_T}{(d\ell)^2} n_{+}^H
- C \rho_z a\big(K_H'' \frac{a}{\ell}+ K_{\ell}^4 (K_H'')^{-4} +  \varepsilon_T^{-1} (dK_{\ell})^2 \big) n_{+}
\end{align}
Notice now, using \eqref{con:eTdK}, \eqref{eq:maerkedobbelt}, \eqref{con:KBKell}, \eqref{cond:disjoint} and \eqref{con:Constants}
\begin{align}
\rho_z a \big(K_H'' \frac{a}{\ell}+ K_{\ell}^4 (K_H'')^{-4} +  \varepsilon_T^{-1} (dK_{\ell})^2 \big)
\ll \ell^{-2}.
\end{align}
Therefore, the above error terms can be absorbed in the energy gap.

Using the definition of ${\mathcal Q}_2^{\rm ex}$ from \eqref{eq:DefQ2ex}, we see from Lemma~\ref{lem:CompareQ2ex} (since $\rmu \approx \rho_z$ by \eqref{eq:Close}), 
\begin{align}
- 2 \rho_z \widehat{W_1 \omega}(0) \sum_j Q_j \chi_{\Lambda}^2(x_j)  Q_j + {\mathcal Q}_2^{\rm ex}
\geq -C \left( \rmu a n_{+}^H + \rmu a d^{-2} \left( \frac{aR}{\ell^2}\right)^{\frac{1}{2}} n_{+}\right).
\end{align}
By \eqref{eq:Someparameters} the $n_{+}$-term can be absorbed in the energy gap. Also $\rmu a = (\varepsilon_T^{-1}  d^2 K_{\ell}^2) \frac{\varepsilon_T}{d^2 \ell^2} \ll \frac{\varepsilon_T}{d^2 \ell^2}$ by \eqref{con:eTdK}.
So the $n_{+}^H$ can be absorbed in the corresponding gap term.

To estimate the error term in \eqref{eq:EstimT2comm} we integrate, using \eqref{eq:CompIntegrals-Final-2},
\begin{align}
- (2\pi)^{-3} \ell^3 \int_{k \in {\mathcal P}_{\rm high}} C  |z|^2
\frac{ \widehat{W}_1(k)^2}{{\mathcal D}_k} K_{\ell}^4 (K_H'')^{-4} d^{-6} \ell^{-6} n_{+}&\geq
-C K_{\ell}^6 (K_H'')^{-4} d^{-6} \frac{1}{\ell^2} n_{+}.
\end{align}
Notice that by using the first and the last estimate in \eqref{con:eTdK}, \eqref{eq:small_and_large} and \eqref{cond:disjoint}, we find
\begin{align}
\label{cond:KHstornok}
K_{\ell}^6 d^{-6} (K_H'')^{-4} \ll 1.
\end{align}
Therefore, the above term can also be absorbed in the energy gap.

This finishes the estimate of $T_{2,{\rm comm}}'(k)$ and therefore establishes the estimate \eqref{eq:FinalT2}  for $T_2(k)$ from \eqref{eq:T2_def}.

\noindent {\bf Proof of \eqref{eq:FinalT1}.}
We now estimate $T_1(k)$ from \eqref{eq:T1_def}.
We clearly have  using \eqref{eq:alphaInPH} and for all $\varepsilon_1, \varepsilon_2 >0$,
\begin{align}
T_1(k)
&\geq - C z \frac{\rho_z^2 a^3}{|k|^4}  \sup_{ s \in {\mathcal P}_{\rm low} } \left(|[b_k^{\dagger}, b_{s-k}^{\dagger}]| \right)
\int f_L(s) \left( \varepsilon_1 \widetilde{a}_{-s}^{\dagger} \widetilde{a}_{-s} + \varepsilon_1^{-1}\right)\,ds \nonumber \\
&\quad - C z \frac{\rho_z^2 a^3}{|k|^4} \sup_{s \in {\mathcal P}_{\rm low} } \left(|[b_k^{\dagger}, \widetilde{a}_{-s}^{\dagger}]|\right) \int f_L(s)
\left( \varepsilon_2 b_{s-k}^{\dagger} b_{s-k} + \varepsilon_2^{-1}\right) \,ds \nonumber \\
&\geq
-C z \frac{\rho_z^3 a^4}{|k|^6} \ell^{-3} \Big(  \varepsilon_1 n_{+} +  \varepsilon_1^{-1} d^{-6}\Big) \nonumber \\
&\quad - C z \frac{\rho_z^2 a^3}{|k|^4} (K_H'')^{-M} \Big( \varepsilon_2 (K_H'')^{-2} \ell^2 \int f_L(s) {\mathcal D}_{s-k} b^{\dagger}_{s-k} b_{s-k} \,ds
+  \varepsilon_2^{-1} d^{-6} \ell^{-3} \Big).
\end{align}
Here we inserted ${\mathcal D}_{s-k}$  for $|k| \geq K_H'' \ell^{-1}$ using \eqref{eq:Dk_in_PH}, estimated $\ell^3 \int \widetilde{a}_{-s}^{\dagger} \widetilde{a}_{-s} \leq C n_{+}$, and 
estimated the commutators for $|k| \geq K_H'' \ell^{-1}$ using Remark~\ref{rem:commutatorsMixed}, \eqref{eq:alphaInPH}, and Lemma~\ref{lem:DecaychiHat}
 as
\begin{align}\label{eq:SomeMoreComms}
\sup_{ s \in {\mathcal P}_{\rm low} } \left(|[b_k^{\dagger}, b_{s-k}^{\dagger}]| \right) \leq C \frac{\rho_z a}{|k|^2}, \qquad \text{ and } \qquad \sup_{s \in {\mathcal P}_{\rm low} } \left(|[b_k^{\dagger}, \widetilde{a}_{-s}^{\dagger}]|\right) 
\leq C (K_H'')^{-M}.
\end{align}
Therefore,
\begin{align}\label{eq:EndEstimateT1}
&\ell^3 (2\pi)^{-3}  \int_{k \in {\mathcal P}_{\rm high}} T_1(k)\,dk \nonumber \\
&\geq -C z \rho_z^3 a^4 (K_H'')^{-3} \ell^5 \Big(  \varepsilon_1 \ell^{-2} n_{+} +  \varepsilon_1^{-1} d^{-6} \ell^{-2} \Big) \nonumber \\
&\quad - C z \rho_z^2 a^3 (K_H'')^{-M-2} d^{-6} \ell^3 \Big( \varepsilon_2 (K_H'')^{-4} \ell^3 \int_{\{|q| \geq \frac{1}{2} K_H'' \ell^{-1}\}}  {\mathcal D}_{q} b^{\dagger}_{q} b_{q} \,dq
+  \varepsilon_2^{-1} (K_H'') \ell^{-2} \Big).
\end{align}
We choose $\varepsilon_1$ in such a way as to be able to absorb the $n_{+}$ term in the positive gap in \eqref{eq:FinalT1}, i.e.
\begin{align*}
\varepsilon_1^{-1}:= C z \rho_z^3 a^4 (K_H'')^{-3} \ell^5. 
\end{align*}
We also choose $\varepsilon_2$ in order to absorb the $\int {\mathcal D}_{q} b^{\dagger}_{q} b_{q} \,dq$-term in the similar positive term 
in \eqref{eq:FinalT1}, i.e.
\begin{align*}
\varepsilon_2^{-1} := C K_{\ell}^{-4} z \rho_z^2 a^3 (K_H'')^{-M-2} d^{-6} \ell^3.
\end{align*}
Therefore, $\ell^3 (2\pi)^{-3}  \int_{k \in {\mathcal P}_{\rm high}} T_1(k)\,dk$ contributes with a total error term of magnitude
\begin{align}
- C \left( z^2 \rho_z^6 a^8 (K_H'')^{-6} \ell^8 d^{-6} + 
z^2 \rho_z^4 a^6 (K_H'')^{-2M-3} d^{-12} \ell^4 K_{\ell}^{-4} 
\right).
\end{align}
Using now that $z^2 \approx \rmu \ell^3$ by \eqref{eq:Close}, we find the final estimate \eqref{eq:FinalT1}.

As observed earlier \eqref{eq:EstimateQ3Tilde1} follows from \eqref{eq:FinalT1}, \eqref{eq:FinalT2} and \eqref{eq:CentralEstimateOfQ31}, so
this finishes the proof of \eqref{eq:EstimateQ3Tilde1}.
\end{proof}

Now we have established all three inequalities \eqref{eq:EstimateQ3Tilde1}, \eqref{eq:EstimateQ3Tilde2+3} and \eqref{eq:EstimateQ3Tilde4}.
This finishes the proof of Lemma~\ref{lem:EstimatesOnQ3-js}.
\end{proof}

In the proof of \eqref{eq:EstimateQ3Tilde1} we used the following result.

\begin{lemma}\label{lem:CompareQ2ex}
Assume that the parameters satisfy 
Assumption~\ref{assump:params} and that $\rmu a^3$ is small enough.
Then, 
\begin{align}\label{eq:CompareQ2ex}
(2\pi)^{-3} \rho_z \ell^3  \int \big(\widehat{W_1 \omega}(k) - \widehat{W_1 \omega}(0)\big)  a_k^{\dagger} a_k
\geq - 4 \frac{\rho_z}{\rmu} \left( \rmu a  n_{+}^H - C \rmu a d^{-2} \left( \frac{aR}{\ell^2}\right)^{\frac{1}{2}} n_{+}\right).
\end{align}
\end{lemma}
Notice that
\begin{align}\label{eq:Someparameters}
\rmu a d^{-2} \left( \frac{aR}{\ell^2}\right)^{\frac{1}{2}} = K_{\ell} d^{-2} (\rmu a^3)^{\frac{1}{4}} (\rmu a R^2)^{\frac{1}{4}} \ll \ell^{-2},
\end{align}
using \eqref{con:ER}, \eqref{con:KBKell}, \eqref{con:sdKellKB}, \eqref{con:eTdK} and \eqref{eq:small_and_large}.

\begin{proof}
Using that $W_1 \omega$ is even and has finite support, as well as \eqref{omegabounds}, we easily get the estimate, for arbitrary $\delta>0$,
\begin{align}
&(2\pi)^{-3} \rho_z \ell^3  \int \big(\widehat{W_1 \omega}(k) - \widehat{W_1 \omega}(0)\big)  a_k^{\dagger} a_k 
\nonumber \\
&
\geq - 2 (2\pi)^{-3}  \rho_z a \ell^3\int_{\{ |k|\geq \delta^{-1} \ell^{-1}\}} a_k^{\dagger} a_k 
- C \rho_z a \delta^{-2} \frac{a R}{\ell^2} n_{+}.
\end{align}
Upon undoing the second quantization we can continue the estimate as
\begin{align}
(2\pi)^{-3} \ell^3\int_{\{ |k|\geq \delta^{-1} \ell^{-1}\}} a_k^{\dagger} a_k &= \sum_j Q_j \chi_{\Lambda}(x_j) \one_{\{|p_j|\geq \delta^{-1} \ell^{-1}\}} \chi_{\Lambda}(x_j)  Q_j \nonumber \\
&\leq 2 n_{+}^H + 2 {\mathcal N} n_{+},
\end{align}
with
\begin{equation}
{\mathcal N}  := \| \one_{\{|p|\leq d^{-2} \ell^{-1}\}} \chi_{\Lambda}(x) \one_{\{|p|\geq \delta^{-1} \ell^{-1}\}} \|^2  \leq C (\delta d^{-2})^{2M'},
\end{equation}
for all $M' \leq M$,
using the regularity of $\chi$.

For simplicity, we choose $M'=1$ and optimize by choosing $\delta^2 = d^2 \left( \frac{aR}{\ell^2}\right)^{\frac{1}{2}} \ll 1$. 
This yields \eqref{eq:CompareQ2ex}.
\end{proof}

\section{Proof of the main theorem}\label{sec:proofCombined}
In this section we will combine the results of the previous sections in order to prove Theorem~\ref{thm:LHY}.
As noted in Section~\ref{sec:Box}, Theorem~\ref{thm:LHY} follows from 
Theorem~\ref{thm:LHY-Background}, which again follows from Theorem~\ref{thm:LHY-Box}.

\begin{proof}[Proof of Theorem~\ref{thm:LHY-Box}]
We will use the concrete choice of parameters set down in \eqref{eq:Xparameter}-\eqref{eq:choiceM} in Appendix~\ref{sec:params}.
Recall in particular the notation $X$ defined in \eqref{eq:Xchoice}.

To prove Theorem~\ref{thm:LHY-Box} let $\Psi \in {\mathcal
  F}_s(L^2(\Lambda))$ be a normalized $n$-particle trial state satisfying \eqref{eq:aprioriPsiHalf}. Notice that if such a state does not exist, then there is nothing to prove. Using Proposition~\ref{prop:LocMatrices} there exists a normalized $n$-particle wave function $\widetilde{\Psi}\in {\mathcal F}_{\rm s}(L^2(\Lambda))$ satisfying \eqref{eq:LocalizedNPlus} and such that
\begin{align}\label{eq:LargeMatricesNew}
\langle\Psi,{\mathcal H}_\Lambda(\rmu)\Psi\rangle \geq
\langle\widetilde\Psi,{\mathcal H}_\Lambda(\rmu)\widetilde\Psi\rangle
- C X^{\frac{1}{4}} \rmu^2 a \ell^3 (\rmu a^3)^{1/2}.
\end{align}
Notice that the error term in \eqref{eq:LargeMatricesNew} is consistent with the error term in Theorem~\ref{thm:LHY-Box}.

Using Proposition~\ref{prop:Hamilton2ndQuant} we find that our localized state $\widetilde\Psi$ satisfies 
\begin{align}
\langle \widetilde{\Psi}, {\mathcal H}_{\Lambda}(\rmu) \widetilde{\Psi} \rangle 
\geq \langle \widetilde{\Psi}, {\mathcal H}_{\Lambda}^{\rm 2nd}(\rmu) \widetilde{\Psi} \rangle  - C \rmu^2 a \ell^3  (\rmu a^3)^{1/2} \left( X^3 + (\rmu a^3)^{1/2}  \right) , 
\end{align}
where the error clearly is consistent with the error term in Theorem~\ref{thm:LHY-Box}.

At this point, we can apply
Theorem~\ref{thm:Kafz} to get the lower bound
\begin{align}\label{eq:IntroK(z)-}
\langle \widetilde{\Psi}, {\mathcal H}_{\Lambda}^{\rm 2nd}(\rmu) \widetilde{\Psi} \rangle
\geq
\inf_{z \in {\mathbb R}_{+} } \inf_{\Phi} \langle \Phi, {\mathcal K(z)} \Phi \rangle
- C \left( X^{\frac{9}{2}} + X^{\frac{1}{2}} \right) K_{\ell}^{3} \rmu a,
\end{align}
where the second infimum is over all normalized $\Phi  \in {\mathcal F}_{\rm s}(\Ran(Q))$ satisfying \eqref{eq:LocalizedAftercNumber} and \eqref{eq:Localized_n_AftercNumber}.

Since 
\begin{align}
K_{\ell}^{3}  \rmu a = \rmu^2 a \ell^3 \sqrt{\rmu a^3},
\end{align}
this implies that we need to prove that
\begin{align}\label{eq:ToProve}
\inf_{\Phi} \langle \Phi, {\mathcal K(z)} \Phi \rangle \geq
-4\pi \rmu^2 a \ell^3 +  4\pi \rmu^2 a  \ell^3 \frac{128}{15\sqrt{\pi}}(\rmu a^3)^{\frac{1}{2}} 
 - C \rmu^2 a \ell^3 (\rmu a^3)^{\frac{1}{2}}
X^{\frac{1}{5}},
\end{align}
for all normalized $\Phi$ satisfying \eqref{eq:LocalizedAftercNumber} and \eqref{eq:Localized_n_AftercNumber}.

We will use that with our choice of parameters \eqref{eq:NewConditionNew} is satisfied.

If $\rho_z = |z|^2/\ell^3$ satisfies \eqref{eq:Close-I}, then Proposition~\ref{prop:FarAway} provides a lower bound on $\langle \Phi, {\mathcal K(z)} \Phi \rangle$ which is larger than needed for \eqref{eq:ToProve} by a factor of $2$ on the LHY-term. Since \eqref{eq:NewConditionNew} is satisfied the assumptions of Proposition~\ref{prop:FarAway} are verified.

If $\rho_z$ satisfies \eqref{eq:Close}---which is opposite inequality to \eqref{eq:Close-I}---and $\Phi$ satisfies \eqref{eq:LocalizedAftercNumber} and \eqref{eq:Localized_n_AftercNumber}, then by \eqref{eq:SimplifyKz} (using again that \eqref{eq:NewConditionNew} is satisfied) and Theorem~\ref{thm:BogHamDiag} we get
\begin{align}\label{eq:Combined1}
\langle \Phi, {\mathcal K(z)} \Phi \rangle 
&\geq - \frac{1}{2} \rmu^2 \ell^{3} \widehat{g}(0)
+ 4\pi \frac{128}{15\sqrt{\pi}} \rho_z a \sqrt{\rho_z a^3}\ell^3 \nonumber \\
&\quad +  \langle \Phi, \left(\frac{b}{4 \ell^2} n_{+} + \varepsilon_T\frac{b}{8 d^2 \ell^2} n_{+}^{H} 
  + {\mathcal Q}_1^{\rm ex}(z)+
  {\mathcal Q}_2^{\rm ex}(z) + {\mathcal Q}_3(z) \right) \Phi \rangle \nonumber \\
  &\quad + (2\pi)^{-3} \ell^3 \langle \Phi, \int {\mathcal D}_k b_k^{\dagger} b_{k}\,dk \, \Phi \rangle - {\mathcal E}_1,
\end{align}
where the error term ${\mathcal E}_1$ satisfies
\begin{align}
{\mathcal E}_1 \leq C \rmu^2 a \ell^3 (\rmu a^3)^{\frac{1}{2}} 
X^{\frac{1}{5}}.
\end{align}
Here the error term in $X^{\frac{1}{5}}$ comes from the $\varepsilon(\rmu)$ in \eqref{eq:LebesgueError-Total}. Notice that this error is compatible with \eqref{eq:ToProve}.

Now we can apply Theorem~\ref{thm:Control3Q} to obtain the inequality
\begin{align}\label{eq:Combined2}
&(2\pi)^{-3} \ell^3 \langle \Phi, \int {\mathcal D}_k b_k^{\dagger} b_{k}\,dk \, \Phi \rangle \nonumber \\
&+  \langle \Phi, \left(\frac{b}{4 \ell^2} n_{+} + \varepsilon_T\frac{b}{8 d^2 \ell^2} n_{+}^{H} 
  + {\mathcal Q}_1^{\rm ex}(z)+
  {\mathcal Q}_2^{\rm ex}(z) + {\mathcal Q}_3(z) \right) \Phi \rangle 
  \geq - {\mathcal E}_2,
\end{align}
with error term
\begin{align}
{\mathcal E}_2 
&\leq C \rmu^2 a \ell^3 \sqrt{\rmu a^3} X^{\frac{23}{2}}.
\end{align}
Here the dominant contribution to the error (with our choice of parameters) comes from 
the $\sqrt{\frac{(K_H'')^2}{\widetilde{K}_{\mathcal M}}} K_{\ell}^{-1}$-term. This error is clearly consistent with \eqref{eq:ToProve}.

Combining \eqref{eq:Combined1} and \eqref{eq:Combined2}, we get
\begin{align}
\langle \Phi, {\mathcal K(z)} \Phi \rangle 
&\geq - \frac{1}{2} \rmu^2 \ell^{3} \widehat{g}(0)
+ 4\pi \frac{128}{15\sqrt{\pi}} \rmu a \sqrt{\rmu a^3}\ell^3 \nonumber \\
&\quad - ( {\mathcal E}_1 + {\mathcal E}_2 + C \left|\rmu a \sqrt{\rmu a^3} - \rho_z a \sqrt{\rho_z a^3}\right|\ell^3).
\end{align}
This establishes \eqref{eq:ToProve} for $\rho_z$ satisfying \eqref{eq:Close}, since by \eqref{eq:Close}, \eqref{eq:NewConditionNew} and \eqref{eq:Xparameter} we have
\begin{align}
\left|\rmu a \sqrt{\rmu a^3} - \rho_z a \sqrt{\rho_z a^3}\right|\ell^3 \leq C \rmu a \sqrt{\rmu a^3} \ell^3 K_{\ell}^{-2} = C \rmu a \sqrt{\rmu a^3} \ell^3 X^3.
\end{align}
This finishes the proof of \eqref{eq:ToProve} and therefore of Theorem~\ref{thm:LHY-Box}.
\end{proof}

\appendix

\section{Bogoliubov method}\label{sec:simplebog}
In this section we recall a simple consequence of the Bogoliubov method (see \cite[Theorem~6.3]{LS} and \cite{BS})

\begin{theorem}[Simple case of Bogoliubov's method]\label{thm:bogolubov-complete}\hfill\\
Let $a_\pm$ be operators on a Hilbert space satisfying $[a_+,a_-]=0$
For $\cA>0$, 
${\mathcal B} \in {\mathbb R}$ satisfying either $|{\mathcal B}| < {\mathcal A}$ or ${\mathcal B} = {\mathcal A}$ and arbitrary 
$\kappa\in\C$, we have the operator identity 
\begin{align}\label{eq:BogIdentity}
&\cA(a^{*}_+ a_+ +a^{*}_{-} a_{-})+{\mathcal B} (a^*_+a^*_{-}+a_+a_{-})+
\kappa(a^*_+ +a_{-})+\overline{\kappa}(a_+ +a_{-}^{*})\nonumber \\
&=
{\mathcal D} (b_{+}^* b_{+} + b_{-}^* b_{-})
-\frac{1}{2}(\cA-\sqrt{\cA^2-{\mathcal B}^2})
([a_{+},a^*_{+}]+[a_{-},a^*_{-}])-\frac{2|\kappa|^2}{{\mathcal A} + {\mathcal B}},
\end{align}

where
\begin{align}
{\mathcal D} := \frac{1}{2} \left( {\mathcal A} + \sqrt{{\mathcal A}^2 - {\mathcal B}^2}\right),
\end{align}
and
\begin{align}\label{eq:bpm}
b_{+}:= a_{+} + \alpha a_{-}^* + \overline{c_0},\qquad 
b_{-}:=a_{-} + \alpha a_{+}^{*} + c_0,
\end{align}
with
\begin{align}
\alpha:= {\mathcal B}^{-1} \left( {\mathcal A} - \sqrt{{\mathcal A}^2 - {\mathcal B}^2}\right),\qquad
c_0:=\frac{2\overline{\kappa}}{{\mathcal A} + {\mathcal B} + \sqrt{{\mathcal A}^2 - {\mathcal B}^2}}.
\end{align}
In particular,
\begin{align}\label{eq:BogIneq}
		&\cA(a^{*}_+ a_+ +a^{*}_{-} a_{-})+\cB(a^*_+a^*_{-}+a_+a_{-})+
		\kappa(a^*_+ +a_{-}^*)+\overline{\kappa}(a_+ +a_{-})\nonumber \\
		&\geq-\frac{1}{2}(\cA-\sqrt{\cA^2-\cB^2})
		([a_{+},a^*_{+}]+[a_{-},a*_{-}])-\frac{2|\kappa|^2}{\cA+\cB}.
\end{align}
\end{theorem}

\begin{proof}
The identity \eqref{eq:BogIdentity} is elementary. From here the
inequality \eqref{eq:BogIneq} follows by dropping the positive
operator term ${\mathcal D} (b_{+}^* b_{+} + b_{-}^* b_{-})$.
\end{proof}

\section{The explicit localization function}\label{sec:chiproperties}

In this section we discuss the explicit choice of the localization function 
$\chi$ and its properties.
Define
$$
\zeta(y) = \begin{cases}
\cos(\pi y), & |y| \leq 1/2,\\ 0,& |y|>1/2,
\end{cases}
$$
and
\begin{align}
\label{eq:Def_chi}
\chi(x) = C_{\cM} \left(\zeta(x_1) \zeta(x_2) \zeta(x_3)\right)^{\cM+2}.
\end{align}
Here $\cM \in 4{\mathbb N}$ is to be chosen large enough.
The condition for the choice of $M$ is given in \eqref{eq:choiceM} below.
The constant $C_{\cM}$ is chosen 
such that the normalization $\int\chi^2=1$ from \eqref{eq:chinormalization} holds.
We have $0\leq\chi \in C^{\cM}(\R^3)$.

\begin{lemma}\label{lem:DecaychiHat}
Let $\chi$ be the localization function from \eqref{eq:Def_chi}.
Recall that by assumption $M/2 \in {\mathbb N}$.
Then, for all $k \in {\mathbb R}^3$,
\begin{equation}
|\widehat{\chi}(k)| \leq C_{\chi} (1+ |k|^2)^{-M/2} ,
\end{equation}
where
\begin{equation}
C_{\chi} = \int \left| (1- \Delta)^{M/2 } \chi 
\right|
\end{equation}

In particular, when $|k| \geq \frac{1}{2} K_H'' \ell^{-1}$, we have
\begin{equation}\label{eq:DecayPH-new}
|\widehat{\chi_{\Lambda}}(k )| =
\ell^3 |\widehat{\chi}(k \ell)|\leq C \ell^3 ( K_H'')^{-M}.
\end{equation}

\end{lemma}

The proof of Lemma~\ref{lem:DecaychiHat} is elementary and will be omitted.

The explicit choice of $\chi$ is important when we analyze the behavior of the small box localization function, i.e. 
it is relevant for the a priori bounds in Theorem~\ref{thm:aprioribounds} the proof of which was given in \cite{FS}.

\section{Calculation of the Bogoliubov integral}
In this appendix, we will estimate the integrals that appear when
using Bogoliubov's method.  The first lemma,
Lemma~\ref{lm:aprioriintegral_new}, is a simple estimate that will be used
several times in our a-priori estimates. The second lemma,
Lemma~\ref{lem:BogIntegral}, gives the correct LHY correction.  For
completeness we give the details of this elementary but technical
argument.

\begin{lemma}\label{lm:aprioriintegral_new}
Consider functions ${\mathcal A},{\mathcal B}:\R^3\to \R$ with 
${\mathcal A}(p)\geq \kappa_A[|p|-P_1]^2_++2K_1 a$ and $|{\mathcal B}(p)|\leq K_2a$ for constants 
$\kappa_A>0$, $0<K_2\leq K_1$, and $0<P_1<a^{-1}$ then there exists a universal constant $C>0$ such that 
\begin{align*}
 & \int \left({\mathcal A}(p)-\sqrt{{\mathcal A}(p)^2-{\mathcal B}(p)^2}\right)dp\nonumber\\
 & \leq  \frac12(1+CP_1a)\kappa_A^{-1}\int_{\R^3} \frac{{\mathcal B}(p)^2}{p^2}dp+C \frac{K_2^2}{K_1} a P_1^{3}\\& 
  +C(K_2a)^2P_1\kappa_A^{-1}\log((P_1a)^{-1})
  +C\min\left\{(K_2 a)^4\kappa_A^{-3}P_1^{-3}, \frac{K_2^2}{K_1^2} \kappa_A^{-1}\int_{\R^3} \frac{{\mathcal B}(p)^2}{p^2}dp\right\}.
\end{align*}
\end{lemma}

\begin{proof}
Using that $|{\mathcal
  B}|/{\mathcal A}\leq 1/2 $ we have 
$$
{\mathcal A}(p)-\sqrt{{\mathcal A}(p)^2-{\mathcal B}(p)^2}\leq C\frac{{\mathcal B}(p)^2}{{\mathcal A}(p)}.
$$
We use this for $|p|<2P_1$ and find
\begin{align*}
  \int_{|p|<2P_1}\left({\mathcal A}(p)-\sqrt{{\mathcal A}(p)^2-{\mathcal B}(p)^2}\right)dp\leq&
  \cst\int_{|p|<2P_1} \frac{{\mathcal B}(p)^2}{{\mathcal A}(p)}dp \\\leq& \cst \frac{K_2^2}{K_1} a \int_{|p|<2P_1}1dp=
    C \frac{K_2^2}{K_1} a P_1^{3}.
\end{align*}
In the range $|p|>2P_1$, we use 
$$
{\mathcal A}(p)-\sqrt{{\mathcal A}(p)^2-{\mathcal B}(p)^2}
\leq \frac12\frac{{\mathcal B}(p)^2}{{\mathcal A}(p)}+C\frac{{\mathcal B}(p)^4}{{\mathcal A}(p)^3}.
$$
For $|p|>2P_1$ we have
\begin{align}\label{eq:B2overA}
  \frac{{\mathcal B}(p)^2}{{\mathcal A}(p)}\leq\kappa_A^{-1}
  \frac{{\mathcal B}(p)^2}{(|p|-P_1)^2}
  \leq \kappa_A^{-1}\frac{{\mathcal B}(p)^2}{p^2}(1+CP_1|p|^{-1})
\end{align}
and hence by splitting the integral over the error in $|p|<a^{-1}$ and $|p|>a^{-1}$ we obtain
$$
\int_{|p|>2P_1}\frac{{\mathcal B}(p)^2}{{\mathcal A}(p)}dp \leq(1+CP_1a)\kappa_A^{-1}
\int_{\R^3} \frac{{\mathcal B}(p)^2}{p^2}dp +C(K_2 a)^2P_1\kappa_A^{-1}\log((P_1a)^{-1}). 
$$ 
Finally, to get the bound in the lemma we estimate either,
\begin{align*}
  \int_{|p|>2P_1}\frac{{\mathcal B}(p)^4}{{\mathcal A}(p)^3}
  \leq \cst (K_2 a)^4\kappa_A^{-3}\int_{|p|>2P_1}  |p|^{-6} dp
  =C(K_2 a)^4\kappa_A^{-3}P_1^{-3},
\end{align*}
or
\begin{equation*}
\int_{|p|>2P_1}\frac{{\mathcal B}(p)^4}{{\mathcal A}(p)^3} \leq  \frac{K_2^2}{4 K_1^2} \int_{|p|>2P_1} \frac{{\mathcal B}(p)^2}{{\mathcal A}(p)}\,dp
\end{equation*}
and use \eqref{eq:B2overA}.
\end{proof}

\begin{lemma}\label{lem:BogIntegral}
Assume that
$\frac{9}{10} \rmu \leq \rho_z \leq \frac{11}{10} \rmu$ as well as Assumption~\ref{assump:params}.
We have the following estimate for sufficiently small $\rmu a^3$, and with ${\mathcal A}, {\mathcal B}$ as defined in \eqref{eq:defABC},
\begin{align}\label{eq:BogIntegral}
&-\frac{1}{2} (2\pi)^{-3} \ell^3 \int  \left( {\mathcal A}(k) - \sqrt{{\mathcal A}(k)^2 - {\mathcal B}(k)^2}\right) \,dk \nonumber \\
&\geq - \frac{\widehat{g\omega}(0)}{2} \rho_z^2 \ell^3
+ 4\pi  \frac{128 }{15\sqrt{\pi}} \rho_z^2 a \sqrt{\rho_z a^3} \ell^3 
- C \rmu^2 a \ell^3\left( 
 \epsilon(\rmu) \sqrt{\rmu a^3}  + \frac{R a}{\ell^2}+\varepsilon_N \right)
\end{align}
where the error term $\epsilon(\rmu)$ was defined in \eqref{eq:LebesgueError-Total}.
\end{lemma}

\begin{proof}
We regularize the integral as
\begin{align}\label{LHYIntegral1}
&\int  {\mathcal A}(k) - \sqrt{{\mathcal A}(k)^2 - {\mathcal B}(k)^2} \,dk \nonumber \\
&=
\int  {\mathcal A}(k) - \sqrt{{\mathcal A}(k)^2 - {\mathcal B}(k)^2} - \rho_z^2 \frac{\widehat{W_1}(k)^2}{2 (1-\varepsilon_N)  k^2}  \,dk
+
\rho_z^2 \int \frac{\widehat{W_1}(k)^2}{2 (1-\varepsilon_N) k^2}  \,dk.
\end{align}

The last integral is controlled by \eqref{eq:I2-integral-new}
and contributes with the first and the last two terms in \eqref{eq:BogIntegral}.

In the regularized integral in \eqref{LHYIntegral1} we perform the change of variables $\sqrt{\rho_z a} \,t = k$.
In this way we get
\begin{align}
&\int  {\mathcal A}(k) - \sqrt{{\mathcal A}(k)^2 - {\mathcal B}(k)^2} - \rho_z^2 \frac{\widehat{W_1}(k)^2}{2(1-\varepsilon_N)  k^2}  \,dk
=
\rho_z^2 \sqrt{\rho_z a^3} a I_1,
\end{align}
with
\begin{align}
I_1 &= \int \alpha(t)  -
\sqrt{\alpha(t)^2 - \beta(t)^2} - \frac{\beta(t)^2}{2 (1-\varepsilon_N)  t^2}\,dt ,\nonumber \\
\alpha(t) &= (1-\varepsilon_N) \widetilde{\tau}(t) + a^{-1} \widehat{W_1}(\sqrt{\rho_z a} t), \nonumber \\
\beta(t) &=  a^{-1} \widehat{W_1}(\sqrt{\rho_z a} t),\nonumber\\
\widetilde{\tau}(t) &= (1-\varepsilon_T) \Big[ |t| - \frac{1}{2K_{\ell}s} (\frac{\rmu}{\rho_z})^{1/2}  \Big]_{+}^2
+ \varepsilon_T \Big[ |t| - \frac{1}{2K_{\ell}ds} (\frac{\rmu}{\rho_z})^{1/2} \Big]_{+}^2.
\end{align}
We further decompose $I_1$ as
\begin{align}
I_1 
&= I_1' + I_1'',
\end{align}
with
\begin{align}\label{eq:SplitIntegrals}
I_1' &:= \int \alpha(t)  - \frac{\beta(t)^2}{2\alpha(t)}-
\sqrt{\alpha(t)^2 - \beta(t)^2} - \frac{\beta(t)^3}{2 (1-\varepsilon_N) t^2 \alpha(t)}, \nonumber\\
I_1''&:= \int \frac{\beta(t)^2}{2} \frac{t^2-\widetilde{\tau}(t)}{t^2 \alpha(t)}.
\end{align}

We will prove that
\begin{align}\label{eq:LebesgueError}
&\left| I_1' - \int_{{\mathbb R}^3} t^2 + 8\pi - \frac{(8\pi)^2}{2 t^2}
- \sqrt{(t^2+8\pi)^2-(8\pi)^2} \,dt \right| \nonumber \\
&\qquad \leq C \left(
 (\rmu a)^{\frac{1}{4}} \sqrt{R} + \varepsilon_T + (K_{\ell} s)^{-1}+ \varepsilon_N \right).
\end{align}
Notice that this is consistent with \eqref{eq:BogIntegral} with the error term $\epsilon(\rmu)$ defined in \eqref{eq:LebesgueError-Total} and that
\begin{align}\label{eq:BogConstantInt}
\int_{{\mathbb R}^3} t^2 + 8\pi - \frac{(8\pi)^2}{2 t^2}
- \sqrt{(t^2+8\pi)^2-(8\pi)^2} \,dt = -64 \pi^4 \frac{128}{15\sqrt{\pi}} .
\end{align}
We artificially introduce the factor $(1-\varepsilon_N)$ in the integral \eqref{eq:BogConstantInt} by scaling to get
\begin{align}\label{eq:BogConstantInt_perturbed}
&\left| \int_{{\mathbb R}^3}(1-\varepsilon_N) t^2 + 8\pi - \frac{(8\pi)^2}{2 (1-\varepsilon_N)t^2}
- \sqrt{((1-\varepsilon_N)t^2+8\pi)^2-(8\pi)^2} \,dt + 64 \pi^4 \frac{128}{15\sqrt{\pi}} \right| \nonumber \\
&\leq C \varepsilon_N.
\end{align}
This contributes with one of the error terms in \eqref{eq:LebesgueError}.

We continue the proof of \eqref{eq:LebesgueError} by noticing that
the part of both integrals---the one in \eqref{eq:BogConstantInt_perturbed} and $I_1'$---where $|t| \leq 10 (K_{\ell } s)^{-1}$ (notice that $(K_{\ell } s)^{-1} \ll 1$ by\eqref{con:sKell})  is bounded by
$$
C (K_{\ell } s)^{-1},
$$
for sufficiently small $\rmu$ (using that $\rho_z \approx \rmu$).
This is another of the error terms in \eqref{eq:LebesgueError}.

For $|t| \geq 10 (K_{\ell } s)^{-1}$ we will use (by \eqref{eq:W1-g-new} and since $W_1$ is even)
\begin{align}\label{eq:tautilde-Est}
|\beta(t) - 8\pi| \leq C \rmu a R^2 |t|^2, \qquad
0 \leq t^2 - \widetilde{\tau}(t) \leq \varepsilon_T t^2 + \frac{1}{K_{\ell}s} (\frac{\rmu}{\rho_z})^{1/2} |t|.
\end{align}
Notice that it follows that $|\beta(t) - 8\pi| \leq C (\rmu a)^{\frac{1}{4}} R^{\frac{1}{2}} |t|^{\frac{1}{2}}$ and also that
$\widetilde{\tau} \geq \frac{1}{2} t^2$ when $\varepsilon_T$ is sufficiently small (since $\frac{\rmu}{\rho_z}$ is close to $1$).

We decompose the resulting difference of integrals as
\begin{align}
&\int_{\{|t| \geq 10 (K_{\ell } s)^{-1}\}} \alpha  - \frac{\beta^2}{2\alpha}-
\sqrt{\alpha^2 - \beta^2} - \frac{\beta^3}{2 (1-\varepsilon_N) t^2 \alpha}\,dt \nonumber \\
&\qquad -
\int_{\{|t| \geq 10 (K_{\ell } s)^{-1}\}}  (1-\varepsilon_N)t^2 + 8\pi - \frac{(8\pi)^2}{2 (1-\varepsilon_N)t^2}
- \sqrt{[(1-\varepsilon_N)t^2+8\pi]^2-(8\pi)^2} \,dt  \nonumber \\
&= J_1 + J_2,
\end{align}
with
\begin{align}
J_2&:= \int_{\{|t| \geq 10 (K_{\ell } s)^{-1}\}}   \frac{ (8 \pi)^3}{2 (1-\varepsilon_N)t^2 [(1-\varepsilon_N)t^2 + 8 \pi]}- \frac{\beta^3}{2 (1-\varepsilon_N)t^2 \alpha} \,dt ,\nonumber \\
J_1&:= \int_{\{|t| \geq 10 (K_{\ell } s)^{-1}\}}  \alpha_1 f\big(\frac{\beta_1}{\alpha_1}\big) - \alpha_2 f\big(\frac{\beta_2}{\alpha_2}\big)\,dt,
\end{align}
with $f(x) := 1 - \frac{1}{2} x^2 - \sqrt{1-x^2}$,
\begin{align}
\alpha_1 &:= \alpha, & \alpha_2&:= (1-\varepsilon_N) t^2 + 8\pi, \nonumber \\
\beta_1 &:=\beta, & \beta_2&:= 8\pi.
\end{align}

The  integral $J_2$ is easily estimated (using that $\epsilon_N \leq 1/2$ and the discussion around \eqref{eq:tautilde-Est}), as
\begin{align}
|J_2| \leq
C \left( (\rmu a)^{\frac{1}{4}} \sqrt{R} + \varepsilon_T + (K_{\ell} s)^{-1}\right),
\end{align}
in agreement with \eqref{eq:LebesgueError}.
We further decompose $J_1$ as
\begin{align}
J_1 = \big( \int_{\{10 (K_{\ell } s)^{-1} \leq |t| \leq 100\}} + \int_{\{ |t| \geq 100\}}\big) \alpha_1 f\big(\frac{\beta_1}{\alpha_1}\big) - \alpha_2 f\big(\frac{\beta_2}{\alpha_2}\big)\,dt =: J_{1,1} + J_{1,2}.
\end{align}
The integral over a bounded set, $J_{1,1}$, is estimated termwise using \eqref{eq:tautilde-Est} and is easily seen to be consistent with \eqref{eq:LebesgueError}.

To estimate $J_{1,2}$ we use that for $0\leq x \leq y \leq \frac{1}{2}$, we have the elementary estimates
\begin{align}
|g(x)|\leq C x^4, \qquad |g(x)-g(y)| \leq C y^3 |x-y|.
\end{align}
Furthermore,
\begin{align}
\left| \frac{\beta_1}{\alpha_1} - \frac{\beta_2}{\alpha_2} \right| \leq \frac{|\beta_1|}{|\alpha_1 \alpha_2|} |\alpha_2-\alpha_1| + \frac{|\beta_1-\beta_2|}{|\alpha_2|},
\end{align}
as well as
\begin{align}
\min\{ (1-\varepsilon_N) \widetilde{\tau} + \beta,  (1-\varepsilon_N)t^2 + 8\pi\} \geq \frac{1}{2} t^2,
\end{align}
for $|t| \geq 100$. Therefore, using also \eqref{eq:tautilde-Est},
\begin{align}
|J_{1,2}| &\leq 
 \int_{\{ |t| \geq 100\}} C t^{-8} \left( \rmu a R^2 |t|^2 + \varepsilon_T t^2 + \frac{1}{K_{\ell}s} (\frac{\rmu}{\rho_z})^{1/2} |t|
 \right) \,dt \nonumber \\
&\quad +  \int_{\{ |t| \geq 100\}} C t^{-6} \left( \frac{1}{t^4}\left[ \rmu a R^2 |t|^2 + \varepsilon_T t^2 + \frac{1}{K_{\ell}s} (\frac{\rmu}{\rho_z})^{1/2} |t|\right]
+ \rmu a R^2 |t|^2
\right)\,dt.
\end{align}
Since the integral is convergent, this is also seen to be in agreement with \eqref{eq:LebesgueError} (using that $\rmu a R^2 \leq 1$ by \eqref{con:ER}.
This finishes the estimate of \eqref{eq:LebesgueError}.

The integral $I_1''$ from \eqref{eq:SplitIntegrals} is split in $3$ parts. For $|t| \leq 10 (K_{\ell}s)^{-1}$, we have $0 \leq t^2 - \widetilde{\tau}(t) \leq t^2$. Therefore,
\begin{align}
\left|  \int_{\{ |t| \leq 10 (K_{\ell}s)^{-1}\} } \frac{\beta^2}{2} \frac{t^2-\widetilde{\tau}}{t^2 \alpha} \right|
\leq C (K_{\ell}s)^{-1},
\end{align}
which is consistent with \eqref{eq:BogIntegral}.

For $10 (K_{\ell} s)^{-1} \leq |t| \leq 10 (K_{\ell} d s)^{-1} $, we have \eqref{eq:tautilde-Est} above.
Therefore,
\begin{align}
\left|  \int_{\{ 10 (K_{\ell} s)^{-1} \leq |t| \leq 10 (K_{\ell} d s)^{-1}\} } \frac{\beta^2}{2} \frac{t^2-\widetilde{\tau}}{t^2 \alpha} \right|
\leq
C \varepsilon_T (K_{\ell} d s)^{-1} + C (K_{\ell} s)^{-1} \log(d^{-1}),
\end{align}
in agreement with \eqref{eq:BogIntegral} .

Finally the case $|t| \geq 10 (K_{\ell} d s)^{-1}$.
Here, $0 \leq t^2 - \widetilde{\tau}(t) \leq C |t| ( (K_{\ell} s)^{-1} + \varepsilon_T (K_{\ell} d s)^{-1})$ and $\alpha \geq \frac{1}{2} t^2$.
Therefore, 
\begin{align}\label{eq:Logarithm}
&\left|  \int_{\{ |t| \geq 10 (K_{\ell} d s)^{-1} \} } \frac{\beta^2}{2} \frac{t^2-\widetilde{\tau}}{t^2 \alpha} \right| 
\nonumber \\
&\leq
C ( (K_{\ell} s)^{-1} + \varepsilon_T (K_{\ell} d s)^{-1}) \int_{\{ 10 (K_{\ell} d s)^{-1}\leq |t| \leq (\rho_z a^3)^{-1/2}\}} |t|^{-3} \nonumber \\
&\quad + C ( (K_{\ell} s)^{-1} + \varepsilon_T (K_{\ell} d s)^{-1}) (\rho_z a^3)^{1/2} a^{-2} 
 \int\frac{\widehat{W}_1(\sqrt{\rho_z a} t)^2}{t^2} \nonumber \\
 &\leq C ( (K_{\ell} s)^{-1} + \varepsilon_T (K_{\ell} d s)^{-1}) \left( \log\big( \frac{K_{\ell} d s }{(\rmu a^3)^{1/2}}\big) + 1 \right). 
\end{align}
Since this estimate is also in agreement with \eqref{eq:BogIntegral} 
this finishes the proof of Lemma~\ref{lem:BogIntegral}.
\end{proof}

\section{The parameters involved in the proof of Theorem~\ref{thm:LHY-Background}}
\label{sec:params}

The data of our problem are the 'chemical potential' parameter $\rmu$ and the potential $v$. The potential will appear in estimates in terms of the scattering length $a$, the radius $R$ of the support and integral $\int v$.

In the proof we need a number of auxilliary dimensionless parameters
\begin{equation}
0 < s,d,\varepsilon_T, \varepsilon_R , K_{\ell}, \widetilde{K}_{\mathcal M},K_H', K_H'', K_B, K_N, K_{\mathcal R}.
\end{equation}
These will all be chosen at the end of this section as powers of 
the dimensionless diluteness parameter
$$
\rmu a^3.
$$
The parameter $K_{\mathcal R}$ also depends on the constant ${\mathcal C}$ from \eqref{eq:assumps_thm_background} and $\varepsilon_R$ depends on $R$.
In the localized kinetic energy defined in \eqref{eq:DefT_tilde} there is also a parameter $b$, which in Theorem~\ref{thm:CompareBoxEnergy} is chosen as a (sufficiently small) universal constant and a parameter $\varepsilon_N$ defined in \eqref{eq:Cond_epsilonN} below in terms of $K_N$.
Furthermore, there is a choice of a specific localization function $\chi$, defined in Appendix~\ref{sec:chiproperties}, that depends on a regularity parameter 
$$
M\in 4 \N,
$$ 
which we will choose explicitly below.
In order for the proof to work, the parameters have to satisfy a number of relations. We have collected these in Assumption~\ref{assump:params} below. After the statement of this assumption, we indicate key places where these assumptions are used. Finally, we make an explicit choice for which all the conditions are satisfied. Here and in the rest of the paper $f\ll g$ is used in the precise meaning that $(f/g)\leq (\rmu
a^3)^{\varepsilon}$ for some positive $\varepsilon$ and likewise for $f \gg g$. Throughout the paper there will also be logarithmic factors. They are ignored 
here as they are always accomodated by the conditions given, since they can always be absorbed in the $\varepsilon$ powers implicit in the $\ll$-signs.

The integral $\int v$ satisfies \eqref{eq:assumps_thm_background}, which we write in terms of the parameter $K_{\mathcal R}$ as
\begin{align}\label{con:intv}
{\mathcal R} =\frac{1}{8\pi
a} \int v \leq K_{\mathcal R} (\rmu a^3)^{-\frac{1}{2}}.
\end{align}
This is simply saying that the parameter $K_{\mathcal R}$ is defined by the exponent $\eta_2$ and the constant ${\mathcal C}$ in Theorem~\ref{thm:LHY-Background} by
\begin{equation}\label{con:KR}
K_{\mathcal R}  := {\mathcal C} (\rmu a^3)^{-\eta_2}.
\end{equation}

\begin{assumption}\label{assump:params}
The parameters satisfy
\begin{align}\label{eq:small_and_large}
0 < s,d,\varepsilon_T, \varepsilon_R \ll 1  \ll K_{\ell}, \widetilde{K}_{\mathcal M},K_H', K_H'', K_B, K_N, K_{\mathcal R}.
\end{align}
Furthermore, they are interdependent by the conditions,
\begin{align}
  (dK_\ell)^2\ll \varepsilon_TK_\ell^{-2}&\ll\varepsilon_T \ll sdK_\ell,\label{con:eTdK}\\
    sK_\ell&\gg 1,\label{con:sKell} \\ 
     sdK_\ell & \gg K_B^{-1}\label{con:sdKellKB}, \\
     d^{-2} \ll K_H''  &\ll K_H', \label{cond:disjoint} \\
     \label{cond:KH3-n}
K_{\ell}^4 (K_H'')^3& \ll \widetilde{K}_{{\mathcal M}}.
\end{align}
Their magnitudes are controlled, in terms of the diluteness parameter, by the conditions,
\begin{align}
 K_B^3K_\ell^2 &\ll (\rmu a^3)^{-\frac{1}{4}}\label{con:KBKell}, \\
 \widetilde{K}_{{\mathcal M}} K_B^3 K_{\ell}^2 &\ll (\rmu a^3)^{-\frac{1}{2}}, \label{con:KBellKM} \\
\label{con:Constants}
K_N^{\frac{1}{2}} K_{\ell} K_B^3 (K_H')^3 &\ll (\rmu a^3)^{-\frac{1}{4}},\\
\label{con:newML}
\widetilde{K}_{{\mathcal M}}  K_{\mathcal R}^{\frac{1}{2}}  K_N^{\frac{1}{8}} (K_H')^{\frac{3}{4}} K_B^{\frac{3}{4}} &\ll K_{\ell}^{\frac{5}{2}} (\rmu a^3)^{-\frac{1}{16}},\\
  \label{eq:Newlambdacondition}
  K_\ell d&\gg (\rmu a^3)^{\frac{1}{6}}.
\end{align}
Their relation to the radius $R$ of the potential is controlled by the conditions,
\begin{align}
  \label{con:RNew}
R/a & \ll  d  K_{\ell} (\rmu a^3)^{-\frac{1}{2}},\\
\label{con:ER}
\rmu a R^2 &\ll \ER \ll K_{\ell}^{-2} K_B^{-3} \ll 1.
\end{align}
Finally, their relation to the regularity parameter $M$ is given by the conditions,\footnote{
The condition \eqref{con:d5s} also appeared in \cite{FS} as (5.3) but with the erroneous power as $d^{-5} s^{M+1}\ll 1$. This minor error affects the choice of $M$ in that paper, that should have been $33$ instead of $30$.
}
\begin{align}
 d^{-5}s^{M-2}&\ll 1,\label{con:d5s}\\
 \label{cond:KHs-ML}
d^{2M} + (K_H'')^{-M/2} + (d^2 K_H'')^{-2M} +  \left( \frac{K_H''}{K_H'} \right)^{M} & \ll \widetilde{K}_{{\mathcal M}}^{-1},\\
\label{eq:dM-small}
d^{2M} &\ll (\rmu a^3)^{\frac{1}{2}}, \\
\label{con:boghamerror}
  (K_H'')^{4-M} K_{\ell}^{\frac{3}{2}} &\ll (\rmu a^3)^{\frac{3}{4}}.
\end{align}

\end{assumption}

\paragraph{Discussion of parameters.}
The proof starts by localizing. This introduces the length scale $\ell$ defined by
\begin{equation}\label{eq:def_ell}
  \ell:=K_\ell(\rmu a)^{-1/2},
\end{equation}
with $K_{\ell} \gg 1$,
and the localization function $\chi$ depending on the regularity parameter $M$.
In order to achieve the necessary a priori bounds in Theorem~\ref{thm:aprioribounds}, we perform a second localization to smaller boxes of scale $d \ell$ with $d \ll 1$.
This double localization procedure requires the additional parameters
$\varepsilon_T, s \ll 1$ and $K_B\gg 1$, satisfying the relations \eqref{con:d5s}, \eqref{con:eTdK}, \eqref{con:sKell}, \eqref{con:RNew}, \eqref{eq:Newlambdacondition}  and \eqref{con:sdKellKB}.
The localization also depends on the parameter $b$ (see Theorem~\ref{thm:CompareBoxEnergy}) which is chosen as a (sufficiently small) universal constant.
The conditions \eqref{eq:Newlambdacondition} and \eqref{con:RNew} are used in
the small boxes in Section~\ref{SmallBoxes}.

By \eqref{con:RNew} the length scale of the small box
has to be longer than the radius of the potential. Condition
\eqref{eq:Newlambdacondition} is a weak condition that assures that
the small boxes are not too small.

The condition \eqref{con:KBKell} ensures that the a priori bounds on the particle number
and expected number of excited particles are both correct to leading order (see \eqref{eq:apriori_n} and \eqref{eq:apriorinn+NEW}). 

Compared to \cite{FS}, the localized kinetic energy in  \eqref{eq:DefT_new} includes a new term with the parameter $\varepsilon_N$. This we define to be
\begin{align}\label{eq:Cond_epsilonN}
\varepsilon_N := K_N^{-1} \sqrt{\rmu a^3},
\end{align}
in terms of the parameter $K_N \gg 1$, in order for this extra term not to change calculations up to LHY precision.

A next important step is to use the technique of localization of large matrices, here in the form of Proposition~\ref{prop:LocMatrices}, to restrict to the subspace where the momentum localized excitations, $n_+^{L}$ is bounded.
For this we need some more parameters.
The operator $n_+^{L}$ itself is defined in \eqref{eq:def_nplusL} in terms of the parameter $K_H' \gg 1$ and the kinetic energy. For later estimates it is preferable to work with a momentum localization which is given in terms of a closely related parameter $K_H'' \gg 1$.
We will assume \eqref{cond:disjoint}, where the first condition makes the regions of low and high momentum in \eqref{eq:Momomta} disjoint, and the second is needed in Lemma~\ref{lem:pseudolocal} to dominate the momentum localization by the kinetic energy localization.
It is useful to note in passing the following consequences of \eqref{con:eTdK} and  \eqref{cond:disjoint},
\begin{align}\label{eq:maerkedobbelt}
&K_{\ell} \ll s d^{-1} \ll d^{-1} \ll d^{-2}\ll K_H'',  \\
\label{cond:newKH'}
&\varepsilon_T (ds)^{-2} \ll (K_H')^2.
\end{align}

The bound from the localization of large matrices is given in terms of $\ML$ which we define as
\begin{equation}\label{eq:DefML}
\ML := \widetilde{K}_{{\mathcal M}}^{-1} \rmu \ell^3,
\end{equation}
with $\widetilde{K}_{{\mathcal M}} \gg 1$, since $\ML$ has to be much smaller than the number of particles in the box.
The condition \eqref{con:KBellKM}
states that the upper bound $\ML$ on the number of excited particles must be much 
bigger than the {\it expected} number of these particles, which 
in Theorem~\ref{thm:aprioribounds} is shown to be not much worse than $K_B^3K_\ell^2\rmu\ell^3(\rmu a^3)^{1/2}$. 

Furthermore, in the localization of large matrices-argument, Proposition~\ref{prop:LocMatrices}, we need \eqref{con:Constants} and  \eqref{con:newML}, where \eqref{con:Constants} is a weak condition that determines the relative size of two error terms, whereas \eqref{con:newML} it crucial for the localization error to be smaller that the LHY-term.

The condition \eqref{cond:KH3-n} is crucial in order to absorb the error term from \eqref{eq:Q3-1-reduce} in the kinetic energy gap.

After introducing second quantization it turns out to be useful to do
$c$-number substitution in the spirit of \cite{LSYc}. 
This is done in Section~\ref{sec:Second}.
After $c$-number
substitution, where the annihilation operator for the constant
functions is replaced by a number $z$ we need to control that the
parameter $\rho_z=|z|^2/\ell^3$, which represents the density in the
$c$-number substituted condensate, is sufficiently close to $\rmu$. This is
done in Section~\ref{sec:rough}. 
To help in this control, in the case of potentials with large support, 
we introduced in Section~\ref{sec:Second} a new parameter $\ER$---satisfying \eqref{con:ER}---and a corresponding quadratic term in \eqref{eq:Kz}. Notice that this is a modification compared to \cite{FS}.

The conditions \eqref{cond:KHs-ML}, \eqref{eq:dM-small} and \eqref{con:boghamerror} are used to control localization errors from Fourier space, using the finite but high regularity $M$ of the function $\chi$.

Notice that the parameter $K_{\mathcal R}$ is defined in \eqref{con:KR} in terms of the parameter $\eta_2>0$ the existence of which has to be established.

\paragraph{Choice of parameters} We will choose to let all the parameters except $\ER$ depend on a small parameter $X\ll1$ in the following way 
\begin{align}\label{eq:Xparameter}
s=X,\ d=X^6, \ \varepsilon_T&=X^{23/4}, \ K_\ell=X^{-3/2},\ K_B=X^{-6},\
K_N= X^{-1},\ K_{\mathcal R}={\mathcal C} X^{-1},\nonumber \\
K_H''&=X^{-13},\ K_H'=X^{-14},\ \widetilde{K}_{{\mathcal M}}=X^{-46}.
\end{align}
Then the conditions \eqref{con:eTdK}--\eqref{cond:KH3-n}
will be satisfied. 

We now choose
\begin{equation}\label{eq:Xchoice}
X= (\rmu a^3)^{\min\{ \frac{1}{928}, \frac{\kappa}{11}\}},
\end{equation}
where $\kappa$ is the parameter from \eqref{eq:assumps_thm_background}.
With this choice \eqref{con:KBKell}
--\eqref{con:RNew}  are satisfied.
Notice that with this choice we get $\eta_2= 5\eta_1$ with $\eta_1$ given by \eqref{eq:Choiceeta2}
as an explicit choice for the parameter $\eta_2$,  the existence of which is part of the statement of Theorem~\ref{thm:LHY-Background}.

We choose
\begin{equation}
\ER:= \sqrt{\rmu a R^2 K_{\ell}^{-2} K_B^{-3}}.
\end{equation}
With this choice \eqref{con:ER} is satisfied.

Finally, \eqref{con:d5s}
--\eqref{con:boghamerror} are satisfied if we choose
\begin{align}\label{eq:choiceM}
M > \max\{ 79, \frac{11}{12} \kappa^{-1}\}.
\end{align}

\section{Proof of Theorem~\ref{thm:aprioribounds}}
\label{SmallBoxes}
In this appendix we will prove Theorem~\ref{thm:aprioribounds}. This
will be done in two steps. First we prove the a priori bound on the
expectation value of $n_+$ in \eqref{eq:apriorinn+NEW} this requires a
localization to smaller boxes where certain errors due to the number
of excited particles can be absorbed into a large Neumann gap.  This
is done in Subsection~\ref{sec:n+apriori}.  Armed with the a priori
estimate on $n_+$ we are able to repeat the analysis in the large box
$\Lambda$ to get the a priori bounds on $n$ and
${\mathcal{Q}}^{\rm ren}_4$ in \eqref{eq:apriori_n} and
  \eqref{eq:Q4apriori_2}. This is done in
  Subsection~\ref{sec:NewAppB}. The arguments in the two subsections
  are very similar and are indeed very similar to Appendix B in
  \cite{FS}.

  \subsection{Proof of  the a priori bound on $n_+$}\label{sec:n+apriori}
The Hamiltonian ${\mathcal H}_{\Lambda}(\rmu)$ defined in \eqref{eq:Def_HB_new}
is localized to the box 
$\Lambda:=\Lambda(0)=[-\ell/2,\ell/2]^3$. In order to arrive at the a priori bounds in Theorem~\ref{thm:aprioribounds} 
we will localize again to boxes with a length scale $\ell d\ll (\rho a)^{-1/2}$. 
The reason for this second localization is that we need a larger Neumann gap in order to absorb errors. 
We therefore introduce a new family of boxes (some of which will have a rectangular shape) given by 
\begin{equation}
B(u)=[-\ell/2,\ell/2]^3\cap \left(\ell d u +[-\ell d/2, \ell d /2]^3\right),\quad u\in \R^3.
\end{equation}
The functions that localize to these boxes are
\begin{equation}\label{eq:chiBdefinition}
        \chi_{B(u)}(x)=\chi\left(\frac{x}{\ell}\right)\chi\left(\frac{x}{d\ell}-u\right), \quad u\in\R^3,
\end{equation}
where $\chi$ is given in \eqref{eq:Def_chi}
in terms of the positive integer $M$. Observe that
\begin{equation}\label{eq:chiBintegral}
\iint\chi_{B(u)}(x)^2dx du= \ell^3.
\end{equation}

We consider the projections 
$$
P_{B(u)}\phi=|B(u)|^{-1}\langle\one_{B(u)},\phi\rangle \one_{B(u)},\quad
Q_{B(u)}\phi=\one_{B(u)}\phi -P_{B(u)}\phi.
$$
In these small boxes we consider the Hamiltonian
\begin{equation}\label{eq:HBU}
  {\mathcal H}_{B(u)}(\rmu)=\sum_{i=1}^N\left((1-\varepsilon_N){\mathcal T}_{B(u),i}-\rmu \int w_{1,B(u)}(x_i,y)dy\right)
  +\frac12\sum_{i\ne j} w_{B(u)}(x_i,x_j),
\end{equation}
where (omitting the index $u$)
\begin{equation}\label{eq:appTB}
  {\mathcal T}_{B}=\frac12\varepsilon_T(1+\pi^{-2})^{-1}(d\ell)^{-2}Q_B
  +Q_B\chi_B[\sqrt{-\Delta}-(ds\ell)^{-1}]_+^{-2}\chi_{B}Q_B,
\end{equation}
and 
\newcommand{\WS}{W^{\rm s}}
\begin{equation}
  w_{B}(x,y)=\chi_B(x)\WS(x-y)\chi_B(y),
  \quad w_{1,B}(x,y)=\chi_B(x)\WS_1(x-y)\chi_B(y),
\end{equation}          
with (where the superscript s refers to small and where we need the definitions \eqref{eq:defW_12})
\begin{equation}
  \WS(x)=\frac{W(x)}{\chi*\chi(x/(d\ell))},\quad 
  \WS_1(x)=\frac{W_1(x)}{\chi*\chi(x/(d\ell))}.
\end{equation}
Here we use that $R< d \ell/2$ by \eqref{con:RNew}.
As in the large boxes we will also need 
\begin{equation}
\quad w_{2,B}(x,y)=\chi_B(x)\WS_2(x-y)\chi_B(y),\quad 
\WS_2(x)=\frac{W_2(x)}{\chi*\chi(x/(d\ell))}.
\end{equation}
Since $\omega\leq 1$ we have 
\begin{equation}\label{eq:WSestimate}
  \int \WS_2\leq 2\int \WS_1(x)\leq C a .
\end{equation}
Similarly to \eqref{eq:W1-g-new} we get
\begin{align}\label{eq:W1-g-new-App}
0 \leq \WS_1(x) -  g(x)  \leq C g(x) \frac{|x|^2}{d^2\ell^2}
\end{align}
and with an identical proof we get a result similar to
\eqref{eq:I2-integral-new}, but with the length scale $\ell$ replaced by
$\ell d$, i.e.,
\begin{align}\label{eq:WS2g}
\left|(2\pi)^{-3}\int_{\R^3} \frac{\widehat{\WS_1}(p)^2}{2p^2}dp - \widehat{g\omega}(0)\right|
 \leq C \frac{Ra^2}{(d\ell)^2}.
\end{align}

We have by a Schwarz inequality that 
\begin{equation}\label{eq:iintw1}
  \iint w_{1,B}(x,y)dxdy\leq \iint\chi_B(x)^2\WS_1(x-y)dxdy\leq \cst a\int\chi_B^2\leq Ca|B|.
\end{equation}
Observe also that 
\begin{equation}\label{eq:appw1sliding}
  \iiint w_{1,B(u)}(x,y)dxdy du= \ell^{3}\int g =8\pi a\ell^{3}.
\end{equation}

It can be inferred from the proof of Theorem 3.10 in \cite{BS} that
the operator ${\mathcal H}_{\Lambda}(\rmu)$ 
can be bounded below
by (we are for the lower bound ignoring the first term in ${\mathcal
  T}$ in \eqref{eq:DefT_new} and the fourth term in
$\widetilde{\mathcal T}$ in \eqref{eq:DefT_tilde})
\begin{equation}\label{eq:BSsliding2}
  {\mathcal H}_{\Lambda}(\rmu)\geq \sum_{i=1}^N (1-\varepsilon_N)\frac{b}{2} Q_{\Lambda,i}\ell^{-2}+
  \int_{\R^3} {\mathcal H}_{B(u)}(\rmu) du,
\end{equation}
if
\begin{equation}\label{eq:BSlidingassp}
\varepsilon_T,s,ds^{-1},\text{ and } (s^{-2}+d^{-3})(sd)^{-2}s^M
\end{equation}
are smaller than some universal constant. Note that, if $\rmu a^3$ is
small enough, this is satisfied for our choices in
Section~\ref{sec:params}, in particular, due to \eqref{con:d5s}.

In the integral above the operators ${\mathcal H}_{B(u)}(\rmu)$
are, however, not unitarily equivalent. Depending on $u$ the boxes $B(u)$ can
be rather small and rectangular. We denote by
$\lambda_1(u)\leq\lambda_2(u)\leq\lambda_3(u)\leq d\ell$ the side
lengths of the boxes $B(u)$. To avoid boxes that are very small, in particular that $\lambda_1(u)\leq R$, we will restrict the integral  in \eqref{eq:BSsliding2} to
$u$ such that 
$$
\|\ell d u\|_\infty\leq \frac\ell2(1+d)-\lambda,
$$
with
\begin{equation}\label{eq:newlambdadef}
  \lambda=\frac1{10}d\ell.
\end{equation}
Note that since the full integral would be over the set where 
$\|\ell d u\|_\infty\leq \frac\ell2(1+d)$ the 
restriction implies that all boxes will satisfy $\lambda_1(u)\geq\lambda$.

For the kinetic energy and the repulsive potential
this restriction will only give a further lower bound. For the chemical potential term we will use the following result. 
\begin{lemma}For all $x\in\Lambda$ we have the estimate\footnote{
A similar estimate was also given in \cite{FS} as Lemma B.1, but with the minor error 
that 7 was replaced by 3.}.  
\begin{align}
-\rmu\iint &w_{1,B(u)}(x,y)dy du \nonumber \\
\geq &
-\rmu\int_{\{\|u\|_\infty-\frac12\left(\frac1d+1\right)\leq - \frac{\lambda}{\ell d}\}}\int w_{1,B(u)}(x,y)dy du\nonumber\\
&-7\rmu\int_{\{-2\frac{\lambda}{\ell d}\leq \|u\|_\infty-\frac12\left(\frac1d+1\right)\leq -\frac{\lambda}{\ell d}\}}\int w_{1,B(u)}(x,y)dy du.
\end{align}
\end{lemma}
\begin{proof} The estimate above follows if we can show that for all $x,y\in\Lambda$ we have 
\begin{align}
  \chi*\chi\left(\frac{x-y}{\ell d}\right)\leq& \int_{\{\|u\|_\infty-\frac12\left(\frac1d+1\right)\leq -\frac{\lambda}{\ell d}\}}
  \chi\left(\frac{x}{\ell d}-u\right)\chi\left(\frac{y}{\ell d}-u\right)du\nonumber\\&
  +7\int_{\{-2\frac{\lambda}{\ell d}\leq \|u\|_\infty-\frac12\left(\frac1d+1\right)\leq -\frac{\lambda}{\ell d}\}}
  \chi\left(\frac{x}{\ell d}-u\right)\chi\left(\frac{y}{\ell d}-u\right)du.
\end{align}
We have 
\begin{align}
\chi*\chi\left(\frac{x-y}{\ell d}\right)-& \int_{\{\|u\|_\infty-\frac12\left(\frac1d+1\right)\leq -\frac{\lambda}{\ell d}\}}
  \chi\left(\frac{x}{\ell d}-u\right)\chi\left(\frac{y}{\ell d}-u\right)du\nonumber\\=&
  \int_{\{\|u\|_\infty-\frac12\left(\frac1d+1\right)\geq-\frac{\lambda}{\ell d}\}}
  \chi\left(\frac{x}{\ell d}-u\right)\chi\left(\frac{y}{\ell d}-u\right)du.
\end{align}
Since $x,y\in\Lambda$, the integral on the right is supported on
$\|u\|_\infty-\frac12\left(\frac1d+1\right)\leq 0$. Using the fact that by \eqref{eq:newlambdadef}
$\lambda< \ell d/4$ and that $\chi$
is a product of symmetric decreasing functions of the coordinates 
$u_1,u_2,u_3$ respectively, we may observe that for fixed $u_2,u_3$ we have
\begin{align}
  &\max_{\frac12\left(\frac1d+1\right)-\frac{\lambda}{\ell d}\leq|u_1|\leq \frac12\left(\frac1d+1\right)}
  \chi\left(\frac{x}{\ell d}-u\right)\chi\left(\frac{y}{\ell d}-u\right)
  \leq\nonumber
\\ &\min_{\frac12\left(\frac1d+1\right)-2\frac{\lambda}{\ell d} \leq|u_1|\leq \frac1{2d}+\frac12-\frac{\lambda}{\ell d}}
  \chi\left(\frac{x}{\ell d}-u\right)\chi\left(\frac{y}{\ell d}-u\right).
\end{align}
Using this repeatedly (also with $u_1,u_2$ and $u_1,u_3$ fixed) gives the result in the lemma. 
\end{proof}
As a consequence of the lemma we find from \eqref{eq:BSsliding2}, if \eqref{eq:BSlidingassp} is satisfied, that
\begin{equation}\label{eq:appsliding}
  {\mathcal H}_{\Lambda}(\rmu)\geq 
  \frac{b}{2}\ell^{-2}\sum_{i=1}^NQ_{\Lambda,i}+\int_{\|\ell d u\|_\infty\leq \frac12 \ell (1+d)
    -\lambda} {\mathcal H}_{B(u)}(m(u)\rmu) du,
\end{equation}
where $m(u)=1$ if 
$\|\ell d u\|_\infty\leq \frac12 \ell (1+d)
    -2\lambda$
and $m(u)=8$ otherwise, i.e., for $u$ near the boundary. 

The goal in the rest of this section is to give a lower bound on the
ground state energy of the operators ${\mathcal H}_{B(u)}(m(u)\rmu)$
to conclude an a priori lower bound on the ground state
energy of ${\mathcal H}_{\Lambda}(\rmu)$. 
We may now assume that the shortest side length of $B(u)$ satisfies
$\lambda_1(u)\geq\lambda$ and we will make use of the fact that the range $R$ of the potential 
satisfies $R\ll \lambda$ by \eqref{con:RNew}.
For simplicity we will often omit the parameter $u$. 
A main ingredient in getting a lower bound is to get a priori bounds on the operators 
\begin{equation}\label{eq:smallboxnn0n+}
n=\sum_{i=1}^N \one_{B,i},\quad n_0=\sum_{i=1}^NP_{B,i},\quad n_+=\sum_{i=1}^NQ_{B,i}.
\end{equation}
Note that the operator $n$ commutes with ${\mathcal H}_{B}$,
so we may consider $n$ a number.

Applying the decomposition of the potential energy in Subsection~\ref{sec:potsplit} 
to the small boxes we arrive at the following lemma. 
\begin{lemma}\label{lm:appinteractionestimate} There is a constant $C>0$ such that on any small box $B$ we have 
\begin{align}\label{eq:SmallsimpleQs}
  -\rmu \sum_{i=1}^N \int w_{1,B}(x,y)\,dy+
   \frac{1}{2} \sum_{i\neq j}  w_B(x_i, x_j)
  \geq A_0+A_2-Ca (\rmu +n_0|B|^{-1})n_+  
\end{align}
where
\begin{align}
  A_0&=\frac{n_0(n_0-1)}{2|B|^2}\iint w_{2,B}(x,y)\,dx dy \nonumber\\
  &- \left(\rmu \frac{n_0}{|B|}+\frac14\left(\rmu
  -\frac{n_0-1}{|B|}\right)^2\right)\iint w_{1,B}(x,y)\,dx dy
  \label{eq:A0}
\intertext{and} 
  A_2&= \frac{1}{2}\sum_{i\neq j} P_{B,i} P_{B,j} w_{1,B}(x_i,x_j) Q_{B,j} Q_{B,i} + h.c.
\end{align} 
\end{lemma}

\begin{proof}
  Lemma~\ref{lm:appinteractionestimate} is identical to \cite[Lemma~B.2]{FS}, so we omit the proof here.
  Moreover, we give an almost identical proof of a slightly stronger result in the case of the large box
  $\Lambda$ in Lemma~\ref{lm:appinteractionestimate_a} below.  
\end{proof}

We express the term $A_2$ from the lemma in second quantization. Introducing the operators
$$
\ta_p^\dagger=|B|^{-1/2}a^\dagger(Q_B\chi_Be^{-ipx})a_0
$$
we can write 
$$
A_2=\frac12 (2\pi)^{-3}\int \widehat{\WS_1}(p)(\ta_p^\dagger \ta_{-p}^\dagger+\ta_{-p}\ta_p)dp.
$$
We shall control $A_2$ using Bogoliubov's method. In order to do this we will add and subtract a term 
\begin{equation}\label{eq:A1F}
  A_1= (2\pi)^{-3}K_{\rm s} a \int (\ta_p^\dagger \ta_{p}+\ta^\dagger_{-p}\ta_{-p})dp,
\end{equation}
with the universal constant $K_{\rm s}>0$ chosen appropriately below. 
Note that we have 
\begin{equation}\label{eq:A1estimate-smallbox}
  A_1\leq K_{\rm s} a \frac{n_0+1}{|B|} n_+\|\chi_B\|_\infty^2\leq CK_{\rm s} a \frac{n_0+1}{|B|} n_+.
\end{equation}

\begin{lemma}[Bogoliubov's method in small boxes]\label{lm:appbogolubov} 
  If the parameters $\varepsilon_N,s,d$, and $K_\ell$ satisfy Assumption~\ref{assump:params} then there exists a constant $C>0$ such that
  if $\rmu a^3\leq C^{-1}$ then 
\begin{align}\label{eq:refinedaprioribog}
  (1-\varepsilon_N)&\sum_{i=1}^NQ_{B,i}\chi_{B,i}[\sqrt{-\Delta_i}-(ds\ell)^{-1}]_+^{-2}\chi_{B,i}Q_{B,i}+A_2
  \geq\nonumber\\&-\frac{1}{2} \widehat{g\omega}(0) \frac{(n+1)n}{|B|^2} \int \chi_B^2
  - C a \left(\frac{Ra}{(d\ell)^2}+a(ds\ell)^{-1}\log(ds\ell a^{-1})\right)\frac{(n+1)n}{|B|^2}\int\chi_B^2\nonumber\\
  &-C\left(a^4(ds\ell)^{3}\left(\frac{n+1}{|B|}\right)^3 +a(ds\ell)^{-3}\right)\frac{n}{|B|}\int\chi_B^2
  -Ca\frac{n+1}{|B|}n_+.
\end{align}
\end{lemma}
\begin{proof}
We add $A_1$ from \eqref{eq:A1F} to the term we want to estimate. Using $n_0\leq n$ we may write
\begin{equation}\label{eq:smallbox-h-op-int}
  (1-\varepsilon_N)\sum_{i=1}^N Q_{B,i}\chi_{B,i}[\sqrt{-\Delta_i}-(ds\ell)^{-1}]_+^{-2}\chi_{B,i}Q_{B,i}+A_1+A_2\geq
  (2\pi)^{-3}\frac12\int h(p) dp,
\end{equation}
where $h$ is the operator
$$
h(p)={\mathcal A}(p)(\ta_p^\dagger \ta_{p}+\ta^\dagger_{-p}\ta_{-p})+{\mathcal B}(p)(\ta_p^\dagger \ta_{-p}^\dagger+\ta_{-p}\ta_p).
$$
with
$$
{\mathcal A}(p)= (1-\varepsilon_N)\frac{|B|}{n+1}[|p|-(ds\ell)^{-1}]^2_++2K_{\rm s} a,\quad
{\mathcal B}(p)= \widehat{\WS_1}(p).
$$
Here we have used that for any function $f(p)$ we have 
\begin{align*}
(2\pi)^{-3}\int f(p) \ta_p^\dagger\ta_pdp &=|B|^{-1}(n_0+1)\sum_j \left(Q_B\chi_Bf(\sqrt{-\Delta})\chi_B Q_B\right)_j \\
&\leq
|B|^{-1}(n+1)\sum_j \left( Q_B\chi_Bf(\sqrt{-\Delta})\chi_B Q_B\right)_j.
\end{align*}

We observe that 
$$
[\ta_p,\ta_p^\dagger]\leq n_0|B|^{-1}\int \chi_B^2\leq n|B|^{-1}\int \chi_B^2.
$$
We have by \eqref{eq:WSestimate} that 
$$
|{\mathcal B}(p)|= |\widehat{\WS_1}(p)|\leq \int\WS_1\leq C_0a .
$$ 
If we therefore choose $K_{\rm s}= C_0$ we see that $|{\mathcal B}|/{\mathcal A}\leq 1/2$, and 
we get the following lower bound from Theorem~\ref{thm:bogolubov-complete}. 
$$
h(p)\geq -\left({\mathcal A}(p)-\sqrt{{\mathcal A}(p)^2-{\mathcal B}(p)^2}\right)n|B|^{-1}\int \chi_B^2.
$$
We also see that all the assumptions are satisfied in order to apply
Lemma~\ref{lm:aprioriintegral_new} to estimate the integral in
\eqref{eq:smallbox-h-op-int}. Note, in particular, that
$\varepsilon_N\leq a(d s \ell)^{-1}\leq 1/2$ by \eqref{eq:small_and_large}, \eqref{con:eTdK},
\eqref{con:sdKellKB}, \eqref{con:KBKell},
and \eqref{eq:Cond_epsilonN} with $K_N\gg 1$. Together with the
estimates \eqref{eq:WS2g} and \eqref{eq:A1estimate-smallbox} this
proves \eqref{eq:refinedaprioribog}.
\end{proof}

To estimate the energy in the small boxes it is important to
establish the result in the following lemma on the integrals in \eqref{eq:A0}.
\begin{lemma}\label{lm:w2>w1} 
  There is a constant $C>0$ such that if $R<\ell d/2$ and if $\lambda_1$ denotes the shortest length of the box
  $B$ then
\begin{equation}\label{eq:w1>}
  \iint w_{1,B}(x,y)dx dy \geq 8\pi a\left(1-C
  \Big(\frac{R}{\lambda_1}\Big)^2\right)\int\chi_B^2.
\end{equation}
\begin{equation}\label{eq:w>g}
  \iint w_{2,B}(x,y)dx dy\geq \iint w_{1,B}(x,y)dx dy +
  \widehat{g\omega}(0)\int\chi_B^2
  -C \frac{Ra}{\lambda_1^2} a \int\chi_B^2.
\end{equation}
\end{lemma}
\begin{proof}
The estimate \eqref{eq:w1>} follows from
\begin{align}
  \iint w_{1,B}(x,y)dxdy 
  &\geq \int W_1^{\rm s}(x)dx \left(\int\chi_B^2-CR^2\|\nabla^2\chi_B\|_\infty\int\chi_B\right)\nonumber\\
  & \geq \left(1-C (R\lambda_1^{-1})^2\right)\int  g(x)dx\int\chi_B^2\nonumber
\end{align}
where we have used that $W$ is spherically symmetric, \eqref{eq:W1-g-new-App},
$|B|^{-1}\left(\int\chi_B\right)^2\leq \int\chi_B^2$, and that 
\begin{equation}\label{eq:chiBestimate}
  \|\partial_i\partial_j\chi_B\|_\infty\leq C_M \lambda_1^{-2}|B|^{-1}\int\chi_B,
\end{equation}
which is a simple exercise (see (B.37) and Appendix C in \cite{FS}). We do not need to assume that $1-C (R\lambda_1^{-1})^2>0$ since otherwise the inequality is trivially true.
We conclude using $\int g=8\pi a$.
The proof of \eqref{eq:w>g} uses the same ideas as well as \eqref{omegabounds}, as follows.
\begin{align}
  \iint w_{2,B}(x,y)dx dy&- \iint w_{1,B}(x,y)dx dy -  \int \omega(x)W_1^{\rm s}(x)dx \int\chi_B^2 \nonumber \\
  &\geq - C \int \omega(x)W_1^{\rm s}(x)|x|^2 dx \|\nabla^2\chi_B\|_\infty\int\chi_B\nonumber\\
  & \geq - C \frac{a^2 R}{\lambda_1^2} \int \chi_B^2,
\end{align}
where we in particular used \eqref{eq:WSestimate}, \eqref{eq:chiBestimate} and \eqref{omegabounds} to get the last estimate.
Finally, \eqref{eq:w>g} now easily follows upon using \eqref{eq:W1-g-new-App}.
\end{proof}
Note that the error $R^2/\lambda_1^2$ in \eqref{eq:w1>} is worse than the error $Ra/\lambda_1^2$ in
\eqref{eq:w>g}. Below we shall however only use \eqref{eq:w1>} to conclude that $\iint w_{1,B}(x,y)dx dy\geq ca\int \chi_B^2$,
which follows if $R/\lambda_1$ is small enough.

We are now ready to give the bound on the energy in the small boxes. 
\begin{theorem}[Lower bound on energy in small boxes]\label{thm:smallbox}
  Assume $B$ is a box with shortest side length
  $\lambda_1\geq \lambda$, where $\lambda$ is given in
  \eqref{eq:newlambdadef}.  Assume moreover that the parameters
  $\varepsilon_T, s, d, K_N, K_\ell$, and $K_B$  satisfy
  Assumption~\ref{assump:params} and that $R$ satisfies \eqref{con:RNew}. 
  Then there are universal constants
  $C>0$ and
  $0<c<1/2$ such that if $\rmu a^3< C^{-1}$
    we have for the Hamiltonian defined in \eqref{eq:HBU} that
  \begin{align}
    {\mathcal H}_B(\rmu)\geq& 
                              -\frac12\rmu^2\iint w_{1,B}(x,y)dxdy
                              -CK_B^3\rmu^2 a\sqrt{\rmu a^3}\int\chi_B^2-C\rmu a.
  \end{align}
\end{theorem}

\begin{proof} From
  \eqref{eq:small_and_large},\eqref{con:eTdK},\eqref{con:sdKellKB},\eqref{con:RNew} and
  \eqref{eq:Cond_epsilonN} we may for all $C'>0$ assume $\rmu a^3$ small enough such that 
  \begin{equation}\label{eq:appdellsassumption}
    (\rmu a)^{-1/2}K_B^{-1}\leq \frac{sd\ell}{\ln(ds\ell a^{-1})}\leq d\ell\leq \frac{\varepsilon_T^{1/2}}{C'} (\rmu a)^{-1/2},\quad
    \varepsilon_N\leq 1/2
  \end{equation}
 This, in particular, implies that
\begin{equation}\label{eq:appgapestimate}
  (1-\varepsilon_N)\sum_{i=1}^N\frac12\varepsilon_T(1+\pi^{-2})^{-1}(\ell d)^{-2}Q_{B,i}\geq \frac12 C'^2(1+\pi^{-2})^{-1}\rmu a n_+.
\end{equation}
We shall make a choice for $C'$ below.

We may of course assume that $n>0$ as the inequality we want to prove
is obviously satisfied if $n=0$ since the operator is $0$ whereas the
lower bound is negative in this case.
Observe that by \eqref{eq:Newlambdacondition}
we may also assume that $\rmu \lambda^3\geq1$ and hence $\rmu |B|\geq 1$. In fact,  $\rmu |B|\gg 1$.

For all 
$\Xi>2$  we can choose an integer
$n'$ in the interval $[\Xi\rmu|B|,(\Xi+1)\rmu|B|)$ and we may write
$n=mn'+n''$ with $m, n',n''$ non-negative integers and
$n''<n'<(\Xi+1)\rmu|B|$.
We will choose $\Xi=3$ below (see estimate \eqref{eq:eBtilden}).

We will get a lower bound on the energy if in
the Hamiltonian we think of dividing the particles in $m$ groups of
$n'$ particles and one group of $n''$ particles ignoring the positive
interaction between the groups. It is not important that the
Hamiltonian is no longer symmetric between the particles as we are not
considering it as an operator on the symmetric subspace, but only
calculating its expectation value in a symmetric state. We arrive at
the conclusion that if we denote by $e_B(n,\rmu)$ the ground state
energy of ${\mathcal H}_B(\rmu)$ restricted to states with $n$
particles in the box $B$, then
\begin{equation}\label{eq:ebnsplit}
  e_B(n,\rmu)\geq me_B(n',\rmu)+e_B(n'',\rmu).
\end{equation}

We will estimate the energy for different groups of particles using Lemmas
\ref{lm:appinteractionestimate} and \ref{lm:appbogolubov}. 
We have that both $n'$ and $n''$ are less than
$(\Xi+1)\rmu|B|\leq 2\Xi\rmu|B|$.  We shall therefore first give a general estimate
on $e_B(\widetilde{n},\rmu)$ assuming that $\widetilde{n}\leq 2\Xi\rmu|B|$.
Notice that in this case the last terms in
\eqref{eq:SmallsimpleQs} and \eqref{eq:refinedaprioribog}
can be absorbed in the positive term from
\eqref{eq:appgapestimate} if we make the choice $C'\geq C\Xi^{1/2}$.
Using \eqref{eq:iintw1} we see that the same is also true for the
errors we get by replacing $\widetilde{n}_0=\widetilde{n}-n_+$ by $\widetilde{n}$
everywhere in $A_0$ in \eqref{eq:A0}.
Notice that by \eqref{eq:small_and_large},
\eqref{eq:newlambdadef}, \eqref{con:RNew},
\eqref{con:eTdK}, and \eqref{con:sdKellKB}
we have
\begin{align}
  \frac{Ra}{(dl)^2}&\leq \frac{Ra}{\lambda^2}\leq  (dK_\ell)^{-1} \sqrt{\rmu a^3}\leq K_B\sqrt{\rmu a^3},
  \label{eq:Ralambda}\\
  a(dsl)^{-3}\rmu&= (dsK_\ell)^{-3}\rmu^2 a\sqrt{\rmu a^3}\leq K_B^3\rmu^2 a\sqrt{\rmu a^3},\label{eq:adsl-3}\\
  a^4(ds\ell)^3\rmu^4&=(dsK_\ell)^3\rmu^2 a\sqrt{\rmu a^3}\leq \rmu^2 a\sqrt{\rmu a^3}\label{eq:a4dsl3}.
\end{align}
Using also \eqref{eq:appdellsassumption} we arrive at
\begin{align}
  e_B(\widetilde{n}, \rmu)\geq& \frac{\widetilde{n}^2}{2|B|^2}\iint w_{2,B}(x,y)\,dx dy 
                                - \left(\rmu \frac{\widetilde{n}}{|B|}+\frac14\left(\rmu
                                -\frac{\widetilde{n}}{|B|}\right)^2\right)\iint w_{1,B}(x,y)\,dx dy\nonumber \\
                              &-\widehat{g\omega}(0)\frac{\widetilde{n}^2}{2|B|^2}\int\chi_B^2\\&
                                -C(\Xi^4 +\Xi^2K_B+\Xi K_B^3) \rmu^2 a\sqrt{\rmu a^3}\int\chi_B^2
                                -C\Xi\rmu a |B|^{-1}\int\chi_B^2.\nonumber
\end{align}
The last term is due to the error in replacing $\widetilde{n}\pm1$ by $\widetilde{n}$ in several terms.
If we now use \eqref{eq:w>g} and \eqref{eq:Ralambda} we obtain
\begin{align}
  e_B(\widetilde{n}, \rmu)\geq& \left(\frac14\left(\frac{\widetilde{n}}{|B|}-\rmu
                                \right)^2-\frac12 \rmu^2\right)
                                \iint w_{1,B}(x,y)\,dx dy\nonumber \\
                              &-C(\Xi^4 +\Xi^2K_B+\Xi K_B^3) \rmu^2 a\sqrt{\rmu a^3}\int\chi_B^2
                                -C\Xi\rmu a |B|^{-1}\int\chi_B^2.\label{eq:eBtilden}
\end{align}
In the case of the $m$ groups of $n'$ particles, i.e., for $\widetilde{n}=n'$  it follows using $n'\geq \Xi\rmu |B|$ that if we choose $\Xi=3$ then
\begin{equation}
  e_B(n', \rmu)\geq \frac12\rmu^2\iint w_{1,B}(x,y)\,dx dy\nonumber -CK_B^3\rmu^2 a\sqrt{\rmu a^3}\int\chi_B^2-C\rmu a |B|^{-1}\int\chi_B^2.
\end{equation}
Since $R/\lambda_1\leq R/\lambda\leq 10R/(d\ell)\ll 1$ by
\eqref{eq:newlambdadef} and \eqref{con:RNew} we may if $\rmu a^3$ is
small enough assume from \eqref{eq:w1>} that
$$
\iint w_{1,B}(x,y)\,dx dy\geq \frac12 a \int\chi_B^2.
$$
It then follows from \eqref{eq:Newlambdacondition}  (which implies $|B|^{-1}\ll \rmu$)
and \eqref{con:KBellKM} (which implies $K_B\ll (\rmu a^3)^{1/6}$ )
that if $\rmu a^3$ is small enough then
$$
e_B(n', \rmu)\geq0.
$$
The estimate in the theorem follows if we now
use \eqref{eq:ebnsplit} together with \eqref{eq:eBtilden} (with the choice $\Xi=3$)
for $\widetilde{n}=n''$ to estimate $e_B(n'', \rmu)$ from below.
\end{proof}

We will now apply the small box estimate from the previous theorem to
get an a priori bound on the energy and on the number of particles $n$
and excited particles $n_+$ in the large box.

\begin{theorem}[A priori bound on $n_+$] \label{thm:aprioriHLambda}
  Assume that the parameters $\varepsilon_T, \varepsilon_N, s, d$,
  $K_\ell$, and $K_B$ satisfy
  Assumption~\ref{assump:params} and that $R$ satisfies
  \eqref{con:RNew}. Then there is a constant $C>0$ such that if
  $\rmu a^3\leq C^{-1}$ then
\begin{align}\label{eq:aprioriHLambda}
{\mathcal H}_\Lambda(\rmu)\geq -4\pi \rmu^2 a\ell^3(1+CK_B^3(\rmu a^3)^{1/2})
+ \frac{b}{2\ell^2} n_{+}. 
\end{align}
Moreover, if there exists a normalized
$\Psi\in{\mathcal F}_{\rm s}(L^2(\Lambda))$ with $n$ particles in $\Lambda$, such that \eqref{eq:aprioriPsi} holds
then the a priori bound on $n_+$ in \eqref{eq:apriorinn+NEW} holds.
\end{theorem}

\begin{proof}
We use \eqref{eq:appsliding} together with the estimate in
Theorem~\ref{thm:smallbox}.  We will denote by $n(u)$, $n_0(u)$, and
$n_+(u)$ the operators defined in \eqref{eq:smallboxnn0n+}. The
corresponding operators in the large box $\Lambda$ will be denoted $n,n_0$, and $n_+$.
On the set 
$$
{\mathcal I}=\Big\{u\in \Big[-\frac12(1+\frac1d),\frac12(1+\frac1d)\Big]^3\ \Big|\ 
\frac12 \ell(1+d)-2\lambda \leq \|\ell du\|_\infty\leq \frac12 \ell(1+d)-\lambda\Big\}
$$ 
we have by \eqref{eq:appsliding} that 
$\rmu$ is replaced by $8\rmu$. On this set we have according to (C.6) in \cite{FS} that 
$|\chi_{B(u)}(x)|\leq C(\lambda/\ell)^M$ with ($C$ depending on $M$) and therefore
\begin{equation}\label{eq:Iintegral}
  \int_{\mathcal I} 
  \int\chi_{B(u)}(x)^2dx\, du\leq C(\lambda/\ell)^{2M}(\ell d)^3d^{-3}\leq C(\lambda/\ell)^{2M}\ell^3.
\end{equation}
If we use Theorem~\ref{thm:smallbox}, \eqref{eq:w1>}, and the assumption
\eqref{con:KBellKM}
to get the
the rough estimate
$$
{\mathcal H}_B(8\rmu)\geq -C\rmu^2 a\int\chi_B^2-C\rmu a
$$
we obtain
\begin{equation}
  \int_{\mathcal I} 
  {\mathcal H}_B(8\rmu)\geq -C\rmu^2 a(\lambda/\ell)^{2M}\ell^3 -C\rmu ad^{-3}
                                                \geq -CK_B^3\rmu^2 a\ell^3\sqrt{\rmu a^3},
  \label{eq:badset}
\end{equation}
where we have used \eqref{eq:Iintegral} and 
that \eqref{eq:newlambdadef}, \eqref{eq:small_and_large} and \eqref{eq:dM-small} give
$$
(\lambda/\ell)^{2M} \leq d^{2M}\leq K_B^3\sqrt{\rmu a^3}
$$
and by \eqref{con:sdKellKB}
$$
\rmu a d^{-3}= (K_\ell d)^{-3}\rmu^2 a\ell^3 \sqrt{\rmu a^3}\leq K_B^3\rmu^2 a\ell^3 \sqrt{\rmu a^3}  
$$
Using again this last estimate together with \eqref{eq:appsliding},
\eqref{eq:appw1sliding}, \eqref{eq:chiBintegral}, \eqref{eq:badset},
and the estimate in Theorem~\ref{thm:smallbox} we arrive at
\eqref{eq:aprioriHLambda}. The a priori bound on $n_+$ follows immediately from this since
$$
\ell^2 K_B^3\rmu^2 a\ell^3 \sqrt{\rmu a^3}  =K_\ell^2K_B^3\rmu\ell^3 \sqrt{\rmu a^3}.
$$
\end{proof}

\subsection{Proof of the a priori bounds
on $n$  and ${\mathcal{Q}}_4^{\rm ren}$}
\label{sec:NewAppB}
The purpose of Theorem~\ref{thm:aprioriHLambda} was to prove an
a priori bound on the number of excited particles $n_+$ in the large
box $\Lambda$. We did this by localizing into smaller boxes $B$ where
errors involving the number of excited particles in the small boxes
could be absorbed into the Neumann gap. Our goal is now to get the
a priori bounds on the total number of particles $n$ and the positive
term ${\mathcal Q}_4^{\rm ren}$ in the large box. We no longer have a
sufficiently large gap at our disposal, but we now have an appropriate
a priori bound on $n_+$

We prove the a priori bounds by giving a lower bound on the Hamiltonian
${\mathcal H}_\Lambda(\rmu)$ defined in \eqref{eq:Def_HB_new}. The approach is
essentially identical to what we did in the small boxes adjusted to
the easier case of the large box.

Applying the decomposition of the potential energy in
Lemma~\ref{lm:potsplit} we arrive at the following lemma by, in
particular, applying a Cauchy-Schwarz inequality to absorb
${\mathcal Q}_3^{\rm ren}$ in the positive
${\mathcal Q}_4^{\rm ren}$-term.  The following lemma is essentially
\cite[Lemma 6.3]{BFS} or \cite[Lemma B.2]{FS}. The only difference is
that we keep a part of $ {\mathcal Q}_4^{\rm ren}$.

\begin{lemma}\label{lm:appinteractionestimate_a} 
There is a constant $C>0$, depending only on the localization function
$\chi$, such that if $\rmu$ satisfies  \eqref{con:RNew},
then
	\begin{align}\label{eq:SmallsimpleQs_a}
	-\rmu \sum_{i=1}^N \int w_{1}(x,y)\,dy+
	\frac{1}{2} \sum_{i\neq j}  w(x_i, x_j)
	\geq A_0+A_2 + \frac{1}{2} {\mathcal Q}_4^{\rm ren} -Ca (\rmu +n_0 \ell^{-3})n_+  
	\end{align}
	where
	\begin{align}
	A_0={}&\frac{n_0(n_0-1)}{2\ell^3}\big(\widehat{g}(0) + \widehat{g\omega}(0)\big)
	- \left(\rmu \frac{n_0}{\ell^3}+\frac14\left(\rmu
	-\frac{n_0-1}{\ell^3}\right)^2\right)\ell^3\widehat{g}(0)
	\label{eq:A0_a}
	\end{align}
	and 
	\begin{equation}\label{eq:defA_2_a}
		A_2= \frac{1}{2}\sum_{i\neq j} P_{i} P_{j} w_{1}(x_i,x_j) Q_{j} Q_{i} + h.c.
	\end{equation} 
\end{lemma}
\begin{proof}
We use the identity \eqref{eq:potsplit} and note that, since $P$ is the projection onto constant functions in the box,
\begin{align}\label{eq:Q_0^ren equation}
	{\mathcal Q}_{0}^{\rm ren}&=\frac{n_0(n_0-1)}{2\ell^{6}}\iint w_{2}(x,y)\,dx dy - \rmu \frac{n_0}{\ell^{3}}\iint w_{1}(x,y)\,dx dy \nonumber \\
	&=\frac{n_0(n_0-1)}{2\ell^{3}} \big(\widehat{g}(0) + \widehat{g\omega}(0)\big)- 
	\rmu n_0 \widehat{g}(0),
\end{align}
where we used \eqref{eq:DefU} and \eqref{eq:w2int} to get the last identity.

We will show below that
\begin{align}\label{eq:Absorb}
{\mathcal Q}_{1}^{\rm ren}+{\mathcal Q}_{3}^{\rm ren}+
\frac{1}{2}{\mathcal Q}_{4}^{\rm ren} \geq
&\,-\frac14\left(\rmu-\frac{n_0-1}{\ell^{3}}\right)^2\iint
w_{1}(x,y)\,dx dy
-Ca(\rmu+ n_0\ell^{-3})n_+,
\end{align}
and that 
\begin{equation}\label{eq:Q2>A2+}
{\mathcal Q}_{2}^{\rm ren}\geq A_2-Ca(\rmu+ n_0\ell^{-3})n_+.
\end{equation}
Combining \eqref{eq:Q_0^ren equation}, \eqref{eq:Absorb}, and \eqref{eq:Q2>A2+} and again using \eqref{eq:DefU} we easily get
\begin{align}
{\mathcal Q}_{0}^{\rm ren}+{\mathcal Q}_{2}^{\rm ren}+{\mathcal Q}_{1}^{\rm ren}+{\mathcal Q}_{3}^{\rm ren}
+ \frac{1}{2} {\mathcal Q}_{4}^{\rm ren}\geq A_0+A_2-Ca (\rmu +n_0 \ell^{-3})n_+,
\end{align}
which inserted into the identity \eqref{eq:potsplit} implies the result in the lemma.

To prove \eqref{eq:Absorb} and \eqref{eq:Q2>A2+} we observe using Lemma~\ref{lem:Simple} that,
$$
0\leq \sum_{i,j}P_iQ_jw_{1}(x_i,x_j)Q_jP_i=n_0\ell^{-3}\sum_j Q_j\chi_{\Lambda}(x_j)W_1*\chi_{\Lambda}(x_j)Q_j
\leq Cn_0 n_+\ell^{-3}a\|\chi_{\Lambda}\|_\infty^2
$$ or more generally using again Cauchy-Schwarz inequalities,
Lemma~\ref{lem:Simple}, $0\leq \omega\leq 1$ (see
\eqref{omegabounds}), and $\|\chi_\Lambda\|_\infty\leq C$, we have for
all $k\in {\mathbb N} \cup \{ 0 \}$
\begin{align}
0\leq \sum_{i,j}P_iQ_j(w_{1}\omega^k)(x_i,x_j)Q_jP_i\leq \,& Cn_0\ell^{-3}an_+,\label{eq:Q'app1}
\end{align}
\begin{align}
 \pm\Bigl(\sum_{i,j}P_iQ_j(w_{1}\omega^k)(x_i,x_j)P_jQ_i+h.c.\Bigr)\leq&\,
 2\sum_{i,j}P_iQ_j(w_{1}\omega^k)(x_i,x_j)Q_jP_i\nonumber\\
\leq &\, Cn_0\ell^{-3}an_+,\label{eq:Q'app2}\\
\pm\Bigl(\sum_{i,j}Q_iP_j(w_1\omega^k)(x_i,x_j)P_jP_i+h.c.\Bigr)\leq &\,
\sum_{i,j}Q_iP_j(w_1\omega^k)(x_i,x_j)P_jQ_i\nonumber\\
&\,+\sum_{i,j}P_iP_j(w_1\omega^k)(x_i,x_j)P_jP_i
\nonumber\\ \leq &\,C
n_0a\ell^{-3}\Bigl(n_+ +n_0\Bigr),
\label{eq:Q'app}
\end{align}
where we have abbreviated $(w_1\omega^k)(x_1,x_2)=w_1(x_1,x_2)\omega(x_1-x_2)^k$. We have
	\begin{align}
	\sum_{i,j}P_iQ_jw_1(x_i,x_j)Q_jQ_i=&\, \sum_{i,j} \Big(P_i
	Q_j w_1(x_i,x_j) \Big[ Q_j Q_i + \omega(x_i-x_j) (P_j P_i + P_j Q_i
	+ Q_j P_i)\Big]  \Big)\nonumber \\ &\,- \sum_{i,j}
	\Big(P_i Q_j w_1(x_i,x_j) \omega(x_i-x_j) (P_j P_i +
	P_j Q_i + Q_j P_i)\Big) \label{eq:3Qto4Q}
	\end{align}
	and the same identity for the Hermitian conjugates.	
        We estimate the first term in \eqref{eq:3Qto4Q} (and its Hermitian conjugate) using a Cauchy-Schwarz inequality
	\begin{align}
	\pm \sum_{i,j} \Big(P_i
	&Q_j w_1(x_i,x_j) \Big[ Q_j Q_i + \omega(x_i-x_j) (P_j P_i + P_j Q_i
	+ Q_j P_i)\Big]+h.c.  \Big) \nonumber \\
	&\,\leq \frac{1}{2}
	{\mathcal Q}_4^{\rm ren}+C\sum_{i\neq j} P_i Q_j
	w(x_i,x_j)(1-\omega(x_i-x_j))^2Q_j P_i \nonumber \\
	&\,\leq \frac{1}{2}
	{\mathcal Q}_4^{\rm ren}+ Cn_0\ell^{-3}an_+,
	\end{align}
	where we have used that $w(x_i,x_j)(1-\omega(x_i-x_j))^2=w_1(x_i,x_j)(1-\omega(x_i-x_j))$ and 
        then \eqref{eq:Q'app1} in the last inequality.
	We estimate the second term in \eqref{eq:3Qto4Q} (and its Hermitian conjugate)
        also using \eqref{eq:Q'app1} and \eqref{eq:Q'app2}  arriving at
        \begin{align}\label{eq:Q3+Q4CSa} 
	{\mathcal Q}_3^{\rm ren} +\frac{1}{2}
	{\mathcal Q}_4^{\rm ren}\geq - \sum_{i\neq j} \Big(
	P_{i} Q_{j} w_{1} \omega(x_i,x_j) P_{j} P_{i} + h.c.\Big)-Can_0 \ell^{-3}n_+.
	\end{align}
	Notice that if we rewrite ${\mathcal Q}_{1}^{\rm ren}$ as in \eqref{eq:Q1n0}
	then the first term on the right side of \eqref{eq:Q3+Q4CSa} cancels the second line of \eqref{eq:Q1n0}.	
Using the exact commutation relation between $n_{+}$ and the operator in question, we estimate the 
remaining part of ${\mathcal Q}_{1}^{\rm ren}$
	\begin{align}
	&\ell^{-3}(n_0-\rmu \ell^{3})\sum_{i} Q_{i} \chi_{\Lambda}(x_i) W_1*\chi_{\Lambda}(x_i)
	P_{i} +
	h.c.\nonumber\\ 
	&=\ell^{-3}({n_0}^{1/2}-(\rmu \ell^{3})^{1/2})\sum_{i}
	Q_{i} \chi_{\Lambda}(x_i) W_1*\chi_{\Lambda}(x_i)
	P_{i}([n_0-1]_+^{1/2}-(\rmu \ell^{3})^{1/2}) +
	h.c.\nonumber\\ 
	&\geq -4 \ell^{-3}\left(n_0^{1/2}+(\rmu \ell^{3})^{1/2}\right)^2\sum_{i}
	Q_{i} \chi_{\Lambda}(x_i) W_1*\chi_{\Lambda}(x_i) Q_{i}\nonumber \\
	&\quad -\frac{1}{4}\ell^{-3}\left([n_0-1]_+^{1/2}-(\rmu \ell^{3})^{1/2}\right)^2\sum_{i}
	P_{i} \chi_{\Lambda}(x_i) W_1*\chi_{\Lambda}(x_i) P_{i}.\label{Q_1 remaining part}
	\end{align}
	The first term above we estimate as
	\begin{align}
	&-4 \ell^{-3} \left(n_0^{1/2}+(\rmu \ell^{3})^{1/2}\right)^2\sum_{i}
	Q_{i} \chi_{\Lambda}(x_i) W_1*\chi_{\Lambda}(x_i) Q_{i}\\
	&\quad \geq -C \ell^{-3}(n_0+\rmu \ell^{3})n_+ a\|\chi_{\Lambda}\|_\infty^2.
	\end{align}
We complete the proof of \eqref{eq:Absorb} by estimating the last term in \eqref{Q_1 remaining part}
	\begin{align}
	&-\frac14\frac{n_0}{\ell^{6}}\left((n_0-1)^{1/2}-(\rmu
          \ell^{3})^{1/2}\right)^2\iint w_{1}(x,y)\,dx
          dy\nonumber\\ &=\,-\frac{1}{4}\frac{n_0}{\ell^{6}}\left((n_0-1)-\rmu
          \ell^{3}\right)^2[(n_0-1)^{1/2}+(\rmu
            \ell^{3})^{1/2}]^{-2}\iint w_{1}(x,y)\,dx
          dy\nonumber\\ &\geq\,-\frac{1}{4}\left(\frac{n_0-1}{\ell^{3}}-\rmu\right)^2\iint
          w_{1}(x,y)\,dx dy,
	\end{align}
	using that $\rmu\ell^3\geq 1$. 
	
	To prove \eqref{eq:Q2>A2+} we recall that $w_2=w_1(1+\omega)\leq 2w_1$ and estimate the
        first two terms in ${\mathcal Q}_{2}^{\rm ren}$ by \eqref{eq:Q'app1} and \eqref{eq:Q'app1}. 
        Finally, the one-body term in ${\mathcal
          Q}_{2}^{\rm ren}$ is estimated as
        \begin{equation}
          \rmu \sum_{i} Q_i \int w_1(x_i,y)\,dy\, Q_i +h.c.\leq Ca\rmu n_+\|\chi_{\Lambda}\|_\infty^2.
        \end{equation}
\end{proof}

We now repeat the simple Bogoliubov argument in
Lemma~\ref{lm:appbogolubov} in our present situation to arrive at the
following lower bound on the kinetic energy $\mathcal{T}$ given in
\eqref{eq:DefT_new}, and appearing in the Hamiltonian
${\mathcal H}_\Lambda(\rmu)$ from \eqref{eq:Def_HB_new}.
\begin{lemma} [Bogoliubov's method in the large box]\label{lm:appbogolubovF} 
  If the parameters $\varepsilon_N,s,d$, and $K_\ell$ satisfy Assumption~\ref{assump:params} then there exists a constant $C>0$ such that
  if $\rmu a^3\leq C^{-1}$ then 
  \begin{align}\label{eq:largeaprioribog}
    \sum_{i=1}^N\mathcal{T}^{(i)} +A_2
  \geq&-\frac{1}{2} \widehat{g\omega}(0) {(n+1)n}\ell^{-3}
  - C a \left(\frac{Ra}{\ell^2}+a(ds\ell)^{-1}\log(ds\ell a^{-1})\right){(n+1)n}\ell^{-3}\nonumber\\
  &-C\left(a^4(ds\ell)^{3}\left((n+1) \ell^{-3}\right)^3 +a(ds\ell)^{-3}\right){n}
  -Ca\ell^{-3}(n+1)n_+.
  \end{align}
\end{lemma}
\begin{proof}
   We have the lower bound 
  $$
  \mathcal{T}\geq (1-\varepsilon_N)\mathcal{T}'\geq
  (1-\varepsilon_N)Q_\Lambda\chi_\Lambda[\sqrt{-\Delta}-(ds\ell)^{-1}]_+^2\chi_\Lambda
  Q_\Lambda .$$
  At this point we again shift to a second quantized formalism to use a simple
  version of Bogoliubov's method.  Introducing the
  operators
  $$
  \ta_p^\dagger=\ell^{-3/2}a^\dagger(Q \chi_{\Lambda}e^{-ipx})a_0
  $$
  we can write 
  $$
  A_2=\frac12 (2\pi)^{-3}\int \widehat{W_1}(p)(\ta_p^\dagger \ta_{-p}^\dagger+\ta_{-p}\ta_p)dp.
  $$
  We shall control $A_2$ using Bogoliubov's method. In order to do this we will add and subtract a term 
  \begin{equation}\label{eq:A1}
    A_1= (2\pi)^{-3}K_\Lambda a \int (\ta_p^\dagger \ta_{p}+\ta^\dagger_{-p}\ta_{-p})dp,
  \end{equation}
  with the constant $K_\Lambda>0$ chosen appropriately. 
  Note that we have 
  \begin{equation}\label{eq:A1estimate}
    A_1\leq K_\Lambda a \frac{n_0+1}{\ell^3} n_+\|\chi_{\Lambda}\|_\infty^2\leq
    CK_\Lambda a \frac{n_0+1}{\ell^3} n_+.
  \end{equation}
  Using $n_0\leq n$ we may write
  $$
  \sum_{i=1}^N\mathcal{T}^{(i)}+A_1+A_2\geq
  (2\pi)^{-3}\frac12\int h(p) dp,
  $$
  where $h$ is the operator
  $$
  h(p)=\left((1-\varepsilon_N)\frac{\ell^3}{n+1}[|p|-(ds\ell)^{-1}]_+^2
    +2K_\Lambda a \right)(\ta_p^\dagger \ta_{p}+\ta^\dagger_{-p}\ta_{-p})+
  \widehat{W_1}(p)(\ta_p^\dagger \ta_{-p}^\dagger+\ta_{-p}\ta_p).
  $$
  We observe that 
  $$
  [\ta_p,\ta_p^\dagger]\leq n_0 \ell^{-3}\int \chi_{\Lambda} ^2\leq n.
  $$
  We will now apply the simple case of Bogoliubov's method in Theorem~\ref{thm:bogolubov-complete} with 
  $$
  {\mathcal A}(p)= (1-\varepsilon_N)\frac{\ell^3}{n+1}[|p|-(ds\ell)^{-1}]_+^2+2K_\Lambda a,\quad
  {\mathcal B}(p)= \widehat{W_1}(p).
  $$
  We have by \eqref{eq:W1-g-new} that 
  $$
  |{\mathcal B}(p)|= |\widehat{W_1}(p)|\leq \int W_1\leq C_0 a .
  $$ 
  If we therefore choose $K_\Lambda\geq C_0$ we see that $|{\mathcal B}|/{\mathcal A}\leq 1/2$, and 
  we get the following lower bound from Theorem~\ref{thm:bogolubov-complete}. 
  $$
  h(p)\geq -\left({\mathcal A}(p)-\sqrt{{\mathcal A}(p)^2-{\mathcal B}(p)^2}\right)n.
  $$
  The final estimate in \eqref{eq:largeaprioribog} now follows from
  Lemma~\ref{lm:aprioriintegral_new} if we use
  \eqref{eq:largeaprioribog} and \eqref{eq:A1estimate}. We note as in the proof of
  Lemma~\ref{lm:appbogolubov} that
  $\varepsilon_N\leq a(d s \ell)^{-1}\leq 1/2$ by
  \eqref{eq:small_and_large}, \eqref{con:eTdK}, \eqref{con:sdKellKB},
  \eqref{con:KBKell}, and \eqref{eq:Cond_epsilonN} with $K_N\gg 1$.
\end{proof}

We are now ready to prove the a priori bounds on $n$ and ${\mathcal Q}^{\rm ren}_4$
from Theorem~\ref{thm:aprioribounds}. The proof is very similar to the proof
Theorem~\ref{thm:smallbox} just adapted to the large box.

\begin{theorem}[A priori bounds on $n$ and ${\mathcal Q}^{\rm ren}_4$]
\label{thm:aprioriHLambdanQ}
Assume that the parameters $\varepsilon_T, s, d$, $K_\ell$, $K_N$ and $K_B$ satisfy
Assumption~\ref{assump:params} and that $R$ satisfies
\eqref{con:RNew}. Then there is a constant $C>0$ such that if
$\rmu a^3\leq C^{-1}$ and if there exists a normalized $n$-particle
$\Psi\in{\mathcal F}_{\rm s}(L^2(\Lambda))$ such that
\eqref{eq:aprioriPsi} holds then the a priori bounds on $n$ and
${\mathcal Q}^{\rm ren}_4$ in \eqref{eq:apriori_n} and
\eqref{eq:Q4apriori_2} hold.
\end{theorem}

\begin{proof}
  Observe that according to \eqref{con:KBellKM} we may assume from
  \eqref{eq:aprioriPsi} that $ \langle\Psi,{\mathcal H}_\Lambda(\rmu)\Psi\rangle<0$. Hence we must have
  $n>0$. Since $\rmu \ell^3=K_\ell^3 (\rmu a^3)^{-1/2}$ we may assume this number to be bigger than 1.
  We proceed as in Theorem~\ref{thm:smallbox}. 
  For all $\Xi>2$  we can choose an integer
  $n'$ in the interval $[\Xi\rmu\ell^3,(\Xi+1)\rmu\ell^3)$ and we may write
  $n=mn'+n''$ with $m, n',n''$ non-negative integers and
  $n''<n'<(\Xi+1)\rmu\ell^3$.
  We will choose $\Xi=3$ below.

  We are aiming at giving a lower bound for
  $\langle\Psi,{\mathcal H}_\Lambda(\rmu)\Psi\rangle$. From Theorem~\ref{thm:aprioriHLambda} we know that \eqref{eq:apriorinn+NEW}
  holds for the expectation value of $n_+$ in $\Psi$.

  As in the proof of Theorem~\ref{thm:smallbox} we divide the
  particles into  $m$ groups of $n'$ particles and one group of $n''$ particles. For $\widetilde{n}\leq n$ we denote by
  $\Gamma_{\widetilde{n}}^\Psi$ the $\widetilde{n}$-particle reduced state of $\Psi$.
  Ignoring the positive interaction between the particles in different groups we have
  \begin{equation}\label{eq:Psinsplit}
    \langle\Psi,{\mathcal H}_\Lambda(\rmu)\Psi\rangle\geq m\tr\left[\Gamma_{n'}^\Psi{\mathcal H}_\Lambda(\rmu)\right]+
    \tr\left[\Gamma_{n''}^\Psi{\mathcal H}_\Lambda(\rmu)\right].
  \end{equation}
  Observe also that for any $\widetilde{n}\leq n$ we have
  \begin{equation}\label{eq:n+reduced}
    \langle\Psi,n_+\Psi\rangle=\langle\Psi\sum_iQ_i\Psi\rangle\geq \tr\left[\Gamma_{\widetilde{n}}^{\Psi} \sum_iQ_i\right]=\tr\left[\Gamma_{\widetilde{n}}^{\Psi} n_+\right].
  \end{equation}
  We have that both $n'$ and $n''$ are less than
  $(\Xi+1)\rmu|B|\leq 2\Xi\rmu|B|$.  To treat both cases simultaneously we shall therefore first give a
  general estimate on $\tr\left[\Gamma_{\widetilde{n}}^\Psi{\mathcal H}_\Lambda(\rmu)\right]$ assuming that
  $\widetilde{n}\leq 2\Xi\rmu|B|$ 
  We estimate the energy using
  Lemmas~\ref{lm:appinteractionestimate_a} (ignoring for now the positive ${\mathcal Q}_4^{\rm ren}$-term) and \ref{lm:appbogolubovF}.
  Observe that using \eqref{eq:apriorinn+NEW} and \eqref{eq:n+reduced}
  we see that the last terms in \eqref{eq:SmallsimpleQs_a} and \eqref{eq:largeaprioribog} can be estimated 
  by
  $$
  C \Xi K_B^3K_\ell^2\rmu^2 a \ell^3\sqrt{\rmu a^3}.
  $$
  We get the same error in replacing $n_0$ by $\widetilde{n}$ everywhere in \eqref{eq:SmallsimpleQs_a}. The error we get by replacing
  $\widetilde{n}\pm 1$ by $\widetilde{n}$ in several terms is bounded by $C\Xi\rmu a$ which is smaller than the above error. Hence using
  \eqref{eq:ScatLengthBasic} we arrive at
  \begin{align}
    \tr\left[\Gamma_{\widetilde{n}}^\Psi{\mathcal H}_\Lambda(\rmu)\right]\geq& -4\pi\rmu^2 a\ell^3 +2\pi(\widetilde{n}\ell^{-3}-\rmu)^2a\ell^3\nonumber\\
                                &-C(\Xi^4 +\Xi^2K_B+\Xi K_B^3K_\ell^2) \rmu^2 a\ell^3\sqrt{\rmu a^3}+\frac12\tr\left[\Gamma_{\widetilde{n}}^\Psi{\mathcal Q}_4^{\rm ren}\right],
  \end{align}
  where we have used \eqref{eq:appdellsassumption}, \eqref{eq:Ralambda} (and $d>1$), \eqref{eq:adsl-3}, and \eqref{eq:a4dsl3}.
  Hence with $\Xi=3$ we see from \eqref{eq:Psinsplit} and \eqref{con:KBKell} that for $\rmu a^3$ small enough 
  \begin{align*}
  \langle\Psi,{\mathcal H}_\Lambda(\rmu)\Psi\rangle\geq&-4\pi\rmu^2 a+ 3m\pi\rmu^2 a\ell^3 +2\pi(n'' \ell^{-3}-\rmu)^2a\ell^3 \nonumber\\
                                &-CK_B^3K_\ell^2\rmu^2 a\ell^3\sqrt{\rmu a^3}+\frac12\tr\left[\Gamma_{{n''}}^\Psi{\mathcal Q}_4^{\rm ren}\right], 
  \end{align*}
  where we used that ${\mathcal Q}_4^{\rm ren}\geq0$. We see from \eqref{eq:aprioriPsi} that we must have $m=0$ and hence $n=n''$ and
  $\Gamma_{{n''}}^\Psi=|\Psi\rangle\langle\Psi|$ and that the a priori bounds \eqref{eq:apriori_n} and \eqref{eq:Q4apriori_2} must hold. 
\end{proof}

\section{Proofs}\label{sec:proofs}
In this section we give the proofs of Lemma~\ref{lem:pseudolocal} and Lemma~\ref{lem:LowMomentaOperator}.

\subsection{Proof of Lemma~\ref{lem:pseudolocal}}
\begin{proof}[Proof of Lemma~\ref{lem:pseudolocal}]
To have a unified proof, we let $\widetilde{\chi}_{\Lambda}$ denote either the function $ \chi_{\Lambda}$ or $\one_{\Lambda}$.
We easily get, using Cauchy-Schwarz and the definition \eqref{eq:Def_QL'} of $Q_L'$,
\begin{align}
Q \widetilde{\chi}_{\Lambda} \one_{\{|p| \leq K_H'' \ell^{-1}\}} \widetilde{\chi}_{\Lambda}Q \leq 2 \| \widetilde{\chi}_{\Lambda}\|_{\infty}^2 Q_L'  + 2 Q_H'  \widetilde{\chi}_{\Lambda}\one_{\{|p| \leq K_H'' \ell^{-1}\}} \widetilde{\chi}_{\Lambda} Q_H' .
\end{align}
So it suffices to prove that
\begin{align}\label{eq:LongCommutator}
\left\| Q_H'  \widetilde{\chi}_{\Lambda} \one_{\{|p| \leq K_H'' \ell^{-1}\}}\right\| \leq C  \left( \frac{K_H''}{K_H'} \right)^{\frac{M}{2}} + C \varepsilon_N^{\frac{3}{4}} .
\end{align}
Where norm denotes the operator norm from $L^2({\mathbb R}^3)$ to $L^2(\Lambda)$.
By scaling, we may assume that $\ell = 1$, so in the rest of the proof $\Lambda$ denotes the unit box, i.e. $\Lambda = (-\frac{1}{2}, \frac{1}{2})^3$.
We will denote by $\widetilde{\chi}$ the scaled $\widetilde{\chi}_{\Lambda}$.
We will in this proof slightly abuse notation and denote by $\theta$ the indicator function of the {\it open} unit box, i.e. $\theta(x) = \one_{(-\frac{1}{2}, \frac{1}{2})^3}(x)$.

From the definition of $Q_H'$ and the spectral theorem, we have the bound
\begin{align}
\label{eq:spectral}
\| {\mathcal T}^{-s} Q_H' \| \leq (K_H')^{-2s},
\end{align}
for all $s>0$.

We rewrite
\begin{align}
{\mathcal T} &= (1-\varepsilon_N) Q \left( A_1 + A_2 + A_3 + A_4 + A_5 \right) Q \nonumber \\
&= (1-\varepsilon_N)  \left( A_5 + A'\right)  ,
\end{align}
with $A'$ defined by the above identity, and 
\begin{align}
A_1&:= (\frac{1}{2}+b) \varepsilon_T d^{-2} + b, \nonumber \\
A_2 &:=- \frac{1}{2} \varepsilon_T d^{-4} (-\Delta^{{\mathcal N}}+d^{-2})^{-1}, \nonumber \\
A_3&:= - b \varepsilon_T d^{-2} \theta \one_{[0,d^{-2}]}(|p|) \theta, \nonumber \\
A_4&:= \chi \tau(p) \chi, \nonumber \\
A_5&:= \varepsilon_N' (-\Delta^{{\mathcal N}}),
\end{align}
where we use the notation, $\varepsilon_N':=\frac{\varepsilon_N}{1-\varepsilon_N}$ and $\tau$ from \eqref{eq:Def_tau} (with $\ell=1$).
Let us note in passing that the self-adjoint operator ${\mathcal T}$ has domain given by the domain of the Neumann-Laplace operator.
Let $\varphi \in L^2(\Lambda), \psi \in L^2({\mathbb R}^3)$ be normalized (where we keep the notation $\Lambda$ for the scaled (i.e. unit)  box $(-\frac{1}{2}, \frac{1}{2})^3$). We will prove below that as long as $2s+2 \leq \frac{M}{2}$, we have that
\begin{align}\label{eq:InducRegul}
\psi_s := (A')^s \widetilde{\chi} \one_{\{|p| \leq K_H'' \}}  \psi \in H^2(\Lambda),
\end{align}
and, for all $\alpha \in {\mathbb N}_0^3$, with $|\alpha|\leq 2$,
\begin{align}\label{eq:InducBound}
\| \partial^{\alpha} \psi_s \|_{L^2(\Lambda)} \leq C_s (K_H'')^{2s+ |\alpha|}.
\end{align}
We start by showing how Lemma~\ref{lem:pseudolocal} follows using \eqref{eq:InducRegul} and \eqref{eq:InducBound}.

We now use that ${\mathcal D}({\mathcal T}) = {\mathcal D}(\Delta^{{\mathcal N}})= {\mathcal D}(A_5)$, hence 
$$
Q_H' =  {\mathcal T} {\mathcal T}^{-1} Q_H' = (1-\varepsilon_N) A_5{\mathcal T}^{-1} Q_H'  + (1-\varepsilon_N) A' {\mathcal T}^{-1} Q_H' .
$$ 
If we use this identity $N = M/2 \in {\mathbb N}$ times, and that $A'$ is symmetric on $H^2(\Lambda) \supseteq {\mathcal D}(\Delta^{{\mathcal N}})$, we get
\begin{align}
&\langle Q_H' \varphi,  \widetilde{\chi}  \one_{\{|p| \leq K_H'' \}} \psi \rangle \nonumber \\
&=
(1 - \varepsilon_N) \langle A_5 {\mathcal T}^{-1} Q_H' \varphi,  \psi_0 \rangle +
(1 - \varepsilon_N)\langle  {\mathcal T}^{-1} Q_H' \varphi, A'  \psi_0 \rangle\nonumber \\
&=(1 - \varepsilon_N)\langle A_5 {\mathcal T}^{-1} Q_H' \varphi,  \psi_0 \rangle +
(1 - \varepsilon_N)^2\langle  A_5 {\mathcal T}^{-2} Q_H' \varphi, \psi_1 \rangle + 
(1 - \varepsilon_N)^2\langle   {\mathcal T}^{-2} Q_H' \varphi, \psi_2 \rangle \nonumber \\
&= \sum_{j=0}^{N-1} (1 - \varepsilon_N)^j \langle A_5 {\mathcal T}^{-j-1} Q_H' \varphi, \psi_j \rangle
+ (1 - \varepsilon_N)^N \langle {\mathcal T}^{-N} Q_H' \varphi, \psi_{N} \rangle.
\end{align}

Here the powers of $(1 - \varepsilon_N)$ are clearly unimportant, since $\varepsilon_N \leq 1/2$.

Using \eqref{eq:InducBound} and 
\eqref{eq:spectral} we immediately get
\begin{align}
\left|  \langle {\mathcal T}^{-N} Q_H' \varphi, \psi_{N} \rangle \right| \leq C (K_H')^{-2N} (K_H'')^{2N},
\end{align}
in agreement with \eqref{eq:LongCommutator}.

For the other terms, we will prove that
\begin{align}\label{eq:A5terms}
| \langle A_5 {\mathcal T}^{-j-1} Q_H' \varphi, \psi_j \rangle | \leq C \varepsilon_N^{\frac{3}{4}} \left( \frac{K_H''}{K_H'} \right)^{2j+2},
\end{align}
and \eqref{eq:LongCommutator} (and therefore the lemma) will follow.

We integrate by parts, using \eqref{eq:InducRegul}, to get
\begin{align}
\langle A_5 {\mathcal T}^{-j-1} Q_H' \varphi, \psi_j \rangle
=
\varepsilon_N' \langle {\mathcal T}^{-j-1} Q_H' \varphi, (-\Delta) \psi_j \rangle
+\varepsilon_N' \int_{\partial \Lambda} \overline{{\mathcal T}^{-j-1} Q_H' \varphi(x)} \partial_n \psi_j \,d\sigma(x),
\end{align}
where $\partial_n$ denotes the normal derivative, $\sigma$ is the surface measure and where we used that ${\mathcal T}^{-1}$ maps into the domain of the Neumann Laplace operator to notice that one of the boundary terms vanishes.
Notice in passing that for $j=0$ the other boundary term also vanishes and this could be used to marginally improve our conclusion, but we will not use this fact.

Therefore, using Cauchy-Schwarz and standard trace theorems,
\begin{align}
&\left| \langle A_5 {\mathcal T}^{-j-1} Q_H' \varphi, \psi_j \rangle \right|\nonumber \\
&\leq
\varepsilon_N' C \left(  \| {\mathcal T}^{-j-1} Q_H' \varphi\| \| \psi_j \|_{H^2(\Lambda)} 
+
\| {\mathcal T}^{-j-1} Q_H' \varphi(x) \|_{H^{1/2}(\Lambda)} \| \psi_j \|_{H^{3/2}(\Lambda)} \right) \nonumber \\
&\leq \varepsilon_N' C \Big(  \| {\mathcal T}^{-j-1} Q_H' \varphi\| \| \psi_j \|_{H^2(\Lambda)}\nonumber \\
&\qquad \qquad+
\| {\mathcal T}^{-j-1} Q_H' \varphi(x) \|^{1/2}
\| \nabla {\mathcal T}^{-j-1} Q_H' \varphi(x) \|^{1/2}
\| \psi_j \|_{H^1(\Lambda)}^{1/2}
\| \psi_j \|_{H^2(\Lambda)}^{1/2}
\Big)  ,
\end{align}
where the last inequality follows from the interpolation inequalities recalled in Lemma~\ref{lem:Interpolation} below, and a Poincar\'{e} inequality using that the range ${\mathcal T}^{-1}$ is orthogonal to constants.

Notice now that
$\| \nabla {\mathcal T}^{-j-1} Q_H' \varphi(x) \|^2 \leq (\varepsilon_N')^{-1} \|  {\mathcal T}^{-j-1/2} Q_H' \varphi(x) \|^2$, since $A_5 \leq {\mathcal T}$.
Inserting this, together with \eqref{eq:InducBound} and the spectral bound 
\eqref{eq:spectral} gives \eqref{eq:A5terms}.

We will now prove the claims  \eqref{eq:InducRegul} and \eqref{eq:InducBound}.
We can write $(A')^s$ as a finite sum of terms of the form
$$
\Big((\frac{1}{2} +b)\varepsilon_T d^{-2} + b\Big)^{s-k} \prod_{j=1}^k Q A_{\sigma(j)} Q,
$$
for some $\sigma: \{1,\ldots,k\} \rightarrow \{2,3,4\}$ and where $k\leq s$.
Notice that the prefactor satisfies $((\frac{1}{2} +b)\varepsilon_T d^{-2} + b)^{s-k} \leq C (K_H'')^{2s-2k}$ by 
\eqref{cond:disjoint}.
So it suffices to prove that
\begin{align}
\left\|  \partial^{\alpha} \prod_{j=1}^k Q A_{\sigma(j)} Q  \widetilde{\chi}  \one_{\{|p| \leq K_H''\}}\right\| \leq C (K_H'')^{2k+|\alpha|},
\end{align}
for all $|\alpha| \leq 2$ and 
where norm denotes the operator norm from $L^2({\mathbb R}^3)$ to $L^2(\Lambda)$.

We will split such a product at each occurrence of a factor of $A_3$. Thereby, it clearly suffices to prove the estimates
\begin{align}\label{eq:ToWork1}
b \varepsilon_T d^{-2} \left\| \partial^{\alpha} \left(\prod_{j=1}^{k} Q A_{\sigma(j)} Q\right)  \theta \one_{[0,d^{-2}]}(|p|)  \right\| \leq C (K_H'')^{2k+2+|\alpha|},
\end{align}
and 
\begin{align}\label{eq:ToWork}
\left\|  \partial^{\alpha} \left(\prod_{j=1}^{k} Q A_{\sigma(j)} Q\right) \widetilde{\chi}  \one_{\{|p| \leq K_H''\}} \right\| \leq C (K_H'')^{2k+|\alpha|},
\end{align}
for all $|\alpha| \leq 2$, $2k+2\leq \frac{M}{2}$ and $\sigma: \{1,\ldots,k\} \rightarrow \{2,4\}$ and
where norms denote the operator norm from $L^2({\mathbb R}^3)$ to $L^2(\Lambda)$.

We will only explicitly consider \eqref{eq:ToWork}, the estimate \eqref{eq:ToWork1} being similar since $d^{-2} \ll K_H''$ by \eqref{cond:disjoint}.
By the interpolation inequality \eqref{eq:H1-interpolation} of Lemma~\ref{lem:Interpolation}, it suffices to consider the cases $|\alpha| \in \{0,2\}$.

Let $f \in C_0^{\frac{M}{2}}({\mathbb R}^3)$ with $f\equiv 0$ outside $\Lambda$.

Consider first, for $s_1, t_1 \in {\mathbb N}$, with $2s_1 \leq \frac{M}{2}$, the operator $(Q\chi \tau \chi Q)^{s_1} (-\Delta^{\mathcal N} + d^{-2})^{-t_1}$. If $s_1\geq t_1$, we have the following identity when acting on $H^{2(s_1-t_1)}({\mathbb R}^3)$,
\begin{align}
&(Q\chi \tau \chi Q)^{s_1} Q(-\Delta^{\mathcal N} + d^{-2})^{-t_1} Qf \nonumber \\
&\quad = \left( (Q\chi \tau \chi Q)^{s_1}  (-\Delta^{\mathcal N} + d^{-2})^{-s_1} Q \right) (-\Delta+ d^{-2})^{s_1-t_1} f.
\end{align}
This uses that $Q$ commutes with $-\Delta^{\mathcal N}$ and that the vanishing on the boundary of the multiplication operator $f$ assures that the boundary conditions are satisfied, so one can replace $\Delta^{\mathcal N}$ by $\Delta$ in the rightmost term. Notice that $(Q\chi \tau \chi Q)^{s_1} (-\Delta^{\mathcal N} + d^{-2})^{-s_1}$ can be bounded as an operator on $L^2(\Lambda)$ using Lemma~\ref{lem:Eigenfunctions_Elliptic}, writing $Q = \one - P$ on $L^2(\Lambda)$ and that $\chi$ has $M$ bounded derivatives.

Iterating this argument, we find with $[t]_{\pm} := \max\{0, \pm t\} \geq 0$,
\begin{align}
&(Q\chi \tau \chi Q)^{s_1} Q(-\Delta^{\mathcal N} + d^{-2})^{-t_1} Q
(Q\chi \tau \chi Q)^{s_2} Q(-\Delta^{\mathcal N} + d^{-2})^{-t_2} Q
f \nonumber \\
&= \left\{ (Q\chi \tau \chi Q)^{s_1} (-\Delta^{\mathcal N} + d^{-2})^{-s_1} Q\right\} (-\Delta^{\mathcal N}+ d^{-2})^{-[s_1-t_1]_{-}} \nonumber \\
&\quad \times \left\{(-\Delta+ d^{-2})^{[s_1-t_1]_{+}} (Q\chi \tau \chi Q)^{s_2}
(-\Delta^{\mathcal N} + d^{-2})^{-s_2-[s_1-t_1]_{+}} Q\right\} \nonumber \\
&\quad 
\times (-\Delta^{\mathcal N}+ d^{-2})^{-\left[s_2-t_2 + [s_1-t_1]_{+}\right]_{-}} \times
(-\Delta+ d^{-2})^{\left[s_2-t_2 + [s_1-t_1]_{+}\right]_{+}} f.
\end{align}
Therefore, for $k'\geq 2$ and $|\alpha| \in \{0,2\}$, consider $s_1, t_{k'} \in {\mathbb N}_{0}$ and $s_2,\ldots, s_{k'}, t_1, \ldots , t_{k'-1} \in {\mathbb N}$ 
with $|\alpha|+ 2\sum_j s_j \leq \frac{M}{2}$,
and define
\begin{align}
\delta_1:= s_1-t_1 + \frac{|\alpha|}{2}, \qquad \delta_j:= \left[ s_j-t_j + [\delta_{j-1}]_{+}\right]_{+}, \quad \text{ for } j \geq 2.
\end{align}
Notice in passing the trivial bound
\begin{align}\label{eq:sums}
\delta_j \leq \frac{|\alpha|}{2} + \sum_{j'=1}^j s_{j'}.
\end{align}
Then, for all $\varphi \in H^{2 [\delta_{k'}]_{+}}$, and with $t_{\rm tot}:= \sum t_j$,
\begin{align}
&\partial^{\alpha} \left( \prod_{j=1}^{k'} (Q\chi \tau \chi Q)^{s_j} Q (\frac{1}{2} \varepsilon_T d^{-4} )^{t_j} (-\Delta^{\mathcal N} + d^{-2})^{-t_j} Q\right) f \varphi \nonumber \\
&=(\frac{1}{2} \varepsilon_T d^{-4} )^{t_{\rm tot}} 
\left(\partial^{\alpha} (Q \chi \tau \chi Q)^{s_1} (-\Delta^{\mathcal N} + d^{-2})^{-s_1 -\frac{|\alpha|}{2}} \right) \nonumber \\
&\quad \times \left( \prod_{j=2}^{k'}  (-\Delta+ d^{-2})^{[\delta_{j-1}]_{+}} (Q\chi \tau \chi Q)^{s_j} (-\Delta^{\mathcal N} + d^{-2})^{-s_j -[\delta_{j-1}]_{+} } 
Q (-\Delta^{\mathcal N} + d^{-2})^{-[\delta_{j}]_{-}}
\right) \nonumber \\
&\quad \times (-\Delta + d^{-2})^{[\delta_{k'}]_{+}}
f \varphi.
\end{align}
In particular, still for all $\varphi \in H^{2 [\delta_{k'}]_{+}}(\Lambda)$,
\begin{align}\label{eq:PulledThrough}
&\left\| \partial^{\alpha} \left( \prod_{j=1}^{k'} (Q \chi \tau \chi Q)^{s_j} Q (\frac{1}{2} \varepsilon_T d^{-4} )^{t_j}(-\Delta^{\mathcal N} + d^{-2})^{-t_j} Q\right) f \varphi 
\right\| \nonumber \\
&\quad \leq C (\frac{1}{2} \varepsilon_T d^{-4} )^{t_{\rm tot}}  \| (-\Delta + d^{-2})^{[\delta_{k'}]_{+}}
f \varphi \|,
\end{align}
since all the other factors are bounded (uniformly in $d\leq 1$) operators on $L^2$.

Now we can prove \eqref{eq:ToWork}. We split the proof in two cases depending on whether $\sigma(k)=2$ or $4$.

If $\sigma(k) = 4$, we clearly have
\begin{align}
A_{\sigma(k)} Q \widetilde{\chi}  \one_{\{|p| \leq K_H''\}} = \chi \tau \chi (1-P)  \widetilde{\chi}  \one_{\{|p| \leq K_H''\}},
\end{align}
where $\tau \chi (1-P)  \widetilde{\chi}  \one_{\{|p| \leq K_H''\}}$ has range in $H^{M-2}({\mathbb R}^3)$. 
Recall that $\widetilde{\chi}$ denotes either $\one_{(-\frac{1}{2}, \frac{1}{2})^3}$ or $\chi$ and that the bound \eqref{eq:sums} implies that in the present case $\delta_{k'}\leq (k-1) + \frac{|\alpha|}{2} \leq \frac{M}{4}
-1$. Therefore, using that $\chi$ satisfies the assumptions on the function $f$ above, we can apply \eqref{eq:PulledThrough}  on the first $k-1$ factors to get
\begin{align}
&\left\|  \partial^{\alpha} \left( \prod_{j=1}^k Q A_{\sigma(j)} Q \right)  \widetilde{\chi}  \one_{\{|p| \leq K_H''\}}\right\| \nonumber \\
&\quad \leq
C  (\frac{1}{2} \varepsilon_T d^{-4} )^{t_{\rm tot}}  \left\| (-\Delta + d^{-2})^{[\delta_{k'}]_{+}}  \chi \tau \chi (1-P)  \widetilde{\chi}  \one_{\{|p| \leq K_H''\}}
\right\|.
\end{align}
Using furthermore \eqref{cond:disjoint}, we therefore clearly get 
the desired estimate \eqref{eq:ToWork} in this case.

If $\sigma(k) = 2$, we will apply Lemma~\ref{lem:EllipticReg_bdry} below.
For this we start by noticing that \eqref{eq:ToWork} is very simple if $\sigma(j) = 2$ for all $j$---using in particular Lemma~\ref{lem:Eigenfunctions_Elliptic} and \eqref{cond:disjoint}.
Therefore, we may assume that there exists a largest $\tilde{k} < k$ such that $\sigma(\tilde{k}) = 4$.
Thus, we write
\begin{align}
&\partial^{\alpha} \left(\prod_{j=1}^k Q A_{\sigma(j)} Q \right)  \widetilde{\chi}  \one_{\{|p| \leq K_H''\}} \nonumber \\
&\quad =
(\frac{1}{2} \varepsilon_T d^{-4} )^{k-\tilde{k}} 
\partial^{\alpha} \left( \prod_{j=1}^{\tilde{k}-1} Q A_{\sigma(j)} Q \right) Q
\chi \tau \chi 
 (-\Delta^{\mathcal N} + d^{-2})^{-(k-\tilde{k})} Q
\widetilde{\chi}  \one_{\{|p| \leq K_H''\}}.
\end{align}
Here we split the last factor of $\chi$ as $\chi = f^2$, using the assumption \eqref{eq:Splitchi} on the structure of $\chi$.
Using Lemma~\ref{lem:EllipticReg_bdry} we see that
$f  (-\Delta^{\mathcal N} + d^{-2})^{-(k-\tilde{k})} Q
\widetilde{\chi}  \one_{\{|p| \leq K_H''\}}$ maps into $H^{\frac{M}{2}}(\Lambda)$.
Therefore, we can apply \eqref{eq:PulledThrough}  on the first $\tilde{k}-1$-factor to get
\begin{align}
&\left\| \partial^{\alpha} \left(\prod_{j=1}^k Q A_{\sigma(j)} Q \right)  \widetilde{\chi}  \one_{\{|p| \leq K_H''\}} \right\| \nonumber \\
&\quad\leq
C (\frac{1}{2} \varepsilon_T d^{-4} )^{t_{\rm tot}}  \left\| \left( (-\Delta + d^{-2})^{[\delta_{k'}]_{+}}
Q
\chi \tau f \right) f
 (-\Delta^{\mathcal N} + d^{-2})^{-(k-\tilde{k})} Q
\widetilde{\chi}  \one_{\{|p| \leq K_H''\}}\right\|,
\end{align}
where, by a slight abuse of notation, $t_{\rm tot}:= | \{j : \sigma(j) = 2\}|$ and where by \eqref{eq:sums} we have $\delta_{k'}+1 \leq  | \{j : \sigma(j) = 4\}| + \frac{|\alpha|}{2}$.
By  Lemma~\ref{lem:EllipticReg_bdry} and \eqref{cond:disjoint} we estimate the last term as
\begin{align}
 \left\| \left( (-\Delta + d^{-2})^{[\delta_{k'}]_{+}}
Q
\chi \tau f \right) f
 (-\Delta^{\mathcal N} + d^{-2})^{-(k-\tilde{k})} Q
\widetilde{\chi}  \one_{\{|p| \leq K_H''\}}\right\|
\leq C (K_H'')^{2 [\delta_{k'}]_{+} + 2}.
\end{align}
Therefore, using \eqref{cond:disjoint} again, we see that \eqref{eq:ToWork} is also true in this case.

This finishes the proof of \eqref{eq:ToWork} and therefore of \eqref{eq:InducRegul} and \eqref{eq:InducBound}, which was all that remained in order to finish the proof of Lemma~\ref{lem:pseudolocal}.
\end{proof}

\begin{lemma}\label{lem:Interpolation}
We have the interpolation inequalities, with $\Lambda=(-\frac{1}{2}, \frac{1}{2})^3$,
\begin{align}\label{eq:H3/2-interpolation}
\| f \|_{H^{3/2}(\Lambda)}  &\leq C \| f \|_{H^{1}(\Lambda)}^{\frac{1}{2}} \| f \|_{H^{2}(\Lambda)}^{\frac{1}{2}}, \\
\label{eq:H1/2-interpolation}
\| f \|_{H^{1/2}(\Lambda)} 
&\leq C \| f \|_{L^2(\Lambda)}^{\frac{1}{2}} \| f \|_{H^{1}(\Lambda)}^{\frac{1}{2}} ,\\
\intertext{and for all $i \in \{1,2,3\}$,}\label{eq:H1-interpolation}
\| \partial_i f \|_{L^2(\Lambda)} &\leq C \left( \| f \|_{L^2(\Lambda)}  + \|  f \|_{L^2(\Lambda)}^{\frac{1}{2}} \| \partial_i^2 f \|_{L^2(\Lambda)}^{\frac{1}{2}} \right).
\end{align}
\end{lemma}

\begin{proof}
These are all standard inequalities, but we state them for easy reference. The last inequality \eqref{eq:H1-interpolation} is given in \cite[Theorem~13.52]{MR3726909}.
 
The proofs of the two other inequalities are similar, so we only explicitly consider \eqref{eq:H3/2-interpolation}.
It follows from results on extensions and the similar inequality in ${\mathbb R}^3$, which is simple by the Fourier transform:
Let $E:L^2(\Lambda) \rightarrow L^2({\mathbb R}^3)$ be a total extension operator, the existence of which is guaranteed by \cite[Theorem 5.24]{adams}. In particular, $E$ restricts as a bounded map $H^{s}(\Lambda) \rightarrow H^s({\mathbb R}^3)$ for $s\in \{1 ,2\}$.
We then have, using the interpolation inequality in ${\mathbb R}^3$,
\begin{align*}
\| f \|_{H^{3/2}(\Lambda)} \leq \| Ef \|_{H^{3/2}({\mathbb R}^3)} \leq C  \| Ef \|_{H^{1}({\mathbb R}^3)}^{\frac{1}{2}} \| Ef \|_{H^{2}({\mathbb R}^3)}^{\frac{1}{2}}  \leq C' \| f \|_{H^{1}(\Lambda)}^{\frac{1}{2}} \| f \|_{H^{2}(\Lambda)}^{\frac{1}{2}} ,
\end{align*}
where the last step uses the boundedness of $E$ in $H^s$ for $s \in \{1 ,2\}$ and which finishes the proof of \eqref{eq:H3/2-interpolation}.
\end{proof}

\begin{lemma}\label{lem:Eigenfunctions_Elliptic}
Let $k \in {\mathbb N}$.
The operator $ (-\Delta^{\mathcal N} + d^{-2})^{-k}$ is bounded from $L^2((-\frac{1}{2}, \frac{1}{2})^3)$ to $H^{2k}((-\frac{1}{2}, \frac{1}{2})^3)$. 
More precisely, we have the bound
\begin{align}
\| \partial^{\beta}  (-\Delta^{\mathcal N} + d^{-2})^{-k}  \| \leq C d^{2k-|\beta|},
\end{align}
with the operator norm from $L^2((-\frac{1}{2}, \frac{1}{2})^3)$ to itself and
for all $|\beta| \leq 2k$ and where the constant is independent of $d\leq 1$.
\end{lemma}

\begin{proof}
This is an easy consequence of writing the operator in terms of its basis of explicit eigenfunctions and we leave the details to the reader.
\end{proof}

\begin{remark}
Lemma~\ref{lem:EllipticReg_bdry} below is a kind of elliptic regularity result.
Notice that for non-smooth domains, such as the box $(-\frac{1}{2}, \frac{1}{2})^3$, the Neumann resolvent $(-\Delta^{\mathcal N} + 1)^{-1}$ is in general not a bounded map from $H^1$ to $H^3$ or between other high order Sobolev spaces (see e.g. \cite{MR961439}). Therefore, one cannot immediately propagate smoothness through (powers of) the resolvent. 
However, this is a problem near the boundary only, which we can circumvent, since in our application we only need regularity after multiplying by a function $f$ that vanishes to high order on the boundary. Lemma~\ref{lem:EllipticReg_bdry} below states that in that case the vanishing on the boundary of $f$ cancels the possible explosion near the boundary of the resolvent.
\end{remark}

\begin{lemma}\label{lem:EllipticReg_bdry}
Let $m \in {\mathbb N}_0$ and $k \in {\mathbb N}$.
Suppose that $f \in C^m({\mathbb R}^3)$, that $\supp f \subseteq [-\frac{1}{2}, \frac{1}{2}]^3$, and let $\Lambda := (-\frac{1}{2}, \frac{1}{2})^3$.
 Suppose that $d\leq 1$ and $d^{-2} \leq K_H''$.

Then, the operator $f (-\Delta^{\mathcal N} + d^{-2})^{-k} \theta \one_{[0,K_H'']}(|p|) $ is bounded from $L^2({\mathbb R}^3)$ to $H^m(\Lambda)$ with
\begin{align}
\left\| f (-\Delta^{\mathcal N} + d^{-2})^{-k} \theta \one_{[0,K_H'']}(|p|)  \right\|_{{\mathcal B}(L^2({\mathbb R}^3),H^m(\Lambda))}
\leq C  \left(  \max_{\{|\alpha| \leq m\}} \| \partial^{\alpha} f \|_{\infty}\right) 
(K_H'')^{[m-2k]_{+}}.
\end{align}
\end{lemma}

\begin{proof}
We will actually prove the slightly stronger statement that for all $|\beta| \leq m + 2k$, the expression
\begin{align}
A_{f,\beta} := f \partial^{\beta}  (-\Delta^{\mathcal N} + d^{-2})^{-k} \theta \one_{[0,K_H'']}(|p|),
\end{align}
defines a bounded map from $L^2({\mathbb R}^3)$ to itself with
\begin{align}\label{eq:InductionBound}
\| A_{f,\beta}  \| \leq C \left(  \max_{\{|\alpha| \leq m\}} \| \partial^{\alpha} f \|_{\infty}\right) (K_H'')^{[|\beta|-2k]_{+}}.
\end{align}
From this we easily get the statement of the lemma using Leibniz' rule.

We will use the following elementary consideration in the proof. For a multi-index $\beta = (\beta_1,\beta_2,\beta_3) \in {\mathbb N}_0^3$, we define
\begin{align}\label{eq:NearBdry}
s_{\beta}(x) = s_{\beta}(x_1,x_2,x_3):= \prod_{j=1}^3 \cos(\pi x_j)^{\beta_j}.
\end{align}
If $f$ satisfies the assumptions of the lemma and $|\beta| \leq m-1$, then the quotient $s_{\beta}^{-1} f$ defines a function which vanishes outside $(-\frac{1}{2}, \frac{1}{2})^3$ and has $m-|\beta|$ continuous derivatives (using e.g. Taylor's formula).

We will prove \eqref{eq:InductionBound} by induction in $m$ for fixed $k$.
For $m=0$ there is nothing to prove by Lemma~\ref{lem:Eigenfunctions_Elliptic}.

Suppose the statement is proved for all $m \leq m_0$ for some $m_0 \geq 0$.
Let $f$ have $m_0+1$ continuous derivatives vanishing on the boundary of the box.
Let $|\beta| \leq  m_0+1 + 2k$ be a multiindex. We may assume that $|\beta| = m_0+1 + 2k$, since otherwise the desired result follows by the induction hypothesis. Decompose $\beta = \beta' + \beta''$, for some multiindices $\beta',\beta''$ with $|\beta''| = 2k$.
Let $\varphi \in L^2({\mathbb R}^3)$ be normalized. Clearly, $\varphi':= \theta \one_{[0,K_H'']}(|p|)  \varphi \in C^{\infty}((-\frac{1}{2}, \frac{1}{2})^3)$.
So by {\it interior} elliptic regularity $(-\Delta^{\mathcal N} + d^{-2})^{-k}  \varphi' \in C^{\infty}((-\frac{1}{2}, \frac{1}{2})^3)$.
Therefore, we only have to prove the suitable bound on the $L^2$-norm on $f \partial^{\beta} (-\Delta^{\mathcal N} + d^{-2})^{-k}  \varphi'$, not the existence of the derivatives.

Let $h \in C^{\infty}({\mathbb R})$ be non-decreasing and satisfy that $h(t) =0$ for all $t\leq \frac{1}{2}$ and $h(t) = 1$ for all $t\geq 1$.
For $T>0$, define 
\begin{align}
h_T(x) = \prod_{j=1}^3 h\left((x_j+\frac{1}{2})T\right)  h\left((\frac{1}{2}-x_j)T\right).
\end{align}
Then, $h_T \in C^{\infty}_0((-\frac{1}{2}, \frac{1}{2})^3)$ and $h_T\equiv 1$ in the box except on a $T^{-1}$-neighborhood of the boundary.
By monotone convergence it suffices to prove the $L^2$-bound on $h_T A_{f,\beta} \varphi$, uniformly in $T\geq 2$.

We will start by proving, using the induction hypothesis, that for all $|\tilde{\beta}|\leq 2k$, and with $C$ independent of $T$,
\begin{align}\label{eq:Leibnitz1}
\left\| \left[ \partial^{\tilde{\beta}}, h_T f\right] \partial^{\beta'}  (-\Delta^{\mathcal N} + d^{-2})^{-k}  \varphi' \right\|
\leq C \left(  \max_{\{|\alpha| \leq m\}} \| \partial^{\alpha} f \|_{\infty}\right) (K_H'')^{[|\beta'|+|\tilde{\beta}|-2k-1]_{+}}.
\end{align}
To prove \eqref{eq:Leibnitz1}
notice first that, by Leibniz' formula, for some constants $c_{\gamma,\delta,\eta}$ and introducing the functions $s_{\delta}$ as defined in \eqref{eq:NearBdry},
\begin{align}
h_T f \partial^{\tilde{\beta}} - \partial^{\tilde{\beta}}  h_T f = \sum_{\gamma+\delta+\eta = \tilde{\beta}, \gamma < \tilde{\beta}} c_{\gamma,\delta,\eta} (s_{\delta} \partial^{\delta} h_T ) (s_{\delta}^{-1} \partial^{\eta} f) \partial^{\gamma}.
\end{align}
Here the function $(s_{\delta} \partial^{\delta} h_T )$ is uniformly bounded in $T$ and the function $(s_{\delta}^{-1} \partial^{\eta} f)$ satisfies the assumptions of the lemma with $m=m_0+1-|\eta|-|\delta|$.
Notice that $|\gamma+\beta'| = |\beta'| +|\tilde{\beta}|-|\delta+\eta| \leq  m_0 +1 + 2k -|\delta| -|\gamma|$.
Due to the sharp inequality in the summation over multiindices, this implies that we can apply the induction hypothesis termwise when the sum acts on $\partial^{\beta'}  (-\Delta^{\mathcal N} + d^{-2})^{-k}  \varphi'$, resulting in \eqref{eq:Leibnitz1}.

Using \eqref{eq:Leibnitz1}, we may write
\begin{align}
h_T A_{f,\beta} \varphi = \partial^{\beta''} h_T f \partial^{\beta'}  (-\Delta^{\mathcal N} + d^{-2})^{-k}  \varphi'
+ \psi'_T,
\end{align}
where $\psi'_T$ satisfies 
\begin{align}\label{eq:FirstLeibnitz}
\| \psi_T' \| \leq C \left(  \max_{\{|\alpha| \leq m\}} \| \partial^{\alpha} f \|_{\infty}\right) (K_H'')^{[|\beta|-2k-1]_{+}}.
\end{align}
So it suffices to prove that (uniformly in $T$),
\begin{align}
\| \partial^{\beta''} h_T f \partial^{\beta'}  (-\Delta^{\mathcal N} + d^{-2})^{-k}  \varphi' \| \leq C \left(  \max_{\{|\alpha| \leq m\}} \| \partial^{\alpha} f \|_{\infty}\right) (K_H'')^{[|\beta|-2k]_{+}}.
\end{align}
For this we use the commutator formula, 
\begin{align}\label{eq:ResolventIdentity}
&\partial^{\beta''} h_T f \partial^{\beta'}   (-\Delta^{\mathcal N} + d^{-2})^{-k}  \varphi' \nonumber \\
&\quad=
\left( \partial^{\beta''}   (-\Delta^{\mathcal N} + d^{-2})^{-k}\right)   h_T f \partial^{\beta'} \varphi' \nonumber \\
&\qquad + \partial^{\beta''}  (-\Delta^{\mathcal N} + d^{-2})^{-k}  [ (-\Delta^{\mathcal N} + d^{-2})^k , h_T f \partial^{\beta'}]  (-\Delta^{\mathcal N} + d^{-2})^{-k}  \varphi' .
\end{align}
Here we can easily estimate the first term on the right as follows, using the presence of $h_T$ to remove the localization $\theta$ in $\varphi'$,
\begin{align}
\|  \partial^{\beta''}   (-\Delta^{\mathcal N} + d^{-2})^{-k}   h_T f \partial^{\beta'} \varphi' \|
&\leq \| \partial^{\beta''}   (-\Delta^{\mathcal N} + d^{-2})^{-k} \| \| h_T f \partial^{\beta'}   \one_{[0,K_H'']}(|p|)  \varphi  \| \nonumber \\
&\leq  C (K_H'')^{|\beta'|},
\end{align}
with $C$ independent of $T$, where we used Lemma~\ref{lem:Eigenfunctions_Elliptic} in the last step.

To estimate the commutator term in \eqref{eq:ResolventIdentity} we notice that due to the presence of $h_T$, the boundary conditions are automatically satisfied and
$$
[ (-\Delta^{\mathcal N} + d^{-2})^k , h_T f \partial^{\beta'}] = [ (-\Delta + d^{-2})^k , h_T f \partial^{\beta'} ]
=  [ (-\Delta + d^{-2})^k , h_T f ]\partial^{\beta'} 
$$ 
is just a differential operator of order $2k-1+|\beta'| = m_0$.
Now we can use \eqref{eq:Leibnitz1} and the assumption $d^{-2} \leq K_H''$ to finish the proof.
\end{proof}

\subsection{Proof of Lemma~\ref{lem:LowMomentaOperator}}

Recall that $R$ is the radius of the support of the potential $v$ and that $\varepsilon_N$ is the parameter defined in  \eqref{eq:Cond_epsilonN} involved in the definition of the kinetic energy 
${\mathcal T}$ from  \eqref{eq:DefT_new}.

\begin{proposition}\label{prop:T-T2}
Suppose that Assumption~\ref{assump:params} is satisfied and $\rmu a^3$ is sufficiently small.
Let ${\mathcal T}$ be as defined in \eqref{eq:DefT_new}.
Then,
\begin{align}\label{eq:LowerT}
{\mathcal T} \geq \varepsilon_N (-\Delta^{\mathcal N}) + \frac{1}{4} Q \chi_{\Lambda} (-\Delta) \chi_{\Lambda} Q - C \varepsilon_T (ds\ell)^{-2},
\end{align}
and
\begin{align}\label{eq:LowerT2}
{\mathcal T}^2& \geq \frac{1}{2} \varepsilon_N^2 (-\Delta^{\mathcal N})^2 + \frac{1}{4}\varepsilon_N Q \chi_{\Lambda} |p|^4 \chi_{\Lambda} Q + \frac{1}{2}(1-\varepsilon_N)^2 ( Q \chi_{\Lambda} \tau(p) \chi_{\Lambda} Q)^2
 - C  \varepsilon_T^2 (ds\ell)^{-4}.
\end{align}
\end{proposition}

\begin{proof}
Notice first that ${\mathcal T}$ is self-adjoint with domain given by ${\mathcal D}(-\Delta^{\mathcal N})$ and quadratic form domain $H^1(\Lambda)$.

We write the kinetic energy as
\begin{align}\label{eq:LowerTprim}
      {\mathcal T} :=
       \varepsilon_N (-\Delta^{\mathcal N})+ (1-\varepsilon_N) Q \chi_{\Lambda} \tau(p) \chi_{\Lambda} Q + A'',
\end{align}
where
\begin{align}  
       A'':= (1-\varepsilon_N) 
       \Big\{
      \frac12 \varepsilon_T (d \ell)^{-2} \frac{-\Delta^{\mathcal N}}{-\Delta^{\mathcal N}+(d\ell)^{-2}} 
      + b \ell^{-2} Q 
      + b\varepsilon_T (d\ell)^{-2} Q\one_{(d^{-2}\ell^{-1},\infty)}(\sqrt{-\Delta})Q\Big\},
\end{align}
is bounded and $\tau(p)$ from \eqref{eq:Def_tau}.

Notice that
\begin{align}\label{eq:tau_p2}
0\leq p^2 - \tau(p) \leq |p| (s\ell)^{-1} \left( 1-\varepsilon_t+\varepsilon_T d^{-1}\right) + \frac{1}{4 (s\ell)^2} \left(1 -\varepsilon_T + \varepsilon_T d^{-2}\right).
\end{align}

Clearly $A''$ is positive, so using \eqref{con:eTdK} and applying a simple bound on $\tau(p)$, we find \eqref{eq:LowerT}.

We next calculate ${\mathcal T} ^2$. Using a Cauchy-Schwarz inequality, and that $\chi_{\Lambda}$ vanishes on $\partial \Lambda$, we see that
\begin{align*}
{\mathcal T}^2& \geq \frac{1}{2} \varepsilon_N^2 (-\Delta^{\mathcal N})^2 + \frac{1}{2}(1-\varepsilon_N)^2 ( Q \chi_{\Lambda} \tau(p) \chi_{\Lambda} Q)^2 \\
&\quad+ \varepsilon_N (1-\varepsilon_N) Q\left(2 \chi_{\Lambda} p^2 \tau(p) \chi_{\Lambda} + [-\Delta^{\mathcal N},\chi_{\Lambda}] \tau(p) \chi_{\Lambda} - 
\chi_{\Lambda} \tau(p)  [-\Delta^{\mathcal N},\chi_{\Lambda}] 
\right)Q
-C( A'')^2.
\end{align*}
We continue the calculation on the commutators,
\begin{align}
 [-\Delta^{\mathcal N},\chi_{\Lambda}] \tau(p) \chi_{\Lambda} - 
\chi_{\Lambda} \tau(p)  [-\Delta^{\mathcal N},\chi_{\Lambda}] 
&= {\mathcal C}_1 + {\mathcal C}_2,
\end{align}
with
\begin{align}
{\mathcal C}_1&:= \left((\Delta \chi_{\Lambda}) \tau(p) \chi_{\Lambda} + \chi_{\Lambda} \tau(p) (\Delta \chi_{\Lambda}) \right), \nonumber \\
 {\mathcal C}_2&:= 2(-i) \sum_{j=1}^3 \left( (\partial_{x_j} \chi_{\Lambda}) p_j \tau(p) \chi_{\Lambda} - \chi_{\Lambda} p_j \tau(p) (\partial_{x_j} \chi_{\Lambda}) \right).
\end{align}
We will prove that 
\begin{align}\label{eq:BoundCommutators}
\pm \left( [-\Delta^{\mathcal N},\chi_{\Lambda}] \tau(p) \chi_{\Lambda} - 
\chi_{\Lambda} \tau(p)  [-\Delta^{\mathcal N},\chi_{\Lambda}]  \right) \leq \chi_{\Lambda} p^4 \chi_{\Lambda} + C (s\ell)^{-4} (\varepsilon_T d^{-2})^2.
\end{align}
Recall for later use that $s \ll 1$ and $\varepsilon_T d^{-2} \gg 1$ by \eqref{con:eTdK}.

By Cauchy-Schwarz and $0 \leq \tau(p) \leq p^2$, we find
\begin{align}
{\mathcal C}_1 \leq \frac{1}{2} \chi_{\Lambda} p^2 \tau(p) \chi_{\Lambda} + C \ell^{-4},
\end{align}
so \eqref{eq:BoundCommutators} is valid for ${\mathcal C}_1$.

We rewrite $ {\mathcal C}_2$ as
\begin{equation}
 {\mathcal C}_2 =  {\mathcal C}_{2,1} +  {\mathcal C}_{2,2},
\end{equation}
with
\begin{align}
{\mathcal C}_{2,1} &:= 2(-i) \sum_{j=1}^3 \left( (\partial_{x_j} \chi_{\Lambda}) p_j (-\Delta) \chi_{\Lambda} - \chi_{\Lambda} p_j (-\Delta) (\partial_{x_j} \chi_{\Lambda}) \right), \nonumber \\
{\mathcal C}_{2,2}  &:= 2(-i) \sum_{j=1}^3 \left( (\partial_{x_j} \chi_{\Lambda}) p_j (\tau(p)-p^2) \chi_{\Lambda} - \chi_{\Lambda} p_j (\tau(p)-p^2) (\partial_{x_j} \chi_{\Lambda}) \right).
\end{align}
On ${\mathcal C}_{2,1}$ we commute again, to get
\begin{align}
{\mathcal C}_{2,1} &= 2(-i) \sum_{j=1}^3\left(  [ \partial_{x_j} \chi_{\Lambda}, p_j (-\Delta)] \chi_{\Lambda} + [ p_j (-\Delta), \chi_{\Lambda}] \partial_{x_j} \chi_{\Lambda} \right) \nonumber \\
&= 2 \sum_{j,k} \left(  [ (\partial_{x_j} \chi_{\Lambda}), \partial_j \partial_k^2] \chi_{\Lambda} + [ \partial_j \partial_k^2, \chi_{\Lambda}] (\partial_{x_j} \chi_{\Lambda}) \right).
\end{align}
At this point we use that $\chi$ can be written as $\chi = f^2$, where $f$ has $M/2$ bounded derivatives. Therefore, with $f_{\Lambda}(x):= f(x/\ell)$, we find $(\partial_j \chi_{\Lambda}) = 2 f_{\Lambda} (\partial_j f_{\Lambda})$, so we can organize the terms as
\begin{equation}
{\mathcal C}_{2,1} = \sum_{j,k} (f_{\Lambda}^2 p_j p_k R^{(0)}_{j,k} + h.c) + ( f_{\Lambda} p_j R^{(1)}_j + h.c.) + R^{(2)},
\end{equation}
where $|R^{(i)}| \leq C \ell^{-i-2}$.
Therefore, using a standard IMS-like commutator formula in the second step,
\begin{align}
\pm {\mathcal C}_{2,1} &\leq \sum_{j,k} \frac{1}{100} f_{\Lambda}^2 p_j^2 p_k^2 f_{\Lambda}^2 + \ell^{-2} \sum_{j=1}^3 f_{\Lambda}p_j^2 f _{\Lambda}+ C \ell^{-4} \nonumber \\
&= \frac{1}{100} \chi_{\Lambda} p^4 \chi_{\Lambda} + \ell^{-2} \sum_{j=1}^3 (f_{\Lambda}^2 p_j^2 + p_j^2 f_{\Lambda}^2 - (\partial_j f_{\Lambda})^2) + C \ell^{-4} \nonumber \\
&\leq \frac{1}{50} \chi_{\Lambda} p^4 \chi_{\Lambda} + C' \ell^{-4}.
\end{align}
So ${\mathcal C}_{2,1} $ is also compatible with \eqref{eq:BoundCommutators}.

To finally estimate ${\mathcal C}_{2,2} $ we use \eqref{eq:tau_p2}.
Therefore, by repeated Cauchy-Schwarz, and that $\varepsilon_T d^{-2} \gg 1$ by \eqref{con:eTdK}, we see that
\begin{align}
\pm {\mathcal C}_{2,2} &\leq \frac{1}{100} \chi_{\Lambda} p^4 \chi_{\Lambda} + C (s\ell)^{-2} (1+ \varepsilon_T d^{-1})^2 \ell^{-2}
+ \ell^{-2} \frac{1}{100} \chi_{\Lambda} p^2 \chi_{\Lambda} + C (s\ell)^{-4}  (1+ \varepsilon_T d^{-2})^2 \nonumber \\
&\leq \frac{1}{50} \chi_{\Lambda} p^4 \chi_{\Lambda} + C (s\ell)^{-4} (\varepsilon_T d^{-2})^2,
\end{align}
in agreement with \eqref{eq:BoundCommutators}.

Therefore, inserting \eqref{eq:BoundCommutators} and using again \eqref{eq:tau_p2}, we find \eqref{eq:LowerT2}.
This finishes the proof of Proposition~\ref{prop:T-T2}.
\end{proof}

From Proposition~\ref{prop:T-T2} and the definition of $Q_L'$ we immediately get the following estimates, since $Q\varphi = Q_L' \varphi = \varphi$ for $\varphi \in  \Ran Q_L'$. 

\begin{corollary}\label{cor:Derivatives}
Suppose that Assumption~\ref{assump:params} is satisfied and $\rmu a^3$ is sufficiently small.
If $\varphi \in \Ran Q_L'$ and is normalized in $L^2$, then
\begin{equation}\label{eq:DerivsPhiTerms}
\| \nabla \varphi \|_2 \leq C \varepsilon_N^{-1/2} K_H' \ell^{-1},\qquad
\| (-\Delta^{\mathcal N}) \varphi \|_2 \leq C \varepsilon_N^{-1} (K_H' \ell^{-1})^2,
\end{equation}
and
\begin{equation}
\label{eq:DerivsChiTerms}
\| \nabla (\chi_{\Lambda} \varphi) \|_2 \leq C  K_H' \ell^{-1},\qquad
\| \Delta(\chi_{\Lambda} \varphi )\|_2 \leq C \varepsilon_N^{-1/2} (K_H' \ell^{-1})^2.
\end{equation}
\end{corollary}

We also need the following Sobolev-type inequality\footnote{The authors are grateful to Rupert Frank for communicating \eqref{eq:Sob-fullspace} to them.}.

\begin{lemma}\label{lem:SobolevHom}
There exists a uniform constant $C>0$ such that the following is true.
For $L>0$, let $\Delta^{\mathcal N}_L$ denote the Neumann Laplace operator on $[-\frac{L}{2}, \frac{L}{2}]^3$.
Then, for all $f \in {\mathcal D}(-\Delta^{\mathcal N}_L)$ with $\int_{ [-\frac{L}{2}, \frac{L}{2}]^3} f(x) \,dx = 0$, we have
\begin{align}\label{eq:Sob-box}
\| f \|_{\infty} \leq C \| \nabla f\|_{L^2([-\frac{L}{2}, \frac{L}{2}]^3)}^{\frac{1}{2}} \| -\Delta^{\mathcal N}_L f\|_{L^2([-\frac{L}{2}, \frac{L}{2}]^3)}^{\frac{1}{2}} .
\end{align}
Also, for all $f \in H^2({\mathbb R}^3)$,
\begin{equation}\label{eq:Sob-fullspace}
\| f \|_{\infty} \leq C \| \nabla f\|_{L^2({\mathbb R}^3)}^{\frac{1}{2}} \| -\Delta f\|_{L^2({\mathbb R}^3)}^{\frac{1}{2}} .
\end{equation}
\end{lemma}

\begin{proof}
We will only prove \eqref{eq:Sob-box}, the case of ${\mathbb R}^3$ in \eqref{eq:Sob-fullspace} follows easily using the same ideas and the Fourier transform.
We use the well-known $L^2$-normalized eigenfunctions $u_n$ of $-\Delta^{\mathcal N}_L$, for $n \in {\mathbb N}_0^3$ with eigenvalues $\lambda_n = \frac{\pi^2 n^2}{L^2}$. Notice the uniform bounds $|u_n(x)| \leq C_0 L^{-3/2}$ for all $x$ and $L$.

Let  $f \in {\mathcal D}(-\Delta^{\mathcal N}_L)$ with $\int_{ [-\frac{L}{2}, \frac{L}{2}]^3} f(x) \,dx = 0$ be given.
By scaling invariance of \eqref{eq:Sob-box}, we may assume that $\| \nabla f\|_{L^2([-\frac{L}{2}, \frac{L}{2}]^3)} =\| -\Delta^{\mathcal N}_L f\|_{L^2([-\frac{L}{2}, \frac{L}{2}]^3)}$, so it suffices to show that
\begin{align}
\| f \|_{\infty}^2 \leq C \Big( \| \nabla f\|_{L^2([-\frac{L}{2}, \frac{L}{2}]^3)}^2 +\| -\Delta^{\mathcal N}_L f\|_{L^2([-\frac{L}{2}, \frac{L}{2}]^3)}^2\Big),
\end{align}
with $C$ independent of $f$ and $L$.

But using the expansion of $f$ in eigenfunctions, with $c_n = \langle f, u_n\rangle$, and noticing that $c_0=0$, we get by Cauchy-Schwarz,
\begin{align}
|f(x)|^2 &\leq C_0 L^{-3}\Big( \sum_{n\neq 0} (\lambda_n + \lambda_n^2)^{-\frac{1}{2}} (\lambda_n + \lambda_n^2)^{\frac{1}{2}} |c_n|\Big)^2 \nonumber \\
&\leq 
C_0 L^{-3} \Big( \sum_{n\neq 0} (\lambda_n + \lambda_n^2)^{-1} \Big)
\Big(  \sum_{n\neq 0} (\lambda_n + \lambda_n^2) |c_n|^2 \Big).
\end{align}
Therefore, it suffices to prove that $L^{-3}\sum_{n\neq 0} (\lambda_n + \lambda_n^2)^{-1} $ is bounded uniformly in $L$. But clearly,
\begin{equation}
L^{-3}\sum_{n\neq 0} (\lambda_n + \lambda_n^2)^{-1}  \leq L^{-3}\sum_{n\neq 0} \lambda_n^{-2} = L \sum_{n\neq 0}\pi^{-4} n^{-4},
\end{equation}
so it suffices to consider the case of large values of $L$.
When $L \rightarrow \infty$ the sum converges to the convergent integral $\int_{{\mathbb R}^3} (\pi^2 k^2 + \pi^4 k^4)^{-1}\,dk$. This finishes the proof.
\end{proof}

\end{document}